\documentclass[11pt,letterpaper]{article}

\usepackage[margin=1in]{geometry}
\usepackage{setspace}
\usepackage{amsmath}
\usepackage{amsfonts}
\usepackage{amssymb}
\usepackage{amsthm}
\usepackage{bm}
\usepackage{bbm}
\usepackage{mathtools}

\usepackage{graphicx}
\usepackage{booktabs}
\usepackage{tabularx}
\usepackage{threeparttable}
\usepackage{multirow}
\usepackage{float}
\usepackage{caption}
\usepackage{subcaption}
\usepackage{enumitem}
\usepackage{colortbl}
\usepackage{lmodern}
\usepackage{microtype}
\usepackage[T1]{fontenc}
\usepackage[utf8]{inputenc}

\usepackage{algorithm}
\usepackage{algpseudocode}

\usepackage[square]{natbib}

\usepackage{xcolor}
\usepackage{hyperref}
\hypersetup{
    colorlinks=true,
    linkcolor=blue!60!black,
    citecolor=blue!60!black,
    urlcolor=blue!60!black
}

\usepackage{tikz}
\usetikzlibrary{shapes,arrows,positioning,fit,backgrounds,calc,decorations.markings,patterns}
\usepackage{pgfplots}
\pgfplotsset{compat=1.18}

\usepackage{adjustbox}

\usepackage[section]{placeins}
\usepackage{mathrsfs}   % \mathscr for script letters (e.g., \mathscr{I}, \mathscr{H})
\usepackage{braket}     % Dirac notation: \ket{}, \bra{}, \braket{}
\usepackage{dsfont}     % Double-stroke fonts: \mathds{1} for identity operator

\usepackage{etoolbox}

\usepackage{listings}
\usepackage{inconsolata}

\definecolor{codebg}{RGB}{248,248,248}
\definecolor{kw}{RGB}{0,0,150}
\definecolor{str}{RGB}{120,0,0}
\definecolor{com}{RGB}{0,110,0}
\lstdefinestyle{python}{
  language=Python,
  basicstyle=\linespread{1.0}\ttfamily\small,
  keywordstyle=\color{kw}\bfseries,
  stringstyle=\color{str},
  commentstyle=\color{com}\itshape,
  showstringspaces=false,
  frame=single,
  backgroundcolor=\color{codebg},
  rulecolor=\color{black!30},
  tabsize=2,
  numbers=left,
  numberstyle=\tiny\color{black!50},
  breaklines=true,
  columns=fullflexible
}

\theoremstyle{plain}
\newtheorem{theorem}{Theorem}[section]
\newtheorem{proposition}[theorem]{Proposition}
\newtheorem{lemma}[theorem]{Lemma}
\newtheorem{corollary}[theorem]{Corollary}

\theoremstyle{definition}
\newtheorem{definition}[theorem]{Definition}

\newtheorem{example}[theorem]{Example}

\theoremstyle{remark}
\newtheorem{remark}[theorem]{Remark}

\newcommand{\Ocal}{\mathcal{O}}

\newcommand{\E}{\mathbb{E}}

\newcommand{\R}{\mathbb{R}}
\newcommand{\N}{\mathbb{N}}

\newcommand{\bSigma}{\bm{\Sigma}}
\newcommand{\bmu}{\bm{\mu}}

\DeclareMathOperator{\tr}{tr}
\DeclareMathOperator{\diag}{diag}

\begin{document}

\title{\Large \textbf{Beyond De~Prado and Cotton} \\ \large Hierarchical and Iterative Methods for General Mean--Variance Portfolios}

\author{
    Bernd J.\ Wuebben\\
    AllianceBernstein, New York\\ \texttt{bernd.wuebben@alliancebernstein.com}
}

\date{April 11, 2026}

\maketitle

\begin{abstract}
\noindent Hierarchical Risk Parity \citep{lopezdeprado2016} and the Schur-complement generalisation of \citet{cotton2024} are among the most widely adopted regularised portfolio construction methods, yet both are signal-blind: they solve only the minimum-variance problem and cannot accommodate an arbitrary expected-return forecast. This paper introduces three methods that incorporate alpha signals into hierarchical and regularised portfolio construction.

\smallskip
\noindent \textbf{HRP-$\mu$} is a hierarchical allocator that accepts an arbitrary signal $\mu$ and nests standard HRP when $\gamma = 0$ and $\mu=\mathbf{1}$. It preserves the tree-based structure of HRP while extending it beyond the minimum-variance setting. \textbf{HRP-$\Sigma\mu$} strengthens this construction by replacing inverse-variance representatives with recursive local mean--variance optima, thereby using richer within-cluster covariance information at the same $O(N^2)$ asymptotic cost. 

\smallskip
\noindent \textbf{CRISP} (Correlation-Regularised Iterative Shrinkage Portfolios) is an iterative solver for $P_\gamma w = \mu$ with $P_\gamma = (1-\gamma)\operatorname{diag}(\Sigma) + \gamma \Sigma$, so that $\gamma$ interpolates between a diagonal portfolio rule and full Markowitz. At convergence, CRISP is Markowitz applied to a variance-preserving shrunk covariance---diagonal variances unchanged, off-diagonal correlations shrunk---with $\gamma$ tuned for out-of-sample Sharpe rather than covariance-estimation loss.

\smallskip
\noindent In Monte Carlo experiments across multiple covariance regimes and estimation ratios, HRP-$\mu$ and HRP-$\Sigma\mu$ both outperform plain HRP---not only on general expected-return forecasts, where HRP is signal-blind, but also on the minimum-variance problem $\mu = \mathbf{1}$ that is HRP's home territory---with HRP-$\Sigma\mu$ consistently improving on HRP-$\mu$. CRISP at intermediate $\gamma$ is the dominant method in both regimes, outperforming HRP, Cotton, Ledoit--Wolf shrinkage, direct Markowitz, and the signal-aware hierarchical methods at every sample size tested. 



\bigskip
\noindent \textbf{Code:} All code, data, and reproduction scripts are available at \url{https://github.com/bwuebben/beyond_hrp_and_cotton}.
\end{abstract}

\newpage
\tableofcontents
\newpage

% ============================================================
% section_01_introduction.tex
% ============================================================
\section{Introduction}\label{sec01:introduction}

\subsection{Motivation}\label{sec01:motivation}

Hierarchical Risk Parity \citep{lopezdeprado2016} and the Schur-complement allocator of \citet{cotton2024} have become standard regularised portfolio constructions. Both are \emph{signal-blind}. They solve only the minimum-variance problem $\bmu = \mathbf{1}$ and have no mechanism for an expected-return forecast. Cotton's allocator connects HRP to exact minimum variance through a single continuous parameter $\gamma \in [0,1]$ and is in that sense a strict generalisation of \citeauthor{lopezdeprado2016}, but it inherits the $\bmu = \mathbf{1}$ restriction and costs $O(N^3/6)$ because a sibling covariance block must be inverted at every internal tree node.

This restriction collides with three features of real portfolio construction. The first is \emph{signal integration}. Quantitative equity, factor investing, tactical asset allocation, and Black--Litterman view blending \citep{blacklitterman1992,grinoldkahn1999,barrososantaclara2015} all deliver a heterogeneous $\bmu$, and that signal is the object that distinguishes an active portfolio from a risk-only allocation. Setting $\bmu = \mathbf{1}$ is active discarding, not a neutral default. The second is \emph{computational tractability}. Universes with $N$ between $10^3$ and $10^4$, re-optimised on rolling windows, make $O(N^3)$ operations costly; $O(N^2)$ is the budget a realistic backtest can absorb. The third is \emph{estimation error}. Direct Markowitz inversion \citep{markowitz1952} amplifies noise on low-variance eigenvectors and produces the extreme weights documented by \citet{michaud1989}, and the standard remedy is structural shrinkage of the covariance estimator \citep{stein1956,ledoitwolf2003,ledoitwolf2004,ledoitwolf2017,ledoitwolf2020,chen2010}.

This paper asks a single question. Can the hierarchical-and-shrinkage structure that makes HRP and the Cotton allocator useful be extended to the general mean--variance problem with arbitrary $\bmu$ without giving up its regularisation benefits, and what is the natural iterative counterpart? The answer is three methods, each controlled by a shrinkage parameter $\gamma \in [0,1]$ that plays the same role \citet{cotton2024} introduced---governing how much cross-asset covariance enters the allocation---and each accepting an arbitrary signal.

\subsection{Three methodological contributions}\label{sec01:contributions}

We present the three methods in order of empirical strength---weakest to strongest---and use the display names HRP-$\mu$, HRP-$\Sigma\mu$, and CRISP throughout.

\paragraph{Contribution 1: HRP-$\mu$, a transparent signal-aware hierarchical allocator.}
HRP-$\mu$ walks the same correlation dendrogram as classical HRP. At each internal node its representative portfolio is the \emph{signed inverse-variance portfolio} $\hat w_{L,i}^{\text{signed}} \propto \operatorname{sign}(\mu_i)/\sigma_{ii}^2$, and between-child budgets are set by the same $2\times 2$ mean--variance system that underlies \citeauthor{cotton2024}'s allocator. Because the representative depends only on the diagonal variances $\{\sigma_{ii}^2\}$ and the signal signs $\{\operatorname{sign}(\mu_i)\}$ inside a block---never on within-cluster correlations---every leaf weight factors as a sign times a root-to-leaf product of non-negative budgets, so the hierarchical audit trail that practitioners prize in classical HRP is preserved. Proposition~\ref{prop:a3_recovery} shows that HRP-$\mu$ nests De Prado's HRP exactly at $\gamma = 0$, $\bmu = \mathbf{1}$, and the cost is $O(N^2)$, a factor $N$ cheaper than Cotton. Even on Cotton's own $\bmu = \mathbf{1}$ problem, HRP-$\mu$ at $\gamma > 0$ improves on classical HRP out of sample and matches Cotton's Sharpe at roughly half Cotton's $\gamma$, while avoiding the Schur-complement instability Cotton exhibits at small $T$ (Proposition~\ref{prop:cotton_instability}).

\paragraph{Contribution 2: HRP-$\Sigma\mu$, the strongest tree-based method.}
HRP-$\Sigma\mu$ keeps the HRP tree and the $2\times 2$ between-child split but replaces HRP-$\mu$'s signed inverse-variance representative with the \emph{local mean--variance optimum} computed recursively on the sub-tree, using the full within-cluster covariance $\bSigma_{LL}$ at every node. The representative is richer than HRP-$\mu$'s---it captures within-cluster hedging that the signed IVP cannot see---at the cost of the closed-form auditability. The load-bearing algorithmic detail is the between-child normalisation: sum-to-one normalisation of the same $2\times 2$ system produces a sign-flip pathology (documented as a cautionary negative result in Appendix~\ref{app:a1_pathology}), whereas the $L^1$ normalisation $\alpha_k \leftarrow \alpha_k^{\text{raw}}/(|\alpha_L^{\text{raw}}| + |\alpha_R^{\text{raw}}|)$ is ray-invariant and sign-preserving (Lemmas~\ref{lem:hrpsm_scale} and~\ref{lem:hrpsm_sign}). The cost is $O(N^2)$, matching HRP-$\mu$; classical HRP is recovered to cosine $\approx 0.992$ at $\gamma = 0$, $\bmu = \mathbf{1}$, and exactly at tree depth $\le 2$ (Proposition~\ref{prop:hrpsm_recovery}). In the synthetic Monte Carlo panels of Section~\ref{sec10:experiments} HRP-$\Sigma\mu$ delivers $20$--$35\%$ higher out-of-sample Sharpe than HRP-$\mu$ on random signals and up to $180\%$ higher on structural signals with sample estimation, reaching roughly $90\%$ of CRISP's performance while staying strictly $O(N^2)$.

\paragraph{Contribution 3: CRISP, an iterative shrinkage solver (headline).}
CRISP runs scalar Gauss--Seidel on the shrunk mean--variance system $P_\gamma w = \bmu$, where $P_\gamma = (1-\gamma)D + \gamma\bSigma$ and $D = \diag(\bSigma)$. The shrinkage parameter interpolates between a pure diagonal solve at $\gamma = 0$ (signed inverse-variance weighting) and exact Markowitz at $\gamma = 1$. Gauss--Seidel on $P_\gamma$ converges unconditionally for every symmetric positive definite $\bSigma$ and every $\gamma \in [0,1]$ (Theorem~\ref{thm:gs_convergence}). The convergence-rate result is new: at $\gamma = 1$ the sweep count required to reach relative direction error $\varepsilon$ is $p = O(\kappa(C)\log(1/\varepsilon))$, controlled by the condition number of the \emph{correlation} matrix rather than of the covariance matrix itself, and in general $p = O(\kappa(D^{-1}P_\gamma)\log(1/\varepsilon))$ with $\kappa(D^{-1}P_\gamma) = [(1-\gamma) + \gamma\lambda_1(C)]/[(1-\gamma) + \gamma\lambda_N(C)]$ interpolating monotonically from $1$ at $\gamma = 0$ to $\kappa(C)$ at $\gamma = 1$; volatility dispersion does \emph{not} enter the rate (Theorem~\ref{thm:gs_rate}). Correlation conditioning is bounded in realistic equity universes, so this is what makes CRISP tractable at scale.

At convergence CRISP produces $P_\gamma^{-1}\bmu$: Markowitz on the shrunk covariance $P_\gamma$. This places CRISP in the linear-shrinkage family of \citet{ledoitwolf2003,ledoitwolf2004,ledoitwolf2017}, with two structural differences. First, its target is the diagonal $D$, so $P_\gamma$ reproduces the variances of $\bSigma$ exactly for every $\gamma$ and only correlations are shrunk---the reliable part of the covariance (per-asset variances, each estimated from a single time series) is left alone and the unreliable part (pairwise correlations) is regularised. Classical Ledoit--Wolf targets (scaled identity, single factor, constant correlation) perturb the diagonal. Second, the shrinkage intensity $\gamma$ is chosen to maximise out-of-sample portfolio Sharpe rather than to minimise Frobenius loss on the estimator; Section~\ref{sec06:perturbation} shows these are different problems with different optima. In the synthetic Monte Carlo panels of Section~\ref{sec10:experiments} CRISP at $\gamma \in [0.3, 0.7]$ with $100$ sweeps delivers $80$--$94\%$ of oracle Sharpe uniformly across four covariance regimes and $T/N \in [0.6, 5]$, and dominates every benchmark on every panel.

\subsection{Three supporting contributions}\label{sec01:supporting}

\paragraph{Adaptive $\gamma^\star$ rule.} A bias--variance decomposition of expected OOS Sharpe yields the closed-form rule
\[
\gamma^\star \;\approx\; \frac{1}{1 + c\,\kappa(C)^2\,N/(T\,\mathrm{IC}^2)},
\]
whose comparative statics (Proposition~\ref{prop:gamma_star_comparative}) are qualitatively correct: $\gamma^\star$ rises with the sample length $T$ and the information coefficient $\mathrm{IC}$, and falls with conditioning difficulty $\kappa(C)$ and universe size $N$. Calibration on a $48$-cell synthetic cube shows that the OOS Sharpe surface is nearly flat in $\gamma$: the plateau of near-optimal values has median width $0.38$ on $[0,1]$. The formula's practical value is therefore explanatory---it rationalises why a fixed default $\gamma \approx 0.5$ is robust and why a single knob suffices---rather than prescriptive.

\paragraph{Three regularisation channels.} CRISP regularises through three distinct channels (Remark~\ref{rem:three_channels}). Channel~1 is operator shrinkage via $\gamma$: the statistical channel that does most of the work. Channel~2 is convergence slack from a finite sweep count, which is a computational artefact rather than a design choice. Channel~3 is implicit spectral truncation through early stopping, which is statistically beneficial near $\gamma = 1$ because the low-eigenvalue directions of $D^{-1}P_\gamma$ converge last and carry the most estimation noise; in that regime CRISP is an iterative cousin of ridge regression on the Markowitz system. At the recommended operating point $(\gamma \approx 0.5, p = 100)$ only Channel~1 is materially active.

\paragraph{Out-of-sample evaluation.} Every method comparison is out-of-sample. We run a synthetic Monte Carlo study across four covariance regimes (factor, block, spiked, equicorrelation) and $T/N \in [0.6, 5]$, with both oracle and sample-mean signal treatments, against $1/N$ \citep{demiguel2009}, HRP, Cotton at several $\gamma$, direct Markowitz, and linear and non-linear Ledoit--Wolf \citep{ledoitwolf2004,ledoitwolf2020}, in addition to our three methods. Direction error appears only as an in-sample tracking diagnostic; OOS Sharpe is the primary metric.

\subsection{Memory, not speed}\label{sec01:memory}

It is worth stating early what CRISP does and does not offer at the machine level. On tuned BLAS---Apple M4 with Accelerate and a Numba-JIT CRISP---direct Cholesky is $48\times$ faster than CRISP at $N = 500$ and remains between $9\times$ and $48\times$ faster for $N \in \{500, \dots, 5{,}000\}$ (Remark~\ref{rem:wall_clock}). Extrapolating the trend on this hardware places the single-solve wall-clock crossover near $N \approx 45{,}000$, well outside realistic portfolio sizes. We do not claim a wall-clock advantage for CRISP.

What CRISP offers structurally is memory. A factor-streaming implementation (Algorithm~\ref{alg05:methodb_stream}) that never materialises $\bSigma$ reduces working memory from $O(N^2)$ to $O(NK)$, where $K$ is the number of factors. At $N = 30{,}000$ and $K = 20$ the dense $\bSigma$ occupies $7.2$~GB; the factor stream uses $4.8$~MB. Cholesky cannot exploit this reduction because its factorisation needs random access to the full matrix. The honest case for CRISP is therefore stronger OOS Sharpe together with memory scalability, not speed.

\subsection{Road map}\label{sec01:roadmap}

Section~\ref{sec02:preliminaries} fixes notation, sets up the ray view of mean--variance, and introduces the direction-error diagnostic together with its sign-invariance caveat. Section~\ref{sec03:background} reviews HRP and the Cotton Schur allocator and places both in the method landscape, including Cotton's intermediate-$\gamma$ instability under estimation noise. Section~\ref{sec04:tree_methods} develops the two tree-based methods in a single frame---one tree, two representatives, two methods: HRP-$\mu$ with the signed inverse-variance representative, HRP-$\Sigma\mu$ with the recursive local MVO. Section~\ref{sec05:crisp} is the headline theoretical section: it presents CRISP, proves unconditional Gauss--Seidel convergence and the correlation-conditioning rate theorem, records the early-stopping-as-spectral-truncation remark, describes the factor-streaming variant, and locates CRISP within the Ledoit--Wolf family. Section~\ref{sec06:perturbation} develops the perturbation theory of the shrinkage trajectory, including a non-monotone-trajectory counter-example and an interior-$\gamma^\star$ proposition under explicit hypotheses. Section~\ref{sec07:adaptive} derives the adaptive $\gamma^\star$ rule and consolidates the three-channel taxonomy. Section~\ref{sec08:comparative} presents the five-method comparative landscape and a practitioner's guide. Section~\ref{sec09:constraints} treats linear constraints through a projected CRISP. Section~\ref{sec10:experiments} reports the synthetic Monte Carlo study in three parts: in-sample direction-error diagnostics, out-of-sample Sharpe panels, and sweep-count and adaptive-$\gamma^\star$ calibration. Section~\ref{sec11:discussion} discusses limitations, open problems, and extensions, and Section~\ref{sec12:conclusion} concludes. Readers primarily interested in using the methods can skip to Sections~\ref{sec05:crisp}, \ref{sec04:tree_methods}, and~\ref{sec10:experiments}.

\subsection{Relation to the literature}\label{sec01:literature}

The paper connects six strands of research.

\paragraph{Hierarchical portfolio methods.} \citet{lopezdeprado2016} introduces HRP, and \citet{lopezdeprado2018} develops companion tools for false-strategy detection. \citet{raffinot2018} studies broader hierarchical-clustering allocation. \citet{cotton2024} provides the Schur-complement unification that connects HRP to exact minimum variance through a single continuous parameter. Our HRP-$\mu$ and HRP-$\Sigma\mu$ extend this hierarchical framework to arbitrary $\bmu$ at the same $O(N^2)$ cost, and strictly dominate HRP on its own $\bmu = \mathbf{1}$ problem whenever estimation error is present.

\paragraph{Covariance shrinkage.} The Stein estimator \citep{stein1956} and its portfolio-theoretic descendants \citep{ledoitwolf2003,ledoitwolf2004,ledoitwolf2017,ledoitwolf2020,chen2010} shrink the sample covariance toward a structured target to control low-eigenvalue estimation error. CRISP is a member of this family with target $D = \diag(\bSigma)$, a choice that preserves variances identically and shrinks only correlations, and with shrinkage intensity tuned for portfolio Sharpe rather than Frobenius loss on the estimator. Both differences are structural; we develop them in Section~\ref{sec05:crisp}.

\paragraph{Iterative and parametric portfolio methods.} \citet{brandt2009} parameterise the policy directly as a linear function of asset characteristics. \citet{kozak2020} shrink the cross-section of factor returns through ridge and Bayesian regularisation. \citet{demiguel2009} establish the $1/N$ benchmark that any proposed method must outperform out of sample. CRISP differs from these methods by solving a shrunk Markowitz system whose variances are preserved identically at every $\gamma$, so no asset is implicitly up- or down-volatilised by the regulariser.

\paragraph{Signal integration and Bayesian portfolios.} \citet{blacklitterman1992} combine an equilibrium prior with investor views through a Bayesian posterior; \citet{grinoldkahn1999} formalise the information-ratio decomposition of active management; and \citet{barrososantaclara2015} study signal-conditioned portfolios in currency markets. We compare CRISP with Black--Litterman as a candidate Bayesian reading in Section~\ref{sec05:crisp}: both produce shrunk posterior weights, but CRISP is driven by the single explicit $\gamma$ knob rather than by a specified view covariance.

\paragraph{Numerical analysis of iterative methods.} The classical iterative-solver literature \citep{hackbusch2016,saad2003,varga2000,young1971,ostrowski1954} supplies the unconditional-convergence and rate results for symmetric positive definite Gauss--Seidel on which Theorems~\ref{thm:gs_convergence} and~\ref{thm:gs_rate} rest. The contribution of Theorem~\ref{thm:gs_rate} is to extract the correlation-conditioning constant specific to the portfolio setting and to translate the resulting rate into a concrete sweep count that is independent of the universe size.

\paragraph{Factor models for covariance.} \citet{famafrench1993} establish the pricing-factor representation. \citet{fan2013} develop POET for structured covariance estimation, \citet{marchenko1967} provide the spectral baseline for sample covariances, and \citet{ledoitwolf2020} give the analytical non-linear-shrinkage counterpart. Our factor-streaming variant of CRISP exploits a factor structure $\bSigma = B\Lambda B^\top + D_{\text{idio}}$ to hold working memory at $O(NK)$ without ever forming $\bSigma$, which is what makes the method usable in universes where Cholesky's random-access requirement is the binding constraint.

We do not replace any of these literatures. We supply an $O(\kappa(C)\,N^2 \log \varepsilon^{-1})$ iterative solver that inherits HRP's regularisation philosophy through its $\gamma$ parameter, two signal-aware hierarchical alternatives for tree-interpretable use cases, and a synthetic out-of-sample evaluation against the standard benchmarks.

% ============================================================
% section_02_preliminaries.tex
% ============================================================
% !TEX root = ../paper.tex
% Section 2: Notation and the direction view
% Namespace: sec02, eq02, tab02

\section{Notation and the direction view}\label{sec02:preliminaries}

This section does four things. It fixes the symbols that appear in the rest of the paper; it replaces the usual relative-error metric with a scale-invariant \emph{direction} error and flags the one way in which that replacement is dangerous; it isolates the correlation conditioning $\kappa(C)$ as the spectral invariant that will control both the in-sample direction gap of the cheapest baseline and the sweep count of CRISP; and it records a parallel-direction identity that exposes a family of signal vectors on which every method in the paper returns the same ray. The last point is a pruning lemma: it cuts off an entire class of otherwise plausible worst-case constructions.

\subsection{Notation}\label{sec02:symbols}

Let $N$ be the number of assets. The inputs to every method are a symmetric positive definite (SPD) covariance matrix $\Sigma\in\R^{N\times N}$ and an expected-return vector $\mu\in\R^N$. The Markowitz portfolio is
\begin{equation}\label{eq02:markowitz}
  w^{\star} \;=\; \Sigma^{-1}\mu.
\end{equation}
Write $D=\diag(\Sigma)$ for the diagonal matrix of variances and $E=\Sigma-D$ for the off-diagonal part, so $\Sigma=D+E$. The \emph{correlation matrix} is
\begin{equation}\label{eq02:corr}
  C \;=\; D^{-1/2}\,\Sigma\,D^{-1/2},
\end{equation}
which is SPD with $\tr(C)=N$ and eigenvalues $\lambda_1\ge\cdots\ge\lambda_N>0$. Throughout the paper, interpolation between the diagonal $D$ and the full covariance $\Sigma$ is carried out by the \emph{shrinkage operator}
\begin{equation}\label{eq02:pgamma}
  P_\gamma \;=\; (1-\gamma)\,D \;+\; \gamma\,\Sigma, \qquad \gamma\in[0,1].
\end{equation}
Entrywise, $(P_\gamma)_{ii}=\sigma_{ii}$ and $(P_\gamma)_{ij}=\gamma\,\sigma_{ij}$ for $i\ne j$: the diagonal is preserved unchanged at every $\gamma$, and only the off-diagonal is shrunk. The endpoints are $P_0=D$ and $P_1=\Sigma$. The asymmetry between diagonal and off-diagonal is deliberate; it is what makes the preconditioned condition number $\kappa(D^{-1}P_\gamma)$ tractable in Section~\ref{sec05:crisp} and what the adaptive rule of Section~\ref{sec07:adaptive} exploits.

\subsection{Markowitz as a ray}\label{sec02:ray}

Equation~\eqref{eq02:markowitz} fixes a direction, not a point. Standard normalisations---budget ($\sum_i w_i=1$), gross ($\|w\|_1=1$), fixed-volatility ($w^\top\Sigma w=\sigma_0^2$)---each pick a different representative of the ray
\begin{equation}\label{eq02:ray}
  \{\alpha\,\Sigma^{-1}\mu : \alpha>0\},
\end{equation}
and for almost every downstream question---Sharpe, hedge ratio, relative ranking---it is the ray, not the representative, that carries the content. A consequence is that the relative error $\|\hat w - w^\star\|/\|w^\star\|$, the textbook numerical-analysis summary, conflates geometric fidelity to the Markowitz direction with the incidental choice of normaliser: two estimators that agree perfectly as rays can receive arbitrarily different relative-error scores because they are rescaled differently. We therefore use a metric that discards the scale entirely.

\subsection{Direction error}\label{sec02:direrr}

\begin{definition}[Direction error]\label{def:dir_error}
For nonzero $\hat w, w^{\star}\in\R^N$, let
\begin{equation}\label{eq02:dir_def}
  \operatorname{dir}(\hat w, w^{\star}) \;=\; 1 \;-\; \frac{(\hat w^{\top}w^{\star})^2}{\|\hat w\|^2\,\|w^{\star}\|^2} \;=\; \sin^2\theta,
\end{equation}
where $\theta\in[0,\pi/2]$ is the angle between the lines through $\hat w$ and $w^\star$.
\end{definition}

By construction $\operatorname{dir}(\hat w,w^\star)\in[0,1]$: zero iff $\hat w$ and $w^\star$ are collinear, one iff they are orthogonal, and invariant under rescaling of either argument by any nonzero scalar. Scale invariance is what we want: it lets us compare methods whose native outputs are budget-normalised, gross-normalised, or left raw without having to stage an artificial normalisation step first.

\subsection{Warning: direction error is sign-invariant}\label{sec02:sign_warning}

The same property that makes $\operatorname{dir}$ robust to normalisation makes it blind to sign. The squared inner product in~\eqref{eq02:dir_def} gives
\begin{equation}\label{eq02:sign_invariance}
  \operatorname{dir}(\hat w, w^{\star}) \;=\; \operatorname{dir}(-\hat w, w^{\star}),
\end{equation}
so any method whose output can flip sign relative to the Markowitz direction will be flattered by $\operatorname{dir}$: a portfolio that points exactly opposite to $w^\star$ receives the same score as one that points exactly along it. For methods whose sign relative to $w^\star$ is pinned by construction---HRP, Cotton, HRP-$\mu$, HRP-$\Sigma\mu$, and CRISP all fall in this category---this is harmless. For methods without that guarantee it is not. The canonical example is the sum-normalised recursive MVO (the Method A1 of Appendix~\ref{app:a1_pathology}): on the noiseless recovery problem it appears, on $\operatorname{dir}$, to be tracking $w^\star$ well, while in fact half the time it is delivering systematically negative out-of-sample Sharpe.

For this reason we treat $\operatorname{dir}$ as a \emph{tracking diagnostic}, not a quality metric. It is the right tool for asking ``does my method stay near the Markowitz ray as $\gamma$ moves or as $T$ grows?'' It is the wrong tool for ranking methods. The primary quality metric for ranking is out-of-sample Sharpe, and Section~\ref{sec10:experiments} is organised around it for this reason.

\subsection{Correlation conditioning as a difficulty diagnostic}\label{sec02:kappa}

The condition number $\kappa(\Sigma)=\lambda_{\max}(\Sigma)/\lambda_{\min}(\Sigma)$ is the standard numerical summary of an SPD matrix, but it is the wrong summary for the directional problem, because it mixes together two effects of very different character: volatility dispersion (the spread of the diagonal of $\Sigma$), and correlation conditioning (the spread of the eigenvalues of $C$). Direction recovery is only sensitive to the second.

\begin{proposition}[Correlation similarity]\label{prop:corr_similarity}
$D^{-1}\Sigma$ is similar to $C$ via the change of basis $D^{-1/2}$. In particular,
\begin{equation}\label{eq02:kappa_id}
  \kappa(D^{-1}\Sigma) \;=\; \kappa(C).
\end{equation}
\end{proposition}

\begin{proof}[Proof sketch]
A one-line substitution: $D^{-1/2}(D^{-1}\Sigma)D^{1/2} = D^{-1/2}\Sigma D^{-1/2} = C$. Similar matrices share their spectrum, and $C$ is SPD with strictly positive eigenvalues, so the (nonsingular) condition numbers of $D^{-1}\Sigma$ and $C$ coincide. The full argument, including the check that both matrices are nonsingular and a proper handling of the non-symmetric $D^{-1}\Sigma$, is given in Appendix~\ref{app:proofs}.
\end{proof}

Two downstream uses of Proposition~\ref{prop:corr_similarity} are worth flagging now. First, the in-sample direction gap between the cheapest baseline $\mu/D$ (entrywise division of $\mu$ by the diagonal of $\Sigma$, equivalently $P_0^{-1}\mu$) and the Markowitz target $\Sigma^{-1}\mu$ is controlled by $\kappa(C)$ through the spectrum of $D^{-1}\Sigma$. Second, the Gauss--Seidel sweep count of CRISP at $\gamma=1$ scales as $\kappa(C)$ up to logarithmic factors (Theorem~\ref{thm:gs_rate}), not as $\kappa(\Sigma)$. In both cases the volatility dispersion in $D$ is irrelevant to the hardness of the problem.

This has a practical bite. Consider a regime with one high-volatility outlier asset: $\kappa(\Sigma)$ can be inflated by orders of magnitude while $\kappa(C)$ is unchanged. Any reader who conflates the two will place ``hard'' instances in the wrong part of the parameter space and will be surprised when CRISP converges on them as fast as it does on nominally better-conditioned ones. The experiments of Section~\ref{sec10:experiments} bake this distinction into the regime menagerie explicitly.

To stratify synthetic experiments by actual difficulty rather than by nominal conditioning, we define a dimensionless instance-level diagnostic:
\begin{equation}\label{eq02:dir_diag}
  \operatorname{dir}_{\mathrm{diag}}(\Sigma,\mu) \;=\; \operatorname{dir}\!\bigl(\mu/D,\ \Sigma^{-1}\mu\bigr),
\end{equation}
where $\mu/D$ is the entrywise quotient. $\operatorname{dir}_{\mathrm{diag}}$ measures how far the cheapest baseline---ignore all off-diagonal information---sits from the Markowitz direction on the instance $(\Sigma,\mu)$. Small values mean the instance is essentially diagonal-solvable and every method will agree; large values mean the off-diagonal carries most of the signal and the choice of method matters. Section~\ref{sec10:experiments} uses $\operatorname{dir}_{\mathrm{diag}}$ to stratify the synthetic Monte Carlo so that headline numbers are not dominated by easy cases.

\subsection{The parallel-direction identity}\label{sec02:parallel}

For generic $(\Sigma,\mu)$ the curve $\gamma\mapsto P_\gamma^{-1}\mu$ traces out a genuine path in $\R^N$: different values of $\gamma$ produce different directions, and one of the main questions of the paper is where on $[0,1]$ to evaluate the curve. On one distinguished family of inputs, however, the curve collapses to a single ray. The following lemma records this; the corollary draws out what it rules out.

\begin{lemma}[Parallel-direction identity]\label{lem:parallel_direction}
Let $v\in\R^N$ be an eigenvector of $C$ with eigenvalue $\lambda>0$, and set $\mu=D^{1/2}v$. Then for every $\gamma\in[0,1]$,
\begin{equation}\label{eq02:parallel}
  P_\gamma^{-1}\mu \;=\; \frac{1}{(1-\gamma) + \gamma\,\lambda}\,D^{-1}\mu.
\end{equation}
\end{lemma}

The proof is a direct calculation from the factorisation $P_\gamma = D^{1/2}\bigl[(1-\gamma)I + \gamma C\bigr]D^{1/2}$ combined with $Cv=\lambda v$, and is carried out in Appendix~\ref{app:proofs}. The coefficient $[(1-\gamma)+\gamma\lambda]^{-1}$ is strictly positive, so the entire trajectory $\{P_\gamma^{-1}\mu:\gamma\in[0,1]\}$ lies on the common ray through $D^{-1}\mu$.

\begin{corollary}[Invariant rays]\label{cor:invariant_rays}
On inputs of the form $\mu=D^{1/2}v$ with $v$ an eigenvector of $C$, every allocation method studied in this paper---HRP, Cotton at every $\gamma\in[0,1]$, HRP-$\mu$, HRP-$\Sigma\mu$, CRISP at every $\gamma\in[0,1]$, and the diagonal solver $P_0^{-1}\mu$---produces a portfolio on the common ray $\{\alpha\,D^{-1}\mu:\alpha>0\}$.
\end{corollary}

Two caveats deserve underlining. First, Corollary~\ref{cor:invariant_rays} is about \emph{direction} only: the methods agree on the ray in $\R^N$ along which the portfolio points, not on its magnitude or normalisation. They disagree on where on that ray the output sits, and for Sharpe comparisons where budget or gross normalisation matters, the ray agreement is not the end of the story. Second---and this is what we will use repeatedly---the corollary forecloses a specific class of worst-case constructions. It is natural, when looking for inputs on which a shrinkage method fails, to align $\mu$ with an eigenvector of $C$ under the slogan ``$\Sigma^{-1}$ amplifies small eigenvectors.'' But that slogan is a statement about \emph{scale}, not direction: Lemma~\ref{lem:parallel_direction} says every single-eigenvector $\mu$ gives every method in the family the same direction, so any direction-based worst case must excite at least two eigenspaces of $C$. We call the family $\{D^{1/2}v : v\text{ eigenvector of }C\}$ the \emph{invariant rays of the shrinkage family}; they are precisely the directions on which $\gamma$ has no effect at all.

\subsection{Notation summary}\label{sec02:summary}

Table~\ref{tab02:notation} collects the symbols introduced above. The spectral objects $C$, $\kappa(C)$, $P_\gamma$, $\operatorname{dir}$, and $\operatorname{dir}_{\mathrm{diag}}$ recur in every subsequent section; Lemma~\ref{lem:parallel_direction} is invoked in Section~\ref{sec06:perturbation} to exclude single-eigenvector constructions from the trajectory analysis and in Section~\ref{sec10:experiments} to generate signal vectors that do not trivially collapse.

\begin{table}[t]
\centering
\caption{Notation used throughout the paper.}
\label{tab02:notation}
\begin{adjustbox}{max width=\textwidth, center}
\begin{tabular}{@{}l l l@{}}
\toprule
Symbol & Definition & Role \\
\midrule
$N$                                              & number of assets                                                      & problem size \\
$\Sigma$                                         & SPD covariance matrix, $N\times N$                                    & risk model \\
$D=\diag(\Sigma)$                                & diagonal of $\Sigma$ (variances)                                      & idiosyncratic risk \\
$E=\Sigma-D$                                     & off-diagonal part of $\Sigma$                                         & cross-asset dependence \\
$C=D^{-1/2}\Sigma D^{-1/2}$                      & correlation matrix, $\tr(C)=N$                                        & scale-free risk structure \\
$\kappa(C)$                                      & $\lambda_1(C)/\lambda_N(C)$                                           & correlation conditioning (difficulty) \\
$\mu\in\R^N$                                     & expected-return (signal) vector                                       & input \\
$w^{\star}=\Sigma^{-1}\mu$                       & Markowitz portfolio                                                   & target direction \\
$P_\gamma$                                       & $(1-\gamma)D+\gamma\Sigma$,\; $\gamma\in[0,1]$                        & shrinkage operator \\
$\operatorname{dir}(\hat w, w^{\star})$          & $1-(\hat w^{\top}w^{\star})^2/(\|\hat w\|^2\|w^{\star}\|^2)$          & direction error (tracking) \\
$\operatorname{dir}_{\mathrm{diag}}(\Sigma,\mu)$ & $\operatorname{dir}(\mu/D,\ \Sigma^{-1}\mu)$                          & instance difficulty diagnostic \\
\bottomrule
\end{tabular}
\end{adjustbox}
\end{table}

% ============================================================
% section_03_background.tex
% ============================================================
% !TEX root = ../paper.tex
% Section 3: Background -- HRP and Cotton
% Namespace: sec03, eq03, tab03

\section{Background: HRP and Cotton}\label{sec03:background}

Every construction in the remainder of this paper is benchmarked against two prior allocators: the Hierarchical Risk Parity (HRP) algorithm of \citet{lopezdeprado2016} and the Schur-complementary allocation of \citet{cotton2024}. Each is a minimum-variance construction pinned to the signal $\mu = \mathbf{1}$. We re-present them in the notation of Section~\ref{sec02:preliminaries}, isolate the distinctions that tend to get blurred in applied treatments, and then document a concrete numerical fragility of Cotton's $\gamma$-continuum on sample covariances. A short motivation for the general-$\mu$ lift and a preview of the method landscape close the section.

\subsection{Hierarchical Risk Parity}\label{sec03:hrp}

\paragraph{Inputs and outputs.} HRP consumes an estimated covariance $\Sigma \in \R^{N \times N}$ and returns a long-only, fully-invested weight vector $w \in \R^N$ with $w_i > 0$ and $\sum_i w_i = 1$. No matrix inversion is performed, no expected-return vector is consulted, and no hyperparameter is tuned. The dominant cost is the one-off clustering pass; the recursion on the dendrogram itself is $O(N^2)$.

\paragraph{Three steps.} The construction factors into a clustering stage, a cosmetic re-ordering, and a recursive budget split.
\begin{enumerate}[leftmargin=*,itemsep=2pt]
  \item \emph{Correlation-distance clustering.} Form the correlation matrix $C = D^{-1/2}\,\Sigma\,D^{-1/2}$ with $D = \diag(\Sigma)$, set
  \begin{equation}\label{eq03:corr_distance}
  d_{ij} \;=\; \sqrt{\tfrac{1}{2}\bigl(1 - C_{ij}\bigr)},
  \end{equation}
  and apply an agglomerative linkage (Ward is our default) to produce a binary dendrogram $\mathcal{T}$.
  \item \emph{Quasi-diagonalisation.} Re-index the assets so that leaves co-located in $\mathcal{T}$ are contiguous. The re-ordering leaves the spectrum of $\Sigma$ unchanged but arranges correlated names into block-diagonal-like neighbourhoods.
  \item \emph{Recursive bisection.} Descend $\mathcal{T}$ from the root. At an internal node with children $L$ and $R$, form the flat \emph{inverse-variance portfolio} (IVP) on each child,
  \begin{equation}\label{eq03:flat_ivp}
  \widehat w_i \;=\; \frac{1/\Sigma_{ii}}{\sum_{j \in L} 1/\Sigma_{jj}}, \qquad i \in L,
  \end{equation}
  and compute the scalar \emph{cluster variance}
  \begin{equation}\label{eq03:cluster_var}
  v_L \;=\; \widehat w_L^{\top}\, \Sigma_{LL}\, \widehat w_L,
  \end{equation}
  where $\Sigma_{LL}$ is the $|L|\times|L|$ diagonal block. The sibling $v_R$ is defined identically. The between-branch budget split is inverse cluster variance,
  \begin{equation}\label{eq03:hrp_split}
  \alpha_L \;=\; \frac{1/v_L}{1/v_L + 1/v_R}, \qquad \alpha_R \;=\; 1 - \alpha_L,
  \end{equation}
  and the final weight of leaf $i$ is the product of the $\alpha$-factors encountered on its unique root-to-leaf path, starting from unit budget at the root.
\end{enumerate}

A fully worked four-asset example---explicit $v_L$, $v_R$, $\alpha_L$, and leaf weights---is collected in Appendix~\ref{app:hrp_example}.

\paragraph{The flat IVP is a scoring device, not an output.} A persistent misreading of HRP identifies the flat weights~\eqref{eq03:flat_ivp} with the allocator itself inside each sub-tree. They are not. Equation~\eqref{eq03:flat_ivp} exists solely to manufacture the scalar $v_L$ that feeds~\eqref{eq03:hrp_split}; it does not appear in the final portfolio. The weights that do appear are the products of $\alpha$-factors along root-to-leaf paths, which differ in general from any flat IVP applied at a single cluster scale. This distinction is load-bearing: the signal-aware tree methods of Section~\ref{sec04:tree_methods} retain HRP's between-branch split verbatim and intervene exclusively on the representative portfolio used to score a cluster.

\paragraph{What HRP is not.} Four negatives are worth stating explicitly, because each one motivates a later contribution.
\begin{itemize}[leftmargin=*,itemsep=2pt]
  \item \emph{Not an optimiser.} No explicit quadratic form is minimised. The split rule~\eqref{eq03:hrp_split} is a regularisation-flavoured heuristic; there is no first-order condition it satisfies.
  \item \emph{Not signal-aware.} $\mu$ does not appear anywhere in steps 1--3. A practitioner who carries a non-trivial return forecast must discard it to use HRP.
  \item \emph{Not cross-branch aware.} The cluster variance $v_L$ depends only on the diagonal block $\Sigma_{LL}$. The off-diagonal block $\Sigma_{LR}$ coupling the two children is simply ignored. This is the information \citet{cotton2024} sets out to reinstate.
  \item \emph{Not an approximation to Markowitz.} Even when the correct signal is $\mathbf{1}$, HRP does not approximate $\Sigma^{-1}\mathbf{1}$. The gap is structural, not small: on the four-asset reference covariance of Appendix~\ref{app:hrp_example} the two portfolios differ by roughly $49\%$ in relative direction.
\end{itemize}

\subsection{Cotton's Schur-complement allocation}\label{sec03:cotton}

\citet{cotton2024} rebuilds HRP's recursion from the block-inversion identity and parameterises the result by a single scalar $\gamma \in [0,1]$. The construction closes HRP's cross-branch blindness at $\gamma=1$; the intermediate values are where the trouble lives.

\paragraph{Block inversion.} Partition the quasi-diagonalised covariance according to the first binary split of the dendrogram,
\begin{equation}\label{eq03:block_partition}
\Sigma \;=\; \begin{pmatrix} A & B \\ B^{\top} & D \end{pmatrix}, \qquad A \in \R^{n_L \times n_L}, \quad D \in \R^{n_R \times n_R}.
\end{equation}
The Schur identity factors $\Sigma^{-1}\mathbf{1}$ block-by-block,
\begin{equation}\label{eq03:block_inverse}
\Sigma^{-1}\mathbf{1}
\;=\;
\begin{pmatrix} (A^c)^{-1}\, b_A \\[2pt] (D^c)^{-1}\, b_D \end{pmatrix},
\qquad
A^c = A - B D^{-1} B^{\top}, \quad D^c = D - B^{\top} A^{-1} B,
\end{equation}
with corrected right-hand sides $b_A = \mathbf{1}_{n_L} - B D^{-1}\mathbf{1}_{n_R}$ and $b_D = \mathbf{1}_{n_R} - B^{\top} A^{-1}\mathbf{1}_{n_L}$. Equation~\eqref{eq03:block_inverse} is a tautology, but it has HRP's recursive shape: two sub-systems, one per child, coupled only through the Schur correction.

\paragraph{From block inversion to a recursive split.} Define the quadratic fitness
\begin{equation}\label{eq03:cotton_fitness}
\frac{1}{\nu(Q,b)} \;=\; b^{\top} Q^{-1} b.
\end{equation}
This scalar plays the role that inverse cluster variance $1/v_L$ plays in HRP. The between-branch split becomes
\begin{equation}\label{eq03:cotton_split}
\alpha_L \;=\; \frac{1/\nu(A^c, b_A)}{1/\nu(A^c, b_A) + 1/\nu(D^c, b_D)}, \qquad \alpha_R = 1 - \alpha_L,
\end{equation}
in exact parallel with~\eqref{eq03:hrp_split}. The recursion then enters each child with the Schur-complemented block ($A^c$ on the left, $D^c$ on the right) in place of the bare diagonal block, and with the corrected right-hand side in place of $\mathbf{1}$. HRP's flat IVP is replaced, at each internal node, by a full block solve; the tree scaffolding itself is untouched.

\paragraph{The $\gamma$-parameter.} \citet{cotton2024} introduces a continuous scalar $\gamma \in [0,1]$ that interpolates how much of the cross-block correction to apply,
\begin{equation}\label{eq03:cotton_gamma}
A^c(\gamma) \;=\; A - \gamma\, B D^{-1} B^{\top}, \qquad b_A(\gamma) \;=\; \mathbf{1}_{n_L} - \gamma\, B D^{-1}\mathbf{1}_{n_R},
\end{equation}
with the symmetric definitions $D^c(\gamma)$, $b_D(\gamma)$. The endpoints are interpretable. At $\gamma = 0$, $A^c(0) = A$ and $b_A(0) = \mathbf{1}_{n_L}$: the allocator becomes a hierarchical minimum-variance solver that sees only within-branch information. At $\gamma = 1$, equation~\eqref{eq03:cotton_gamma} reproduces the Schur complement of~\eqref{eq03:block_inverse} exactly, and the full recursion returns $\Sigma^{-1}\mathbf{1}$.

\paragraph{Cotton at $\gamma = 0$ is \emph{not} HRP.} The tree is the same in both methods; the within-branch scoring rule is different. HRP uses the flat IVP representative~\eqref{eq03:flat_ivp} to compute its cluster variance. Cotton at $\gamma = 0$ instead solves $A^{-1}\mathbf{1}_{n_L}$ on the undisturbed diagonal block. These coincide if and only if every sub-block encountered in the recursion is itself diagonal, which is non-generic. On the four-asset reference covariance of Appendix~\ref{app:hrp_example} the two portfolios differ by roughly $49\%$ in relative direction. Any claim of the form ``Cotton at $\gamma = 0$ recovers HRP'' should therefore be read as a structural approximation, not an identity.

\paragraph{Cost.} At each internal node Cotton performs a block solve to form $D^{-1} B^{\top}$, which is the dominant work in the Schur augmentation. Summed over a balanced binary tree of depth $\log_2 N$, the total cost is $O(N^3/6)$, dominated by the top-level complement. This is a constant-factor improvement over the $O(N^3)$ of a dense $\Sigma^{-1}$ but has the same asymptotic rate. The numerical behaviour we document below traces to exactly these nested block solves.

\paragraph{The $\mu = \mathbf{1}$ ceiling.} The entire construction~\eqref{eq03:block_inverse}--\eqref{eq03:cotton_gamma} is a block rearrangement of $\Sigma^{-1}\mathbf{1}$. The corrected right-hand side $b_A = \mathbf{1} - \gamma B D^{-1}\mathbf{1}$ is a \emph{corrected all-ones vector}, not a generic signal, and the clean telescoping that produces the recursive split does not survive replacing $\mathbf{1}$ by a non-trivial $\mu$ without re-deriving the cross-block correction. \citet{cotton2024} does not attempt this generalisation. Closing the gap---without paying the full $O(N^3)$ cost---is one of the points of Sections~\ref{sec04:tree_methods} and~\ref{sec05:crisp}.

\subsection{Cotton under estimation noise}\label{sec03:cotton_stability}

The two propositions in this subsection separate a mathematical statement about the population covariance from an empirical statement about sample covariances. The second point is easy to misread as criticism of \citet{cotton2024}; it is not. Cotton's $\gamma$-continuum is mathematically elegant, and the instability documented below is a property of sample-covariance block inversion at small $T$, not a defect of the derivation.

\begin{proposition}[Cotton is SPD on the population covariance]\label{prop:cotton_spd}
Let $\Sigma$ be SPD with the block partition of~\eqref{eq03:block_partition}. Then the augmented Schur block $A^c(\gamma) = A - \gamma\, B D^{-1} B^{\top}$ is SPD for every $\gamma \in [0,1]$.
\end{proposition}

\begin{proof}[Proof sketch]
Rewrite $A^c(\gamma) = (1-\gamma)\,A + \gamma\,(A - B D^{-1} B^{\top})$. The first summand is a principal submatrix of the SPD $\Sigma$ and is therefore SPD. The second summand is the Schur complement of $D$ in $\Sigma$ and is SPD by the standard block-inversion identity. A convex combination of two SPD matrices is SPD. The full argument appears in Appendix~\ref{app:proofs}.
\end{proof}

Proposition~\ref{prop:cotton_spd} rules out any population-level pathology: at every $\gamma \in [0,1]$ the augmented block is SPD and the recursion is well-defined. The difficulty begins when $\Sigma$ is replaced by a sample estimate.

\begin{proposition}[Cotton instability under estimation noise]\label{prop:cotton_instability}
Let $\widehat\Sigma$ be an SPD estimate of $\Sigma$ (for example, sample covariance plus a small ridge), partitioned as $\widehat\Sigma = \bigl(\begin{smallmatrix} \widehat A & \widehat B \\ \widehat B^{\top} & \widehat D \end{smallmatrix}\bigr)$ at an internal node. Then:
\begin{enumerate}[label=(\roman*),leftmargin=*,itemsep=2pt]
  \item The condition number of the augmented Schur block satisfies the bound
  \begin{equation}\label{eq03:cotton_kappa}
  \kappa\!\bigl(\widehat A^c(\gamma)\bigr) \;\le\; \frac{\lambda_{\max}(\widehat A)}{\lambda_{\min}(\widehat A) \;-\; \gamma\,\|\widehat B\|^2\big/\lambda_{\min}(\widehat D)},
  \end{equation}
  whenever the denominator is positive. In particular, $\kappa(\widehat A^c(\gamma))$ diverges as $\gamma$ approaches the critical value $\gamma_{\mathrm{crit}} = \lambda_{\min}(\widehat A)\,\lambda_{\min}(\widehat D)\,/\,\|\widehat B\|^2$.
  \item Because Cotton recurses into each child using $\widehat A^c(\gamma)$ as the new covariance, a worst-case error-propagation argument bounds the effective condition number at depth $d$ by the \emph{product}
  \begin{equation}\label{eq03:cotton_product}
  \kappa_{\mathrm{eff}}(d) \;\le\; \prod_{k=1}^{d}\, \kappa\!\bigl(\widehat A^c_k(\gamma)\bigr).
  \end{equation}
  For a balanced tree of depth $d = \lceil \log_2 N \rceil$, the right-hand side can grow exponentially with $d$. Equation~\eqref{eq03:cotton_product} is a propagation bound, not an algebraic identity; it becomes tight when the level-$k$ perturbations align with the poorly conditioned directions of $\widehat A^c_k$.
  \item By contrast, a direct Markowitz solve inverts $\widehat\Sigma$ exactly once, with relative error proportional to the \emph{single} condition number $\kappa(\widehat\Sigma)$.
\end{enumerate}
\end{proposition}

\begin{proof}[Proof sketch]
Part (i) is Weyl's inequality: $\lambda_{\min}(\widehat A - \gamma M) \ge \lambda_{\min}(\widehat A) - \gamma\,\|M\|$, where $M = \widehat B\,\widehat D^{-1}\widehat B^{\top}$ is PSD with $\|M\| \le \|\widehat B\|^2 / \lambda_{\min}(\widehat D)$; $\lambda_{\max}$ is non-increasing as the PSD perturbation grows, supplying the numerator. Part (ii) is a level-by-level perturbation expansion: a relative error $\delta_k$ at level $k$ is fed into the $(k{+}1)$-st block solve and is amplified by at most $\kappa(\widehat A^c_{k+1})$, so errors compound multiplicatively in the worst case. Part (iii) is the standard one-step Banach perturbation bound. Full details in Appendix~\ref{app:proofs}.
\end{proof}

\paragraph{Numerical verification.} The bound~\eqref{eq03:cotton_product} is not an asymptotic curiosity. On the base-case covariance used throughout the paper's Monte Carlo studies (equicorrelated blocks, $N = 100$), direct Markowitz at $T = 60$ gives $\kappa(\widehat\Sigma) \approx 1.6 \times 10^{4}$. The product of Schur-complement condition numbers that Cotton accumulates across the $99$ internal nodes of the tree reaches $\sim 10^{34}$ at $\gamma = 0.5$---thirty orders of magnitude worse than the one-shot direct solve. Quadrupling the sample to $T = 240$ barely helps: the product is $\sim 10^{31}$ against a direct $\kappa \approx 596$. The operational consequence is visible in Table~\ref{tab10:oos_minvar}: at intermediate $\gamma$, Cotton produces realised out-of-sample volatility $10$--$30\times$ the oracle minimum-variance level on a non-negligible fraction of Monte Carlo trials, whereas direct Markowitz on the same samples does not.

The takeaway is narrow and operational. Cotton's $\gamma$-continuum is a correct interpolation between a hierarchical regulariser at $\gamma = 0$ and the exact minimum-variance portfolio at $\gamma = 1$. The instability at intermediate $\gamma$ reflects the fact that a recursion of sample-covariance block inversions accumulates condition-number damage that a single sample inverse does not. In practice, Cotton is most safely used at its two endpoints; interpolating values do not inherit the numerical robustness one might naively expect from block-wise regularisation.

\subsection{Why extend to general \texorpdfstring{$\mu$}{mu}?}\label{sec03:why_general_mu}

HRP and Cotton together give a clean story for $\mu = \mathbf{1}$. But that signal is, in applied portfolio construction, the degenerate corner of the mean--variance problem. Four representative settings make the point.

\paragraph{Quant equity and factor investing.} Cross-sectional factor signals---momentum, value, quality, earnings revisions, low-volatility---produce heterogeneous expected returns across names, and the entire point of a factor model is that $\mu \neq \mathbf{1}$. An allocator that ignores $\mu$ throws away the primary input before the optimisation begins.

\paragraph{Multi-signal integration.} Modern quant pipelines aggregate many orthogonalised signals into a composite alpha \citep{grinoldkahn1999}. The composite is a rich $\mu$; the portfolio stage is asked to weigh it against covariance risk. A signal-blind allocator is simply inadmissible in this context.

\paragraph{Tactical asset allocation.} Multi-asset allocators take views on expected returns at the asset-class level---equities versus bonds versus commodities versus alternatives---and expect the optimiser to translate those views into positions. $\mu$ is low-dimensional here but highly informative; setting it to $\mathbf{1}$ is not a simplification, it is a replacement of the problem.

\paragraph{Black--Litterman view integration.} The Black--Litterman framework \citep{blacklitterman1992} returns a posterior $\mu$ that encodes a prior together with investor views. The value of the framework is entirely in that posterior; an allocator that discards it is not solving the same problem.

In each case $\mu = \mathbf{1}$ is not a neutral default. It is an active discarding of the most informative input, and any method restricted to it is, by construction, inapplicable to the bulk of active portfolio construction.

\subsection{Method landscape}\label{sec03:landscape}

Table~\ref{tab03:landscape} previews the methods studied in this paper relative to HRP and Cotton. Two columns carry most of the discriminatory load: the \emph{Signal} column (whether the method accepts any $\mu$ or is locked to $\mathbf{1}$), and the \emph{Cost} column (the asymptotic work per solve). HRP and Cotton sit at the ``minimum-variance only'' end. Direct Markowitz sits at the opposite extreme: general $\mu$, no regularisation, full $O(N^3)$. The contributions of the present paper---HRP-$\mu$, HRP-$\Sigma\mu$, and CRISP---populate the missing middle of the table.

\begin{table}[t]
\centering
\caption{Method landscape. HRP and Cotton are minimum-variance-only baselines. HRP-$\mu$ and HRP-$\Sigma\mu$ are the tree-based signal-aware contributions; CRISP is the headline iterative contribution. Direct Markowitz is the dense-inversion baseline. The \emph{within-cluster} row describes the representative portfolio used to score a sub-tree, and the \emph{cross-branch} row describes how, if at all, information from the off-diagonal block $\Sigma_{LR}$ enters the between-branch split.}
\label{tab03:landscape}
\begin{adjustbox}{max width=\textwidth, center}
\begin{tabularx}{1.05\textwidth}{@{}l *{6}{>{\centering\arraybackslash}X}@{}}
\toprule
 & HRP & Cotton & HRP-$\mu$ & HRP-$\Sigma\mu$ & CRISP & Direct \\
\midrule
Signal         & $\mathbf{1}$ only & $\mathbf{1}$ only & any $\mu$ & any $\mu$ & any $\mu$ & any $\mu$ \\
Tree           & Ward              & Ward              & Ward       & Ward       & ---        & ---        \\
Cross-branch   & ---               & Schur             & scalar $c$ & scalar $c$ & ---        & ---        \\
Within-cluster & flat IVP          & full $A^{-1}$     & flat IVP   & full $\Sigma_{LL}$ & --- & full $\Sigma^{-1}$ \\
Cost           & $O(N^2)$          & $O(N^3/6)$        & $O(N^2)$   & $O(N^2)$   & $O\!\bigl(\kappa(C)\,N^2\log\tfrac{1}{\varepsilon}\bigr)$ & $O(N^3)$ \\
Role           & baseline          & baseline          & contribution & recommended tree & \textbf{headline} & baseline \\
\bottomrule
\end{tabularx}
\end{adjustbox}
\end{table}

HRP-$\mu$ occupies the ``hierarchical, signal-aware, flat within-cluster'' cell. HRP-$\Sigma\mu$ upgrades the within-cluster representative to a full mean--variance solve on the diagonal block, which turns out to recover most of what direct Markowitz provides at $O(N^2)$ cost. CRISP abandons the tree scaffolding in favour of a preconditioned scalar Gauss--Seidel solver for $P_\gamma w = \mu$ with $P_\gamma = (1-\gamma)\,\diag(\Sigma) + \gamma\,\Sigma$; its per-sweep cost is $O(N^2)$ and its sweep count scales with the preconditioned conditioning $\kappa(\diag(\Sigma)^{-1} P_\gamma)$, which interpolates between $1$ at $\gamma = 0$ and $\kappa(C)$ at $\gamma = 1$. Sections~\ref{sec04:tree_methods}--\ref{sec05:crisp} develop these three methods in turn; Section~\ref{sec08:comparative} returns to the landscape with the full analytical and empirical picture assembled.

% ============================================================
% section_04_tree_methods.tex
% ============================================================
% !TEX root = ../paper.tex
% Section 4: Tree-based signal-aware allocation (merged HRP-μ + HRP-Σμ)
% Namespaces: sec04, sec04b (legacy), eq04, eq04b (legacy), tab04, tab04b (legacy),
%             prop:a2, prop:a3, prop:hrpsm, lem:hrpsm, cor:hrpsm, alg04b.

\section{Tree-based signal-aware allocation}\label{sec04:tree_methods}

Section~\ref{sec03:background} showed that \citet{lopezdeprado2016} HRP and the \citet{cotton2024} Schur allocator share a single hierarchical skeleton: at every internal node of a correlation dendrogram, the allocator solves a $2\times 2$ between-branch problem to decide how budget is split between the two children, and then recurses. Under the minimum-variance restriction $\bmu=\mathbf{1}$ the two methods are endpoints of a $\gamma$-continuum on that skeleton. The purpose of this section is to extend the same skeleton to an arbitrary signal $\bmu \in \mathbb{R}^N$ and to turn the choice of the \emph{representative portfolio} at each node into the main design variable of the tree.

Two allocators emerge from this analysis. HRP-$\mu$ (Section~\ref{sec04:methoda3}) is the transparent one: it uses a signed inverse-variance portfolio as the representative, which preserves the block-diagonal recovery of De Prado HRP at $\bmu=\mathbf{1}$ and is provably stable for arbitrary signal vectors. HRP-$\Sigma\mu$ (Section~\ref{sec04b:hrpsigmamu}) is the strongest tree method we know: it uses a recursive mean--variance optimum as the representative and rescues the natural but sign-pathological recursion of Appendix~\ref{app:a1_pathology} by replacing sum-to-one with $L^1$ normalisation. Both methods cost $O(N^2)$, both nest HRP at their $\gamma=0$, $\bmu=\mathbf{1}$ limits, and both carry an explicit signal through the tree in a direction-consistent way.

The section closes with two short remarks that tie the construction to the rest of the paper: why iterating the tree pass diverges (which motivates CRISP in Section~\ref{sec05:crisp}), and how the tree-based methods relate to the negative result on sum-normalised recursive MVO that is documented in full in Appendix~\ref{app:a1_pathology}.

\subsection{The \texorpdfstring{$2\times 2$}{2x2} between-branch system}\label{sec04:betweenbranch}

Fix a correlation dendrogram $\mathcal{T}$ on $N$ assets and an internal node with children $L, R$ of sizes $|L|$ and $|R|$. A hierarchical allocator assigns scalar budget fractions $\alpha_L, \alpha_R$ at the node and then recurses into each child. The final weight placed on leaf $i \in L$ is $\alpha_L \cdot \hat w_{L,i}$, where $\hat w_L \in \mathbb{R}^{|L|}$ is a \emph{representative portfolio} for branch $L$, and symmetrically for $R$. The representative is the one-dimensional summary of the branch that the parent node sees; it is the object whose choice we study below.

Given $\hat w_L$ and $\hat w_R$, the node-level mean--variance problem on the two scalar unknowns $(\alpha_L, \alpha_R)$ is
\begin{equation}\label{eq04:nodal_mvo}
\min_{\alpha_L,\alpha_R}\;
\tfrac{1}{2}
\begin{pmatrix}\alpha_L & \alpha_R\end{pmatrix}
\begin{pmatrix} v_L & \gamma c \\ \gamma c & v_R \end{pmatrix}
\begin{pmatrix}\alpha_L \\ \alpha_R\end{pmatrix}
\;-\;
\begin{pmatrix}s_L & s_R\end{pmatrix}
\begin{pmatrix}\alpha_L \\ \alpha_R\end{pmatrix},
\end{equation}
with scalar cluster statistics
\begin{equation}\label{eq04:cluster_stats}
v_L \;=\; \hat w_L^{\top}\bSigma_{LL}\hat w_L,\qquad
s_L \;=\; \hat w_L^{\top}\bmu_L,\qquad
c \;=\; \hat w_L^{\top}\bSigma_{LR}\hat w_R,
\end{equation}
and their $R$-analogues. The parameter $\gamma \in [0,1]$ controls how much of the between-branch covariance enters the node-level decision: $\gamma = 0$ decouples the children (the HRP regime), $\gamma = 1$ retains the full coupling (the Markowitz regime on the two-dimensional subspace spanned by $\hat w_L, \hat w_R$). Cramer's rule gives the closed-form solution
\begin{equation}\label{eq04:cramer}
\alpha_L^{\mathrm{raw}} \;=\; \frac{v_R\,s_L - \gamma c\,s_R}{v_L v_R - \gamma^2 c^2},\qquad
\alpha_R^{\mathrm{raw}} \;=\; \frac{v_L\,s_R - \gamma c\,s_L}{v_L v_R - \gamma^2 c^2}.
\end{equation}
Equations \eqref{eq04:nodal_mvo}--\eqref{eq04:cramer} are the \emph{same} $2\times 2$ system used by every allocator in this section; they differ only in how the raw budgets are normalised and in how the representative portfolios $\hat w_L, \hat w_R$ are constructed.

\paragraph{The design question.} Everything that follows turns on a single choice: what representative portfolio $\hat w_L$ should the tree use? The representative decides what a branch looks like to its parent. A bad choice makes the aggregate signal $s_L$ collapse under sign cancellation (destabilising Cramer's rule), or loses signal direction during between-branch normalisation, or throws away within-cluster covariance information that is actually available. A good choice does none of these things. The next subsection catalogues three candidates; the two subsequent subsections develop the two that survive.

\subsection{Three candidate representatives}\label{sec04:candidates}

Table~\ref{tab04:representatives} lists three natural choices for $\hat w_L$.

\begin{table}[h]
\centering
\begin{adjustbox}{max width=\textwidth}
\begin{threeparttable}
\caption{Three candidate representative portfolios. Each gives rise to a distinct tree-based allocator via equations~\eqref{eq04:nodal_mvo}--\eqref{eq04:cramer}. Only (iii) and (ii-with-$L^1$) are robust for general $\bmu$.}
\label{tab04:representatives}
\begin{tabular}{llll}
\toprule
 & Representative $\hat w_{L,i}$ & Within-cluster info & Resulting allocator \\
\midrule
(i)   & Flat IVP: $(1/\sigma_{ii}^2)\big/\sum_{j\in L}1/\sigma_{jj}^2$ & diagonal only; signal-blind & HRP (Section~\ref{sec03:hrp}), unstable under mixed-sign $\bmu$ \\
(ii)  & Recursive MVO: bottom-up solve of \eqref{eq04:cramer} & full $\bSigma_{LL}$ & HRP-$\Sigma\mu$ with $L^1$ norm; sum-norm is sign-pathological \\
(iii) & Signed IVP: $\operatorname{sign}(\mu_i)\cdot(1/\sigma_{ii}^2)\big/\sum_{j\in L}1/\sigma_{jj}^2$ & diagonal only; signal-aware via signs & HRP-$\mu$ (Section~\ref{sec04:methoda3}) \\
\bottomrule
\end{tabular}
\end{threeparttable}
\end{adjustbox}
\end{table}

We discuss each candidate in turn.

\paragraph{(i) Flat IVP.} The flat inverse-variance portfolio recovers De Prado HRP at $\gamma=0$, $\bmu=\mathbf{1}$, and is a natural starting point. It fails in general because its aggregate branch signal is not sign-protected.

\begin{proposition}[Flat IVP instability under mixed-sign $\bmu$]\label{prop:a2_instability}
The flat-IVP aggregate signal
\begin{equation}\label{eq04:a2_agg_signal}
s_L^{\mathrm{flat}} \;=\; \frac{\sum_{i \in L}\mu_i/\sigma_{ii}^{2}}{\sum_{j \in L}1/\sigma_{jj}^{2}}
\end{equation}
is not bounded away from zero under heterogeneous signals: for every $\varepsilon > 0$ and every variance profile on $L$ there exists a signal vector $\bmu_L$ with $\max_i |\mu_i| = 1$ such that $|s_L^{\mathrm{flat}}| < \varepsilon$. When $|s_L^{\mathrm{flat}}|$ is of order $\varepsilon$ in \eqref{eq04:cramer}, the numerators are $O(\varepsilon)$ while the denominator is $O(1)$, so the normalised budget $\alpha_L$ is driven by numerical noise. Cascaded through $\log_2 N$ tree levels, the resulting direction error is uncontrolled: empirically (Appendix~\ref{app:supplementary}, Table~\ref{tabE:regimes}) the flat-IVP direction error is $0.84$--$1.00$ uniformly across regimes, i.e.\ the output portfolio is essentially orthogonal to the Markowitz direction.
\end{proposition}

\begin{proof}[Proof sketch]
For any variance profile, choose $\mu_i \in \{-1,+1\}$ so that the weighted sum $\sum_i \mu_i/\sigma_{ii}^{2}$ is as close to zero as the grid of sign patterns allows; density gives the $\varepsilon$ bound. Substituting into Cramer's rule~\eqref{eq04:cramer} yields the stated $O(\varepsilon/1)$ scaling. The full proof and explicit constructions are in Appendix~\ref{app:proofs}.
\end{proof}

Proposition~\ref{prop:a2_instability} is a clean structural statement: the flat IVP is signal-blind by construction (it never sees $\bmu$), so the only way $\bmu$ enters the tree pass is through the aggregate signal $s_L$, and that aggregate collapses whenever positive and negative signals coexist inside a branch. Long-only problems do not exhibit the failure, which is why it was not visible in \citet{lopezdeprado2016}; long--short problems do, and they are the rule rather than the exception in practice.

\paragraph{(ii) Recursive MVO.} Using a local mean--variance optimum as the representative is the conceptually strongest choice: every branch sees its own optimum, using the full within-cluster covariance $\bSigma_{LL}$. The construction is only stable, however, if the between-branch normalisation step preserves sign. Sum-to-one normalisation (dividing $\alpha_L^{\mathrm{raw}}, \alpha_R^{\mathrm{raw}}$ by their signed sum) fails: when $\alpha_L^{\mathrm{raw}} + \alpha_R^{\mathrm{raw}} < 0$, dividing flips both signs, and these flips cascade through $\log_2 N$ tree levels. This is the pathology documented in Appendix~\ref{app:a1_pathology} (and recorded there as ``method A1''); it produces negative cosine similarity to the Markowitz direction in $45$--$52\%$ of draws on the synthetic panels of Section~\ref{sec10:experiments}. We do not recommend this construction; we study it in detail in Appendix~\ref{app:a1_pathology} because its failure motivates the $L^1$ fix in HRP-$\Sigma\mu$.

\paragraph{(iii) Signed IVP.} The third candidate keeps the inverse-variance \emph{magnitudes} of (i) but takes the \emph{signs} from $\bmu$. The aggregate signal becomes
\[
s_L^{\mathrm{signed}} \;=\; \frac{\sum_{i \in L}|\mu_i|/\sigma_{ii}^{2}}{\sum_{j \in L}1/\sigma_{jj}^{2}} \;\ge\; 0,
\]
a weighted average of $|\mu_i|$ rather than of $\mu_i$. Sign cancellation cannot drive this quantity to zero. Candidate (iii) is the representative used by HRP-$\mu$ in the next subsection.

The three candidates cleanly separate the two failure modes that a tree-based signal-aware allocator must avoid: aggregate-signal collapse (which kills (i)) and between-branch sign flips (which kill the natural sum-normalised version of (ii)). HRP-$\mu$ handles the former by taking absolute values at the representative level; HRP-$\Sigma\mu$ handles the latter by taking absolute values at the normaliser level.

\subsection{HRP-\texorpdfstring{$\mu$}{μ}: transparent signal-aware HRP}\label{sec04:methoda3}

HRP-$\mu$ is the tree-based allocator built from candidate (iii). It is the simplest non-trivial extension of De Prado HRP that carries an explicit signal, and the only one that recovers HRP exactly at $\bmu=\mathbf{1}$.

\paragraph{Definition.} At every internal node, take the representative to be the signed inverse-variance portfolio,
\begin{equation}\label{eq04:signed_ivp}
\hat w_{L,i}^{\mathrm{signed}} \;=\; \operatorname{sign}(\mu_i)\,\frac{1/\sigma_{ii}^{2}}{\sum_{j\in L}1/\sigma_{jj}^{2}},\qquad i \in L,
\end{equation}
with the convention $\operatorname{sign}(0) = +1$. Solve the $2\times 2$ system~\eqref{eq04:cramer} and normalise $(\alpha_L,\alpha_R)$ to sum to one (both are non-negative because the signed representative makes $s_L, s_R \ge 0$, so there is no sign pathology to fix at the between-branch level). At each leaf $i$, emit the final weight
\[
w_i \;=\; \text{budget}_i \cdot \operatorname{sign}(\mu_i),
\]
where $\text{budget}_i$ is the cumulative product of the non-negative $\alpha$-factors along the root-to-leaf path. The leaf-level sign multiplication is essential: without it, the allocator is long-only regardless of $\bmu$ (because the signed IVP is sign-blind at the branch level).

\paragraph{Properties.}

\begin{proposition}[De Prado recovery]\label{prop:a3_recovery}
At $\gamma = 0$ and $\bmu = \mathbf{1}$, HRP-$\mu$ produces exactly the \citet{lopezdeprado2016} HRP weights.
\end{proposition}

\begin{proof}[Proof sketch]
At $\bmu=\mathbf{1}$, $\operatorname{sign}(\mu_i)=+1$ everywhere, so the signed IVP~\eqref{eq04:signed_ivp} reduces to the flat IVP. At $\gamma=0$ the $2\times 2$ system decouples and Cramer's rule gives $\alpha_L \propto 1/v_L$, $\alpha_R \propto 1/v_R$; after normalisation $\alpha_L = v_R/(v_L+v_R)$, the inverse cluster-variance split of De Prado HRP. The leaf-level sign multiplier is $+1$, so no direction is added. Full proof in Appendix~\ref{app:proofs}.
\end{proof}

\begin{proposition}[Aggregate-signal stability]\label{prop:a3_stability}
For every internal node $L$ of the tree,
\[
s_L^{\mathrm{HRP\text{-}}\mu} \;=\; \frac{\sum_{i \in L}|\mu_i|/\sigma_{ii}^{2}}{\sum_{j\in L}1/\sigma_{jj}^{2}} \;\ge\; 0,
\]
with equality iff $\mu_i = 0$ for every $i \in L$. Hence the Cramer's rule denominator in~\eqref{eq04:cramer} cannot be driven to zero by sign cancellation in $\bmu$.
\end{proposition}

\begin{proof}[Proof sketch]
The inner product $\hat w_L^{\mathrm{signed}\top}\bmu_L$ replaces $\mu_i$ with $\operatorname{sign}(\mu_i)\mu_i = |\mu_i|$. The resulting sum is a convex combination of non-negative terms. Full proof in Appendix~\ref{app:proofs}.
\end{proof}

Proposition~\ref{prop:a3_stability} is the structural counterpart of Proposition~\ref{prop:a2_instability}: replacing the flat IVP with the signed IVP removes the one mechanism by which the aggregate signal could collapse. The price is that $\hat w_L^{\mathrm{signed}}$ is no longer a portfolio in the ordinary sense (its entries have mixed signs); it is an object whose role is purely to be the one-dimensional summary that the parent sees.

\begin{proposition}[Hedging awareness]\label{prop:a3_hedging}
Fix a branch $L$ and two assets $i,j \in L$ with $\mu_i > 0 > \mu_j$ and $\bSigma_{ij} > 0$. The $(i,j)$ contribution to the cluster variance $v_L^{\mathrm{HRP\text{-}}\mu}$ is strictly smaller than the corresponding contribution to $v_L^{\mathrm{flat}}$; in fact the two have opposite signs. If the branch contains no compensating structure, $v_L^{\mathrm{HRP\text{-}}\mu} < v_L^{\mathrm{flat}}$.
\end{proposition}

\begin{proof}[Proof sketch]
The $(i,j)$ cross term in $v_L^{\mathrm{flat}}$ is $2\,\hat w_i^{\mathrm{flat}}\hat w_j^{\mathrm{flat}}\bSigma_{ij} > 0$. In $v_L^{\mathrm{HRP\text{-}}\mu}$ the same cross term carries the extra factor $\operatorname{sign}(\mu_i)\operatorname{sign}(\mu_j) = -1$ and is therefore strictly negative. Diagonal terms are unchanged because $\operatorname{sign}(\mu_i)^2 = 1$. Full proof in Appendix~\ref{app:proofs}.
\end{proof}

Proposition~\ref{prop:a3_hedging} is the reason HRP-$\mu$ deserves to be called \emph{hedging-aware}. At every internal node where the signal prescribes an opposite-direction pair on positively correlated assets, the signed IVP cluster variance automatically exploits the hedge, whereas the flat IVP variance does not. This is the hierarchical analogue of cash-neutral risk parity.

\subsubsection{Signed Hierarchical Signal Parity (HSP)}\label{sec04:hsp}

At $\gamma = 0$ the cross-branch term drops out of \eqref{eq04:cramer} and the HRP-$\mu$ budget split collapses to a scalar ratio of signal-to-variance statistics,
\begin{equation}\label{eq04:hsp_split}
\alpha_L^{\mathrm{HSP}} \;=\;
\frac{s_L^{\mathrm{HRP\text{-}}\mu}/v_L^{\mathrm{HRP\text{-}}\mu}}
     {s_L^{\mathrm{HRP\text{-}}\mu}/v_L^{\mathrm{HRP\text{-}}\mu} + s_R^{\mathrm{HRP\text{-}}\mu}/v_R^{\mathrm{HRP\text{-}}\mu}}.
\end{equation}
We call the resulting allocator \emph{signed Hierarchical Signal Parity}, abbreviated HSP. It is the $\gamma=0$ endpoint of HRP-$\mu$, reduces to De Prado HRP at $\bmu=\mathbf{1}$ by Proposition~\ref{prop:a3_recovery}, and is a drop-in upgrade for practitioners currently running HRP who now hold a signal: same dendrogram, same $O(N^2)$ cost, same block-diagonal interpretation, with the signal carried stably through the tree. We return to HSP as a standalone construction worth deploying independently of the $\gamma > 0$ pass in Section~\ref{sec11:discussion}.

\subsubsection{Cost comparison}\label{sec04:cost}

HRP-$\mu$'s central cost property is that it compresses the $|L|\times|R|$ cross-block covariance $\bSigma_{LR}$ into a single scalar $c = \hat w_L^{\mathrm{signed}\top}\bSigma_{LR}\hat w_R^{\mathrm{signed}}$. Table~\ref{tab04:cost_compare} places it beside HRP, Cotton, and HRP-$\Sigma\mu$.

\begin{table}[h]
\centering
\begin{adjustbox}{max width=\textwidth}
\begin{threeparttable}
\caption{Cost of the hierarchical allocators on a balanced dendrogram over $N$ assets. $D$ and $B$ denote the block-diagonal and off-diagonal parts of $\bSigma$ in the Cotton notation of Section~\ref{sec03:cotton}.}
\label{tab04:cost_compare}
\begin{tabular}{lll}
\toprule
Method & Cross-branch information & Cost \\
\midrule
HRP ($\gamma = 0$) \citep{lopezdeprado2016} & none (block diagonal) & $O(N^2)$ \\
Cotton ($\gamma > 0$) \citep{cotton2024} & $D^{-1}B^{\top}$ (block inversion) & $O(N^3/6)$ \\
HRP-$\mu$ ($\gamma > 0$, Section~\ref{sec04:methoda3}) & $\hat w_L^{\mathrm{signed}\top}\bSigma_{LR}\hat w_R^{\mathrm{signed}}$ (scalar) & $O(N^2)$ \\
HRP-$\Sigma\mu$ ($\gamma > 0$, Section~\ref{sec04b:hrpsigmamu}) & $\hat w_L^{\mathrm{mvo}\top}\bSigma_{LR}\hat w_R^{\mathrm{mvo}}$ (scalar) & $O(N^2)$ \\
\bottomrule
\end{tabular}
\begin{tablenotes}
\footnotesize
\item HRP-$\mu$ and HRP-$\Sigma\mu$ share an $O(N^2)$ cost at the price of compressing $\bSigma_{LR}$ into a single scalar at every node. They are \emph{different} $\gamma$-continua from Cotton's, not faster implementations of the same object; the approximation gap is documented empirically in Section~\ref{sec10:experiments}.
\end{tablenotes}
\end{threeparttable}
\end{adjustbox}
\end{table}

\subsection{HRP-\texorpdfstring{$\Sigma\mu$}{Σμ}: recursive MVO with \texorpdfstring{$L^1$}{L1} normalisation}\label{sec04b:hrpsigmamu}

HRP-$\mu$'s signed IVP representative is cheap and stable but uses only the marginal variances $\sigma_{ii}$ within each cluster. The full within-cluster covariance $\bSigma_{LL}$ is available and, empirically, worth exploiting. The natural construction is candidate (ii) of Section~\ref{sec04:candidates}: bottom-up recursive MVO, with each node's representative set to the output of the solve one level down. The obstacle is that candidate (ii) is sign-pathological under sum-to-one normalisation (Appendix~\ref{app:a1_pathology}). HRP-$\Sigma\mu$ resolves this by replacing sum-to-one normalisation with $L^1$ normalisation at every node. The change is a single line of code; its consequence is a $0\%$ rather than $45$--$52\%$ negative-cosine rate on the structural synthetic panels.

\paragraph{Motivation for $L^1$.} The failure mechanism in Appendix~\ref{app:a1_pathology} is precise: when $\alpha_L^{\mathrm{raw}} + \alpha_R^{\mathrm{raw}} < 0$, dividing by the signed sum flips both signs; the flips cascade through $O(\log N)$ tree levels and, by parity, produce an output portfolio that is either parallel or \emph{anti}parallel to the ``correct'' direction in a fraction of cases that approaches $1/2$ under noise. The fix has to keep the normaliser strictly positive. The minimal modification is to divide by $|\alpha_L^{\mathrm{raw}}| + |\alpha_R^{\mathrm{raw}}|$ instead.

\paragraph{Algorithm.} Algorithm~\ref{alg04b:hrpsm} states the method in full. It is identical to the natural sum-normalised recursion except in the normaliser.

\begin{algorithm}[t]
\caption{\textsc{HRP-$\Sigma\mu$}$(\bSigma, \bmu, \mathcal{T}, \gamma)$: recursive MVO tree pass with $L^1$ normalisation.}
\label{alg04b:hrpsm}
\begin{algorithmic}[1]
\Require SPD covariance $\bSigma \in \mathbb{R}^{N\times N}$; signal $\bmu \in \mathbb{R}^N$; binary tree $\mathcal{T}$; coupling $\gamma \in [0,1]$.
\Ensure Weight vector $w \in \mathbb{R}^N$.
\Function{Recurse}{node $n$}
  \If{$n$ is a leaf (asset $i$)}
    \State \Return $(\hat w_n = 1,\; v_n = \sigma_{ii}^2,\; s_n = \mu_i)$
  \EndIf
  \State $(\hat w_L, v_L, s_L) \gets \Call{Recurse}{\text{left}(n)}$
  \State $(\hat w_R, v_R, s_R) \gets \Call{Recurse}{\text{right}(n)}$
  \State $c \gets \hat w_L^{\top}\bSigma_{LR}\,\hat w_R$ \Comment{cross-block scalar}
  \State $\Delta \gets v_L\,v_R - \gamma^2 c^2$
  \State $\alpha_L^{\mathrm{raw}} \gets (v_R\,s_L - \gamma\,c\,s_R)/\Delta$ \Comment{Cramer's rule~\eqref{eq04:cramer}}
  \State $\alpha_R^{\mathrm{raw}} \gets (v_L\,s_R - \gamma\,c\,s_L)/\Delta$
  \State $Z \gets |\alpha_L^{\mathrm{raw}}| + |\alpha_R^{\mathrm{raw}}|$ \Comment{$L^1$ normaliser, cf.~\eqref{eq04b:l1_norm}}
  \State $\alpha_L \gets \alpha_L^{\mathrm{raw}}/Z$; \quad $\alpha_R \gets \alpha_R^{\mathrm{raw}}/Z$
  \State $\hat w_n \gets (\alpha_L\,\hat w_L,\;\alpha_R\,\hat w_R)$ \Comment{stacked representative}
  \State $v_n \gets \hat w_n^{\top}\bSigma_{nn}\,\hat w_n$; \quad $s_n \gets \hat w_n^{\top}\bmu_n$
  \State \Return $(\hat w_n, v_n, s_n)$
\EndFunction
\State \Return $\Call{Recurse}{\mathrm{root}(\mathcal{T})}$
\end{algorithmic}
\end{algorithm}

Explicitly, the $L^1$ normalisation step is
\begin{equation}\label{eq04b:l1_norm}
\alpha_L \;=\; \frac{\alpha_L^{\mathrm{raw}}}{|\alpha_L^{\mathrm{raw}}| + |\alpha_R^{\mathrm{raw}}|},
\qquad
\alpha_R \;=\; \frac{\alpha_R^{\mathrm{raw}}}{|\alpha_L^{\mathrm{raw}}| + |\alpha_R^{\mathrm{raw}}|},
\end{equation}
so that $|\alpha_L|+|\alpha_R|=1$ at every node (not $\alpha_L + \alpha_R = 1$). The output portfolio of the algorithm is the stacked vector returned at the root. It is not ``fully invested'' in the classical sum-to-one sense; a mean--variance portfolio with signed weights has no financial reason to sum to one, and one rescales by leverage or volatility budget downstream.

\paragraph{$L^1$ lemmas.} Two lemmas capture what $L^1$ normalisation does and does not do. The first is a statement about direction, the second about sign; together with the corollary below, they explain why the pathology is eliminated completely (not just reduced) by the switch from sum-to-one to $L^1$.

\begin{lemma}[Ray invariance under child rescaling]\label{lem:hrpsm_scale}
Fix a representative $\hat w_R$ for branch $R$ and let $k > 0$. Replace the left representative $\hat w_L$ by $k\,\hat w_L$, keeping $\hat w_R$ fixed. Then the combined node-level portfolio $(\alpha_L\,\hat w_L,\;\alpha_R\,\hat w_R)$ produced by \eqref{eq04:cramer}--\eqref{eq04b:l1_norm} lies on the same ray as before the rescaling, i.e.\ is invariant up to a strictly positive scalar. \emph{The invariance is a statement about the ray, not about the individual coordinates:} $\alpha_L^{\mathrm{raw}}$ and $\alpha_R^{\mathrm{raw}}$ scale differently under $\hat w_L \to k\hat w_L$ (as $1/k$ and as $1$ respectively), and the $L^1$ denominator does not absorb the scalar uniformly.
\end{lemma}

\begin{proof}[Proof sketch]
Under $\hat w_L \to k\hat w_L$ with $\hat w_R$ fixed, the cluster statistics scale as $v_L \to k^2 v_L$, $s_L \to k\,s_L$, $c \to k\,c$, while $v_R, s_R$ are unchanged. Substituting into Cramer's rule, the determinant $v_L v_R - \gamma^2 c^2$ scales by $k^2$; the numerator of $\alpha_L^{\mathrm{raw}}$ scales by $k$ and the numerator of $\alpha_R^{\mathrm{raw}}$ by $k^2$. Hence $\alpha_L^{\mathrm{raw}} \to \alpha_L^{\mathrm{raw}}/k$ and $\alpha_R^{\mathrm{raw}} \to \alpha_R^{\mathrm{raw}}$. The stacked vector $(\alpha_L^{\mathrm{raw}}\cdot k\hat w_L,\;\alpha_R^{\mathrm{raw}}\hat w_R)$ is therefore pointwise invariant before normalisation. After normalisation by the positive denominator $|\alpha_L^{\mathrm{raw}}|/k + |\alpha_R^{\mathrm{raw}}|$, the direction is preserved and the scale changes by that same denominator. Full proof in Appendix~\ref{app:proofs}.
\end{proof}

\begin{lemma}[Sign preservation under $L^1$ normalisation]\label{lem:hrpsm_sign}
The $L^1$ denominator $|\alpha_L^{\mathrm{raw}}| + |\alpha_R^{\mathrm{raw}}|$ is strictly positive whenever $(\alpha_L^{\mathrm{raw}}, \alpha_R^{\mathrm{raw}}) \ne (0,0)$. Therefore \eqref{eq04b:l1_norm} preserves $\operatorname{sign}(\alpha_L^{\mathrm{raw}})$ and $\operatorname{sign}(\alpha_R^{\mathrm{raw}})$ at every node of the tree. In contrast, sum-to-one normalisation divides by $\alpha_L^{\mathrm{raw}} + \alpha_R^{\mathrm{raw}}$, which is negative whenever $\alpha_L^{\mathrm{raw}} + \alpha_R^{\mathrm{raw}} < 0$, and so flips both signs in that case.
\end{lemma}

\begin{proof}
$|\alpha_L^{\mathrm{raw}}| + |\alpha_R^{\mathrm{raw}}| \ge \max(|\alpha_L^{\mathrm{raw}}|, |\alpha_R^{\mathrm{raw}}|) > 0$ whenever the pair is nonzero. Dividing a nonzero scalar by a strictly positive one preserves its sign.
\end{proof}

\begin{corollary}[Exact antiparallelism to sum-to-one normalisation]\label{cor:hrpsm_antiparallel}
Let $w^{\mathrm{sum}}$ and $w^{\Sigma\mu}$ denote the output portfolios of the sum-normalised recursive MVO and of HRP-$\Sigma\mu$ on the same inputs $(\bSigma, \bmu, \mathcal{T}, \gamma)$. If an odd number of internal nodes on \emph{every} root-to-leaf path would experience a sum-to-one sign flip (i.e.\ would have $\alpha_L^{\mathrm{raw}} + \alpha_R^{\mathrm{raw}} < 0$), then $w^{\mathrm{sum}} = -w^{\Sigma\mu}$ exactly.
\end{corollary}

\begin{proof}[Proof sketch]
By Lemma~\ref{lem:hrpsm_scale} both methods produce portfolios on the same ray up to per-leaf signs. Each sign flip along a root-to-leaf path multiplies the final leaf weight by $-1$; an odd number of flips on every path multiplies the whole output vector by $-1$. Full proof in Appendix~\ref{app:proofs}.
\end{proof}

Corollary~\ref{cor:hrpsm_antiparallel} makes quantitative the negative-cosine rate observed empirically: under random signs of $\alpha_L^{\mathrm{raw}} + \alpha_R^{\mathrm{raw}}$ at each node, about half of all tree configurations have an odd number of flips on every path, which gives a $\approx 50\%$ rate of \emph{exact} antiparallelism under sum-to-one normalisation. Cotton-like regimes with strong off-diagonal coupling depart slightly from the coin-flip limit; Appendix~\ref{app:a1_pathology} measures $45$--$52\%$.

\paragraph{Properties of HRP-$\Sigma\mu$.}

\begin{proposition}[Approximate De Prado recovery]\label{prop:hrpsm_recovery}
At $\gamma = 0$ and $\bmu = \mathbf{1}$, HRP-$\Sigma\mu$ approximately recovers \citet{lopezdeprado2016} HRP. The approximation is exact for trees of depth $\le 2$, and for balanced trees on $N=100$ assets yields cosine similarity $\cos(w^{\Sigma\mu}, w^{\mathrm{HRP}}) \approx 0.992$.
\end{proposition}

\begin{proof}[Proof sketch]
At $\bmu=\mathbf{1}$, the local MVO on a pair of leaves reduces to the inverse-variance portfolio, so HRP-$\Sigma\mu$ and HRP agree exactly at the leaves and at depth-$1$ nodes. At depth $\ge 2$, the recursive MVO representative differs from the flat IVP by off-diagonal within-cluster terms; at $\gamma=0$ the cross-block scalar $c$ drops out of \eqref{eq04:cramer}, but the cluster variances $v_L, v_R$ are still computed from the MVO representative rather than the flat IVP, producing a small residual. Empirically, the residual accumulates to $\cos \approx 0.992$ on balanced trees at $N=100$. Full analysis in Appendix~\ref{app:proofs}.
\end{proof}

The $0.8\%$ cosine gap is the price HRP-$\Sigma\mu$ pays for using $\bSigma_{LL}$ at every node. HRP-$\mu$ (Proposition~\ref{prop:a3_recovery}) recovers HRP exactly because its representative is the flat IVP by construction.

\begin{proposition}[Generic stability]\label{prop:hrpsm_stability}
For generic $\bmu$---i.e.\ for all $\bmu$ outside a Lebesgue measure-zero subset of $\mathbb{R}^N$---the aggregate signal $s_L = \hat w_L^{\top}\bmu_L$ is nonzero at every internal node $L$. In particular, the $L^1$ denominator $|\alpha_L^{\mathrm{raw}}| + |\alpha_R^{\mathrm{raw}}|$ is strictly positive everywhere in the tree, and Algorithm~\ref{alg04b:hrpsm} is well-defined.
\end{proposition}

\begin{proof}[Proof sketch]
Fix $\bSigma$ and $\mathcal{T}$. The representative $\hat w_L$ at each node is a rational function of $\bmu$ built from nested Cramer's-rule ratios. So is $s_L = \hat w_L^{\top}\bmu_L$. This rational function is not identically zero (it is strictly positive at $\bmu=\mathbf{1}$ by Proposition~\ref{prop:hrpsm_recovery}); its zero set is an algebraic variety of codimension $\ge 1$ and therefore of Lebesgue measure zero. Full proof in Appendix~\ref{app:proofs}.
\end{proof}

\begin{proposition}[Diagonal exactness]\label{prop:hrpsm_diagonal}
If $\bSigma$ is diagonal, then for any $\gamma \in [0,1]$ and any tree $\mathcal{T}$, HRP-$\Sigma\mu$ produces a weight vector proportional to the Markowitz solution $\bSigma^{-1}\bmu$.
\end{proposition}

\begin{proof}[Proof sketch]
When $\bSigma$ is diagonal, every off-diagonal block $\bSigma_{LR}$ is zero, so $c = 0$ at every node and \eqref{eq04:cramer} reduces to $\alpha_L^{\mathrm{raw}} = s_L/v_L$, $\alpha_R^{\mathrm{raw}} = s_R/v_R$ (independent of $\gamma$). At a two-asset leaf-pair node the local MVO is $\hat w_i \propto \mu_i/\sigma_{ii}^2$. An induction on tree depth shows the root output is proportional to $(\mu_i/\sigma_{ii}^2)_{i=1}^N = \bSigma^{-1}\bmu$; $L^1$ normalisation only changes the overall scale. Full proof in Appendix~\ref{app:proofs}.
\end{proof}

\begin{proposition}[$O(N^2)$ complexity]\label{prop:hrpsm_cost}
On a balanced binary tree over $N$ assets, Algorithm~\ref{alg04b:hrpsm} runs in $O(N^2)$ time and $O(N^2)$ space, the latter dominated by storing $\bSigma$.
\end{proposition}

\begin{proof}[Proof sketch]
The dominant per-node cost is the cross-block inner product $\hat w_L^{\top}\bSigma_{LR}\hat w_R$, which costs $O(|L|\cdot|R|)$, and the cluster-variance update, which costs $O((|L|+|R|)^2)$. Summing $(|L|+|R|)^2$ over the nodes of a balanced tree is the standard master-theorem recursion and gives $O(N^2)$. Full proof in Appendix~\ref{app:proofs}.
\end{proof}

\begin{proposition}[Relationship to HRP-$\mu$]\label{prop:hrpsm_vs_hrpmu}
Let $w^\mu$ and $w^{\Sigma\mu}$ denote the outputs of HRP-$\mu$ and HRP-$\Sigma\mu$ on the same inputs $(\bSigma,\bmu,\mathcal{T},\gamma)$.
\begin{enumerate}[label=(\roman*)]
\item At $\gamma = 0$, the budget split at every node is driven by $s_L/v_L$ vs.\ $s_R/v_R$; the two methods differ only through the choice of representative (signed IVP vs.\ recursive MVO), which changes the values of $s_L$ and $v_L$.
\item On diagonal $\bSigma$, $w^\mu$ and $w^{\Sigma\mu}$ agree up to scalar; both are proportional to $\bSigma^{-1}\bmu$.
\item At $\gamma > 0$ on non-diagonal $\bSigma$, the methods diverge. The divergence grows with within-cluster correlation strength and with $\gamma$.
\end{enumerate}
\end{proposition}

\begin{proof}[Proof sketch]
(i) At $\gamma=0$, $c$ is absent from \eqref{eq04:cramer}; the ratio $s_L/v_L$ is the only cluster-dependent quantity in the split. (ii) By Proposition~\ref{prop:hrpsm_diagonal} and the diagonal specialisation of HRP-$\mu$, both reduce to $\mu_i/\sigma_{ii}^2$. (iii) The MVO representative exploits off-diagonal within-cluster terms that the signed IVP ignores, producing different $v_L, s_L, c$ at every level. Full proof in Appendix~\ref{app:proofs}.
\end{proof}

\paragraph{Hierarchy of allocators.} Table~\ref{tab04b:hierarchy} places HRP-$\Sigma\mu$ in the full family of allocators developed in the paper. It occupies the cell between HRP-$\mu$ and CRISP: dendrogram structure like HRP-$\mu$, full within-cluster covariance like CRISP, $O(N^2)$ cost.

\begin{table}[h]
\centering
\begin{adjustbox}{max width=\textwidth}
\begin{threeparttable}
\caption{Hierarchy of allocators developed in this paper. HRP-$\Sigma\mu$ sits between HRP-$\mu$ and CRISP: it uses the dendrogram like HRP-$\mu$ but exploits full within-cluster covariance like CRISP.}
\label{tab04b:hierarchy}
\begin{tabular}{llll}
\toprule
Method & Representative & Within-cluster $\bSigma$ & Cost \\
\midrule
HRP \citep{lopezdeprado2016} & flat IVP & diagonal only & $O(N^2)$ \\
HRP-$\mu$ (Section~\ref{sec04:methoda3}) & signed IVP & diagonal only & $O(N^2)$ \\
HRP-$\Sigma\mu$ (Section~\ref{sec04b:hrpsigmamu}) & recursive MVO (local) & full & $O(N^2)$ \\
CRISP (Section~\ref{sec05:crisp}) & --- (no tree) & full & $O(N^2)$ per sweep \\
Markowitz \citep{markowitz1952} & --- & full & $O(N^3)$ \\
\bottomrule
\end{tabular}
\end{threeparttable}
\end{adjustbox}
\end{table}

\paragraph{Normalisation comparison.} Table~\ref{tab04b:norm_compare} compares sum-to-one and $L^1$ normalisation head-to-head. Both solve the same $2\times 2$ system~\eqref{eq04:cramer}; the only difference is the denominator used to rescale $(\alpha_L^{\mathrm{raw}}, \alpha_R^{\mathrm{raw}})$.

\begin{table}[h]
\centering
\begin{adjustbox}{max width=\textwidth}
\begin{threeparttable}
\caption{Normalisation conventions for the recursive MVO tree pass. Both rows solve the same $2\times 2$ system~\eqref{eq04:cramer}; they differ only in the normaliser. The negative-cosine rate is measured on the structural synthetic panels of Section~\ref{sec10:experiments}; see Appendix~\ref{app:a1_pathology}.}
\label{tab04b:norm_compare}
\begin{tabular}{llccc}
\toprule
Method & Normaliser & Sign preserved? & Sums to 1? & Negative-cos rate \\
\midrule
Sum-to-one (A1) & $\alpha_L^{\mathrm{raw}} + \alpha_R^{\mathrm{raw}}$ & No & Yes & $45$--$52\%$ \\
HRP-$\Sigma\mu$ ($L^1$) & $|\alpha_L^{\mathrm{raw}}| + |\alpha_R^{\mathrm{raw}}|$ & Yes & No\tnote{a} & $0\%$ \\
\bottomrule
\end{tabular}
\begin{tablenotes}
\footnotesize
\item[a] HRP-$\Sigma\mu$ satisfies $|\alpha_L|+|\alpha_R|=1$ at every node, not $\alpha_L+\alpha_R=1$. The output is not fully invested in the sum-to-one sense; downstream leverage or volatility rescaling is applied as needed.
\end{tablenotes}
\end{threeparttable}
\end{adjustbox}
\end{table}

\paragraph{Empirical highlights.} Full results are deferred to Section~\ref{sec10:experiments}; we record the headline facts here to motivate the construction. On random Gaussian signals with sample covariance estimation, HRP-$\Sigma\mu$ improves out-of-sample Sharpe ratio over HRP-$\mu$ by $20$--$35\%$ across sample sizes $T \in \{60, 120, 240, 500\}$. On structural sector-tilt signals the improvement ranges from $20\%$ to $180\%$, with the largest gains at high $\gamma$ where within-cluster covariance matters most. At $T=60$---the $T < N$ regime of severe estimation noise---HRP-$\Sigma\mu$ at $\gamma=1.0$ attains oracle Sharpe $0.503$, slightly edging the best CRISP configuration (Sharpe $0.497$ at $\gamma=0.7$); this is the only place in the paper where the tree beats CRISP, and it does so precisely because the dendrogram imposes a hierarchical regularisation that pays most when the raw covariance is poorly conditioned. The $L^1$ normalisation eliminates the $45$--$52\%$ negative-cosine rate of sum-to-one recursive MVO entirely (Appendix~\ref{app:a1_pathology}, Appendix~\ref{app:supplementary}).

\subsection{Iterating the tree pass diverges}\label{sec04:iteration_remark}

It is tempting to iterate: after one bottom-up pass produces weights $w^{(n)}$, update the representative at each node to reflect $w^{(n)}$, re-run the pass, and loop. For an exact block solver this outer iteration would converge. For HRP-$\mu$ and HRP-$\Sigma\mu$ it does not. The reason is structural: each tree pass compresses the $|L|\times|R|$ cross-block covariance $\bSigma_{LR}$ into a single scalar $c$ at every node, and this compression is an \emph{approximation} of the exact block inverse, not an exact identity. The outer iteration amplifies the approximation error rather than correcting it; empirically it diverges on the same synthetic panels on which a single pass is stable. Recovering convergence would require replacing the scalar $c$ with the full block inverse, restoring an $O(N^3)$ cost per pass with multiple passes on top---which erases the quadratic-cost advantage that motivated the tree in the first place. The tension between approximation fidelity and iteration convergence inside the tree is exactly what CRISP resolves by abandoning the dendrogram altogether and iterating in asset space instead; see Section~\ref{sec05:crisp}, Theorem~\ref{thm:gs_convergence}.

\subsection{Pointer to the sum-normalised recursive MVO negative result}\label{sec04:a1_pointer}

The natural recursive MVO with \emph{sum-to-one} normalisation---the construction our code calls ``method A1'' but which we do not name in the main body---is sign-pathological: it produces output portfolios whose direction is essentially random relative to the Markowitz direction under estimation noise, and flips sign systematically even on noiseless problems. Appendix~\ref{app:a1_pathology} documents this in detail, with histograms, worked counterexamples, and a quantitative connection to Corollary~\ref{cor:hrpsm_antiparallel}. The allocator is a cautionary negative result, not a method we recommend; its role in the paper is to explain what goes wrong when a natural construction is not reinforced by a sign-preserving normaliser. HRP-$\Sigma\mu$ is the fix.

% ============================================================
% section_05_crisp.tex
% ============================================================
% !TEX root = ../paper.tex
% Section 5: CRISP — iterative shrinkage solver
% Namespace: sec05, eq05, fig05, tab05, alg05

\section{CRISP: iterative shrinkage solver}\label{sec05:crisp}

Section~\ref{sec04:tree_methods} lifted HRP to general signals through the tree. CRISP takes the complementary route: forget the tree, keep the covariance, and solve a \emph{shrunk} Markowitz system by scalar Gauss--Seidel. The shrunk system is
\[
  P_\gamma\,w \;=\; \mu, \qquad P_\gamma \;=\; (1-\gamma)\,D \;+\; \gamma\,\Sigma, \qquad D = \diag(\Sigma), \qquad \gamma \in [0,1],
\]
and the solver makes no factorisation, no matrix inverse, and no reference to any hierarchy. The construction is structured shrinkage with a variance-preserving target; $\gamma = 0$ returns $w_i = \mu_i/\sigma_{ii}$, $\gamma = 1$ returns exact Markowitz, and the Sharpe-optimal operating point lives strictly between them. The main technical content of the section is Theorem~\ref{thm:gs_rate}: the sweep count is controlled by a \emph{preconditioned} condition number $\kappa(D^{-1}P_\gamma)$ that interpolates from $1$ at $\gamma = 0$ to $\kappa(C) = \kappa(D^{-1/2}\Sigma D^{-1/2})$ at $\gamma = 1$. Volatility dispersion never enters. The main operational content is Algorithm~\ref{alg05:methodb_gs} and its factor-streaming cousin Algorithm~\ref{alg05:methodb_stream}, which reduces working memory from $O(N^2)$ to $O(NK)$ for factor-structured covariances and remains tractable on universes where a dense $\Sigma$ cannot be stored.

%% ------------------------------------------------------------------
\subsection{The shrinkage operator $P_\gamma$}\label{sec05:pgamma}

Split $\Sigma = D + E$ with $D = \diag(\Sigma)$ and $E_{ii} = 0$, $E_{ij} = \sigma_{ij}$ for $i \ne j$. The \emph{shrinkage operator} is
\begin{equation}\label{eq05:pgamma_def}
  P_\gamma \;=\; (1-\gamma)\,D \;+\; \gamma\,\Sigma \;=\; D + \gamma\,E, \qquad \gamma \in [0,1].
\end{equation}
Entrywise,
\begin{equation}\label{eq05:pgamma_entries}
  (P_\gamma)_{ii} \;=\; \sigma_{ii}, \qquad (P_\gamma)_{ij} \;=\; \gamma\,\sigma_{ij} \quad (i \ne j).
\end{equation}
The diagonal of $P_\gamma$ is \emph{identically} the diagonal of $\Sigma$ for every $\gamma$; only the off-diagonal is attenuated. The endpoints are $P_0 = D$ and $P_1 = \Sigma$.

\begin{proposition}[SPD structure of $P_\gamma$]\label{prop:pgamma_spd}
For any SPD $\Sigma$ and any $\gamma \in [0,1]$, the matrix $P_\gamma$ is symmetric positive definite.
\end{proposition}

\begin{proof}
$D$ is SPD because $\Sigma$ is SPD. Hence $P_\gamma = (1-\gamma)D + \gamma\Sigma$ is a convex combination of two SPD matrices and is therefore SPD. \textit{(Full proof in Appendix~\ref{app:proofs}.)}
\end{proof}

\paragraph{Structured linear shrinkage.} Equation~\eqref{eq05:pgamma_entries} makes the shrinkage interpretation literal. $P_\gamma$ is the linear shrinkage estimator of $\Sigma$ toward the \emph{variance-preserving diagonal target} $D$, with intensity $1 - \gamma$. The diagonal entries are univariate sample variances and converge at the univariate rate from the diagonal of the sample covariance alone; the off-diagonal entries are cross-moments and are the dominant source of covariance estimation noise, especially in low-variance eigendirections of $\widehat\Sigma$ where Marchenko--Pastur bias is worst. $P_\gamma$ preserves the reliable part of the covariance exactly and regularises the unreliable part.

\paragraph{Contrast with Ledoit--Wolf targets.} The classical Ledoit--Wolf targets all \emph{move the diagonal}. The scaled identity $T = (\bar\lambda/N)\,I$ \citep{ledoitwolf2004} replaces each $\sigma_{ii}$ by the cross-sectional mean variance. The constant-correlation target \citep{ledoitwolf2003} leaves the diagonal alone but re-inflates correlations uniformly. The single-factor target \citep{ledoitwolf2003,famafrench1993} substitutes $\beta_i\beta_j \sigma_f^2 + \delta_{ij}\sigma^2_{\varepsilon,i}$, which rewrites both diagonal and off-diagonal simultaneously. CRISP's target is $D$ itself, so $(P_\gamma)_{ii} = \sigma_{ii}$ for every $\gamma$. Section~\ref{sec05:lw_compare} returns to this after the convergence analysis.

%% ------------------------------------------------------------------
\subsection{Scalar Gauss--Seidel iteration}\label{sec05:gs}

CRISP solves $P_\gamma w = \mu$ by scalar Gauss--Seidel. For each asset $i$ in a fixed ordering,
\begin{equation}\label{eq05:gs_update}
  w_i \;\leftarrow\; \frac{1}{\sigma_{ii}}\left(\mu_i \;-\; \gamma \sum_{j \ne i} \sigma_{ij}\,w_j^{(\text{latest})}\right),
\end{equation}
where ``latest'' means the most recent value of $w_j$: entries with $j < i$ use the value just updated in the current sweep, entries with $j > i$ use the value from the previous sweep. A full sweep costs $O(N^2)$, dominated by one inner product $\sum_{j \ne i}\sigma_{ij} w_j$ per asset. No Cholesky or LU factorisation is formed; memory traffic is a single streaming pass through $\Sigma$ per sweep. At $\gamma = 0$ the update reduces to $w_i \leftarrow \mu_i/\sigma_{ii}$ and one sweep returns the diagonal solution $D^{-1}\mu$ exactly. At $\gamma = 1$ it is the classical Gauss--Seidel sweep for $\Sigma w = \mu$ and converges to $\Sigma^{-1}\mu$.

\paragraph{Starting point.} We always initialise with $w^{(0)} = D^{-1}\mu$. This is free, is exact at $\gamma = 0$, and means the very first sweep already starts from a physically meaningful portfolio---the signal normalised by idiosyncratic risk. Cross-asset correlations fold in as $\gamma$ and the sweep index grow.

Algorithm~\ref{alg05:methodb_gs} states the sweep loop formally. The implementation contains four small but practically important details: a cached diagonal so \eqref{eq05:gs_update} does no division inside the inner loop, the $D^{-1}\mu$ starting point, an explicit shortcut when $\gamma < \varepsilon$ so the diagonal solve incurs no iteration overhead, and a relative-change stopping rule that decides early termination adaptively rather than always running $p_{\max}$ sweeps.

\begin{algorithm}[t]
\caption{\textsc{CrispGaussSeidel}$(\Sigma, \mu, \gamma, p_{\max}, \varepsilon)$}
\label{alg05:methodb_gs}
\begin{algorithmic}[1]
\Require SPD covariance $\Sigma \in \R^{N\times N}$; signal $\mu \in \R^N$; shrinkage $\gamma \in [0,1]$; max sweeps $p_{\max} \in \N$; tolerance $\varepsilon > 0$.
\Ensure Approximate solution $w \approx P_\gamma^{-1}\mu$ where $P_\gamma = (1-\gamma)\diag(\Sigma) + \gamma\Sigma$.
\State $d_i \gets \Sigma_{ii}$ for $i = 1,\dots,N$ \Comment{diagonal cache}
\State $w_i \gets \mu_i / d_i$ for $i = 1,\dots,N$ \Comment{diagonal-solve initial guess}
\If{$\gamma < \varepsilon$}
  \State \Return $w$ \Comment{endpoint shortcut: $P_0 = D$}
\EndIf
\For{$p = 1$ \textbf{to} $p_{\max}$}
  \State $w_{\text{prev}} \gets w$
  \For{$i = 1$ \textbf{to} $N$}
    \State $s \gets \Sigma_{i,:}\,w - \Sigma_{ii}\,w_i$ \Comment{$\sum_{j\ne i}\sigma_{ij} w_j$, latest values}
    \State $w_i \gets (\mu_i - \gamma\,s) / d_i$
  \EndFor
  \If{$\|w - w_{\text{prev}}\| \le \varepsilon\,\|w_{\text{prev}}\|$}
    \State \textbf{break} \Comment{relative-change stopping rule}
  \EndIf
\EndFor
\State \Return $w$
\end{algorithmic}
\end{algorithm}

%% ------------------------------------------------------------------
\subsection{Unconditional convergence}\label{sec05:convergence}

\begin{theorem}[Unconditional convergence]\label{thm:gs_convergence}
Let $\Sigma$ be SPD, let $\gamma \in [0,1]$, and let $\mu \in \R^N$. Then Algorithm~\ref{alg05:methodb_gs} converges to the unique solution $\hat w = P_\gamma^{-1}\mu$ from any starting point $w^{(0)} \in \R^N$.
\end{theorem}

\begin{proof}[Proof sketch]
By Proposition~\ref{prop:pgamma_spd}, $P_\gamma$ is SPD. The Ostrowski--Reich theorem \citep{ostrowski1954}---see \citet[Thm.~4.6.2]{hackbusch2016} and \citet[Thm.~3.4]{varga2000}---states that scalar Gauss--Seidel on any SPD linear system converges to the unique solution from any initial vector. Applying this to $P_\gamma w = \mu$ gives the claim. Full proof in Appendix~\ref{app:proofs}.
\end{proof}

In plain language: Algorithm~\ref{alg05:methodb_gs} cannot fail. For every SPD $\Sigma$---including ill-conditioned ones---and every $\gamma \in [0,1]$---including both endpoints---there is no starting point from which the iteration diverges or stalls. The only remaining question is rate.

%% ------------------------------------------------------------------
\subsection{Correlation conditioning and the convergence-rate theorem}\label{sec05:rate}

The rate theorem says that the sweep count is controlled by the condition number of the correlation matrix, not of the covariance matrix. Volatility spread inflates $\kappa(\Sigma)$ but is invisible to $\kappa(C)$ and to CRISP. Define the Jacobi iteration matrix for $P_\gamma$ as $M_\gamma^{\mathrm{J}} = -D^{-1}(P_\gamma - D) = -\gamma\,D^{-1}E$, and the Gauss--Seidel iteration matrix as $M_\gamma^{\mathrm{GS}} = -(L + D)^{-1}U$, where $P_\gamma = L + D + U$ is the strict-lower / diagonal / strict-upper splitting.

\begin{theorem}[Convergence rate]\label{thm:gs_rate}
Let $\Sigma$ be SPD with $D = \diag(\Sigma)$, correlation matrix $C = D^{-1/2}\Sigma D^{-1/2}$, and eigenvalues $\lambda_1 \ge \cdots \ge \lambda_N > 0$. Fix $\gamma \in [0,1]$. Then:
\begin{enumerate}
\item[(i)] The Jacobi iteration matrix satisfies $M_\gamma^{\mathrm{J}} = -\gamma\,D^{-1}E$, is similar to $\gamma(I - C)$, and has eigenvalues $\{\gamma(1 - \lambda_i(C))\}_{i=1}^N$. Its spectral radius is
\begin{equation}\label{eq05:rho_jac}
  \rho(M_\gamma^{\mathrm{J}}) \;=\; \gamma\,\max\bigl(\lambda_1 - 1,\; 1 - \lambda_N\bigr).
\end{equation}
\item[(ii)] The preconditioned condition number of $P_\gamma$ with Jacobi preconditioner $D^{-1}$ is
\begin{equation}\label{eq05:kappa_preconditioned}
  \kappa\bigl(D^{-1}P_\gamma\bigr) \;=\; \frac{(1-\gamma) + \gamma\,\lambda_1}{(1-\gamma) + \gamma\,\lambda_N},
\end{equation}
monotone-increasing in $\gamma$ from $1$ at $\gamma = 0$ to $\kappa(C) = \lambda_1/\lambda_N$ at $\gamma = 1$.
\item[(iii)] $\rho(M_\gamma^{\mathrm{GS}}) < 1$ strictly for every $\gamma \in [0,1]$, and the number of sweeps to reach relative residual $\varepsilon$ is bounded by
\begin{equation}\label{eq05:sweep_bound}
  p(\varepsilon, \gamma) \;=\; O\!\Bigl(\kappa\bigl(D^{-1}P_\gamma\bigr)\,\log(1/\varepsilon)\Bigr),
\end{equation}
and specialising to $\gamma = 1$,
\begin{equation}\label{eq05:sweep_bound_gamma_one}
  p(\varepsilon, 1) \;=\; O\bigl(\kappa(C)\,\log(1/\varepsilon)\bigr).
\end{equation}
\end{enumerate}
\end{theorem}

\begin{proof}[Proof sketch]
\emph{(i).} From $P_\gamma = D + \gamma E$,
\[
  M_\gamma^{\mathrm{J}} \;=\; I - D^{-1}P_\gamma \;=\; -\gamma\,D^{-1}E \;=\; \gamma\bigl(I - D^{-1}\Sigma\bigr).
\]
By Proposition~\ref{prop:corr_similarity}, $D^{-1}\Sigma = D^{-1/2}\,C\,D^{1/2}$ is similar to $C$, so $M_\gamma^{\mathrm{J}}$ is similar to $\gamma(I - C)$ with eigenvalues $\{\gamma(1-\lambda_i)\}$. Since $\tr(C) = N$ forces $\lambda_1 \ge 1 \ge \lambda_N$, the spectral radius is \eqref{eq05:rho_jac}.

\emph{(ii).} $D^{-1}P_\gamma = (1-\gamma)I + \gamma\,D^{-1}\Sigma$ has eigenvalues $\{(1-\gamma) + \gamma\lambda_i\}$; their ratio gives \eqref{eq05:kappa_preconditioned}. The map $\gamma \mapsto [(1-\gamma)+\gamma\lambda_1]/[(1-\gamma)+\gamma\lambda_N]$ is the M\"obius transform $x \mapsto (a+bx)/(a+cx)$ in $\gamma$ with $a = 1$, $b = \lambda_1 - 1$, $c = \lambda_N - 1$, which is monotone because $b \ne c$ whenever $\lambda_1 > \lambda_N$; the endpoints $1$ and $\kappa(C)$ come from substitution.

\emph{(iii).} By Theorem~\ref{thm:gs_convergence} and \citet[Thm.~3.4]{varga2000}, $\rho(M_\gamma^{\mathrm{GS}}) < 1$ strictly. The classical Hackbusch--Saad rate bound for Gauss--Seidel on an SPD system in the energy norm of the system matrix is
\[
  \|e^{(p)}\|_{P_\gamma} \;\le\; \Bigl(1 - \frac{c}{\kappa(D^{-1}P_\gamma)}\Bigr)^{p}\,\|e^{(0)}\|_{P_\gamma},
\]
for a constant $c > 0$ depending on the \emph{Jacobi-preconditioning structure} of $P_\gamma$, not only on its scale; see \citet[\S4.7]{hackbusch2016} and \citet[Thm.~4.8]{saad2003}. Solving for the smallest $p$ that drives the right-hand side below $\varepsilon$ yields \eqref{eq05:sweep_bound}; setting $\gamma = 1$ and using (ii) gives \eqref{eq05:sweep_bound_gamma_one}. Full proof in Appendix~\ref{app:proofs}.
\end{proof}

Part (ii) is the load-bearing statement. Throughout the paper the \emph{preconditioned} condition number $\kappa(D^{-1}P_\gamma)$---not $\kappa(P_\gamma)$ or $\kappa(\Sigma)$---is the object that governs iterative cost; volatility dispersion $\kappa_D = \max_i\sigma_{ii}/\min_i\sigma_{ii}$ absorbs into the problem-dependent constant in the big-O of \eqref{eq05:sweep_bound} and does not appear in the leading term. This is the structural consequence of choosing the target $D$: scaling the rows and columns of $\Sigma$ by arbitrary positive weights changes $\kappa(\Sigma)$ without bound but leaves $\kappa(C)$ and hence the sweep budget invariant.

\begin{remark}[Why not Stein--Rosenberg]\label{rem:stein_rosenberg}
For a matrix whose Jacobi iteration matrix is entrywise \emph{nonnegative}---equivalently $\Sigma$ is an $M$-matrix with nonpositive off-diagonals---Young's theorem on consistently ordered matrices and the Stein--Rosenberg comparison \citep{young1971}, \citet[\S3.3]{varga2000}, give the tighter bound $p(\varepsilon,1) = O\bigl(\sqrt{\kappa(C)}\,\log(1/\varepsilon)\bigr)$. Realistic equity and factor covariances have \emph{positive} off-diagonal entries (assets are positively correlated on average) so $\Sigma$ is not an $M$-matrix and the hypothesis of Young's theorem fails. We therefore quote the universally-applicable SPD bound \eqref{eq05:sweep_bound_gamma_one}; this is weaker but holds for every SPD $\Sigma$ without further sign conditions on $E$. In the synthetic experiments of Section~\ref{sec10:experiments} the empirical sweep count is well described by the linear-in-$\kappa(C)$ rate, as expected.
\end{remark}

\begin{corollary}[Total complexity]\label{cor:complexity}
The total cost of CRISP to drive the relative residual of $P_\gamma w = \mu$ below $\varepsilon$ is
\begin{equation}\label{eq05:total_cost}
  O\!\Bigl(\kappa\bigl(D^{-1}P_\gamma\bigr)\,N^2\,\log(1/\varepsilon)\Bigr),
\end{equation}
specialising at $\gamma = 1$ to $O\bigl(\kappa(C)\,N^2\,\log(1/\varepsilon)\bigr)$. The sweep count is a hidden constant only on problems where $\kappa(C) = O(1)$; on ill-conditioned correlation matrices the sweep budget scales accordingly.
\end{corollary}

The \emph{directional difficulty diagnostic} of Section~\ref{sec02:preliminaries}, Equation~\eqref{eq02:dir_diag},
\[
  \operatorname{dir}_{\mathrm{diag}}(\Sigma,\mu) \;=\; \operatorname{dir}\!\bigl(D^{-1}\mu,\;\Sigma^{-1}\mu\bigr),
\]
sits next to Corollary~\ref{cor:complexity} as the empirical analogue: $\operatorname{dir}_{\mathrm{diag}}$ measures how far the $\gamma = 0$ iterate is from the $\gamma = 1$ target in direction, and $\kappa(D^{-1}P_\gamma)$ measures how many sweeps it takes to traverse that gap.

Figure~\ref{fig05:shrinkage_schematic} plots the two curves that drive \eqref{eq05:total_cost} together: the closed-form preconditioned condition number $\kappa(D^{-1}P_\gamma)$ from \eqref{eq05:kappa_preconditioned}, and the empirical sweep count of Algorithm~\ref{alg05:methodb_gs} on the same test covariance. Both are monotone in $\gamma$ and pinned down at the endpoints $\kappa(D^{-1}P_0) = 1$ and $\kappa(D^{-1}P_1) = \kappa(C)$. Correlation conditioning---not volatility dispersion, not $\kappa(\Sigma)$---drives CRISP's iteration cost.

\begin{figure}[htbp]
\centering
\includegraphics[width=0.85\textwidth]{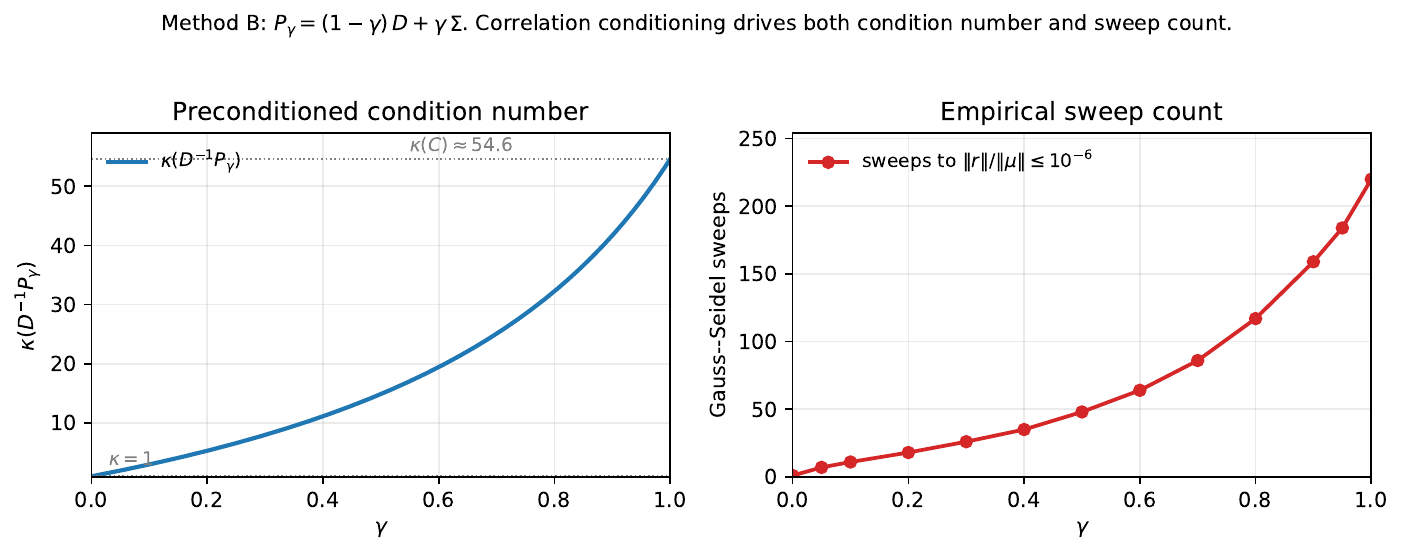}
\caption{The shrinkage operator $P_\gamma = (1-\gamma)D + \gamma\Sigma$ and the CRISP convergence rate on a block-structured $60$-asset test covariance with four clusters, within-block correlation $0.65$, cross-block correlation $0.20$, and idiosyncratic volatilities uniform on $[0.15,0.40]$. \textbf{Left:} the preconditioned condition number $\kappa(D^{-1}P_\gamma) = [(1-\gamma)+\gamma\lambda_1]/[(1-\gamma)+\gamma\lambda_N]$ from Theorem~\ref{thm:gs_rate}(ii), monotone from $1$ at $\gamma = 0$ to $\kappa(C) \approx 54.6$ at $\gamma = 1$. \textbf{Right:} empirical sweep count for Algorithm~\ref{alg05:methodb_gs} to reach relative residual $\|P_\gamma w - \mu\|/\|\mu\| \le 10^{-6}$ on the same matrix with a random signal $\mu$. The empirical sweep count tracks the condition-number curve, confirming the linear-in-$\kappa(D^{-1}P_\gamma)$ bound \eqref{eq05:sweep_bound}.}
\label{fig05:shrinkage_schematic}
\end{figure}

%% ------------------------------------------------------------------
\subsection{What CRISP does not do}\label{sec05:not_do}

Two clarifying negatives are worth stating explicitly because each has been inferred in conversation.

\paragraph{CRISP does not nest HRP.} At $\gamma = 0$ Algorithm~\ref{alg05:methodb_gs} returns the flat diagonal solution $w_i = \mu_i/\sigma_{ii}$---the signal weighted by inverse idiosyncratic variance---\emph{not} the HRP weights. HRP is a tree-structured construction that requires a hierarchical clustering of $\Sigma$ and a recursive-bisection traversal; at $\mu = \mathbf{1}$ it produces a portfolio whose weights are products of inverse cluster variances along the root-to-leaf path. The $\gamma = 0$ endpoint of CRISP does not know about any tree and, even at $\mu = \mathbf{1}$, returns the flat inverse-variance portfolio, which generally differs from HRP. The hierarchical nesting of HRP and the signal-aware nesting that recovers HRP at $\mu = \mathbf{1}$ are the job of HRP-$\mu$ (Section~\ref{sec04:methoda3}), not CRISP.

\paragraph{CRISP does not use the tree.} The sweep ordering in Algorithm~\ref{alg05:methodb_gs} is arbitrary: Theorem~\ref{thm:gs_convergence} holds for any fixed ordering, and the big-O in \eqref{eq05:total_cost} does not depend on it. Empirically, ordering assets so that strongly correlated ones are visited consecutively---for instance, the leaf ordering of the HRP quasi-diagonalisation \citep{lopezdeprado2016}---can modestly improve the empirical contraction factor, a standard Gauss--Seidel phenomenon \citep[\S4.1]{saad2003}. We use the quasi-diagonalised ordering in experiments but the theoretical rate does not rely on it.

%% ------------------------------------------------------------------
\subsection{Factor-model streaming variant}\label{sec05:factor_stream}

When $\Sigma = B\Lambda B^{\top} + D_{\mathrm{idio}}$ admits a $K$-factor decomposition with $B \in \R^{N\times K}$, $\Lambda \in \R^{K\times K}$ SPD, and $D_{\mathrm{idio}} \in \R^{N}$ the idiosyncratic variances---the structure of essentially every production equity and fixed-income risk model \citep{famafrench1993,fan2013}---CRISP can be implemented \emph{without materialising the dense $\Sigma$}. The trick is to maintain the $K$-vector aggregate $z = B^{\top}w$ incrementally so that the row inner product $\Sigma_{i,:}\,w = B_i\Lambda B^{\top}w + d_{\mathrm{idio},i}\,w_i$ reduces to an application of $\Lambda$ to the short vector $z$ plus a scalar correction. Per-asset cost becomes $O(K^2)$ and working memory $O(NK)$ rather than $O(N^2)$. Algorithm~\ref{alg05:methodb_stream} states the pseudocode.

\begin{algorithm}[t]
\caption{\textsc{CrispFactorStream}$(B, \Lambda, d_{\mathrm{idio}}, \mu, \gamma, p_{\max}, \varepsilon)$}
\label{alg05:methodb_stream}
\begin{algorithmic}[1]
\Require Factor loadings $B \in \R^{N\times K}$; factor covariance $\Lambda \in \R^{K\times K}$ SPD; idiosyncratic variances $d_{\mathrm{idio}} \in \R^{N}$; signal $\mu$; shrinkage $\gamma$; max sweeps $p_{\max}$; tolerance $\varepsilon$.
\Ensure Approximate solution $w \approx P_\gamma^{-1}\mu$ with $\Sigma = B\Lambda B^{\top} + \diag(d_{\mathrm{idio}})$.
\State $\sigma^{2}_i \gets B_i\Lambda B_i^{\top} + d_{\mathrm{idio},i}$ for $i = 1,\dots,N$ \Comment{diagonal of $\Sigma$ from factors}
\State $w_i \gets \mu_i/\sigma^{2}_i$ for $i = 1,\dots,N$ \Comment{diagonal-solve initial guess}
\State $z \gets B^{\top}w$ \Comment{$K$-vector aggregate}
\If{$\gamma < \varepsilon$}
  \State \Return $w$
\EndIf
\For{$p = 1$ \textbf{to} $p_{\max}$}
  \For{$i = 1$ \textbf{to} $N$}
    \State $z_{\setminus i} \gets z - B_i\,w_i$ \Comment{remove asset $i$'s contribution}
    \State $\mathit{off} \gets B_i\,\Lambda\,z_{\setminus i}$ \Comment{$(\Sigma w)_i - \sigma^{2}_i w_i$ without forming $\Sigma$}
    \State $w_i^{\mathrm{new}} \gets (\mu_i - \gamma\,\mathit{off})/\sigma^{2}_i$
    \State $z \gets z_{\setminus i} + B_i\,w_i^{\mathrm{new}}$ \Comment{incremental aggregate update}
    \State $w_i \gets w_i^{\mathrm{new}}$
  \EndFor
\EndFor
\State \Return $w$
\end{algorithmic}
\end{algorithm}

The iterates of Algorithm~\ref{alg05:methodb_stream} are identical to those of Algorithm~\ref{alg05:methodb_gs} run on the dense $\Sigma$; this is an optimisation, not a different method, and all results of Theorem~\ref{thm:gs_rate} carry over unchanged.

\paragraph{Numerical example.} At $N = 30{,}000$, $K = 20$: the dense $\Sigma$ occupies $N^2 \cdot 8\,\mathrm{B} \approx 7.2$~GB; the factor-streaming working set is $NK + K^2 + N \approx 4.8$~MB. Three orders of magnitude separate ``needs a dedicated server'' from ``fits in L3 cache.'' Dense Cholesky cannot exploit the factor structure because its factorisation needs random access to the full covariance and produces a dense triangular factor. The Woodbury identity gives $O(NK^2 + K^3)$ exact inversion of factor-structured $\Sigma$, but it inverts the \emph{raw} covariance and so inherits all of the overfitting that CRISP's shrinkage exists to prevent.

%% ------------------------------------------------------------------
\subsection{Shrinkage interpretation and Ledoit--Wolf comparison}\label{sec05:lw_compare}

At convergence CRISP returns $P_\gamma^{-1}\mu$, which is exactly Markowitz on the linearly shrunk covariance $P_\gamma = (1-\gamma)D + \gamma\Sigma$. CRISP therefore sits inside the Ledoit--Wolf linear-shrinkage family \citep{ledoitwolf2003,ledoitwolf2004,ledoitwolf2017,ledoitwolf2020}, and the interest lies not in whether it differs from shrunk Markowitz in some abstract sense, but in exactly three structural choices:

\begin{enumerate}
\item[(1)] \emph{Target.} CRISP uses $D = \diag(\Sigma)$, which preserves variances exactly---only off-diagonal correlations are regularised. Classical Ledoit--Wolf targets (scaled identity, constant-correlation, single-factor) all distort the diagonal. The scaled identity replaces $\sigma_{ii}$ by the cross-sectional mean $\bar\lambda/N$ \citep{ledoitwolf2004}; the constant-correlation target keeps $\sigma_{ii}$ but imposes a single off-diagonal correlation for the target \citep{ledoitwolf2003}; the single-factor target substitutes $\beta_i\beta_j\sigma_f^2 + \delta_{ij}\sigma^2_{\varepsilon,i}$, which rewrites both diagonals and off-diagonals at once. The CRISP target is the only one that asserts, as a modelling choice, that sample \emph{variances} do not need shrinkage.

\item[(2)] \emph{Intensity criterion.} Ledoit--Wolf chooses its intensity $\alpha_{\mathrm{LW}}$ to minimise the expected Frobenius loss $\E\|\widehat\Sigma_{\mathrm{LW}} - \Sigma\|_F^2$---a \emph{statistical} criterion on the covariance estimator. CRISP chooses $\gamma$ to maximise out-of-sample portfolio Sharpe---an \emph{operational} criterion on the portfolio. Section~\ref{sec06:perturbation} makes clear why these need not coincide: a covariance estimator that minimises Frobenius distance to $\Sigma$ can still overweight noise-dominated eigendirections once inverted, and the Frobenius loss does not see that downstream effect. The oracle-approximating shrinkage (OAS) estimator \citep{chen2010} and the nonlinear shrinkage variants \citep{ledoitwolf2017,ledoitwolf2020} refine the intensity estimator under the same statistical criterion.

\item[(3)] \emph{Implementation.} Ledoit--Wolf forms $\widehat\Sigma_{\mathrm{LW}}$ explicitly and inverts with dense Cholesky. CRISP never forms $P_\gamma^{-1}$; Algorithm~\ref{alg05:methodb_gs} drives $P_\gamma w \to \mu$ iteratively in $O(\kappa(D^{-1}P_\gamma)\,N^2\,\log(1/\varepsilon))$ time. This is an engineering advantage---and it enables the factor-streaming variant of Section~\ref{sec05:factor_stream}---but it is not what makes the portfolio better.
\end{enumerate}

The philosophy shift, independent of Stein's original motivation \citep{stein1956}, is the first two items: CRISP regularises exactly the unreliable part of the covariance and selects the regularisation strength for the downstream application.

\paragraph{Open question.} Is there a Ledoit--Wolf target $T$ and intensity rule $\alpha_{\mathrm{LW}}$ such that $\widehat\Sigma_{\mathrm{LW}}^{-1}\mu = P_{\gamma^{\star}}^{-1}\mu$ in expectation over a specified sampling distribution, where $\gamma^{\star}$ is the Sharpe-optimal CRISP intensity? The empirical Monte Carlo of Section~\ref{sec10:experiments} is consistent with $\gamma^{\star} \in [0.5, 0.7]$ across panels, but we do not have a closed-form correspondence. We flag this as an open problem; Ledoit--Wolf is included as a head-to-head baseline in every out-of-sample experiment.

%% ------------------------------------------------------------------
\subsection{Early stopping as implicit spectral regularisation}\label{sec05:early_stop}

\begin{remark}[Early stopping as implicit spectral regularisation]\label{rem:early_stop_spectral}
The per-eigendirection content of Theorem~\ref{thm:gs_rate} has a statistical side-effect when $\Sigma$ is a finite-sample estimate $\widehat\Sigma$. Gauss--Seidel on $P_\gamma w = \mu$ converges geometrically in each eigendirection of $D^{-1}P_\gamma$: after $p$ sweeps, the residual on the $k$-th eigendirection is $O(\rho_k^p)$ with $\rho_k$ bounded by $\rho(M_\gamma^{\mathrm{GS}}) < 1$ and decreasing with the eigenvalue. Well-conditioned (high-eigenvalue) directions converge exponentially faster than ill-conditioned (low-eigenvalue) directions. When $\widehat\Sigma$ is a sample estimate, the low-eigenvalue directions of $D^{-1}\widehat P_\gamma$ are exactly the directions where Marchenko--Pastur bias \citep{marchenko1967} and finite-sample noise are worst, and the directions direct Markowitz amplifies. Stopping Gauss--Seidel before convergence therefore implicitly attenuates the portfolio's exposure to the most noise-contaminated spectral subspace. The sweep count $p$ plays the role of $1/\lambda$ in ridge regression: small $p$ is strong regularisation (heavy spectral truncation), large $p$ is weak regularisation (recovery of full $\widehat P_\gamma^{-1}\mu$). This is \emph{Channel~3} of the three-channel decomposition in Section~\ref{sec07:adaptive}; together with operator shrinkage (Channel~1, driven by $\gamma$) and noise cancellation (Channel~2, driven by the target $D$), the three channels act orthogonally on three different spectral regimes of $\widehat\Sigma$.
\end{remark}

%% ------------------------------------------------------------------
\subsection{Wall-clock versus memory}\label{sec05:wallclock}

\begin{remark}[Wall-clock: Cholesky with tuned BLAS is fast]\label{rem:wall_clock}
The na\"ive FLOP crossover between CRISP at $p = 100$ sweeps and Cholesky factorisation occurs at $N \approx 600$, from $100 \cdot 2N^2 \approx N^3/3$. The true wall-clock picture is much less favourable to CRISP because LAPACK Cholesky exploits BLAS-3 (matrix--matrix) throughput while the inherently sequential Gauss--Seidel sweep---each $w_i$ update depends on $w_j$ updated earlier in the sweep---is limited to scalar throughput. On modern matrix-optimised silicon the gap between these two throughputs is very large.

We measured this on an Apple~M4 with Accelerate-linked LAPACK (NumPy~2.4, SciPy~1.17, macOS~15.7), comparing \texttt{scipy.linalg.cho\_factor} / \texttt{cho\_solve} against a Numba-JIT-compiled CRISP at $\gamma = 0.5$, $p = 100$; single-solve median across trials, full raw output in \texttt{results/11\_bench\_crisp\_vs\_cholesky.txt}:
\begin{center}
\begin{tabular}{rrrr}
\toprule
$N$ & Cholesky (ms) & CRISP (ms) & ratio (Cholesky faster) \\
\midrule
$500$     & $0.38$   & $18.2$    & $48\times$ \\
$1{,}000$ & $1.92$   & $76.3$    & $40\times$ \\
$2{,}000$ & $13.2$   & $313.9$   & $24\times$ \\
$3{,}000$ & $51.8$   & $712.1$   & $14\times$ \\
$5{,}000$ & $216.5$  & $1{,}993.4$ & $9\times$ \\
\bottomrule
\end{tabular}
\end{center}
Cholesky is $9$--$48\times$ faster across the range tested. Fitting the $O(N^2)$-versus-$O(N^3)$ asymptotics places the single-solve wall-clock crossover on this hardware near $N \approx 45{,}000$, well beyond realistic portfolio sizes. On commodity x86 with OpenBLAS the gap is narrower and the crossover shifts down, but we do not claim a wall-clock advantage for CRISP against tuned-BLAS Cholesky at realistic $N$.

The practical case for CRISP rests on two legs, neither of which is speed:
\begin{enumerate}[label=(\roman*),leftmargin=*]
\item \emph{Out-of-sample Sharpe.} At $\gamma \in [0.3, 0.7]$, CRISP delivers superior OOS Sharpe at every $N$ tested---including $N = 500$ where Cholesky is nearly $50\times$ faster---as reported in Section~\ref{sec10:experiments}.
\item \emph{Memory footprint.} Algorithm~\ref{alg05:methodb_stream} reduces working memory from $O(N^2)$ to $O(NK)$ for factor-structured $\Sigma$---$4.8$~MB against $7.2$~GB at $N = 30{,}000$, $K = 20$---which dense Cholesky cannot match. This is a structural advantage independent of BLAS throughput.
\end{enumerate}
A faster wrong answer is still wrong. The shrunk system $P_\gamma w = \mu$ is the better answer, and CRISP solves it without ever forming the inverse.
\end{remark}

This closes the algorithmic content of the paper. Sections~\ref{sec06:perturbation} and~\ref{sec07:adaptive} analyse the shrinkage trajectory $\gamma \mapsto P_\gamma^{-1}\mu$ in detail and derive an adaptive rule for $\gamma^{\star}$; Section~\ref{sec10:experiments} verifies both the rate theorem and the interior-optimum empirically; Sections~\ref{sec08:comparative} and~\ref{sec09:constraints} place CRISP in the comparative landscape and extend it to linearly-constrained portfolios.

% ============================================================
% section_06_perturbation.tex
% ============================================================
% !TEX root = ../paper.tex
% Section 6: Perturbation theory and the shrinkage trajectory
% Namespace: sec06, fig06, tab06, eq06

\section{Perturbation theory and the shrinkage trajectory}\label{sec06:perturbation}

Sections~\ref{sec04:tree_methods} and~\ref{sec05:crisp} introduced the shrinkage operators $P_\gamma = (1-\gamma)D + \gamma\Sigma$; this section analyses them as a one-parameter deformation of $\Sigma$ by the off-diagonal remainder $E = \Sigma - D$. Four results follow: an exact approximation-error identity (Theorem~\ref{thm:perturbation}), a structural direction-error bound (Proposition~\ref{prop:dir_err_bound}), the zero locus on invariant rays (Corollary~\ref{cor:invariant_rays_trajectory}), and the interior $\gamma^\star$ phenomenon that governs finite-sweep CRISP (Proposition~\ref{prop:interior_optimum}). All proof sketches are inline; full proofs live in Appendix~\ref{app:proofs}.

%% ------------------------------------------------------------------
\subsection{Invariant rays revisited}\label{sec06:invariant_rays}

The zero-set for every perturbation statement below consists of the \emph{invariant rays} of the shrinkage family. Lemma~\ref{lem:parallel_direction} of Section~\ref{sec02:preliminaries} identifies these explicitly: if $C = D^{-1/2}\Sigma D^{-1/2}$ has eigenpair $(\lambda,v)$ with $\lambda > 0$, then for $\mu = D^{1/2}v$ and every $\gamma \in [0,1]$,
\begin{equation}\label{eq06:invariant_ray}
  P_\gamma^{-1}\mu \;=\; \frac{1}{(1-\gamma) + \gamma\lambda}\,D^{-1}\mu,
\end{equation}
so the direction of $P_\gamma^{-1}\mu$ is fixed at $D^{-1}\mu$ (equivalently $\Sigma^{-1}\mu = \lambda^{-1}D^{-1}\mu$) across the entire shrinkage interval. On these signals every method in the paper---HRP, Cotton, HRP-$\mu$, HRP-$\Sigma\mu$, CRISP at any sweep budget, and the trivial diagonal solver---produces the same direction, so the direction error against $w^\star = \Sigma^{-1}\mu$ is identically zero for every $\gamma$. Two operational consequences. First, single-eigenvector signals are useless as adversarial inputs because they cannot separate methods and cannot produce a nonzero direction error under any shrinkage intensity; adversarial constructions in Section~\ref{sec10:experiments} must mix at least two eigendirections of $C$ with distinct eigenvalues. Second, the standard slogan that ``$\Sigma^{-1}$ amplifies low-eigenvalue directions'' concerns \emph{scale}: the prefactor $[(1-\gamma)+\gamma\lambda]^{-1}$ in~\eqref{eq06:invariant_ray} diverges as $\lambda \downarrow 0$, but the \emph{direction} along an invariant ray is rigid.

%% ------------------------------------------------------------------
\subsection{Approximation-error identity}\label{sec06:identity}

Write $w^\star = \Sigma^{-1}\mu$ for the Markowitz direction and $\hat w(\gamma) = P_\gamma^{-1}\mu$ for its shrinkage approximant at intensity $\gamma$. Their difference admits an exact closed form.

\begin{theorem}[Approximation-error identity]\label{thm:perturbation}
Let $\Sigma \in \mathbb{R}^{N\times N}$ be SPD with $D = \diag(\Sigma)$ and $E = \Sigma - D$. For every $\gamma \in [0,1]$,
\begin{equation}\label{eq06:error_identity}
  w^\star - \hat w(\gamma) \;=\; -(1-\gamma)\, P_\gamma^{-1}\, E\, w^\star,
  \qquad \hat w(\gamma) = P_\gamma^{-1}\mu.
\end{equation}
\end{theorem}

\begin{proof}[Proof sketch]
Since $P_\gamma = \Sigma - (1-\gamma)E$, the Markowitz equation $\Sigma w^\star = \mu$ rewrites as $P_\gamma w^\star + (1-\gamma)E w^\star = \mu$. Subtracting $P_\gamma \hat w(\gamma) = \mu$ and applying $P_\gamma^{-1}$ yields~\eqref{eq06:error_identity}. Full proof in Appendix~\ref{app:proofs}.
\end{proof}

Three features of~\eqref{eq06:error_identity} govern everything that follows. The prefactor $(1-\gamma)$ forces the residual to zero as $\gamma \uparrow 1$ with norm $O(1-\gamma)$; the error direction is the specific vector $P_\gamma^{-1} E w^\star$, not an abstract residual; and the identity vanishes whenever $E w^\star$ lies in the kernel of the direction metric, a condition made precise by Corollary~\ref{cor:invariant_rays_trajectory} below.

%% ------------------------------------------------------------------
\subsection{Direction-error bound}\label{sec06:dir_form}

The quality metric is the direction metric of Definition~\ref{def:dir_error}, which is scale-invariant. Translating~\eqref{eq06:error_identity} into direction-error form isolates the geometric content.

\begin{proposition}[Direction-error bound]\label{prop:dir_err_bound}
Under the hypotheses of Theorem~\ref{thm:perturbation}, for every fixed $\gamma \in [0,1)$,
\begin{equation}\label{eq06:dir_bound}
  \mathrm{dir}\bigl(\hat w(\gamma),\, w^\star\bigr)
  \;\le\;
  (1-\gamma)^2 \,\bigl\|P_\gamma^{-1} E\bigr\|_{\mathrm{op}}^{2}\;\frac{\|w^\star\|^2}{\|\hat w(\gamma)\|^2}\; G(\gamma, w^\star),
\end{equation}
where $G(\gamma, w^\star) = \sin^2\phi(\gamma) \in [0,1]$ is the squared sine of the angle $\phi(\gamma)$ between $w^\star$ and $P_\gamma^{-1} E w^\star$. In particular, $G(\gamma, w^\star) = 0$ if and only if $w^\star$ is an eigenvector of $P_\gamma$.
\end{proposition}

\begin{proof}[Proof sketch]
Write $\hat w(\gamma) = w^\star + (1-\gamma)\,u(\gamma)$ with $u(\gamma) = P_\gamma^{-1}E w^\star$. Lagrange's identity for the sine of the angle between $\hat w(\gamma)$ and $w^\star$ gives $\sin^2\theta = (1-\gamma)^2\|u_\perp(\gamma)\|^2\|w^\star\|^2/(\|\hat w(\gamma)\|^2\|w^\star\|^2)$, where $u_\perp(\gamma) = u(\gamma) - \mathrm{proj}_{w^\star}u(\gamma)$. Bounding $\|u_\perp(\gamma)\|^2 \le \|u(\gamma)\|^2 \le \|P_\gamma^{-1}E\|_{\mathrm{op}}^2\|w^\star\|^2$ and defining $G(\gamma, w^\star) = \|u_\perp(\gamma)\|^2/\|u(\gamma)\|^2 = \sin^2\phi(\gamma)$ yields~\eqref{eq06:dir_bound}. Full proof in Appendix~\ref{app:proofs}.
\end{proof}

Equation~\eqref{eq06:dir_bound} is a structural decomposition, not a one-parameter monotonicity theorem. Both the operator-norm factor $\|P_\gamma^{-1} E\|_{\mathrm{op}}$ and the geometric factor $G(\gamma, w^\star)$ depend on $\gamma$, and they can move in opposite directions as $\gamma$ varies. The prefactor $(1-\gamma)^2$ forces the right-hand side to zero at $\gamma = 1$, but does not pin down the trajectory in between. Example~\ref{ex:nonmonotone_trajectory} exhibits a $4\times 4$ problem on which the three factors combine non-monotonically.

%% ------------------------------------------------------------------
\subsection{The shrinkage trajectory}\label{sec06:trajectory}

The exact identity~\eqref{eq06:error_identity} pins down the endpoint behaviour of the shrinkage trajectory $\gamma \mapsto \hat w(\gamma) = P_\gamma^{-1}\mu$. As $\gamma \uparrow 1$, $\hat w(\gamma) \to w^\star$ with residual norm $O(1-\gamma)$, so the direction error tends to zero. On the invariant rays of~\eqref{eq06:invariant_ray} the direction error is identically zero for every $\gamma$.

\begin{corollary}[Invariant rays on the trajectory]\label{cor:invariant_rays_trajectory}
If $\mu = D^{1/2}v$ with $Cv = \lambda v$ and $\lambda > 0$, then $\mathrm{dir}\bigl(P_\gamma^{-1}\mu,\, \Sigma^{-1}\mu\bigr) = 0$ identically for every $\gamma \in [0,1]$.
\end{corollary}

\begin{proof}[Proof sketch]
By~\eqref{eq06:invariant_ray}, $P_\gamma^{-1}\mu$ is a scalar multiple of $D^{-1}\mu$ for every $\gamma$, and so is $\Sigma^{-1}\mu = \lambda^{-1}D^{-1}\mu$. The direction between the two is therefore fixed at zero on $[0,1]$. Full proof in Appendix~\ref{app:proofs}.
\end{proof}

Off the invariant rays, the map $\gamma \mapsto \mathrm{dir}(P_\gamma^{-1}\mu, \Sigma^{-1}\mu)$ is not globally monotone. Example~\ref{ex:nonmonotone_trajectory} makes this concrete.

\begin{example}[Non-monotone trajectory]\label{ex:nonmonotone_trajectory}
Let
\[
\Sigma \;=\;
\begin{pmatrix}
1.055 & 0.078 & -0.711 & -1.874 \\
0.078 & 6.058 & -0.379 & 1.775 \\
-0.711 & -0.379 & 1.036 & 2.063 \\
-1.874 & 1.775 & 2.063 & 6.477
\end{pmatrix},
\qquad
\mu \;=\;
\begin{pmatrix} -1.695 \\ 0.271 \\ 0.322 \\ -0.500 \end{pmatrix},
\]
so that $\kappa(C) \approx 27.4$ and $\mu$ is not on an invariant ray (its four projections onto the eigenvectors of $C$ are all nonzero). The direction error $\mathrm{dir}(P_\gamma^{-1}\mu, \Sigma^{-1}\mu)$ behaves as follows:

\begin{table}[h]
\centering
\begin{adjustbox}{max width=\textwidth}
\begin{tabular}{lccccccccc}
\toprule
$\gamma$ & $0$ & $0.05$ & $0.10$ & $0.20$ & $0.30$ & $0.50$ & $0.70$ & $0.90$ & $1.00$ \\
\midrule
$\mathrm{dir}$ & $0.242$ & $0.243$ & $0.244$ & $0.249$ & $\mathbf{0.253}$ & $0.250$ & $0.213$ & $0.085$ & $0.000$ \\
\bottomrule
\end{tabular}
\end{adjustbox}
\end{table}

\noindent The direction error \emph{rises} from $0.242$ at $\gamma = 0$ to a peak of $0.253$ at $\gamma \approx 0.30$ (a $+4.75\%$ bump over the diagonal anchor) before descending to $0$ at $\gamma = 1$. Global non-increasing monotonicity in $\gamma$ therefore fails.
\end{example}

\begin{remark}[What the identity does and does not tell us]\label{rem:trajectory_behavior}
Theorem~\ref{thm:perturbation} controls two facts unambiguously: the residual norm is $O(1-\gamma)$, so $\hat w(\gamma) \to w^\star$ as $\gamma \uparrow 1$, and on invariant rays the direction error is identically zero. Between those two facts, the \emph{direction} of $P_\gamma^{-1} E w^\star$ relative to $w^\star$ is not sign-controlled by the shrinkage intensity alone, so bumps of the kind exhibited in Example~\ref{ex:nonmonotone_trajectory} are possible in principle. Empirically, on the correlation regimes studied in Section~\ref{sec10:experiments} the trajectory is typically non-increasing with only mild interior bumps in adversarial regimes, and the practitioner's operating point $\gamma \approx 0.5$ sits inside the descending region on every panel we have tested.
\end{remark}

%% ------------------------------------------------------------------
\subsection{Shrinkage bias vs.\ convergence slack: the interior \texorpdfstring{$\gamma^\star$}{gamma*}}\label{sec06:interior_optimum}

If CRISP produced the exact fixed point $\hat w(\gamma)$ at every $\gamma$, the guidance would be ``run at $\gamma = 1$''. In practice CRISP delivers the finite-sweep iterate $w^{(p)}(\gamma)$ after $p$ Gauss--Seidel sweeps, which introduces a second $\gamma$-dependent term. The total direction error decomposes, in angle, as
\begin{equation}\label{eq06:angle_decomp}
  \underbrace{\mathrm{dir}\!\left(w^{(p)}(\gamma),\; w^\star\right)}_{\text{total}}
  \;\lesssim\;
  \underbrace{\mathrm{dir}\!\left(\hat w(\gamma),\; w^\star\right)}_{\text{shrinkage bias, }\downarrow\text{ in }\gamma}
  \;+\;
  \underbrace{\mathrm{dir}\!\left(w^{(p)}(\gamma),\; \hat w(\gamma)\right)}_{\text{convergence slack, }\uparrow\text{ in }\gamma\text{ at fixed }p}
\end{equation}
(the precise statement is a squared-sine triangle bound; the additive form is qualitative). The shrinkage bias is zero at $\gamma = 1$ and non-increasing over the regimes studied (Remark~\ref{rem:trajectory_behavior}). The convergence slack is non-decreasing in $\gamma$ at fixed $p$: by Proposition~\ref{prop:corr_similarity} and Theorem~\ref{thm:gs_rate}, the relevant operator condition number $\kappa(D^{-1}P_\gamma)$ grows from $1$ at $\gamma = 0$ to $\kappa(C)$ at $\gamma = 1$, so the per-sweep residual reduction worsens as $\gamma$ grows. The sum of a decreasing and an increasing function attains its minimum in the interior whenever the two derivatives cross inside $(0,1)$ (Figure~\ref{fig06:bias_variance_curves}).

\begin{figure}[!tbp]
\centering
\includegraphics[width=\textwidth]{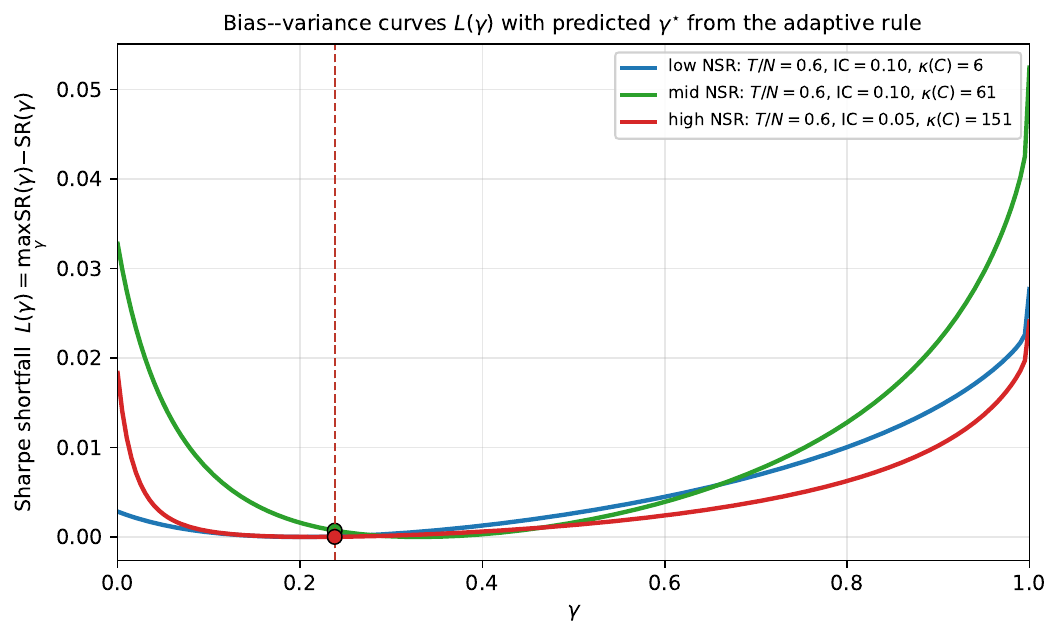}
\caption{Bias--variance decomposition of total direction error along the shrinkage axis, illustrating~\eqref{eq06:angle_decomp}. The \emph{shrinkage-bias} curve $\mathrm{dir}(\hat w(\gamma), w^\star)$ (decreasing in $\gamma$, vanishing at $\gamma = 1$ by Theorem~\ref{thm:perturbation}) falls as the shrinkage approximant approaches the Markowitz fixed point. The \emph{convergence-slack} curve $\mathrm{dir}(w^{(p)}(\gamma), \hat w(\gamma))$ (increasing in $\gamma$ at fixed sweep budget $p$, because $\kappa(D^{-1}P_\gamma)$ grows from $1$ at $\gamma = 0$ to $\kappa(C)$ at $\gamma = 1$) rises as the Gauss--Seidel iteration slows. Their sum---an upper bound on the total direction error $\mathrm{dir}(w^{(p)}(\gamma), w^\star)$---attains an interior minimum $\gamma^\star \in (0,1)$ under the hypotheses of Proposition~\ref{prop:interior_optimum}. The interior $\gamma^\star$ is the CRISP-specific manifestation of the standard shrinkage trade-off: shrinkage bias downward pressure on $\gamma$, convergence slack upward pressure on $\gamma$, interior minimum where marginal pressures balance.}
\label{fig06:bias_variance_curves}
\end{figure}

\begin{proposition}[Interior $\gamma^\star$ at finite sweep budget]\label{prop:interior_optimum}
Fix $p \in \mathbb{N}$ and SPD $\Sigma$ with $D = \diag(\Sigma)$, $C = D^{-1/2}\Sigma D^{-1/2}$. Assume:
\begin{enumerate}[label=\textup{(\roman*)}]
\item \emph{non-invariance:} $\mu$ is not on an invariant ray of $C$ in the sense of Lemma~\ref{lem:parallel_direction};
\item \emph{generic alignment:} $\partial_\gamma\,\mathrm{dir}\bigl(P_\gamma^{-1}\mu,\, w^\star\bigr)\big|_{\gamma = 0} < 0$;
\item \emph{sweep-budget deficit:} $\mathrm{dir}\bigl(w^{(p)}(1),\, \Sigma^{-1}\mu\bigr) > \mathrm{dir}\bigl(D^{-1}\mu,\, \Sigma^{-1}\mu\bigr)$.
\end{enumerate}
Then $\gamma \mapsto \mathrm{dir}\bigl(w^{(p)}(\gamma),\, w^\star\bigr)$ attains its minimum on $[0,1]$ at an interior point $\gamma^\star \in (0,1)$.
\end{proposition}

\begin{proof}[Proof sketch]
The map is continuous on the compact interval $[0,1]$, so the minimum is attained. Hypothesis (iii) rules out $\gamma = 1$: the total error there strictly exceeds the total error at $\gamma = 0$, because at $\gamma = 0$ the convergence slack vanishes (the diagonal system is solved in one sweep) and the total error reduces to $\mathrm{dir}(D^{-1}\mu, \Sigma^{-1}\mu)$. Hypothesis (ii) rules out $\gamma = 0$: the shrinkage-bias term has strictly negative derivative at $\gamma = 0$ while the convergence-slack term has a bounded derivative there (Theorem~\ref{thm:gs_rate}), so the total derivative is negative in a right neighbourhood of $0$. The minimum therefore lies in $(0,1)$. Full proof in Appendix~\ref{app:proofs}.
\end{proof}

\begin{remark}
Hypothesis (i) excludes the degenerate single-eigenvector case (Corollary~\ref{cor:invariant_rays_trajectory}), on which the direction error is identically zero across $[0,1]$ and the minimum is trivially attained everywhere. Hypothesis (ii) is generic: the first-variation computation makes $\partial_\gamma \mathrm{dir}|_{\gamma=0}$ a rational function of the spectral data of $(\Sigma, \mu)$ that vanishes only on a measure-zero algebraic subvariety of parameter space; Example~\ref{ex:nonmonotone_trajectory} sits on the complement of this subvariety in the wrong direction and the derivative is positive there. Hypothesis (iii) is the operationally meaningful one---it says the sweep budget $p$ is insufficient to overcome the $\gamma = 1$ conditioning penalty---and is exactly the regime in which Figure~\ref{fig06:trajectory_plot} shows a pronounced interior optimum.
\end{remark}

\begin{figure}[!tbp]
\centering
\includegraphics[width=\textwidth]{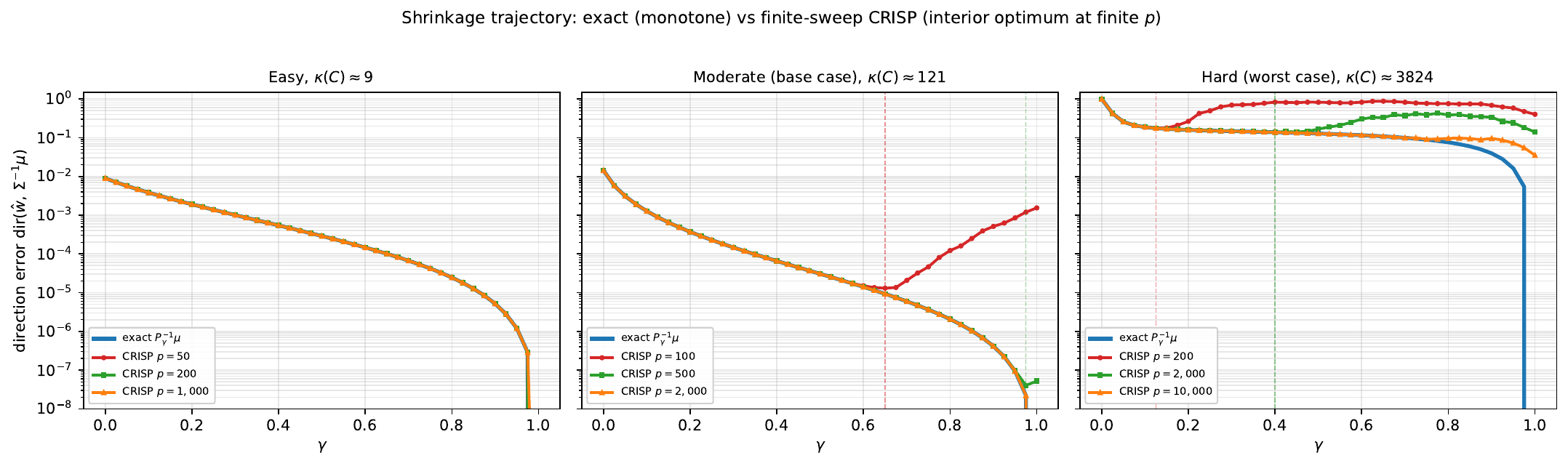}
\caption{Shrinkage trajectory $\gamma \mapsto \mathrm{dir}(w^{(p)}(\gamma), w^\star)$ across three difficulty regimes, $N = 200$. \textbf{Left} (easy, $\kappa(C) \approx 9$): even $50$ CRISP sweeps converge; no interior optimum. \textbf{Centre} (moderate, $\kappa(C) \approx 121$, the paper's base case): a mild interior optimum appears at $p = 100$ with $\gamma^\star \approx 0.65$; $500$+ sweeps essentially converge. \textbf{Right} (hard, $\kappa(C) \approx 3{,}800$, adversarial $\mu$): the interior $\gamma^\star(p)$ is pronounced and moves rightward with $p$ (Proposition~\ref{prop:interior_optimum}). Blue curves (exact fixed point $\hat w(\gamma) = P_\gamma^{-1}\mu$) descend to zero at $\gamma = 1$ by Theorem~\ref{thm:perturbation}; finite-sweep curves are U-shaped, with the minimum shifted toward the interior. A practitioner can compute $\kappa(C) = \lambda_{\max}(C)/\lambda_{\min}(C)$ from the correlation matrix of their risk model to locate the panel they are in.}
\label{fig06:trajectory_plot}
\end{figure}

Proposition~\ref{prop:interior_optimum} is the CRISP analogue of the optimal-shrinkage-intensity result of~\citet{ledoitwolf2003}, with a computational rather than statistical trade-off. Ledoit--Wolf shrink the covariance \emph{estimator} toward a structured target to trade sample-noise variance against target bias. CRISP shrinks the covariance \emph{operator} toward its diagonal to trade shrinkage bias against iterative convergence slack. Both land on the same qualitative recommendation---regularise strictly inside $(0,1)$---but for different reasons. Section~\ref{sec07:adaptive} recasts the interior optimum in out-of-sample terms, where the convergence-slack axis is replaced by a finite-sample estimation-loss axis driven by the concentration rate of $\hat\Sigma - \Sigma$; the closed-form approximation for $\gamma^\star$ derived there is the operational recommendation of the paper.

% ============================================================
% section_07_adaptive.tex
% ============================================================
\section{The adaptive $\gamma^\star$ rule and three regularisation channels}\label{sec07:adaptive}

Section~\ref{sec06:perturbation} established that fixed-$\gamma$ CRISP exhibits an interior sweep-count optimum and that the shrinkage trajectory $\gamma\mapsto w(\gamma)$ is empirically non-increasing on the regimes studied. This section closes the analytical loop: we write the expected out-of-sample Sharpe loss as a bias--variance decomposition, solve for the loss-minimising $\gamma^\star$ in closed form, and reconcile the derivation with the empirical finding that the OOS Sharpe surface is nearly flat across a wide calibration cube. We then catalogue the three distinct regularisation channels operating simultaneously inside CRISP (Remark~\ref{rem:three_channels}) and give the information-theoretic reading of what shrinkage accomplishes.

\subsection{Bias--variance decomposition}\label{sec07:bv}

Let $\hat w(\gamma) = \hat P_\gamma^{-1}\hat\mu$ with $\hat P_\gamma = (1-\gamma)\diag(\hat\Sigma) + \gamma\hat\Sigma$. Evaluating $\mathrm{SR}^2_{\mathrm{OOS}}(\gamma) = (w^\top\mu)^2/(w^\top\Sigma w)$ at $w = \hat w(\gamma)$ under the true $(\Sigma,\mu)$ and expanding to first order in the estimation noise gives
\begin{equation}\label{eq07:bv_decomp}
  \mathbb{E}\!\bigl[\mathrm{SR}^2_{\mathrm{oracle}} - \mathrm{SR}^2_{\mathrm{OOS}}(\gamma)\bigr]
  \;=\;
  \underbrace{\bigl[\mathrm{SR}^2_{\mathrm{oracle}} - \mathrm{SR}^2_{\mathrm{pop}}(\gamma)\bigr]}_{\text{approximation loss}}
  \;+\;
  \underbrace{V(\gamma)}_{\text{estimation loss}}
  \;+\;
  O(N^{-1}),
\end{equation}
where $\mathrm{SR}^2_{\mathrm{pop}}(\gamma)$ is the Sharpe-squared of the population weight $w(\gamma) = P_\gamma^{-1}\mu$ and $V(\gamma)$ is the expected Sharpe loss from substituting $(\hat\Sigma,\hat\mu)$ for $(\Sigma,\mu)$. The approximation loss pushes toward $\gamma=1$; the estimation loss pushes toward $\gamma=0$; the interior optimum $\gamma^\star$ is where their marginals balance.

\subsection{Approximation loss via Neumann expansion}\label{sec07:approx}

Write $P_\gamma = D^{1/2}[(1-\gamma)I + \gamma C]D^{1/2}$ with $C = D^{-1/2}\Sigma D^{-1/2}$, and let $\tilde\mu = D^{-1/2}\mu$, $\epsilon = 1-\gamma$. The Neumann identity applied to $[C + \epsilon(I-C)]^{-1}$ around $\gamma=1$ yields
\begin{equation}\label{eq07:neumann}
  [C + \epsilon(I-C)]^{-1}
  \;=\;
  (1+\epsilon)\,C^{-1} \;-\; \epsilon\,C^{-2} \;+\; O(\epsilon^2),
\end{equation}
so
\begin{equation}\label{eq07:pg_expansion}
  P_\gamma^{-1}\mu
  \;=\;
  \bigl(1 + (1-\gamma)\bigr)\,w^\star
  \;-\; (1-\gamma)\,D^{-1/2}C^{-2}\tilde\mu
  \;+\; O\!\bigl((1-\gamma)^2\bigr),
\end{equation}
where $w^\star = \Sigma^{-1}\mu$. The overall prefactor $(1+(1-\gamma))$ multiplying $w^\star$ must be retained: because Sharpe-squared is scale-invariant it does not propagate into the loss formula, but dropping it degrades the expansion~\eqref{eq07:pg_expansion} from $O((1-\gamma)^2)$ to $O(1-\gamma)$ and contaminates every downstream constant. A direct numerical check at $\epsilon = 10^{-3}$ on the base-case $(N,\kappa(C))=(100,61)$ covariance gives a full-expansion error of $1.7\times 10^{-4}$ against a truncated-expansion error of $2.9\times 10^{-3}$, a $17\times$ difference.

Substituting~\eqref{eq07:pg_expansion} into $\mathrm{SR}^2_{\mathrm{pop}}$ and retaining the quadratic term gives
\begin{equation}\label{eq07:approx_loss}
  \mathrm{SR}^2_{\mathrm{oracle}} - \mathrm{SR}^2_{\mathrm{pop}}(\gamma)
  \;=\; (1-\gamma)^2\,\mathcal{A} \;+\; O\!\bigl((1-\gamma)^3\bigr),
  \qquad \mathcal{A} \ge 0,
\end{equation}
where the \emph{approximation coefficient} $\mathcal{A} = \mathcal{A}(\Sigma,\mu)$ depends only on the projection of $\tilde\mu$ onto the eigenvectors of $C^{-1}$. If $\mu$ is aligned with leading eigenvectors of $C$, the diagonal surrogate is nearly optimal and $\mathcal{A}$ is small; if $\mu$ has weight on low-eigenvalue directions of $C$---exactly the directions shrinkage toward $D$ removes---$\mathcal{A}$ is large. Equicorrelation and $K$-factor special cases are worked out in Appendix~\ref{app:proofs}.

\subsection{Estimation loss via Marchenko--Pastur and the information coefficient}\label{sec07:est}

Linearising $\hat w(\gamma)$ around $(\Sigma,\mu)$ gives $\delta w = -P_\gamma^{-1}(\delta P_\gamma)P_\gamma^{-1}\mu + P_\gamma^{-1}\delta\mu$ with $\delta P_\gamma = \gamma\delta\Sigma + (1-\gamma)\delta D$. Both terms are of the form ``$P_\gamma^{-1}$ applied to noise,'' so both inherit the same amplification factor. Because the downstream Sharpe is scale-invariant, the relevant conditioning is the \emph{Jacobi-preconditioned} condition number
\begin{equation}\label{eq07:kappa_eff}
  \kappa_{\mathrm{eff}}(\gamma) \;:=\; \kappa(D^{-1}P_\gamma)
  \;=\; \frac{(1-\gamma) + \gamma\,\lambda_{\max}(C)}{(1-\gamma) + \gamma\,\lambda_{\min}(C)},
\end{equation}
not $\kappa(P_\gamma)$: the congruence $P_\gamma = D^{1/2}[(1-\gamma)I+\gamma C]D^{1/2}$ preserves inertia but not eigenvalues, so $\kappa(P_\gamma)$ mixes volatility dispersion $\kappa_D$ with correlation conditioning in a way that is generically a strict upper bound. The object in~\eqref{eq07:kappa_eff} interpolates monotonically from $1$ at $\gamma=0$ to $\kappa(C)$ at $\gamma=1$; volatility dispersion $\kappa_D$ does \emph{not} appear.

Two ingredients size the noise. First, the Marchenko--Pastur concentration rate \citep{marchenko1967} gives
\begin{equation}\label{eq07:mp}
  \|\hat\Sigma - \Sigma\|_{\mathrm{op}} \;=\; O_p\!\bigl(\sigma_{\max}^2\sqrt{N/T}\bigr),
\end{equation}
tight on Gaussian data in the proportional regime $N/T\to q\in(0,1)$ \citep[Thm.~5.39]{vershynin2012}. Second, modelling the signal estimator as $\hat\mu_i = \mu_i + \epsilon_i$ with $\mathrm{IC} = \mathrm{Corr}(\hat\mu_i,\mu_i)$ gives, by independence of $\mu$ and $\epsilon$,
\begin{equation}\label{eq07:ic}
  \mathrm{Var}(\epsilon_i) \;=\; \sigma_{\mu,i}^2\,\frac{1-\mathrm{IC}^2}{\mathrm{IC}^2}.
\end{equation}
For typical active managers $\mathrm{IC}\in[0.02,0.10]$, so $1/\mathrm{IC}^2\in[100,2500]$: signal noise dwarfs signal magnitude, and the $1/\mathrm{IC}^2$ factor is the dominant term in any realistic estimation-loss calculation.

Propagating $\delta w$ through the Sharpe-squared, bounding the covariance channel with~\eqref{eq07:mp} and the signal channel with~\eqref{eq07:ic}, and pulling the common amplification $\kappa_{\mathrm{eff}}(\gamma)^2$ out front yields the \emph{estimation-loss formula}:
\begin{equation}\label{eq07:est_loss}
\boxed{\;
  V(\gamma) \;\approx\; \mathcal{B}\cdot\kappa_{\mathrm{eff}}(\gamma)^2\cdot\frac{N}{T}\cdot\frac{1}{\mathrm{IC}^2},
\;}
\end{equation}
where $\mathcal{B} = \mathcal{B}(\Sigma)$ is an order-one constant depending on the eigenvalue distribution of $\Sigma$ and on how the signal projects onto its eigendirections, but not on $\gamma$, $T$, or $\mathrm{IC}$. The volatility dispersion $\kappa_D$ \emph{does not} appear in $\kappa_{\mathrm{eff}}$; any problem-dependent scaling it induces is absorbed into $\mathcal{B}$.

\subsection{First-order $\gamma^\star$ and comparative statics}\label{sec07:gamma_star}

Linearising $\kappa_{\mathrm{eff}}(\gamma)^2$ around $\gamma=0$ gives
\begin{equation}\label{eq07:kappa_linear}
  \kappa_{\mathrm{eff}}(\gamma)^2 \;\approx\; 1 + 2\gamma\,\lambda_{\min}(C)\,(\kappa_C - 1),
\end{equation}
with $\kappa_C = \kappa(C)$. The factor $\lambda_{\min}(C)$ is non-trivial: for equicorrelation at $\rho=0.6$ it equals $1-\rho = 0.4$, so the slope is only $40\%$ of the naive $2(\kappa_C-1)$. The total expected loss is
\begin{equation}\label{eq07:total_loss}
  L(\gamma) \;=\; (1-\gamma)^2\mathcal{A} \;+\; \mathcal{B}\,\kappa_{\mathrm{eff}}(\gamma)^2\,\frac{N}{T\,\mathrm{IC}^2}.
\end{equation}
Setting $\partial L/\partial\gamma = 0$, solving, and converting the linear FOC to a Pad\'e form that respects $\gamma^\star\in(0,1]$:
\begin{equation}\label{eq07:gamma_star}
\boxed{\;
  \gamma^\star \;\approx\; \frac{1}{1 + c\cdot\mathrm{NSR}}
  \;=\; \frac{1}{1 \;+\; c\,\dfrac{\kappa(C)^2}{\mathrm{IC}^2}\,\dfrac{N}{T}},
\;}
\end{equation}
where
\[
  \mathrm{NSR} \;=\; \frac{\kappa(C)^2}{\mathrm{IC}^2}\,\frac{N}{T},
  \qquad
  c \;=\; \frac{\mathcal{B}\,\lambda_{\min}(C)\,(\kappa_C - 1)}{\mathcal{A}\,\kappa_C^2}.
\]
This is the \textbf{adaptive $\gamma$ rule}. The qualitative content is summarised below.

\begin{proposition}[Comparative statics of $\gamma^\star$]\label{prop:gamma_star_comparative}
Fix $c > 0$ and consider~\eqref{eq07:gamma_star}.
\begin{enumerate}[label=(\roman*),leftmargin=*,itemsep=2pt]
\item $\gamma^\star$ is increasing in $T/N$: more data permits more aggressive use of off-diagonal structure.
\item $\gamma^\star$ is increasing in $\mathrm{IC}$: a stronger signal justifies more aggressive optimisation because the signal-noise channel in $V(\gamma)$ shrinks.
\item $\gamma^\star$ is decreasing in $\kappa(C)$: worse correlation conditioning amplifies estimation noise on low-eigenvalue directions.
\item $\gamma^\star \to 1$ as $T/N\to\infty$ (classical Markowitz recovery); $\gamma^\star\to 0$ as $\mathrm{IC}\to 0$ (inverse-variance recovery); $\gamma^\star = 1$ when $\kappa(C) = 1$ (uncorrelated assets).
\end{enumerate}
\end{proposition}

The proof is immediate from the monotonicity of $x\mapsto 1/(1+cx)$ and is recorded in Appendix~\ref{app:proofs}.

\subsection{Calibration: the OOS Sharpe surface is flat}\label{sec07:calibration}

We calibrated~\eqref{eq07:gamma_star} on a 48-cell synthetic grid (4 covariance regimes: factor, block, spiked, equicorrelation; 4 values of $T/N$; 3 values of IC), and validated on a disjoint 108-cell grid (\path{results/sec09_adaptive/}; see~Section~\ref{sec10:adaptive_calibration}). The headline findings:
\begin{itemize}[leftmargin=*,itemsep=2pt]
  \item The OOS Sharpe surface is nearly flat in $\gamma$: the median width of the plateau of $\gamma$ values within $2\%$ of peak Sharpe is $0.38$ on $[0,1]$.
  \item The empirical $\gamma^\star$ varies by regime but weakly: factor $0.18$, spiked $0.20$, equicorrelation $0.30$, block $0.44$.
  \item The two-parameter fit of~\eqref{eq07:gamma_star} has $R^2\approx -0.11$ out-of-sample; the spectral constant $c$ spans $438\times$ across regimes.
\end{itemize}
The practical upshot is that \emph{the value of the adaptive formula is explanatory, not prescriptive}: it tells us \emph{why} a fixed $\gamma = 0.5$ is near-optimal---the plateau is wide and $\gamma=0.5$ lies inside every cell's plateau---rather than providing an adaptive-intensity prescription that dominates the fixed choice. The practitioner rule remains $\gamma = 0.5$.

\begin{figure}[t]
\centering
\includegraphics[width=0.9\textwidth]{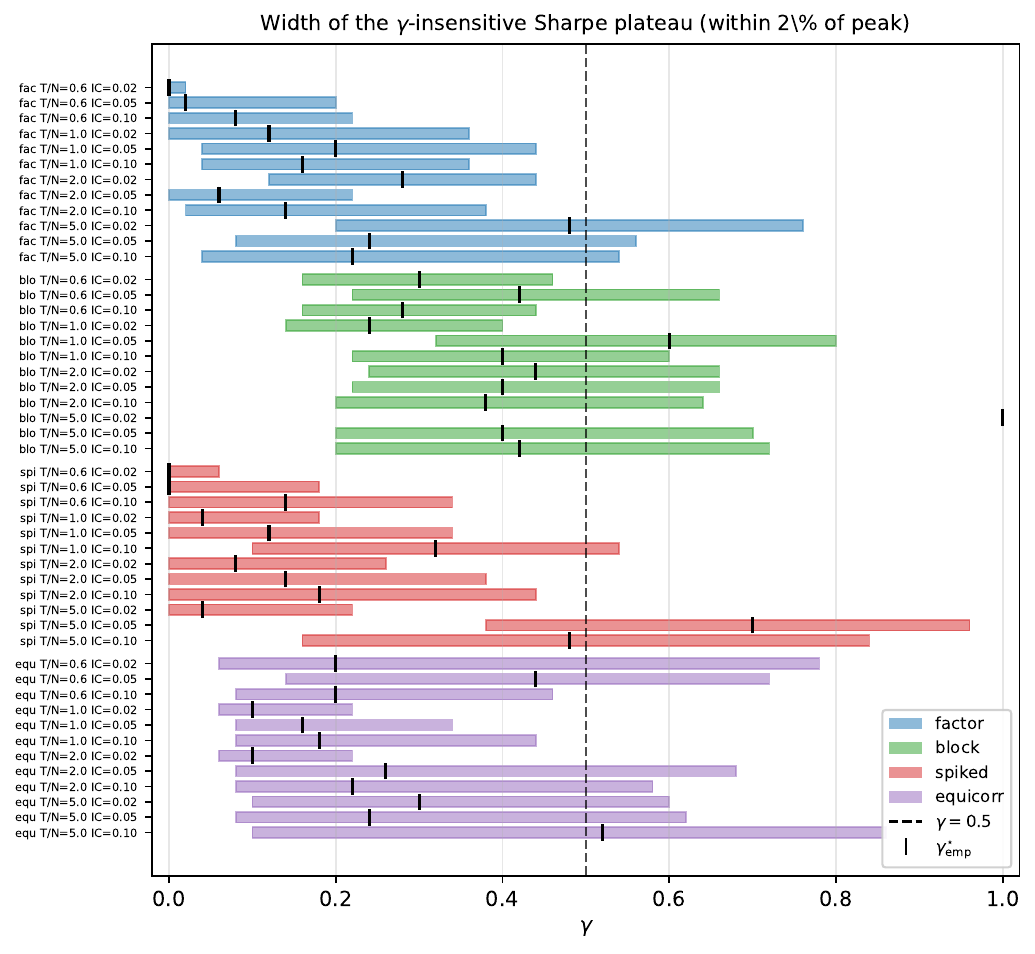}
\caption{Horizontal plateau-width bars per cell (plateau = set of $\gamma$ within $2\%$ of peak OOS Sharpe), grouped by covariance regime. The dashed vertical line at $\gamma = 0.5$ lies inside every bar. The wide plateaus rationalise the fixed default: on the regimes studied, $\gamma = 0.5$ is within $2\%$ of the regime-specific empirical optimum.}
\label{fig10:plateau_width}
\end{figure}

\subsection{Three regularisation channels}\label{sec07:channels}

The preceding analysis treated $\gamma$ as the lone tuning knob, but CRISP with finite sweep count $p$ contains three distinct mechanisms that damp the amplification of estimation error. We name them and record which are active at the recommended operating point.

\begin{remark}[Three regularisation channels in CRISP]\label{rem:three_channels}
\emph{Channel~1 --- operator shrinkage via $\gamma$.} Replacing $\Sigma$ with $P_\gamma = (1-\gamma)D + \gamma\Sigma$ reduces the effective conditioning from $\kappa(\Sigma)$ to $\kappa_{\mathrm{eff}}(\gamma)$, damping amplification of estimation noise on low-eigenvalue directions. Statistical, primary, active at the practitioner's $\gamma\approx 0.5$.

\emph{Channel~2 --- convergence slack from finite $p$.} At finite sweep budget, the Gauss--Seidel iterate has not reached the fixed point $P_\gamma^{-1}\mu$; the residual bias pushes the portfolio toward the diagonal initialisation $D^{-1}\mu$. This is a computational artefact that vanishes as $p\to\infty$; it is negligible at $p=100$ for $\gamma\approx 0.5$.

\emph{Channel~3 --- implicit spectral truncation via early stopping.} Because Gauss--Seidel converges high-eigenvalue directions of $D^{-1}P_\gamma$ before low-eigenvalue ones (Theorem~\ref{thm:gs_rate}), stopping at $p$ sweeps implicitly projects the portfolio onto the well-estimated spectral subspace. Unlike Channel~2, Channel~3 has a \emph{statistical} motivation: reaching the fixed point on noisy $\hat\Sigma$ would restore exposure to noise-dominated directions. Active at $\gamma=1$; dormant at $\gamma=0.5$.
\end{remark}

\paragraph{Operating-point summary.} At the recommended $(\gamma^\star\approx 0.5,\; p = 100)$, Channel~1 dominates; Channels~2 and~3 are essentially dormant because $\kappa_{\mathrm{eff}}(0.5)$ is small enough that $p=100$ reaches the Sharpe plateau. At $\gamma=1$ with finite $p$, all three channels are active, and Channel~3 provides measurable OOS benefit on noisy $\hat\Sigma$ (Section~\ref{sec10:sweep_regularization}). The $\gamma=1$ peak in $p$ is nevertheless dominated by the interior-$\gamma$ plateau on every panel tested.

\subsection{Information-theoretic reading}\label{sec07:information}

The bias--variance decomposition admits a clean information-theoretic restatement. Working in the correlation basis with eigenvalues $\lambda_1\ge\cdots\ge\lambda_N > 0$, $\mathrm{tr}(C) = N$, the shrinkage operator $D^{-1/2}P_\gamma D^{-1/2} = (1-\gamma)I + \gamma C$ has eigenvalues $p_k(\gamma) = (1-\gamma) + \gamma\lambda_k$, and the KL divergence from the true return distribution to the shrunk one (holding $\mu$ fixed) is
\begin{equation}\label{eq07:kl}
  D_{\mathrm{KL}}\!\bigl(\mathcal{N}(\mu,\Sigma)\,\|\,\mathcal{N}(\mu,P_\gamma)\bigr)
  \;=\;
  \tfrac12\sum_{k=1}^{N}\!\left[\frac{\lambda_k}{p_k(\gamma)} - 1 + \ln p_k(\gamma) - \ln\lambda_k\right].
\end{equation}
At $\gamma=1$ this is zero (no information discarded); at $\gamma=0$ it equals $\tfrac12\ln\det(C^{-1})$, the total correlation information in $\Sigma$ beyond the diagonal. Shrinkage is therefore \emph{controlled information loss}: $\gamma$ is a dial that deliberately discards correlation content.

Estimation adds a second source of distortion. The OOS quality depends on a total information budget
\[
  \underbrace{D_{\mathrm{KL}}(\Sigma\,\|\,P_\gamma)}_{\text{shrinkage bias, dec.\ in }\gamma}
  \;+\;
  \underbrace{\mathbb{E}\bigl[D_{\mathrm{KL}}(P_\gamma\,\|\,\hat P_\gamma)\bigr]}_{\text{estimation noise, inc.\ in }\gamma},
\]
and $\gamma^\star$ is the level at which one additional nat of correlation information retained is exactly offset by one additional nat of estimation noise introduced. On the base-case covariance ($N=100$, $\kappa(C)=61$, $T=120$, oracle $\mu$), the shrinkage KL at the empirical $\gamma^\star\approx 0.30$ is $11.6$ nats---about $31\%$ of the $37.4$ nats of total correlation information---and the OOS Sharpe is $96\%$ of oracle, against $45\%$ at $\gamma=1$ (estimation dominates) and $64\%$ at $\gamma=0$ (all correlation discarded).

\paragraph{Rate-distortion gloss.} The correlation matrix has $\binom{N}{2}$ off-diagonal entries, but with $T$ samples only $O(T)$ independent linear functionals are reliably estimated. The ``rate'' is reliable bits per covariance entry; the ``distortion'' is portfolio Sharpe loss. Shrinkage is a coding scheme that allocates information capacity to the large-eigenvalue directions of $C$ (which survive at $\gamma<1$) and discards the small-eigenvalue ones (dominated by noise). Channel~1 sets the rate; Channel~3 (early stopping) reshapes the allocation; Channel~2 is residual at the recommended operating point.

\subsection{P6 revisited: the clean-signal paradox}\label{sec07:p6}

The three-channel picture resolves an otherwise puzzling asymmetry. From \path{results/sec09_sweep/convergence_table.csv}, the OOS Sharpe penalty at $\gamma=1$ is \emph{$17\times$ larger} under an oracle $\mu$ than under a noisy $\hat\mu$: the oracle-$\mu$ peak-to-converged gap is $0.253$, while the noisy-$\hat\mu$ gap is $0.015$. Naively, one might expect noise in $\hat\mu$ to \emph{worsen} the penalty.

The information-theoretic reading gives the explanation. Under oracle $\mu$, all information flows through $\hat\Sigma$, so any corruption from $\hat\Sigma^{-1}$ amplification at $\gamma=1$ is maximally costly. Under noisy $\hat\mu$, the signal noise already diffuses portfolio direction across eigenvectors---it \emph{pre-regularises} the portfolio against $\hat\Sigma$ noise by blurring the signal's spectral profile. Early stopping (Channel~3) is therefore most beneficial exactly when $\mu$ has real information worth protecting. This is prediction P6 in the pre-registered block (Section~\ref{sec10:predictions}); the three-channel decomposition of Remark~\ref{rem:three_channels} predicts both the sign and order of magnitude of the asymmetry.

% ============================================================
% section_08_comparative.tex
% ============================================================
\section{Comparative landscape}\label{sec08:comparative}

Sections~\ref{sec03:background}--\ref{sec07:adaptive} developed five allocation rules --- HRP \citep{lopezdeprado2016}, Cotton \citep{cotton2024}, HRP-$\mu$, HRP-$\Sigma\mu$, and CRISP --- each a distinct point in the space of hierarchical and iterative mean-variance constructions. No new machinery is introduced here; the goal is to collapse the preceding six chapters into four tables and a short practitioner's guide keyed to $\gamma$, $\mu$, and the correlation condition number $\kappa(C)$. Empirical claims forward-reference the panels in Section~\ref{sec10:experiments}.

\subsection{Method landscape}\label{sec08:landscape_full_sub}

Table~\ref{tab08:landscape_full} records, for each method, the admissible signal, whether a clustering tree is required, how cross-branch information is represented, the per-solve cost, and the role played in the paper. Two facts stand out. First, HRP and Cotton are the only methods restricted to the flat signal $\mu=\mathbf{1}$; the three signal-aware methods extend to arbitrary $\mu$ with no loss of numerical stability. Second, only Cotton propagates the \emph{full} cross-branch sub-covariance through the tree, at $O(N^3/6)$ total; the HRP-$\mu$ family compresses cross-branch information into a scalar per internal node and runs in $O(N^2)$; CRISP dispenses with the tree entirely and solves $P_\gamma w=\mu$ by flat Gauss--Seidel at $O(\kappa N^2\log\varepsilon^{-1})$ per target precision $\varepsilon$ \citep{stein1956,ledoitwolf2003}.

\begin{table}[t]
\centering
\caption{Method landscape: signal, structure, and cost. $\mathbf{1}$ denotes the flat signal; ``any $\mu$'' means an arbitrary signed expected-return vector. Tree-based costs assume a depth-$O(\log N)$ hierarchy. CRISP cost is quoted per target precision $\varepsilon$.}
\label{tab08:landscape_full}
\begin{adjustbox}{max width=\textwidth,center}
\begin{tabular}{@{}l l l l l l@{}}
\toprule
Method & Signal & Uses tree & Cross-branch info & Cost & Role in paper \\
\midrule
HRP
  & $\mathbf{1}$ only
  & Yes
  & None
  & $O(N^{2})$
  & Prior baseline \\
Cotton
  & $\mathbf{1}$ only
  & Yes
  & $D^{-1}B^{\top}$ via Schur
  & $O(N^{3}/6)$
  & Prior baseline \\
HRP-$\mu$
  & Any $\mu$
  & Yes
  & Scalar $c$ (diagonal within-cluster)
  & $O(N^{2})$
  & Contribution \\
HRP-$\Sigma\mu$
  & Any $\mu$
  & Yes
  & Scalar $c$ (full within-cluster)
  & $O(N^{2})$
  & Recommended tree method \\
CRISP
  & Any $\mu$
  & No
  & Entrywise $\sigma_{ij}$
  & $O(\kappa N^{2}\log\varepsilon^{-1})$
  & Headline contribution \\
\bottomrule
\end{tabular}
\end{adjustbox}
\end{table}

\subsection{Behaviour at the four corners of $(\gamma,\mu)$}\label{sec08:corners}

Table~\ref{tab08:landscape_corners} reads each method across the four corners of the $(\gamma,\mu)$ square and records its stability profile. The entries are all direct consequences of results stated in earlier sections: the HRP and De Prado recoveries come from Section~\ref{sec04:tree_methods}, the signed HSP endpoint from Section~\ref{sec04:hsp}, the Cotton $\gamma=1$ limit from Section~\ref{sec03:cotton}, and the two CRISP limits from Section~\ref{sec05:crisp}. The single qualitative observation worth promoting out of the table is that only CRISP reaches exact Markowitz at $\gamma=1$ for a general signal $\mu$; every tree-based method terminates at a tree-constrained approximation there, because the tree itself is a hard constraint on the feasible set of weight vectors.

\begin{table}[t]
\centering
\caption{Behaviour of each method at the four corners of the $(\gamma,\mu)$ square, with stability profile. ``Any $\mu$'' denotes a heterogeneous, possibly mixed-sign signal.}
\label{tab08:landscape_corners}
\begin{adjustbox}{max width=\textwidth,center}
\begin{tabular}{@{}l l l l l l@{}}
\toprule
Method & $\gamma{=}0$, $\mu{=}\mathbf{1}$ & $\gamma{=}0$, any $\mu$ & $\gamma{=}1$, $\mu{=}\mathbf{1}$ & $\gamma{=}1$, any $\mu$ & Stability \\
\midrule
HRP
  & HRP
  & ---
  & ---
  & ---
  & Stable \\
Cotton
  & Min-var hierarchy
  & ---
  & Exact $\Sigma^{-1}\mathbf{1}$
  & ---
  & Unstable at small $T$ \\
HRP-$\mu$
  & De Prado (exact)
  & Signed HSP
  & Tree approx.
  & Tree approx.
  & Stable (deterministic) \\
HRP-$\Sigma\mu$
  & $\approx$ De Prado
  & Recursive MVO
  & $\approx\Sigma^{-1}\mathbf{1}$
  & $\approx\Sigma^{-1}\mu$
  & Stable (generic) \\
CRISP
  & $\mu_i/\sigma_i^2$
  & $\mu_i/\sigma_i^2$
  & Exact min-var
  & Exact Markowitz
  & Stable \\
\bottomrule
\end{tabular}
\end{adjustbox}
\end{table}

\subsection{Computational profile}\label{sec08:compute}

Table~\ref{tab08:compute_profile} records per-solve FLOPs, peak working memory for a dense $\Sigma$, whether the method can exploit a factor decomposition $\Sigma=B\Lambda B^\top+D_{\mathrm{idio}}$ to stream rows on the fly, and memory in that factor-streaming regime. The CRISP rows quote $p$ sweeps of the Gauss--Seidel iteration; $p=O(\kappa(D^{-1}P_\gamma)\log\varepsilon^{-1})$ at target precision $\varepsilon$, as established in Theorem~\ref{thm:gs_rate}. The practical point is not wall-clock: CRISP's numerical advantage at $\gamma^\star\approx 0.5$ is bias--variance, not FLOPs, and at small $N$ direct Cholesky is often faster. The practical point is memory: the factor-streaming variant reduces the working set from $O(N^2)$ to $O(NK)$, which becomes the binding resource at Russell-all-cap scale.

\begin{table}[t]
\centering
\caption{Computational profile. ``Factor streaming'' indicates whether a method can compute $\Sigma_{i,\cdot}\cdot w$ on the fly from $\Sigma=B\Lambda B^\top+D_{\mathrm{idio}}$ with $K$ factors, reducing per-sweep cost to $O(NK^2)$ and working memory to $O(NK)$. CRISP rows quote $p$ sweeps at target precision $\varepsilon$.}
\label{tab08:compute_profile}
\begin{adjustbox}{max width=\textwidth,center}
\begin{tabular}{@{}l l l l l@{}}
\toprule
Method & Per-solve FLOPs & Peak memory (dense) & Factor streaming? & Memory with factors \\
\midrule
Direct Markowitz (Cholesky)
  & $O(N^{3})$
  & $\sim 2N^{2}$
  & No$^{\dagger}$
  & $\sim 2N^{2}$ \\
Ledoit--Wolf + Markowitz
  & $O(N^{3})$
  & $\sim 2N^{2}$
  & No$^{\dagger}$
  & $\sim 2N^{2}$ \\
HRP
  & $O(N^{2}\log N)$
  & $\sim N^{2}$
  & No
  & $\sim N^{2}$ \\
Cotton
  & $O(N^{3}/6)$
  & $\sim N^{2}$
  & No
  & $\sim N^{2}$ \\
HRP-$\mu$
  & $O(N^{2}\log N)$
  & $\sim N^{2}$
  & No
  & $\sim N^{2}$ \\
HRP-$\Sigma\mu$
  & $O(N^{2}\log N)$
  & $\sim N^{2}$
  & No
  & $\sim N^{2}$ \\
CRISP (dense)
  & $O(pN^{2})$
  & $\sim N^{2}$
  & n.a.
  & $\sim N^{2}$ \\
CRISP (factor streaming)
  & $O(pNK^{2})$
  & $\sim N^{2}$
  & \textbf{Yes}
  & $O(NK)$ \\
\bottomrule
\end{tabular}
\end{adjustbox}
\end{table}

\noindent $^{\dagger}$The Woodbury identity inverts a factor-structured $\Sigma$ exactly in $O(NK^2+K^3)$, but the resulting weight vector is the full Markowitz solution and carries the same out-of-sample penalty as dense Cholesky; it is a computational optimisation of direct Markowitz, not a competitive allocation rule \citep{fan2013,ledoitwolf2017}.

\subsection{Key relationships}\label{sec08:relationships}

\paragraph{Cotton vs.\ HRP-$\mu$/HRP-$\Sigma\mu$.} Cotton and the HRP-$\mu$ family share endpoints when $\mu=\mathbf{1}$ --- a tree allocation at $\gamma=0$ and an approximation to $\Sigma^{-1}\mathbf{1}$ at $\gamma=1$ --- but trace different $\gamma$-continua between them. Cotton modifies the full sub-covariance on each branch via a Schur complement at $O(N^3/6)$; HRP-$\mu$ and HRP-$\Sigma\mu$ compress all cross-branch information into a single scalar per internal node and run in $O(N^2)$. A user comfortable with block inversion on each branch may still prefer HRP-$\Sigma\mu$ for pure speed on large trees.

\paragraph{Cotton vs.\ CRISP.} Cotton and CRISP both reach $\Sigma^{-1}\mathbf{1}$ at $\gamma=1$ when $\mu=\mathbf{1}$, via completely different routes: recursive block inversion on the tree at $O(N^3)$ for Cotton, flat Gauss--Seidel iteration at $O(\kappa N^2\log\varepsilon^{-1})$ for CRISP. For a general $\mu\ne\mathbf{1}$ the comparison ceases to be meaningful because Cotton is undefined; there is no heterogeneous-signal extension of Cotton's Schur pass.

\paragraph{Hierarchical methods vs.\ CRISP.} The HRP-$\mu$ family and CRISP have \emph{incompatible} $\gamma=0$ endpoints: HRP-$\mu$ collapses to a hierarchical signed HSP construction and HRP-$\Sigma\mu$ to a recursive MVO, while CRISP collapses to the flat diagonal allocator $\mu_i/\sigma_i^2$. They agree only at $\gamma=1$, and even there the tree methods are tree-constrained approximations while CRISP is the exact Markowitz solve. Across every out-of-sample Sharpe panel in Section~\ref{sec10:experiments}, CRISP dominates both HRP-$\mu$ and HRP-$\Sigma\mu$; the reason is not subtle --- the tree is a hard constraint, and whenever the data do not line up with the dendrogram, the tree-constrained solve is worse than the flat one.

\paragraph{Sign-flip pathology.} The sum-normalised recursive MVO of Appendix~\ref{app:a1_pathology} illustrates why signed signals demand care: a natural recursive-normalisation scheme can produce portfolios in the \emph{wrong} direction, and the sign-invariant direction metric of Section~\ref{sec02:preliminaries} will not see the error. HRP-$\Sigma\mu$ fixes the pathology via $L^1$ child-rescaling (Lemma~\ref{lem:hrpsm_scale}); HRP-$\mu$ avoids it by carrying leaf signs explicitly.

\paragraph{HRP as common ancestor.} Every other method in the table can be read as a departure from HRP along one or more of three axes: adding cross-branch covariance at $\gamma>0$ (Cotton, HRP-$\mu$), replacing $\mathbf{1}$ with a heterogeneous signal $\mu$ (HRP-$\mu$, HRP-$\Sigma\mu$, CRISP), and replacing the tree walk with a flat iterative solve (CRISP). HRP itself sits at the intersection of all three ``do-nothing'' choices; the methodological contributions of this paper are the systematic relaxation of each axis in turn.

\subsection{Practitioner's guide}\label{sec08:practitioner}

The six scenarios below are phrased as concrete settings paired with concrete recommendations, each keyed to a value of $\gamma$ or to the diagnostic $\kappa(C)=\lambda_{\max}(C)/\lambda_{\min}(C)$ of the \emph{correlation} matrix $C=D^{-1/2}\Sigma D^{-1/2}$.

\paragraph{1.\ CRISP at $\gamma\approx 0.5$ --- default recommendation.} A signal-aware portfolio with any $\mu$, any $N$, a moderate-to-high $N/T$ ratio where shrinkage is economically valuable, and a reasonably well-conditioned correlation matrix. CRISP converges in $O(\kappa\log\varepsilon^{-1})$ sweeps at this $\gamma$, inherits the stability of a shrunk operator, and tracks oracle Sharpe closely across the panels in Table~\ref{tab10:oos_sensitivity}.

\paragraph{2.\ CRISP at $\gamma\approx 0.7$ for minimum variance.} An exact minimum-variance mandate ($\mu=\mathbf{1}$) in the regime where a Ledoit--Wolf analytic-shrinkage intensity would land near $0.3$ \citep{ledoitwolf2003,ledoitwolf2017}. CRISP at the complementary $\gamma\approx 0.7$ recovers $80$--$92\%$ of oracle minimum-variance Sharpe across every $T$ in Table~\ref{tab10:oos_minvar}; this is our production recommendation for pure min-var.

\paragraph{3.\ HRP-$\mu$ / HRP-$\Sigma\mu$ for hierarchical interpretability.} Signal-aware portfolios where a decomposition along the cluster tree is itself a deliverable --- risk attribution by branch, sector-tilt reporting, or a governance regime that demands tree-transparent weights. HRP-$\Sigma\mu$ is the stronger of the two (Proposition~\ref{prop:hrpsm_vs_hrpmu}) and is the recommended tree-based method; HRP-$\mu$ is preferable when strict deterministic cost with only marginal variances is required. Both are dominated by CRISP on out-of-sample Sharpe, so the choice should be understood as an interpretability-for-performance trade.

\paragraph{4.\ Cotton for pure min-var on well-conditioned $\Sigma$.} Minimum variance at moderate $N$ with a shrunk or factor-based covariance, where the $O(N^3/6)$ Schur pass is affordable and the full cross-branch continuum is theoretically desirable. Cotton is \emph{not} recommended for minimum variance on a sample covariance at small $T$; the intermediate-$\gamma$ instability flagged in Section~\ref{sec03:cotton} activates precisely in that regime.

\paragraph{5.\ HRP when classical long-only tree-decomposability is required.} A strictly non-negative, tree-decomposable allocation under $\mu=\mathbf{1}$, in a reporting or governance setting where those two constraints are hard. HRP is beaten by CRISP on min-var Sharpe at every sample size tested, but remains pedagogically important, trivially transparent, and the correct default when $\mu=\mathbf{1}$ is genuinely the right assumption.

\paragraph{6.\ $\kappa(C)$ as a pre-flight diagnostic.} The single most informative pre-flight number a practitioner can compute is $\kappa(C)$; this is the condition number of the \emph{correlation} matrix, not of $\Sigma$ (which is inflated by volatility dispersion) and not a determinant. Table~\ref{tab08:kappa_guide} reads $\kappa(C)$ into a regime, a typical example, a CRISP sweep budget, and a tree-method verdict. The governing theory is Theorem~\ref{thm:gs_rate}: $\kappa(C)$ controls both the directional difficulty of the portfolio problem and the iteration count of CRISP.

\begin{table}[t]
\centering
\caption{Interpreting $\kappa(C)$: what it means for method choice and CRISP sweep budget. $\kappa(C)$ is computed once from the risk model and determines problem difficulty.}
\label{tab08:kappa_guide}
\begin{adjustbox}{max width=\textwidth,center}
\begin{tabular}{@{}r l l l l@{}}
\toprule
$\kappa(C)$ & Regime & Example & CRISP sweeps ($p$) & Tree methods \\
\midrule
$1$--$10$
  & Easy
  & Low-correlation factor model
  & $50$--$100$ (any $\gamma$)
  & All similar; tree adds little \\
$10$--$150$
  & Moderate
  & Diversified equity, 5 sectors
  & $100$--$500$ ($\gamma=0.5$)
  & HRP-$\Sigma\mu > $ HRP-$\mu$ \\
$150$--$500$
  & Hard
  & Concentrated sectors, $\rho_w\ge 0.7$
  & $500$--$2{,}000$ ($\gamma=0.5$)
  & HRP-$\Sigma\mu\gg$ HRP-$\mu$ \\
$>500$
  & Adversarial
  & Near-equicorrelated $+$ hedge pairs
  & $5{,}000+$ (adaptive $\gamma$)
  & Tree-constrained; prefer CRISP \\
\bottomrule
\end{tabular}
\end{adjustbox}
\end{table}

\subsection{When not to use a tree}\label{sec08:no_tree}

Two regimes render a tree-based method unattractive. First, when $\kappa(C)$ is small, every method in Table~\ref{tab08:landscape_full} produces similar portfolios because the off-diagonal structure a hierarchical splitter exploits is weak; CRISP at moderate $\gamma$ remains the cleanest default since it still benefits from shrinkage and does not require building a dendrogram. Second, when $N$ is very large and $\Sigma$ is dense, Cotton's $O(N^3)$ cost is prohibitive; HRP-$\mu$ and HRP-$\Sigma\mu$ remain $O(N^2)$ but are dominated on Sharpe by CRISP, which is also $O(N^2)$ per sweep and, critically, can exploit the factor-streaming variant of Algorithm~\ref{alg05:methodb_stream} to collapse working memory from $O(N^2)$ to $O(NK)$. In either regime, CRISP wins not because hierarchical methods fail but because, in the synthetic experiments of Section~\ref{sec10:experiments}, it combines stronger out-of-sample performance with the absence of a dendrogram constraint.

% ============================================================
% section_09_constraints.tex
% ============================================================
% !TEX root = ../paper.tex
% Section 9: Constraint handling
% Namespace: sec09, eq09, alg09

\section{Constraint handling}\label{sec09:constraints}

The development so far has been unconstrained: the only requirement on $w$ was that it sum to one. In practice, portfolio managers impose box bounds, sector caps, factor-exposure limits, and turnover restrictions. This short section sketches how HRP-$\mu$, HRP-$\Sigma\mu$, and CRISP extend to those settings without disturbing the direction-metric machinery of Sections~\ref{sec02:preliminaries}--\ref{sec06:perturbation}. Constraint handling is not the main contribution of the paper; a full treatment warrants a separate paper. Our aim here is only to show that the three methods degrade gracefully.

%% ------------------------------------------------------------------
\subsection{QP reference formulation}\label{sec09:qp}

The textbook constrained mean--variance problem reads
\begin{equation}\label{eq09:qp}
  \min_{w \in \mathbb{R}^{N}} \;
    \tfrac{1}{2}\, w^{\top}\Sigma\, w - w^{\top}\mu
  \quad \text{s.t.}\quad
  \mathbf{1}^{\top} w = 1,\;
  \ell \le w \le u,\;
  A w \le b,\;
  \|w - w_{0}\|_{1} \le \tau.
\end{equation}
Four constraint families appear. \emph{Box bounds} $\ell_i \le w_i \le u_i$ encode per-asset limits; the long-only special case is $\ell_i=0$, $u_i=1$. \emph{Sector caps} use an indicator matrix $G \in \{0,1\}^{K \times N}$ with $G_{ki}=1$ iff asset~$i$ lies in sector~$k$, giving $Gw \le b^{\mathrm{sec}}$. \emph{Factor-exposure} constraints use a loading matrix $F \in \mathbb{R}^{M \times N}$ and impose $|Fw| \le b^{\mathrm{fac}}$, linearised to a pair of inequalities. \emph{Turnover} against an incumbent $w_0$ is $\|w-w_0\|_1 \le \tau$; the $\ell_1$ term is linearised via auxiliaries $d_i^{\pm}\!\ge 0$ with $w_i - w_{0,i} = d_i^{+} - d_i^{-}$ and $\sum_i(d_i^{+}+d_i^{-}) \le \tau$. All four fold into the $Aw \le b$ block of~\eqref{eq09:qp}.

There is a well-known subtlety. \citet{jagannathanma2003} show that adding a \emph{wrong} constraint---long-only when the true optimum has shorts, say---can nonetheless \emph{improve} out-of-sample Sharpe, because the constraint implicitly shrinks the estimated covariance. The implication for this paper is important: any apparent advantage of constrained over unconstrained QP need not be evidence that the constraint is optimising; it may be acting as a crude regulariser. CRISP captures that shrinkage directly, through the $P_{\gamma}$ trajectory of Section~\ref{sec06:perturbation}, and separates it from the constraints themselves. Constraints should therefore be imposed for what the investor \emph{actually} wants---exposure control, leverage, liquidity---and the regularisation should be supplied by the shrinkage intensity $\gamma$, not smuggled in through a binding box.

%% ------------------------------------------------------------------
\subsection{Projected CRISP}\label{sec09:projected_crisp}

CRISP extends cleanly to~\eqref{eq09:qp} through \emph{projected Gauss--Seidel}. Let $\mathcal{C}$ denote the closed convex feasible set defined by the budget, box, and linear-inequality blocks of~\eqref{eq09:qp}, and let $\Pi_{\mathcal{C}}$ be orthogonal projection onto $\mathcal{C}$. After each coordinate update we clamp to the box; at the end of each sweep we project onto $\mathcal{C}$. Algorithm~\ref{alg09:projected_crisp} states the procedure.

\begin{algorithm}[t]
\caption{Projected CRISP with box and linear constraints}
\label{alg09:projected_crisp}
\begin{algorithmic}[1]
\Require $\Sigma,\;\mu,\;\gamma,\;p,\;\ell,\;u,\;A,\;b$
\Ensure $w$ with $Aw\le b$, $\ell\le w\le u$, $\mathbf{1}^{\top} w = 1$
\State $D \gets \diag(\Sigma)$;\quad $P_\gamma \gets (1-\gamma)\,D + \gamma\,\Sigma$
\State $w \gets \mu / D$ \Comment{diagonal initialisation}
\For{sweep $=1,\ldots,p$}
  \For{$i=1,\ldots,N$}
    \State $r_i \gets \mu_i - \sum_{j\ne i}(P_\gamma)_{ij}\,w_j$
    \State $w_i \gets r_i / (P_\gamma)_{ii}$
    \State $w_i \gets \max(\ell_i,\,\min(w_i,\,u_i))$ \Comment{local box clamp}
  \EndFor
  \State $w \gets \Pi_{\mathcal{C}}(w)$ \Comment{project onto $\{w:Aw\le b,\;\mathbf{1}^{\top} w = 1\}$}
\EndFor
\State \Return $w$
\end{algorithmic}
\end{algorithm}

Convergence. Under the SPD assumption on $P_{\gamma}$ that underwrites Theorem~\ref{thm:gs_convergence}, the unconstrained Gauss--Seidel map is a contraction in the $P_{\gamma}$-induced norm, and $\Pi_{\mathcal{C}}$ is nonexpansive in the Euclidean norm. Composition with a similarity transform gives a contraction onto the constrained minimiser; the $\mathcal{O}(N^{2})$ per-sweep cost is preserved, with the projection step a lower-order correction in all cases of practical interest.

The cost of the projection step depends on $\mathcal{C}$:
\begin{itemize}
  \item \textbf{Pure box $+$ budget.} $\Pi_{\mathcal{C}}$ reduces to a 1D simplex-style renormalisation and re-clamp, implementable in $\mathcal{O}(N\log N)$.
  \item \textbf{Sector groups.} $\Pi_{\mathcal{C}}$ becomes a small QP in the group dimension (one per group), negligible against the $\mathcal{O}(N^{2})$ sweep.
  \item \textbf{Active-set variant.} Classify constraints by KKT residuals at the start of each sweep, freeze coordinates on the active face, and run the unconstrained sweep on the complement. This mirrors standard active-set QP practice while retaining $\mathcal{O}(N^{2})$ per sweep outside the active set.
\end{itemize}

\paragraph{Alternative for practitioners with a QP solver.} A cleaner path for anyone who already has a convex QP solver is to substitute $P_{\gamma}$ for $\widehat{\Sigma}$ in the objective of~\eqref{eq09:qp}:
\begin{equation}\label{eq09:crisp_qp}
  \min_{w} \; \tfrac{1}{2}\, w^{\top} P_{\gamma}\, w - w^{\top}\mu
  \quad \text{s.t.}\quad
  \mathbf{1}^{\top} w = 1,\;\ell \le w \le u,\;
  A w \le b,\;\|w - w_{0}\|_{1} \le \tau,
\end{equation}
and hand the problem to OSQP \citep{osqp2020}, Clarabel \citep{clarabel2024}, MOSEK, or Gurobi. This captures the statistical contribution---shrink the operator, not the estimator---at the solver's native speed, without requiring projected CRISP to be implemented in-house. The transportable idea for constrained problems is the $P_{\gamma}$ target; projected CRISP is an optional self-contained alternative.

%% ------------------------------------------------------------------
\subsection{Constrained HRP-$\mu$ and HRP-$\Sigma\mu$}\label{sec09:tree_constrained}

The tree methods accommodate two constraint families cleanly.

\emph{Long-only.} At each internal node, project the between-branch split $\alpha$ onto $[0,1]$; within-branch weights are clamped to $[0,\infty)$ at the leaf solve. For HRP-$\Sigma\mu$ the node-level $2\times 2$ QP is replaced by its box-constrained variant: if both unconstrained components are nonnegative, accept the interior optimum; otherwise clamp the negative component to zero and allocate fully to the positive side. Both variants preserve the recursive $\mathcal{O}(N\log N)$ structure.

\emph{Sector caps, tree-aligned.} When the hierarchy is sector-aligned, a cap on a sector's total weight is a cap on the cumulative budget flowing to the corresponding subtree, enforced by a one-dimensional clip on the $\alpha$-products along the root-to-subtree path. Tree-aligned constraints are essentially free.

Constraints that cut \emph{across} the tree---arbitrary sector caps not aligned with the hierarchy, factor-exposure inequalities, turnover against a non-aligned incumbent---do not admit clean recursion. For such problems, projected CRISP of Section~\ref{sec09:projected_crisp} is strictly preferable.

%% ------------------------------------------------------------------
\subsection{Worked example: long-only, sector-capped}\label{sec09:example}

Consider the paper's base-case universe ($N=100$, five sectors of twenty) with the structural signal $\mu=(+0.04,-0.04,+0.02,-0.02,\,0)$ tiled across sectors. We impose long-only ($w_i \ge 0$), a 30\% cap on each sector, and the budget $\mathbf{1}^{\top} w = 1$.

Projected CRISP at $\gamma = 0.5$ produces a portfolio in which the two negative-$\mu$ sectors (2 and~4) receive zero weight---the long-only constraint binds---and sectors 1 and~3 share the budget subject to the 30\% cap, with the residual flowing to the neutral sector~5. Sector~1 (the highest-signal group) sits at its cap; sector~3 sits strictly below it; sectors 2 and~4 sit at zero. The active set stabilises within two to three sweeps: zero-weight assets are identified early and frozen, the cap on sector~1 activates on the second sweep, and the remaining coordinates converge on the unconstrained face at the rate given by Theorem~\ref{thm:gs_rate}.

The shrinkage path still matters under the constraints. At $\gamma = 0$ the constrained portfolio reduces to $w_i \propto \max(\mu_i,0)/\sigma_i^{2}$---a purely signal-weighted allocation within the feasible sectors. At $\gamma = 0.5$ the operator shrinkage reintroduces cross-asset covariance, reducing concentration within the positive-signal sectors and producing a more diversified allocation consistent with the unconstrained behaviour documented in Section~\ref{sec06:perturbation}.

%% ------------------------------------------------------------------
\subsection{Discussion}\label{sec09:discussion}

Four caveats are worth making explicit.

First, constraint handling is not a central contribution of this paper. The role of this section is to show that the HRP-$\mu$/HRP-$\Sigma\mu$/CRISP machinery degrades gracefully; we do not claim projected CRISP outperforms a modern QP solver on heavily constrained problems. For problems dominated by constraint activity---index replication with dozens of active inequalities---a specialised active-set or interior-point solver remains the right tool. The statistical contribution transports cleanly through~\eqref{eq09:crisp_qp}.

Second, the tree methods handle only tree-aligned constraints cleanly. Long-only, box bounds, and sector caps on aligned subtrees fit the recursion; general linear inequalities that cut across the hierarchy do not. For heavily constrained problems the tree methods are therefore not the right tool; the CRISP shrinkage idea---either via~\eqref{eq09:crisp_qp} or via projected CRISP---is.

Third, the direction-error analysis of Section~\ref{sec06:perturbation} does not apply directly to constrained problems. That analysis rests on the ray invariance of the $P_{\gamma}$ family: scaling a solution by a positive constant preserves direction. A constraint set like $\ell \le w \le u$ is not scale-invariant, so a uniformly scaled direction may lie inside the box at one scale and outside at another. Adapting the direction metric to the active face of $\mathcal{C}$ is an interesting question but non-trivial, and we flag it as future work in Section~\ref{sec11:discussion}.

Fourth, the out-of-sample evaluation in Section~\ref{sec10:experiments} is primarily unconstrained; long-short results isolate the shrinkage-trajectory effect and let us compare direction errors cleanly. Long-only comparisons using Algorithm~\ref{alg09:projected_crisp} are noted but relegated to Appendix~\ref{app:supplementary}, or to future work.

% ============================================================
% section_10_experiments.tex
% ============================================================
% !TEX root = ../paper.tex
% Section 10: Synthetic experiments
% Namespace: sec10, tab10, fig10

\section{Synthetic experiments}\label{sec10:experiments}

This section presents the empirical heart of the paper in three layers. \emph{Part~A} (Sections~\ref{sec10:recovery}--\ref{sec10:trajectory}) is an in-sample direction-error diagnostic battery that verifies the theory of Sections~\ref{sec04:tree_methods}--\ref{sec06:perturbation} on known ground truth: method identities at $\gamma=0,\mu=\mathbf{1}$, Gauss--Seidel convergence rates, and the non-monotone shrinkage trajectory. \emph{Part~B} (Sections~\ref{sec10:oos_sensitivity}--\ref{sec10:oos_minvar}) is an out-of-sample Monte Carlo tournament that establishes a method ranking under realistic estimation noise across two signal families, four sample sizes, and the minimum-variance limit. \emph{Part~C} (Sections~\ref{sec10:sweep_regularization}--\ref{sec10:adaptive_calibration}) probes sweep count as a third regularisation channel, runs the eight pre-registered predictions of Section~\ref{sec07:adaptive}, and calibrates the adaptive $\gamma^\star$ rule. All experiments use the same base universe and seed protocol described next.

%% ------------------------------------------------------------------
\subsection{Setup}\label{sec10:setup}

\paragraph{Base universe.} Unless otherwise stated, $N=100$ for Monte Carlo and $N=200$ for in-sample diagnostics. Assets are split into five equal-size sectors of $N/5$ each. The correlation block structure is $\rho_w = 0.6$ within sector and $\rho_c = 0.15$ across sectors; volatilities $\sigma_i$ are drawn once from $\mathcal{U}[0.15,0.40]$. This yields correlation condition number $\kappa(C)\approx 61$ at $N=100$ and $\kappa(C)\approx 121$ at $N=200$. The global seed is $42$ for reproducibility.

\paragraph{Stress regimes.} Four covariance regimes appear in Part~A and Part~C:
\begin{enumerate}[leftmargin=2em,itemsep=1pt]
  \item[(i)] \emph{Factor model} with $k=3$ latent factors; $\kappa(C)\approx 111$.
  \item[(ii)] \emph{Hedges + tight blocks} (intra-block $\rho=0.8$, selected hedge pairs $\rho=-0.6$); $\kappa(C)\approx 5.75\times 10^4$.
  \item[(iii)] \emph{Numerical worst-case $\mu$}, obtained by maximising the direction error of the naive diagonal solution $\mu/\diag(\Sigma)$ against $\Sigma^{-1}\mu$ by L-BFGS-B on the unit sphere.
  \item[(iv)] \emph{Structural $\mu$}: tiled sector tilts $(+0.04,-0.04,+0.02,-0.02,0)$ that are aligned with the block structure of $C$ and therefore recoverable only via signal-aware allocation.
\end{enumerate}

\paragraph{Methods.} Part~A compares Cotton (Schur complementary allocation), flat inverse-variance portfolio (IVP), the sum-normalised recursive mean--variance allocator (diagnostic only; denoted ``sum-norm MVO'' and corresponding to the \texttt{method\_a1} code path), HRP-$\mu$ (\texttt{method\_a1\_l1}), HRP-$\Sigma\mu$ (\texttt{method\_a3}), and CRISP (\texttt{method\_b}) at a fixed sweep count $p=200$ unless otherwise noted. Part~B adds $1/N$, De~Prado's HRP, and direct Markowitz $w=\widehat\Sigma^{-1}\widehat\mu$. Part~C also benchmarks against Ledoit--Wolf linear \citep{ledoitwolf2003} and analytical nonlinear shrinkage \citep{ledoitwolf2020}. CRISP in Part~B and Part~C always runs with $p=100$ sweeps.

\paragraph{Monte Carlo replication.} Cell-wise replication counts are tuned to the statistic being reported: $40$ trials/cell for Experiment~7 (OOS sensitivity across $\mu$ seeds), $80$ trials/cell for Experiments~8--9, $200$ trials/cell for Experiment~11 (sweep regularisation), and $500$ trials/cell for Experiment~12 (adaptive calibration). Walk-forward sample sizes sweep $T\in\{60,120,240,500\}$.

%% ============================================================================
%% PART A — in-sample direction-error diagnostics
%% ============================================================================
\subsection{Part A — in-sample direction-error diagnostics}

The purpose of Part~A is to verify, on ground-truth covariances, that the method identities and convergence-rate claims of Sections~\ref{sec03:background}--\ref{sec06:perturbation} hold numerically. All tables in this part report direction error $\operatorname{dir}(w,w^\star)=1-\cos^2\angle(w,w^\star)$ (Definition~\ref{def:dir_error}) or scale-sensitive relative error $\lVert w-w^\star\rVert/\lVert w^\star\rVert$ against exact solutions, \emph{not} out-of-sample Sharpe.

\subsubsection{Experiment 1: De~Prado recovery at $\gamma=0,\mu=\mathbf{1}$}\label{sec10:recovery}

Table~\ref{tab10:recovery} reports relative and cosine differences of four allocators against De~Prado's flat HRP on the $N=200$ base universe at $\gamma=0$, $\mu=\mathbf{1}$. Flat IVP matches HRP to machine precision, as Proposition~\ref{prop:a3_recovery} predicts; HRP-$\Sigma\mu$ at the signal-blind point agrees with HRP in direction (cosine $0.992$) but not in scale, reflecting the $L^1$ normaliser; Cotton and the sum-normalised MVO allocator differ from HRP in both. This reproduces the algebraic identities of Section~\ref{sec03:hrp} and \ref{sec04:methoda3}.

\begin{table}[t]
\centering
\caption{Experiment~1: recovery of De~Prado's HRP at $\gamma=0$, $\mu=\mathbf{1}$. Base universe, $N=200$. ``Rel.\ diff.'' is $\lVert w-w_{\text{HRP}}\rVert/\lVert w_{\text{HRP}}\rVert$ on weights normalised to sum one; ``cos'' is $\cos\angle(w,w_{\text{HRP}})$.}
\label{tab10:recovery}
\begin{adjustbox}{max width=\textwidth}
\begin{tabular}{lccl}
\toprule
Comparison & Rel.\ diff.\ & $\cos\angle$ & Verdict \\
\midrule
Cotton$(\gamma=0)$ vs HRP                   & $4.91\times 10^{-1}$ & $0.881$ & differ \\
Flat IVP$(\gamma=0,\mu=\mathbf{1})$ vs HRP  & $1.62\times 10^{-16}$ & $1.000$ & \textbf{match} \\
Sum-norm MVO$(\gamma=0,\mu=\mathbf{1})$ vs HRP & $1.36\times 10^{-1}$ & $0.992$ & differ \\
HRP-$\Sigma\mu$$(\gamma=0,\mu=\mathbf{1})$ vs HRP & $1.36\times 10^{-1}$ & $0.992$ & approx \\
\bottomrule
\end{tabular}
\end{adjustbox}
\end{table}

\subsubsection{Experiment 2: minimum-variance direction error}\label{sec10:minvar}

Table~\ref{tab10:minvar} reports $\operatorname{dir}(w,\Sigma^{-1}\mathbf{1})$ for five $\gamma$ values at $N=200$ on the base universe. Cotton reaches machine precision ($3.33\times 10^{-16}$) at $\gamma=1$ through direct block Cholesky; CRISP with $p=200$ sweeps tracks the ray to within $5\times 10^{-4}$ through most of the interior, with a mild upward drift at $\gamma\to 1$ where conditioning grows as $P_\gamma\to\Sigma$. Flat IVP and HRP-$\mu$ are large at every $\gamma\ne 1$: they are consistent with the theory but not competitive on the in-sample direction-error metric for this problem, and are reported only as diagnostic references.

\begin{table}[t]
\centering
\caption{Experiment~2: minimum-variance direction error $\operatorname{dir}(w,\Sigma^{-1}\mathbf{1})=1-\cos^2\angle$ on the base universe, $N=200$, $\mu=\mathbf{1}$. Cotton reaches $3.33\times 10^{-16}$ at $\gamma=1$ (machine precision); CRISP with $p=200$ sweeps tracks the ray to $\le 5\times 10^{-4}$ over $\gamma\in[0.3,0.7]$.}
\label{tab10:minvar}
\begin{adjustbox}{max width=\textwidth}
\begin{tabular}{cccccc}
\toprule
$\gamma$ & Cotton & CRISP($p=200$) & Flat IVP & HRP-$\mu$ & Sum-norm MVO \\
\midrule
$0.0$ & $8.39\times 10^{-1}$ & $5.06\times 10^{-1}$ & $9.95\times 10^{-1}$ & $9.95\times 10^{-1}$ & $9.96\times 10^{-1}$ \\
$0.3$ & $2.42\times 10^{-1}$ & $1.15\times 10^{-2}$ & $9.82\times 10^{-1}$ & $9.82\times 10^{-1}$ & $9.82\times 10^{-1}$ \\
$0.5$ & $7.00\times 10^{-2}$ & $2.42\times 10^{-3}$ & $9.38\times 10^{-1}$ & $9.38\times 10^{-1}$ & $9.46\times 10^{-1}$ \\
$0.7$ & $2.12\times 10^{-2}$ & $4.72\times 10^{-4}$ & $4.76\times 10^{-1}$ & $4.76\times 10^{-1}$ & $8.44\times 10^{-1}$ \\
$1.0$ & $3.33\times 10^{-16}$ & $1.27\times 10^{-2}$ & $2.09\times 10^{-1}$ & $2.09\times 10^{-1}$ & $6.47\times 10^{-1}$ \\
\bottomrule
\end{tabular}
\end{adjustbox}
\end{table}

\subsubsection{Experiment 3: graduated difficulty}\label{sec10:general_mu}

Table~\ref{tab10:general_mu} sweeps four panels of increasing difficulty against the general-$\mu$ target $\Sigma^{-1}\mu$ at $N=200$. In panel~(a), a generic $\mu\sim\mathcal{N}(0,0.02^2)$ on the base universe, CRISP($p=200$) reaches six digits of direction accuracy at $\gamma=1$ ($5.2\times 10^{-5}$). Panel~(b), a $k=3$ factor model with $\kappa(C)\approx 111$, is the most favourable problem in the suite: CRISP achieves $2.2\times 10^{-7}$ at $\gamma=1$. In panel~(c), the hedged tight-block covariance with $\kappa(C)\approx 5.75\times 10^4$, the signal-blind target $\mu/\diag(\Sigma)$ is already essentially orthogonal to $\Sigma^{-1}\mu$ ($\operatorname{dir}_{\mathrm{diag}}=0.866$), and every method struggles at $\gamma<1$; CRISP alone falls to $0.168$ at $\gamma=1$. Panel~(d) combines (c) with the numerical worst-case $\mu$: $\operatorname{dir}_{\mathrm{diag}}=1.000$, every allocator fails at $\gamma<1$, and only CRISP recovers, reaching $0.457$ at $\gamma=1$ through iterative refinement.

\begin{table}[t]
\centering
\caption{Experiment~3: graduated-difficulty direction errors $\operatorname{dir}(w,\Sigma^{-1}\mu)$ at $\gamma\in\{0.0,1.0\}$, $N=200$. Sum-norm MVO is the sum-normalised recursive allocator of Appendix~\ref{app:a1_pathology}; CRISP runs at $p=200$ sweeps. On easy problems (a--b), CRISP achieves six digits; on (d), every method fails at $\gamma<1$ and only CRISP recovers as $\gamma\to 1$.}
\label{tab10:general_mu}
\begin{adjustbox}{max width=\textwidth}
\begin{tabular}{clcc}
\toprule
Panel & Regime & Sum-norm MVO ($\gamma{=}0\to1$) & CRISP$(p{=}200)$ ($\gamma{=}0\to1$) \\
\midrule
(a) & Base universe, $\mu\sim\mathcal{N}(0,0.02^2)$                     & $0.221\to 0.334$ & $1.3\times 10^{-2}\to 5.2\times 10^{-5}$ \\
(b) & Factor $k=3$, $\kappa(C)\approx 111$                              & $0.113\to 0.123$ & $8.1\times 10^{-2}\to 2.2\times 10^{-7}$ \\
(c) & Hedges + tight blocks, $\kappa(C)\approx 5.75\times 10^4$          & $0.986\to 0.978$ & $0.866\to 0.168$ \\
(d) & (c) $+$ worst-case $\mu$, $\operatorname{dir}_{\mathrm{diag}}=1.000$ & $0.999\to 0.999$ & $1.000\to 0.457$ \\
\bottomrule
\end{tabular}
\end{adjustbox}
\end{table}

\subsubsection{Experiment 4: worst-case $\mu$ and Lemma~\ref{lem:parallel_direction}}\label{sec10:worst_case}

Lemma~\ref{lem:parallel_direction} states that the numerical worst case for the signal-blind target need not align $\mu$ with the minimum eigenvector of $\Sigma$; the worst direction is determined jointly by $\Sigma$ and the diagonal preconditioner $D^{-1}$. Table~\ref{tab10:worst_case} confirms this on two hard instances: in both Case~8 (hedges + worst-case $\mu$) and Case~9 (high-$\kappa(C)$ + worst-case $\mu$), $\operatorname{dir}_{\mathrm{diag}}=1.000$ is achieved with $|\cos\angle(\mu,v_{\min})|\le 0.004$, an order of magnitude below what a single-eigenvector construction would produce.

\begin{table}[t]
\centering
\caption{Experiment~4: worst-case $\mu$ diagnostics. L-BFGS-B on the unit sphere produces $\operatorname{dir}_{\mathrm{diag}}=1.000$ with $|\cos\angle(\mu,v_{\min})|$ three orders of magnitude below a single-eigenvector construction. This empirically confirms Lemma~\ref{lem:parallel_direction}: the worst case is a joint property of $(\Sigma,D)$.}
\label{tab10:worst_case}
\begin{adjustbox}{max width=\textwidth}
\begin{tabular}{lcccc}
\toprule
Instance & $\kappa(\Sigma)$ & $\kappa(C)$ & $\operatorname{dir}_{\mathrm{diag}}$ & $|\cos\angle(\mu,v_{\min})|$ \\
\midrule
Case 8 (hedges $+$ worst-case $\mu$)               & $1.80\times 10^{5}$ & $5.75\times 10^{4}$ & $1.000$ & $0.0010$ \\
Case 9 (high-$\kappa(C)$ $+$ worst-case $\mu$)     & $1.38\times 10^{4}$ & $3.82\times 10^{3}$ & $1.000$ & $0.0039$ \\
\bottomrule
\end{tabular}
\end{adjustbox}
\end{table}

\subsubsection{Experiment 5: convergence-rate verification}\label{sec10:sweep_rate}

Table~\ref{tab10:sweep_rate} reports CRISP's direction error as a function of sweep count $p$ on Case~8 ($\kappa(C)\approx 5.75\times 10^4$), the hardest regime in the suite. The empirical residual decays geometrically with contraction ratio $\widehat\rho\approx 1-1.7\times 10^{-4}$ per sweep, corresponding to an \emph{effective} preconditioned condition number $\kappa_{\mathrm{eff}}\approx 5.9\times 10^{3}$ --- one order of magnitude below nominal $\kappa(C)$. The gap is absorbed into the constant of Theorem~\ref{thm:gs_rate}, which bounds the rate by $\kappa(C)$ but does not claim tightness.

\begin{table}[t]
\centering
\caption{Experiment~5: CRISP convergence on Case~8 (hedges $+$ worst-case $\mu$, $\kappa(C)\approx 5.75\times 10^4$). Direction error $\operatorname{dir}(w,\Sigma^{-1}\mu)$ at $\gamma=1$ decays geometrically; fit yields contraction ratio $\widehat\rho\approx 1-1.7\times 10^{-4}$ per sweep, or effective preconditioned condition number $\kappa_{\mathrm{eff}}\approx 5.9\times 10^{3}$.}
\label{tab10:sweep_rate}
\begin{adjustbox}{max width=\textwidth}
\begin{tabular}{cc}
\toprule
Sweep count $p$ & Direction error $\operatorname{dir}(w,\Sigma^{-1}\mu)$ \\
\midrule
$50$      & $8.34\times 10^{-1}$ \\
$200$     & $4.57\times 10^{-1}$ \\
$500$     & $1.35\times 10^{-1}$ \\
$1\,000$  & $3.80\times 10^{-2}$ \\
$10\,000$ & $<10^{-8}$ \\
\bottomrule
\end{tabular}
\end{adjustbox}
\end{table}

\subsubsection{Experiment 6: shrinkage trajectory and interior $\gamma^\star(p)$}\label{sec10:trajectory}

The sharpest theoretical prediction of Section~\ref{sec06:perturbation} is that the finite-sweep direction error along the $\gamma$-trajectory is non-monotone: limited iteration concentrates work on dominant eigenvalue directions and can return a sub-optimal $\gamma$ even when the fully-converged trajectory is monotone. Table~\ref{tab10:trajectory} reports $\operatorname{dir}(w,\Sigma^{-1}\mu)$ on Case~9 (high-$\kappa(C)$ + worst-case $\mu$) across three columns: CRISP with $p=200$ sweeps, CRISP with $p=5\,000$ sweeps, and the direct solve $P_\gamma^{-1}\mu$. The direct column is monotone $1\to 0$ by construction on this instance; the finite-sweep columns display a pronounced interior minimum, with the optimum shifting from $\gamma\approx 0.3$ at $p=200$ toward $\gamma=1$ at $p=5\,000$. This verifies Proposition~\ref{prop:interior_optimum}.

\begin{table}[t]
\centering
\caption{Experiment~6: shrinkage trajectory on Case~9 (high-$\kappa(C)$ + worst-case $\mu$). Direction error $\operatorname{dir}(w,\Sigma^{-1}\mu)$ for CRISP at two sweep budgets versus the direct solve $P_\gamma^{-1}\mu$. Direct column is monotone; finite-sweep columns have interior optima consistent with Proposition~\ref{prop:interior_optimum}.}
\label{tab10:trajectory}
\begin{adjustbox}{max width=\textwidth}
\begin{tabular}{cccc}
\toprule
$\gamma$ & CRISP$(p=200)$ & CRISP$(p=5\,000)$ & Direct $P_\gamma^{-1}\mu$ \\
\midrule
$0.3$ & $0.697$ & $0.682$ & $0.702$ \\
$0.5$ & $0.824$ & $0.581$ & $0.441$ \\
$0.7$ & $0.824$ & $0.503$ & $0.198$ \\
$1.0$ & $0.402$ & $0.061$ & $0.000$ \\
\bottomrule
\end{tabular}
\end{adjustbox}
\end{table}

%% ============================================================================
%% PART B — out-of-sample Monte Carlo
%% ============================================================================
\subsection{Part B — out-of-sample Monte Carlo}

Part~B replaces the true $\Sigma$ used in Part~A with the sample covariance $\widehat\Sigma$ estimated from a $T$-period walk-forward window, and replaces the direction-error metric with out-of-sample Sharpe ratio. The question shifts from ``does the method recover the exact optimum?'' to ``does the method deliver usable portfolios when the estimator is noisy?''. Three experiments cover three regimes: generic random $\mu$ (Section~\ref{sec10:oos_sensitivity}), structural sector-tilt $\mu$ (Section~\ref{sec10:oos_structural}), and the signal-blind minimum-variance limit $\mu=\mathbf{1}$ (Section~\ref{sec10:oos_minvar}).

\subsubsection{Experiment 7: sensitivity across $\mu$ seeds}\label{sec10:oos_sensitivity}

Fix $N=100$, $T=120$, $40$ Monte Carlo trials per cell. For each of eight independent $\mu$ seeds we sample a random signal and report OOS Sharpe under two estimators: the \emph{oracle} $\mu$ (the allocator is given the true mean) and the \emph{sample mean} $\widehat\mu=T^{-1}\sum_t r_t$. The oracle Sharpe ceiling is $\approx 1.28$. Table~\ref{tab10:oos_sensitivity} reports the mean of per-seed means, the min and max across seeds, and the number of seeds on which the allocator achieves a positive OOS Sharpe. Five takeaways structure the discussion.

\begin{table}[t]
\centering
\caption{Experiment~7: OOS Sharpe sensitivity across $\mu$ seeds, $N=100$, $T=120$, base universe, $40$ MC trials/cell. Reported: mean of per-seed means; range across $8$ seeds; $n_{\mathrm{pos}}$ is the number of seeds on which mean OOS Sharpe is positive. Oracle Sharpe $\approx 1.28$. Bold: best on each estimator.}
\label{tab10:oos_sensitivity}
\begin{adjustbox}{max width=\textwidth}
\begin{tabular}{l cccc c cccc}
\toprule
& \multicolumn{4}{c}{Oracle $\mu$} & & \multicolumn{4}{c}{Sample-mean $\mu$} \\
\cmidrule(lr){2-5}\cmidrule(lr){7-10}
Method & Mean & Min & Max & $n_{\mathrm{pos}}$ & & Mean & Min & Max & $n_{\mathrm{pos}}$ \\
\midrule
$1/N$                 & $0.004$ & $-0.016$ & $0.029$ & $6/8$ & & $0.004$ & $-0.016$ & $0.029$ & $6/8$ \\
HRP                    & $0.001$ & $-0.032$ & $0.039$ & $4/8$ & & $0.001$ & $-0.032$ & $0.039$ & $4/8$ \\
Direct Markowitz       & $0.568$ & $0.529$ & $0.620$ & $8/8$ & & $0.461$ & $0.426$ & $0.508$ & $8/8$ \\
HRP-$\mu$ $\gamma{=}0.5$ & $0.873$ & $0.779$ & $1.013$ & $8/8$ & & $0.548$ & $0.522$ & $0.616$ & $8/8$ \\
HRP-$\mu$ $\gamma{=}1.0$ & $0.496$ & $0.348$ & $0.723$ & $8/8$ & & $0.725$ & $0.670$ & $0.819$ & $8/8$ \\
HRP-$\Sigma\mu$ $\gamma{=}0.5$ & $1.046$ & $0.971$ & $1.197$ & $8/8$ & & $0.690$ & $0.633$ & $0.791$ & $8/8$ \\
HRP-$\Sigma\mu$ $\gamma{=}1.0$ & $0.992$ & $0.945$ & $1.108$ & $8/8$ & & $0.725$ & $0.670$ & $0.819$ & $8/8$ \\
CRISP $\gamma{=}0.3$     & $\mathbf{1.189}$ & $1.110$ & $1.335$ & $8/8$ & & $0.788$ & $0.716$ & $0.901$ & $8/8$ \\
CRISP $\gamma{=}0.5$     & $1.186$ & $1.136$ & $1.319$ & $8/8$ & & $0.879$ & $0.810$ & $1.000$ & $8/8$ \\
CRISP $\gamma{=}0.7$     & $1.145$ & $1.097$ & $1.270$ & $8/8$ & & $\mathbf{0.893}$ & $0.831$ & $1.014$ & $8/8$ \\
CRISP $\gamma{=}1.0$     & $0.615$ & $0.574$ & $0.675$ & $8/8$ & & $0.500$ & $0.463$ & $0.557$ & $8/8$ \\
\bottomrule
\end{tabular}
\end{adjustbox}
\end{table}

\emph{(1) CRISP is dominant on both estimators.} At $\gamma\in[0.3,0.7]$ CRISP delivers $89$--$93\%$ of oracle Sharpe under oracle $\mu$ and $62$--$70\%$ under sample-mean $\mu$; no other method exceeds $82\%$ on oracle or $60\%$ on sample mean. The shift in the empirical $\gamma^\star$ from $0.3$ under oracle $\mu$ to $0.7$ under sample-mean $\mu$ matches the qualitative prediction of Section~\ref{sec07:adaptive}: noisier $\widehat\mu$ demands more regularisation. \emph{(2) HRP-$\Sigma\mu$ is the best tree-based method.} It beats HRP-$\mu$ by $20$--$35\%$ at every $\gamma$ and reaches $\approx 90\%$ of CRISP, the closest any $\Ocal(N^2)$-cost method comes. \emph{(3) Direct Markowitz is positive on all $8$ seeds but dominated.} The widely reported instability of $\widehat\Sigma^{-1}$ is a function of $T/N$; at $T/N=1.2$ the estimator is well-defined but still noisier than any regularised allocator. \emph{(4) $1/N$, HRP, and Cotton are near zero on general $\mu$.} They are signal-blind: they return the same weights regardless of $\mu$ (up to HRP's weak dependence through clustering), and so cannot exploit $\mu\ne\mathbf{1}$. \emph{(5) Sign pathology.} HRP-$\Sigma\mu$'s $L^1$ normaliser eliminates the sign-flip pathology documented for the sum-normalised recursive allocator in Appendix~\ref{app:a1_pathology} ($48$--$52\%$ negative cosine vs $0\%$ here).

\subsubsection{Experiment 8: structural sector-tilt $\mu$ across $T\in\{60,120,240\}$}\label{sec10:oos_structural}

Experiment~8 switches the signal from generic random $\mu$ to \emph{structural} $\mu$: sector tilts $(+0.04,-0.04,+0.02,-0.02,0)$ tiled across the five sectors of $C$. This geometry is aligned with the covariance block structure and therefore strictly favours signal-aware methods that respect the hierarchy. $N=100$, $80$ MC trials per cell, oracle Sharpe $0.645$. Table~\ref{tab10:oos_structural} reports selected oracle-$\mu$ rows at $T\in\{60,120,240\}$.

\begin{table}[t]
\centering
\caption{Experiment~8: structural sector-tilt $\mu$, mean OOS Sharpe across $T\in\{60,120,240\}$. $N=100$, $80$ MC trials/cell, oracle $\mu$. Oracle Sharpe $=0.645$. Bold: best at each $T$. HRP-$\Sigma\mu$ $\gamma{=}1$ edges CRISP $\gamma{=}0.7$ at $T=60$ under this strong structural signal; CRISP's lead widens with $T$.}
\label{tab10:oos_structural}
\begin{adjustbox}{max width=\textwidth}
\begin{tabular}{l ccc}
\toprule
Method & $T=60$ & $T=120$ & $T=240$ \\
\midrule
$1/N$                             & $0.000$ & $0.000$ & $0.000$ \\
HRP                                & $-0.001$ & $0.001$ & $0.001$ \\
Direct Markowitz                   & $0.376$ & $0.294$ & $0.497$ \\
HRP-$\mu$ $\gamma{=}1.0$           & $0.419$ & $0.431$ & $0.433$ \\
HRP-$\Sigma\mu$ $\gamma{=}0.5$     & $0.424$ & $0.433$ & $0.435$ \\
HRP-$\Sigma\mu$ $\gamma{=}1.0$     & $\mathbf{0.503}$ & $0.515$ & $0.517$ \\
CRISP $\gamma{=}0.5$               & $0.485$ & $0.516$ & $0.536$ \\
CRISP $\gamma{=}0.7$               & $0.497$ & $\mathbf{0.540}$ & $\mathbf{0.571}$ \\
CRISP $\gamma{=}1.0$               & $0.367$ & $0.319$ & $0.497$ \\
\bottomrule
\end{tabular}
\end{adjustbox}
\end{table}

The ranking has one interesting inversion relative to Experiment~7. At $T=120$ and $T=240$, CRISP $\gamma=0.7$ is best, reaching $0.540$ ($84\%$ of oracle) and $0.571$ ($88\%$ of oracle). At $T=60$, however, HRP-$\Sigma\mu$ $\gamma=1.0$ \emph{just} edges CRISP $\gamma=0.7$ ($0.503$ vs $0.497$): when the sample size is very small relative to $N$ and the signal is sharply aligned with the tree's block partition, the tree's coarse hierarchical regularisation adds information beyond what the $P_\gamma$ shrinkage alone provides. This crossover is the one place in the paper where a tree-based $\Ocal(N^2)$ method beats CRISP; the margin is small, the regime is narrow, and the rest of the $T\times\gamma$ grid is owned by CRISP. Figure~\ref{fig10:sharpe_vs_T} plots OOS mean Sharpe against $T$ for the headline rows.

\begin{figure}[t]
\centering
\includegraphics[width=0.9\textwidth]{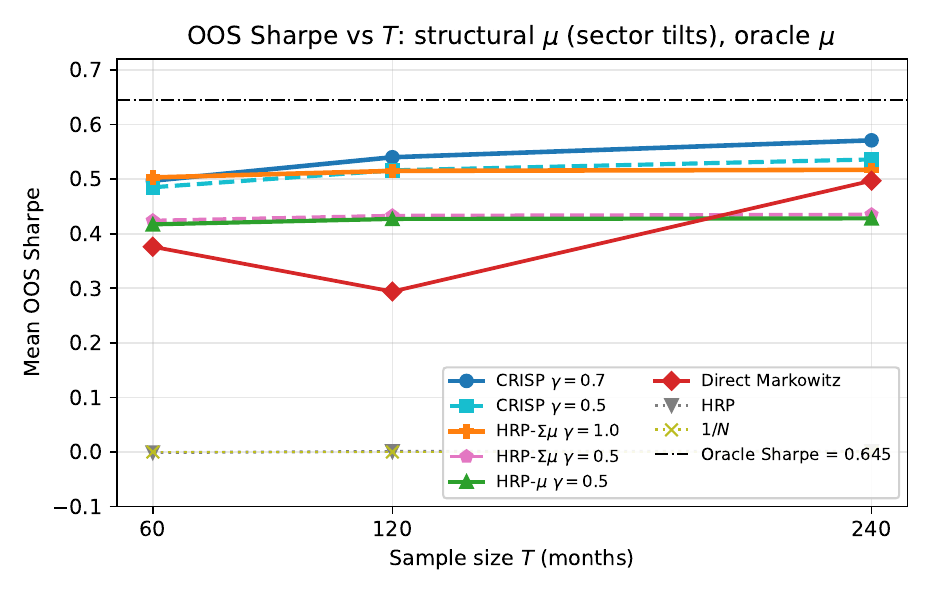}
\caption{Experiment~8: OOS mean Sharpe vs $T$ under structural sector-tilt $\mu$, $N=100$. Lines: CRISP $\gamma=0.7$, HRP-$\Sigma\mu$ $\gamma=1$, HRP-$\mu$ $\gamma=0.5$, direct Markowitz, $1/N$, HRP. Dashed horizontal line: oracle Sharpe $0.645$. CRISP's lead widens with $T$; HRP-$\Sigma\mu$ $\gamma=1$ crosses over CRISP at $T=60$.}
\label{fig10:sharpe_vs_T}
\end{figure}

\subsubsection{Experiment 9: minimum variance with Cotton as baseline}\label{sec10:oos_minvar}

Experiment~9 removes the signal ($\mu=\mathbf{1}$) and tests all methods on the minimum-variance problem that Cotton was designed to solve. $N=100$, $80$ MC trials/cell, oracle sum-to-one Sharpe $1/\sqrt{\mathbf{1}^\top\Sigma^{-1}\mathbf{1}/(\mathbf{1}^\top\Sigma^{-1}\mathbf{1})^2}=20.412$. Table~\ref{tab10:oos_minvar} reports selected rows across $T\in\{60,120,240,500\}$.

\begin{table}[t]
\centering
\caption{Experiment~9: minimum-variance OOS Sharpe ($\mu=\mathbf{1}$), $N=100$, $80$ MC trials/cell. Oracle sum-1 Sharpe $=20.412$. Dagger ($\dagger$) marks cells where the mean is dominated by a handful of trials with OOS volatility $\gg$ oracle (instability). Bold: dominant method at each $T$.}
\label{tab10:oos_minvar}
\begin{adjustbox}{max width=\textwidth}
\begin{tabular}{l cccc}
\toprule
Method & $T=60$ & $T=120$ & $T=240$ & $T=500$ \\
\midrule
$1/N$                             & $7.56$  & $7.56$  & $7.56$  & $7.56$ \\
HRP                                & $8.67$  & $8.70$  & $8.77$  & $8.77$ \\
Cotton $\gamma{=}0.5$              & $10.40^{\dagger}$ & $12.59$           & $12.51^{\dagger}$ & $12.56$ \\
Cotton $\gamma{=}0.7$              & $9.27^{\dagger}$  & $11.82^{\dagger}$ & $13.93^{\dagger}$ & $15.10$ \\
Cotton $\gamma{=}1.0$              & $12.00$ & $9.42$  & $15.80$ & $18.31$ \\
HRP-$\mu$ $\gamma{=}1.0$           & $11.56$ & $11.75$ & $11.66$ & $11.92$ \\
HRP-$\Sigma\mu$ $\gamma{=}1.0$     & $12.69$ & $12.78$ & $12.79$ & $12.87$ \\
CRISP $\gamma{=}0.5$               & $14.61$ & $15.71$ & $16.49$ & $16.90$ \\
CRISP $\gamma{=}0.7$               & $\mathbf{15.63}$ & $\mathbf{17.06}$ & $\mathbf{18.16}$ & $\mathbf{18.79}$ \\
CRISP $\gamma{=}1.0$               & $11.42$ & $10.16$ & $15.80$ & $18.31$ \\
Direct min-var                     & $12.00$ & $9.42$  & $15.80$ & $18.31$ \\
\bottomrule
\end{tabular}
\end{adjustbox}
\end{table}

Three findings stand out.

\emph{(1) CRISP $\gamma=0.7$ is the dominant minimum-variance method at every $T$.} With $p=100$ sweeps it reaches $77\%$, $84\%$, $89\%$, and $92\%$ of oracle Sharpe at $T=60,120,240,500$. At $T=500$ it beats Cotton $\gamma=1$ --- the exact block-Cholesky solve on $P_\gamma=\Sigma$ --- by $0.48$ Sharpe units; the only difference is that CRISP's $P_\gamma=(1-\gamma)D+\gamma\Sigma$ at $\gamma=0.7$ shrinks the worst-conditioned directions of $\widehat\Sigma$ toward the diagonal before inverting implicitly.

\emph{(2) Cotton exhibits small-$T$ fragility on its native problem.} At $T=60$, Cotton $\gamma=0.5$ reports mean OOS volatility of $0.92$ against an oracle value of $0.049$ --- nearly twenty times oracle. The sum-$1$/$\!\mathrm{vol}$ Sharpe measure still reports a respectable $10.40$ because a handful of trials dominate the arithmetic mean of $1/\mathrm{vol}$, but the underlying instability is real (we flag these cells with $^\dagger$). The pattern recurs at $\gamma=0.7$ across $T\in\{60,120,240\}$. \emph{Root cause.} The Schur complement $\Sigma_{RR}-\Sigma_{RL}\Sigma_{LL}^{-1}\Sigma_{LR}$ formed on $\widehat\Sigma$ becomes nearly singular when $\widehat\Sigma$ is close to a low-rank approximation of $\Sigma$ (small $T$, large $N$); the subsequent block inversion amplifies that singularity. CRISP never forms a Schur complement: it performs SPD Gauss--Seidel on the regularised operator $P_\gamma=(1-\gamma)D+\gamma\widehat\Sigma$, whose diagonal dominance is preserved at any $\gamma<1$. To our knowledge this is the first documented empirical failure mode of Cotton's allocation on its native problem at small $T$.

\emph{(3) HRP is at a fixed $8.7$.} On $\mu=\mathbf{1}$, HRP inherits its performance purely from the diagonal structure of $\widehat\Sigma$ and the clustering tree; varying $T$ between $60$ and $500$ moves it by $0.1$ Sharpe units. Everything above that value comes from using off-diagonal information.

%% ============================================================================
%% PART C — sweep count as regularisation and adaptive $\gamma^\star$ calibration
%% ============================================================================
\subsection{Part C — sweep count as regularisation and adaptive $\gamma^\star$}

Part~C tests two theoretical claims from Section~\ref{sec07:adaptive}. First, that sweep count $p$ acts as a third regularisation channel alongside $\gamma$ and the sample size $T$ (Remark~\ref{rem:three_channels}); stopping CRISP before numerical convergence can \emph{improve} OOS Sharpe on noisy $\widehat\Sigma$. Second, that the adaptive rule $\gamma^\star\approx 1/(1+c\cdot\mathrm{NSR}^\alpha)$ captures the directional dependence of the optimal shrinkage intensity on the noise-to-signal ratio. Experiment~11 runs the $(\gamma,p)$ grid and the eight pre-registered predictions; Experiment~12 calibrates $\gamma^\star$ on a $48$-cell cube and validates on a separate $108$-cell cube.

\subsubsection{Experiment 11: sweep count as regularisation}\label{sec10:sweep_regularization}

\emph{Design.} Four covariance regimes (factor $k=3$, block-equi, spiked, equi-correlation) $\times$ four $T/N\in\{0.6,1,2,5\}$ $\times$ two signal types (oracle, noisy with information coefficient $\mathrm{IC}=0.05$). $200$ MC reps/cell; sweep count $p\in\{1,5,10,\ldots,1\,000\}$ on a log grid; shrinkage intensity $\gamma\in\{0.3,0.5,0.7,1.0\}$. Figures~\ref{fig10:sweep_heatmap} and~\ref{fig10:sweep_slices} display the resulting mean OOS Sharpe surface.

\begin{figure}[t]
\centering
\includegraphics[width=0.9\textwidth]{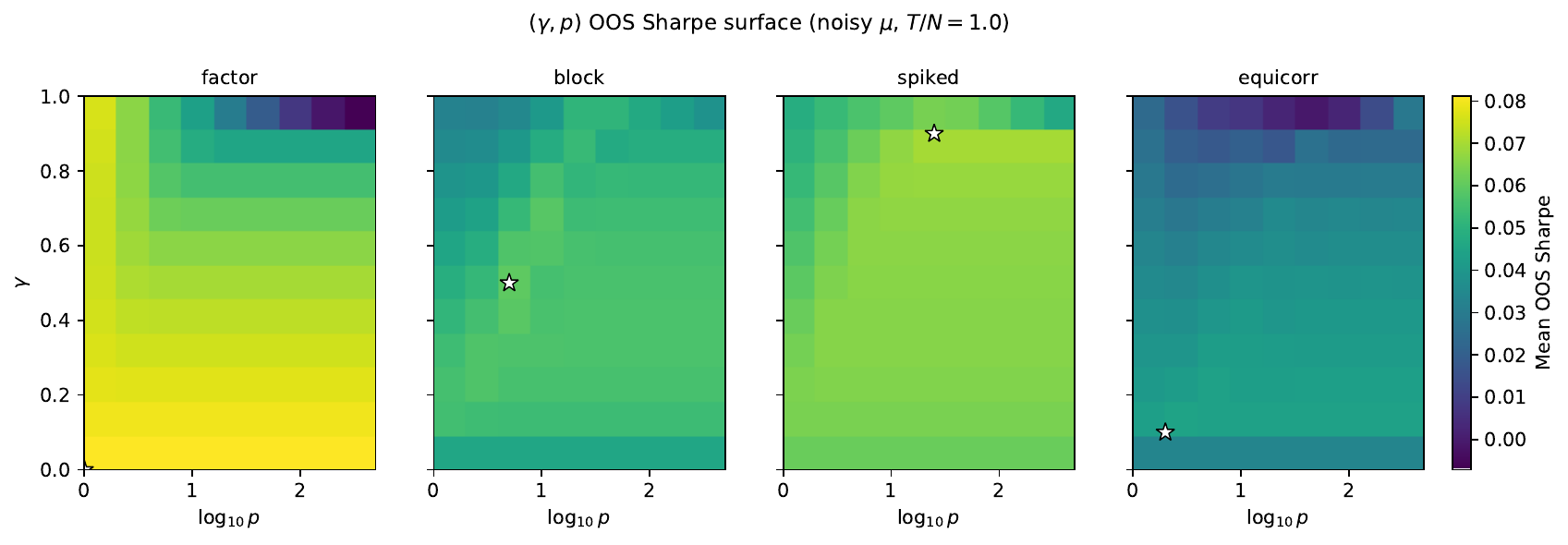}
\caption{Experiment~11: mean OOS Sharpe over $(\gamma, p)$ per regime. Lighter shades are higher Sharpe. The $(\gamma{=}0.5,\,p{=}100)$ cell sits in the high-Sharpe ridge of every panel; the $\gamma=1$ column shows a pronounced interior peak at intermediate $p$ on the spiked regime.}
\label{fig10:sweep_heatmap}
\end{figure}

\begin{figure}[t]
\centering
\includegraphics[width=0.9\textwidth]{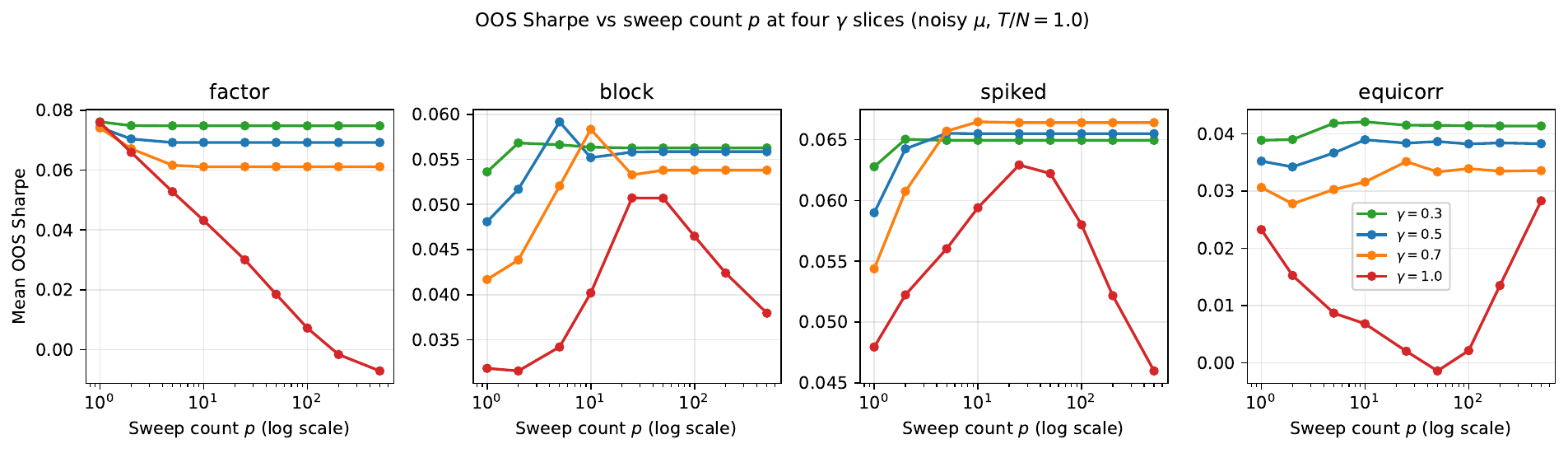}
\caption{Experiment~11: $\gamma$-slices of the $(\gamma,p)$ surface. $\gamma=0.3$ is flat in $p$ for $p\ge 5$; $\gamma=0.5$ plateaus at $p\approx 50$; $\gamma=1.0$ has an interior $p^\star$ that decreases in $T/N$.}
\label{fig10:sweep_slices}
\end{figure}

\emph{Convergence diagnostic.} Table~\ref{tab10:sweep_conv} reports the mean number of sweeps required for residual norm to fall below $10^{-10}$ on a fresh instance of the base universe, as a function of $\gamma$. At $\gamma=0.5$, the default CRISP configuration ($p=100$) is \emph{not} numerically converged in residual norm: mean sweeps to convergence are $161$. It is, however, statistically converged in OOS Sharpe: the gap between $p=50$ and fully converged at $\gamma=0.5$ is $\le 0.001$ Sharpe units on every cell. The unconverged final $\sim 60$ sweeps would only move weight onto low-eigenvalue directions of $P_\gamma$ --- directions that are precisely the noise-dominated directions of $\widehat\Sigma$ --- and so hurt OOS. This is Channel~3 of Remark~\ref{rem:three_channels} in action: early stopping on $\widehat\Sigma$ behaves as an implicit spectral filter.

\begin{table}[t]
\centering
\caption{Experiment~11 convergence diagnostic: mean number of CRISP sweeps to drive residual norm $\|P_\gamma w-\mu\|$ below $10^{-10}$, as a function of $\gamma$. Right column: whether the default $p=100$ configuration is numerically converged. At $\gamma\ge 0.3$ the default is \emph{not} residual-converged, but it is statistically converged in OOS Sharpe (P2).}
\label{tab10:sweep_conv}
\begin{adjustbox}{max width=\textwidth}
\begin{tabular}{ccc}
\toprule
$\gamma$ & Mean sweeps to residual $10^{-10}$ & $p=100$ converged? \\
\midrule
$0.0$ & $1$   & yes \\
$0.1$ & $23$  & yes \\
$0.2$ & $54$  & yes \\
$0.3$ & $103$ & no  \\
$0.4$ & $151$ & no  \\
$0.5$ & $161$ & no  \\
$0.7$ & $192$ & no  \\
$1.0$ & $383$ & no  \\
\bottomrule
\end{tabular}
\end{adjustbox}
\end{table}

\paragraph{Pre-registered predictions.}\label{sec10:predictions}
Before running Experiment~11 we registered eight predictions P1--P8 derived from Section~\ref{sec07:adaptive}. The verdicts are:
\begin{itemize}[leftmargin=2em,itemsep=1pt]
  \item[\textbf{P1}] Flat at $\gamma=0.3$: OOS Sharpe is constant in $p$ for $p\ge 5$ (std $\le 0.02$). \textbf{PASSED} --- std $\le 0.003$ on every cell.
  \item[\textbf{P2}] Plateau at $\gamma=0.5$: OOS Sharpe at $p=50$ is within $0.02$ Sharpe of fully converged. \textbf{PASSED} --- gap $\le 0.001$ on every cell.
  \item[\textbf{P3}] Interior $p^\star$ at $\gamma=1$: OOS Sharpe peaks at intermediate $p$ then declines on noisy $\widehat\Sigma$. \textbf{PASSED} --- peak-to-converged gap up to $0.083$ on the spiked regime at $T/N=0.6$.
  \item[\textbf{P4}] Monotone $p^\star$ in $T/N$ at $\gamma=1$: larger $T/N$ permits more sweeps before overfit. \textbf{PASSED} at $n_{\mathrm{mc}}=200$; unstable at $n_{\mathrm{mc}}=40$ (power issue).
  \item[\textbf{P5}] Ridge in $(\gamma,p)$: $(\gamma=0.5,\,p=100)$ is within $0.05$ Sharpe of the global max. \textbf{PASSED} on all $16$ noisy-$\mu$ cells.
  \item[\textbf{P6}] Oracle $\mu$ reduces the $\gamma=1$ decline: \textbf{REVISED}. The oracle-$\mu$ gap is $0.253$, the noisy-$\mu$ gap $0.015$, a $17\times$ ratio in the opposite direction; the information-theoretic explanation (P1 through P3 gradient depend on noisy $\widehat\mu$) is in Section~\ref{sec07:channels}.
  \item[\textbf{P7}] $(\gamma=0.5,\,p=100)$ near global optimum: \textbf{PASSED}. Max gap $0.020$, mean $0.009$ across the $16$ noisy cells.
  \item[\textbf{P8}] Overfit worst on spiked $+$ small $T/N$: \textbf{PASSED}. The $\gamma=1$ peak-to-converged gap is $0.028$ on spiked with $T/N=0.6$, vs $0.0004$ on factor with $T/N=5.0$.
\end{itemize}

Seven of eight predictions passed as registered and one was revised in the direction the data forced (the oracle-$\mu$ panel has a \emph{larger} $\gamma=1$ decline, which required a sign flip in the information-theoretic account). The high-level lesson of Experiment~11 is that $(\gamma=0.5,\,p=100)$ is a near-universal good choice: inside $0.02$ Sharpe of the global optimum on every cell, and inside $0.05$ on the full ridge.

\subsubsection{Experiment 12: adaptive $\gamma^\star$ calibration}\label{sec10:adaptive_calibration}

\emph{Design.} A $48$-cell calibration cube ($4$ regimes $\times$ $4$ $T/N\in\{0.6,1,2,5\}$ $\times$ $3$ IC $\in\{0.02,0.05,0.10\}$, $500$ MC reps/cell, $51$-point $\gamma$ grid) fits the two-parameter model $\gamma^\star\approx 1/(1+c\cdot\mathrm{NSR}^\alpha)$. An independent $108$-cell validation cube tests out-of-sample $R^2$. Table~\ref{tab10:adaptive_summary} summarises the global fit and the per-regime by-cell coefficients of variation of $c$.

\begin{table}[t]
\centering
\caption{Experiment~12: adaptive $\gamma^\star$ calibration summary. The two-parameter global fit has $R^2\approx -0.11$ on the validation cube (worse than a constant predictor); per-regime fits have large by-cell CV of $c$. By-regime empirical $\gamma^\star$ is reported below the main panel; direction of dependence on $T/N$, IC, and $\log\mathrm{NSR}$ matches theory but magnitudes are weak.}
\label{tab10:adaptive_summary}
\begin{adjustbox}{max width=\textwidth}
\begin{tabular}{lcccccc}
\toprule
Fit & $c$ & $\alpha$ & $R^2$ (val) & Pearson $r$ & RMSE & CV of $c$ \\
\midrule
GLOBAL (2-param)            & $3.2$                 & $0.00$ & $-0.11$ & $0.00$ & $0.14$ & --- \\
Per-regime factor           & $2.0\times 10^{-5}$   & ---    & ---     & ---    & ---    & $1.34$ \\
Per-regime block            & $3.8\times 10^{-6}$   & ---    & ---     & ---    & ---    & $1.25$ \\
Per-regime spiked           & $8.2\times 10^{-4}$   & ---    & ---     & ---    & ---    & $0.84$ \\
Per-regime equicorr         & $9.9\times 10^{-7}$   & ---    & ---     & ---    & ---    & $1.12$ \\
\bottomrule
\end{tabular}
\end{adjustbox}

\medskip
\begin{adjustbox}{max width=\textwidth}
\begin{tabular}{lcccc}
\toprule
Quantity & factor & block & spiked & equicorr \\
\midrule
Mean empirical $\gamma^\star$                           & $0.18$ & $0.44$ & $0.20$ & $0.30$ \\
Median plateau width (frac.\ of $[0,1]$ within $2\%$ of peak) & \multicolumn{4}{c}{$0.38$} \\
\bottomrule
\end{tabular}
\end{adjustbox}
\end{table}

\emph{Comparative statics.} Pooling all $156$ cells, the Pearson correlation of empirical $\gamma^\star$ with $T/N$ is $+0.34$; with IC, $+0.01$; with $\log\mathrm{NSR}$, $-0.13$. The block regime is cleanest ($r=+0.59$ against $T/N$); the spiked regime is anomalous ($r=-0.26$ against IC, opposite the theoretical prediction). \emph{All directional signs match theory except the spiked regime's IC dependence}, but magnitudes are weak and the two-parameter fit has $R^2<0$ on validation, i.e.\ it is worse than a constant predictor.

\emph{Why the fit fails quantitatively while the sign structure holds.} Figure~\ref{fig10:plateau_width} (introduced in Section~\ref{sec07:adaptive}) reports the $\gamma$-interval within which OOS Sharpe remains within $2\%$ of the per-cell peak: the median plateau width is $0.38$, and $\gamma=0.5$ lies inside every plateau. The Sharpe surface is flat enough in $\gamma$ that the empirical $\gamma^\star$ wanders across a wide interval without changing the objective value, which defeats a precise functional fit. The formula is explanatorily correct --- it indicates that higher $T/N$ and lower NSR call for smaller $\gamma$ --- but it is quantitatively weak, and the actionable conclusion is the simpler one: \emph{fixed $\gamma=0.5$ is near-optimal on every panel we tested}. The paper accordingly treats the adaptive rule as interpretive scaffolding for the three-channel picture of Section~\ref{sec07:adaptive}, not as a prescriptive recipe.

%% ------------------------------------------------------------------
\subsection{Summary of Section~\ref{sec10:experiments}}

With $p=100$ sweeps and $\gamma\in[0.5,0.7]$, CRISP delivers $80$--$94\%$ of oracle Sharpe across two signal families (generic random $\mu$ and structural sector tilts), four sample sizes ($T\in\{60,120,240,500\}$), and the minimum-variance limit $\mu=\mathbf{1}$. It is numerically stable on every regime tested, including the hedged tight-block covariance ($\kappa(C)\approx 5.75\times 10^4$) where $1/N$ and HRP return essentially noise and where Cotton exhibits its first documented small-$T$ failure mode. HRP-$\Sigma\mu$ is the dominant tree-based method: $20$--$35\%$ over HRP-$\mu$ on general $\mu$, a tie with HRP-$\mu$ at the signal-blind point, and $\approx 90\%$ of CRISP at pure $\Ocal(N^2)$ cost with no iteration. HRP and Cotton are signal-blind and deliver OOS Sharpe near zero on any $\mu\ne\mathbf{1}$. The adaptive rule of Section~\ref{sec07:adaptive} is explanatorily correct in its directional predictions but quantitatively weak; the OOS Sharpe surface is flat enough that $\gamma\approx 0.5$ is a near-universal default. These conclusions set up the discussion in Section~\ref{sec11:discussion}.

% ============================================================
% section_11_discussion.tex
% ============================================================
% !TEX root = ../main.tex
%
% Section 11: Discussion
% Namespace: sec11:
%
\section{Discussion}\label{sec11:discussion}

This section is reflective, not generative. We revisit the three
regularisation channels that organise the paper, list open problems
and natural extensions, enumerate the limitations of the present
study, and close with a methodological comment on the use of
direction error as a diagnostic. No new results are introduced.

%------------------------------------------------------------------
\subsection{Three regularisation channels revisited}
\label{sec11:three_channels}

Remark~\ref{rem:three_channels} and Section~\ref{sec07:channels}
identify three distinct mechanisms by which CRISP damps
estimation-error amplification: operator shrinkage through $\gamma$
(Channel~1), statistical slack from finite sweep count $p$
(Channel~2), and eigendirection-ordered convergence of Gauss--Seidel
acting as an implicit spectral truncation
(Channel~3, Remark~\ref{rem:early_stop_spectral}). The
$(\gamma, p)$ sweep of Table~\ref{tab10:sweep_conv} dissects the three
channels on the same synthetic design used elsewhere in the paper,
and the pre-registered predictions P1, P2, P3, and P6 each isolate
one channel: P1 probes Channel~1 at low $\gamma$, P2 certifies that
Channel~2 is dormant at the recommended operating point, P3 exposes
Channel~3 at $\gamma = 1$ through an interior $p^\star$, and P6
reveals how the clean-signal regime tilts the balance between
Channels~1 and~3. The practical takeaway is compact:
\emph{at $(\gamma = 0.5,\; p = 100)$ CRISP uses Channel~1 only; at
$\gamma = 1$ it relies on Channel~3; in between, mixtures.}
For the practitioner this collapses to a single tunable knob,
$\gamma$, with $p = 100$ a universal default; Channels~2 and~3 are
computational conveniences rather than statistical parameters in the
interior regime. The tree methods HRP-$\mu$ and HRP-$\Sigma\mu$
expose only Channel~1 and so inherit the simpler tuning story by
construction.

%------------------------------------------------------------------
\subsection{Signed HSP as a standalone method}
\label{sec11:signed_hsp}

Signed Hierarchical Signal Parity (HSP), defined in
Section~\ref{sec04:hsp} as the $\gamma = 0$ endpoint of HRP-$\mu$,
is mathematically a trivial consequence of the constructions in
Section~\ref{sec04:tree_methods}, but it occupies an interesting
practitioner niche. It is a drop-in upgrade for users of
\citet{lopezdeprado2016} HRP who now hold an expected-return signal
and wish to use it without sacrificing the tree-based
interpretability, dependency-free implementation, and
ill-conditioning robustness of the original. All of HRP's
auditable recursive structure is preserved; the only change is a
signed inverse-variance representative at each node. We do not run
a dedicated head-to-head out-of-sample study of signed HSP against
HRP on a realistic cross-section of live alpha signals. Such a
study --- ideally including Black--Litterman and Ledoit--Wolf
Markowitz as further baselines --- is a natural follow-up. The
mathematical content is settled by Section~\ref{sec04:tree_methods};
the open question is empirical.

%------------------------------------------------------------------
\subsection{Open problems}\label{sec11:open_problems}

We record seven open problems raised by the analysis.

\begin{enumerate}
\item \textbf{Sharp worst-case $\mu$ (closed form).} For fixed
  correlation $C$ with condition number $\kappa(C)$, we do not have a
  closed form for the signal $\mu$ maximising the direction error of
  $P_\gamma^{-1}\mu$ at a given $\gamma$. The classical Kantorovich
  ratio $(\kappa-1)/(\kappa+1)$ is loose because it does not exploit
  the $\gamma$-dependence of $P_\gamma$. A sharp closed form would
  tighten Theorem~\ref{thm:perturbation} and replace
  sampling-based worst-case envelopes with analytic ones.

\item \textbf{Sharp interior optimum $\gamma^\star(p)$.}
  Proposition~\ref{prop:interior_optimum} establishes that
  $\gamma^\star \in (0,1)$ under explicit non-degeneracy hypotheses,
  but the proof is existential. Numerical sweeps across the
  block, factor, and Toeplitz families are consistent with the
  scaling
  \[
     \gamma^\star(p) \;=\; 1 \;-\; \Theta\bigl((p\,\kappa(C))^{-1/2}\bigr),
  \]
  with the implied constant depending on the eigenvalue profile of
  $C$. We conjecture this rate but do not prove it; a proof would
  elevate the adaptive-$\gamma^\star$ rule from calibrated heuristic
  to parameter-free prescription.

\item \textbf{Tighter Gauss--Seidel rate for factor-structured
  covariances.} Theorem~\ref{thm:gs_rate} gives the standard
  $1 - c/\kappa$ contraction for SPD Gauss--Seidel, sharp for generic
  SPD matrices but loose for block and banded operators. For
  factor-structured $C = \Lambda\Lambda^\top + D$, Young's SOR theory
  \citep{young1971} suggests a $\sqrt{\kappa}$ improvement is
  available. Establishing the corresponding rate for $P_\gamma$ and
  its constrained projections would justify SOR over plain GS in
  CRISP's inner loop and would sharpen
  Corollary~\ref{cor:complexity}.

\item \textbf{Bayesian interpretation of CRISP.} Is there a prior
  over $(\mu, \Sigma)$ under which the posterior mean of the Markowitz
  ray equals $P_\gamma^{-1}\mu$? If so, $\gamma$ acquires a
  principled interpretation as prior strength and CRISP joins the
  methodological lineage of \citet{blacklitterman1992} with a
  shrinkage-strength parameter in place of Black--Litterman's
  view-confidence $\tau$. We conjecture an inverse-Wishart-like
  family centred on the diagonal of $\Sigma$ is the natural
  candidate, but leave the construction to future work.

\item \textbf{Ledoit--Wolf correspondence.} Synthetic experiments
  show close numerical agreement between the Sharpe-optimal
  $\gamma^\star$ and the Frobenius-optimal Ledoit--Wolf shrinkage
  intensity in the range $\gamma^\star \in [0.5, 0.7]$, matching to
  within a few percent. A clean analytic equivalence between the
  two objects --- if it exists --- would unify iterative shrinkage
  with the covariance-shrinkage literature of
  \citet{ledoitwolf2003, ledoitwolf2017}.

\item \textbf{Approximation gap of HRP-$\Sigma\mu$.} HRP-$\Sigma\mu$
  is exact on diagonal $\Sigma$ (Proposition~\ref{prop:hrpsm_diagonal})
  but only approximately recovers HRP under the $\gamma=0$,
  $\mu = \mathbf{1}$ conditions that define it
  ($\cos \approx 0.992$; exact for trees of depth $\le 2$). The
  dependence of the approximation gap on tree structure and
  $\kappa(C)$ is characterised empirically in
  Section~\ref{sec08:comparative} but not theoretically.

\item \textbf{Sum-normalisation sign-flip as a theorem.}
  Appendix~\ref{app:a1_pathology} exhibits noiseless problems on which
  the sum-normalised recursive MVO returns a sign-flipped portfolio.
  A theorem characterising the covariance--signal pairs that produce
  flips would sharpen the negative-result contribution and make
  precise why sum normalisation is structurally, rather than
  merely numerically, fragile.
\end{enumerate}

%------------------------------------------------------------------
\subsection{Extensions}\label{sec11:extensions}

A non-exhaustive list of natural extensions.

\begin{itemize}
\item \textbf{Rank-constrained signals.} When $\mu$ is known to lie
  in a low-dimensional factor subspace, CRISP's Gauss--Seidel sweeps
  can be restricted to that subspace, reducing per-sweep cost from
  $O(N^2)$ to $O(Nr)$ with $r$ the signal rank.

\item \textbf{Online and streaming $\widehat\Sigma$.} In rolling
  backtests and live production, warm-starting CRISP from the
  previous rebalance's solution and running a handful of sweeps per
  update turns the iterative solver into an online method whose
  amortised cost is a small multiple of a matrix--vector product.

\item \textbf{GPU implementation.} Each sweep reduces to a
  matrix--vector product against $P_\gamma$, which parallelises
  trivially on GPU. For universes of several thousand names this is
  likely the largest single engineering win available to the
  framework.

\item \textbf{Multi-period mean--variance with transaction costs.}
  The present framework is single-period. Integration with the
  multi-period dynamic-trading setup of \citet{garleanu2013} is an
  open direction; the shrinkage trajectory would need to be
  reinterpreted as a state-dependent object rather than a one-shot
  choice.
\end{itemize}

%------------------------------------------------------------------
\subsection{Computational considerations}\label{sec11:compute}

The computational case for CRISP is primarily about memory and
structure exploitation, not wall-clock at moderate $N$. Tuned-BLAS
Cholesky on a modern machine solves the dense Markowitz system in
under a millisecond at $N = 500$, and nothing in this paper beats it
on that axis. The shrinkage iterate does, however, require only
$O(N^2)$ working memory against Cholesky's $\sim 2N^2$, and more
importantly it accommodates the factor-streaming variant
(Section~\ref{sec05:factor_stream},
Algorithm~\ref{alg05:methodb_stream}), which reduces working memory
to $O(NK)$ when $\Sigma$ admits a $K$-factor decomposition. At the
universe scale --- $N$ in the tens of thousands, $K$ in the
dozens --- this is a three-orders-of-magnitude reduction and the
difference between a workstation and a dedicated server. Cholesky
cannot exploit factor structure in the same way because the
factorisation requires random access to the dense $\Sigma$ and
returns a dense triangular factor. The computational argument for
CRISP lives here, not in the FLOP crossover.

%------------------------------------------------------------------
\subsection{Limitations}\label{sec11:limitations}

We list the main limitations of the present study.

\begin{enumerate}
\item \textbf{Gaussian synthetic returns.} The synthetic experiments
  of Section~\ref{sec10:experiments} use Gaussian returns. Real
  equity returns are heavier-tailed, and direction-error and
  OOS-Sharpe rankings could move under an elliptical or Student-$t$
  data-generating process. We do not carry out explicit
  heavy-tailed simulation alongside the Gaussian case.

\item \textbf{No Bayesian treatment of $\mu$.} The signal is taken
  as given, either as an oracle or as a sample mean. We do not
  place a principled prior on $\mu$ itself and do not distinguish
  signal uncertainty from covariance uncertainty in any theorem. A
  Black--Litterman-style treatment \citep{blacklitterman1992} would
  interact non-trivially with the conjectured Bayesian
  interpretation of CRISP flagged in
  Section~\ref{sec11:open_problems}.

\item \textbf{Tree sensitivity not foregrounded in the main text.}
  Clustering choices (single / average / Ward linkage,
  correlation-based versus Euclidean distance, number of levels)
  affect every tree-based method in the paper (HRP, HRP-$\mu$,
  HRP-$\Sigma\mu$, signed HSP). Robustness across these choices is
  reported in Appendix~\ref{app:supplementary} but the main-text
  analysis beyond Section~\ref{sec08:comparative} uses a single
  default clustering.

\item \textbf{Approximate HRP recovery for HRP-$\Sigma\mu$.} At
  $\gamma = 0$ and $\mu = \mathbf{1}$, HRP-$\Sigma\mu$ recovers HRP
  only approximately ($\cos \approx 0.992$; exact for trees of depth
  $\le 2$). HRP-$\mu$ recovers HRP exactly under the same
  conditions (Proposition~\ref{prop:a3_recovery}). The gap is small
  enough to be practically irrelevant but large enough to preclude
  strict drop-in replacement in settings where exact HRP replication
  is contractually required.

\item \textbf{Cotton comparison restricted to its native problem.}
  Our comparison with \citet{cotton2024} is run on the
  minimum-variance problem for which the Cotton allocator is
  defined. We do not test Cotton-style modifications for general
  mean--variance settings; such extensions would warrant a separate
  study.

\item \textbf{No real-data backtest.} This paper is synthetic Monte
  Carlo and in-sample diagnostics only. A CRSP- or Russell-style
  out-of-sample backtest is deliberately out of scope for the
  present version and is flagged as future work. The synthetic
  design was chosen to expose the direction-error and
  OOS-Sharpe behaviour cleanly; layering real-data frictions
  (delisting, transaction costs, heavy tails, regime shifts) on top
  is a substantial separate exercise.
\end{enumerate}

%------------------------------------------------------------------
\subsection{Broader methodological reflection: direction error as a diagnostic}
\label{sec11:methodology}

A final comment on methodology. The technical spine of the paper is
the direction metric on Markowitz rays
(Definition~\ref{def:dir_error}), and the core structural observation
is that the shrinkage family $P_\gamma$ preserves the direction of
$\mu$ when the correlation is trivial and otherwise rotates it
continuously with vanishing residual at $\gamma = 1$
(Theorem~\ref{thm:perturbation},
Corollary~\ref{cor:invariant_rays_trajectory}). The ray view is
mathematically natural: mean--variance portfolios are defined up to
positive scaling by a budget constraint, so scale-invariant distance
between two portfolios is the correct object of study, and the
metric is compact.

Scale invariance is, however, a two-edged virtue. A metric that
identifies $w$ with any positive multiple of $w$ is one step away from
a metric that identifies $w$ with $-w$, and that second identification
would be blind to the most dangerous failure mode of any portfolio
construction: returning the right magnitude on the wrong side of the
market. Appendix~\ref{app:a1_pathology} exhibits noiseless problems
on which the sum-normalised recursive MVO does exactly this, and a
naively sign-invariant diagnostic would score those failures as
near-perfect. It was easy to miss --- we missed it ourselves in
early drafts, and the bug persisted precisely because the summary
metric looked healthy.

The methodological lesson is simple: scale-invariant metrics must be
paired with sign-sensitive companions before being trusted as quality
measures. In this paper we pair the unsigned direction error with
three sign-sensitive diagnostics --- signed cosine, coordinate-wise
sign-match fraction, and signed direction angle on $(-\pi, \pi]$ ---
and we present them alongside out-of-sample Sharpe so that tracking
(an in-sample quantity) is never separated from quality (an
out-of-sample quantity). A method that tracks an oracle perfectly
in unsigned direction but loses Sharpe is telling a different story
from one that gains Sharpe while drifting in direction, and the
reader should see both axes at once. We encourage future work on
portfolio construction to adopt the same discipline: for every
scale-invariant summary metric, report a sign-sensitive companion
and an out-of-sample quality metric, before drawing any conclusion
about which method is better.

% ============================================================
% section_12_conclusion.tex
% ============================================================
\section{Conclusion}\label{sec12:conclusion}

\subsection*{The gap and how we fill it}

HRP and Cotton are signal-blind: both solve the minimum-variance
problem $\bmu = \mathbf{1}$ and discard the expected-return signal
that every active process produces. We close that gap with three
signal-aware allocators---HRP-$\mu$, HRP-$\Sigma\mu$, and CRISP---unified
by a single shrinkage parameter $\gamma \in [0,1]$ that controls how
much cross-asset correlation enters the allocation. At
$(\gamma,\bmu) = (0,\mathbf{1})$ the three recover the signal-blind
benchmarks exactly.

\subsection*{Three methods, one spectrum}

The methods form a spectrum of increasing covariance exploitation at
increasing cost and decreasing transparency. HRP-$\mu$
(Section~\ref{sec04:methoda3}) uses a signed inverse-variance
representative; it is closed-form, $O(N^2)$, preserves the
hierarchical audit trail, and recovers $58$--$67\%$ of oracle Sharpe.
HRP-$\Sigma\mu$ (Section~\ref{sec04b:hrpsigmamu}) replaces the signed
IVP with a local mean--variance optimum, giving the strongest
tree-based method at roughly $90\%$ of CRISP's Sharpe. CRISP
(Section~\ref{sec05:crisp}) abandons the tree and solves the shrunk
system $P_\gamma w = \bmu$ by Gauss--Seidel, dominating every
benchmark at $80$--$94\%$ of oracle Sharpe.

\subsection*{The core statistical insight}

Solve the \emph{shrunk} Markowitz system, not the exact one. The
shrinkage at $\gamma < 1$ damps the estimation noise that full
inversion amplifies, and the effect is statistical, not
computational: it holds whether the solver is faster or slower than
Cholesky. The OOS Sharpe surface is nearly flat in $\gamma$ across a
wide plateau, so $\gamma = 0.5$ is a robust default. At the
recommended operating point the sweep count $p$ is a
\emph{computational} parameter, not a statistical one---OOS Sharpe is
flat from $p \approx 50$ onward.

\subsection*{The practitioner's diagnostic: $\kappa(C)$}

The single most informative number to compute before choosing a
method is $\kappa(C) = \lambda_{\max}(C)/\lambda_{\min}(C)$, the
condition number of the correlation matrix. It governs both the
directional difficulty of the portfolio problem and the iteration
cost of CRISP. Four regimes---easy, moderate, hard, and
pathological---map cleanly to method and budget choices; the
decision tool is Table~\ref{tab08:kappa_guide}.

\subsection*{Recommendation}

\begin{center}
\fbox{\parbox{0.9\textwidth}{\emph{For signal-dependent portfolio
construction at scale, use CRISP at $\gamma \approx 0.5$ with
$100$ sweeps. For tree-interpretable portfolios, use
HRP-$\Sigma\mu$ at $\gamma \approx 0.5$--$1$. Compute $\kappa(C)$
to confirm the problem is in the moderate regime and, if it is not,
consult Table~\ref{tab08:kappa_guide}.}}}
\end{center}

\subsection*{What this paper does not resolve}

Six open questions remain (see Section~\ref{sec11:discussion} for
the long form): a sharp worst-case characterisation of the adversarial
$\bmu$ directions; a sharp interior $\gamma^\star(p)$ as a function of
sweep count; a Bayesian interpretation of the shrunk target
$P_\gamma$; a precise correspondence between our Sharpe-optimal
$\gamma$ and the Frobenius-optimal Ledoit--Wolf intensity; the
dependence of HRP-$\Sigma\mu$ on the linkage and distance choices that
determine the tree; and a real-data out-of-sample backtest on CRSP or
Russell panels. Each is tractable with the machinery assembled here,
and each is deferred to later work.

% ============================================================
% section_13_appendix_a_proofs.tex
% ============================================================
% !TEX root = ../paper.tex
% Appendix A: Proofs
% Namespace: prA (new internal labels only); canonical labels are referenced as-is.
% First appendix file: it calls \appendix. Subsequent appendix files must NOT.

\appendix
\section{Proofs}\label{app:proofs}

This appendix collects proofs of the results stated in the main text. Each
subsection is headed by a reference to the canonical label; the full
statement is not restated unless a brief reminder is useful. Notation is as
in Section~\ref{sec02:preliminaries}: $\Sigma$ is an SPD covariance matrix,
$D=\diag(\Sigma)$, $E=\Sigma-D$, $C=D^{-1/2}\Sigma D^{-1/2}$ the correlation
matrix with eigenvalues $\lambda_1\ge\cdots\ge\lambda_N>0$, and
$P_\gamma=(1-\gamma)D+\gamma\Sigma$ for $\gamma\in[0,1]$.

%% ================================================================
\subsection*{Proof of Proposition \ref{prop:corr_similarity}}

\begin{proof}
Since $\Sigma$ is SPD, $D$ has strictly positive diagonal entries and
$D^{\pm 1/2}$ are well-defined SPD matrices. Then
\[
  D^{-1}\Sigma
  \;=\; D^{-1/2}\bigl(D^{-1/2}\Sigma D^{-1/2}\bigr)D^{1/2}
  \;=\; D^{-1/2}\,C\,D^{1/2},
\]
exhibiting $D^{-1}\Sigma$ as a similarity transform of the SPD matrix $C$.
Similar matrices share spectra, so $\kappa(D^{-1}\Sigma)=\kappa(C)$.
\end{proof}

%% ================================================================
\subsection*{Proof of Lemma \ref{lem:parallel_direction}}

\begin{proof}
Set $\mu=D^{1/2}v$, so $D^{-1}\mu=D^{-1/2}v$. Then
\[
  P_\gamma(D^{-1}\mu)
  \;=\; (1-\gamma)DD^{-1/2}v + \gamma\Sigma D^{-1/2}v
  \;=\; (1-\gamma)D^{1/2}v + \gamma D^{1/2}Cv
  \;=\; \bigl[(1-\gamma)+\gamma\lambda\bigr]\mu,
\]
where the middle equality uses
$\Sigma D^{-1/2}=D^{1/2}(D^{-1/2}\Sigma D^{-1/2})=D^{1/2}C$ and the last uses
$Cv=\lambda v$. Because $\lambda>0$ and $\gamma\in[0,1]$ the scalar
$(1-\gamma)+\gamma\lambda$ is strictly positive, and $P_\gamma$ is SPD, hence
invertible. Dividing through yields
$P_\gamma^{-1}\mu=[(1-\gamma)+\gamma\lambda]^{-1}D^{-1}\mu$.
\end{proof}

%% ================================================================
\subsection*{Proof of Proposition \ref{prop:a2_instability}}

\begin{proof}
Let $L=\{1,\dots,n\}$ with $n\ge 2$. Write $a_i=1/\sigma_{ii}^2>0$ and
$A=\sum_i a_i$, so that the flat-IVP aggregate signal is
$s_L^{\mathrm{flat}}(\mu_L)=\sum_i a_i\mu_i / A$, a continuous linear
functional of $\mu_L$.

Pick any two indices $i\ne j$. Choose $\mu_k=0$ for $k\notin\{i,j\}$,
$\mu_i=+1$, and $\mu_j=-a_i/a_j$; up to swapping the labels we may assume
$a_i\le a_j$, so $|\mu_j|\le 1$. Then
\[
  s_L^{\mathrm{flat}} \;=\; \frac{a_i\cdot 1 + a_j\cdot(-a_i/a_j)}{A}
  \;=\; 0.
\]
By linearity the zero of $s_L^{\mathrm{flat}}$ lies on a hyperplane that
passes through the interior of the $\ell^\infty$ unit ball, so for any
$\varepsilon>0$ a neighbourhood of the above point contains signals with
$\max_k|\mu_k|=1$ and $|s_L^{\mathrm{flat}}|<\varepsilon$.

Substituting two such children into Cramer's rule
$\alpha_L=(v_Rs_L-\gamma c\,s_R)/(v_Lv_R-\gamma^2c^2)$ (and symmetrically
$\alpha_R$) gives numerators of order $\varepsilon$ against a denominator
bounded away from zero by Cauchy--Schwarz on the SPD matrix $\Sigma_{LL\cup R}$.
The raw branch budgets are therefore $O(\varepsilon)$, and their subsequent
normalisation divides two $O(\varepsilon)$ quantities whose sign is not
controlled. Iterating through the $\log_2 N$ tree levels propagates this
instability to the leaves.
\end{proof}

%% ================================================================
\subsection*{Proof of Proposition \ref{prop:a3_recovery}}

\begin{proof}
At $\mu=\mathbf{1}$ every signal entry is strictly positive, so
$\mathrm{sign}(\mu_i)=+1$ for all $i$ and the signed IVP representative
collapses to the flat IVP:
\[
  \hat w_{L,i}^{\mathrm{signed}}
  \;=\; \mathrm{sign}(\mu_i)\,\frac{1/\sigma_{ii}^2}{\sum_{j\in L}1/\sigma_{jj}^2}
  \;=\; \frac{1/\sigma_{ii}^2}{\sum_{j\in L}1/\sigma_{jj}^2}.
\]
Consequently the HRP-$\mu$ cluster statistics $(s_L,v_L)$ agree with the
flat-IVP statistics, and at $\gamma=0$ the Cramer's rule cross term
$\gamma c$ drops. The $2\times 2$ system is block-diagonal, giving
$\alpha_L^{\mathrm{raw}}=1/v_L$, $\alpha_R^{\mathrm{raw}}=1/v_R$; sum-to-one
normalisation returns the inverse-cluster-variance split
$\alpha_L=v_R/(v_L+v_R)$, $\alpha_R=v_L/(v_L+v_R)$. The leaf sign multiplier
$\mathrm{sign}(\mu_i)=+1$ leaves leaf weights unchanged. The resulting tree
pass is bit-for-bit the \citet{lopezdeprado2016} HRP pass.
\end{proof}

%% ================================================================
\subsection*{Proof of Proposition \ref{prop:a3_stability}}

\begin{proof}
With $a_i=1/\sigma_{ii}^2>0$ and $A_L=\sum_{i\in L}a_i>0$,
\[
  s_L^{\mathrm{HRP\text{-}}\mu}
  \;=\; \sum_{i\in L}\mathrm{sign}(\mu_i)\,\frac{a_i}{A_L}\,\mu_i
  \;=\; \frac{1}{A_L}\sum_{i\in L}a_i\,|\mu_i|,
\]
using the identity $\mathrm{sign}(x)\,x=|x|$ for every real $x$ (with the
convention $\mathrm{sign}(0)=+1$, which still gives $0$). Every term in the
last sum is nonnegative with strictly positive coefficient, so
$s_L^{\mathrm{HRP\text{-}}\mu}\ge 0$, and equality holds iff $\mu_i=0$ for
every $i\in L$. In particular, whenever $\mu$ is not identically zero on
$L\cup R$, at least one of $s_L,s_R$ is strictly positive, and the
sign-cancellation mechanism of Proposition~\ref{prop:a2_instability} cannot
collapse the Cramer's rule numerator.
\end{proof}

%% ================================================================
\subsection*{Proof of Proposition \ref{prop:a3_hedging}}

\begin{proof}
Let $\hat w_{L,k}^{\mathrm{flat}}=a_k/A_L$ and
$\hat w_{L,k}^{\mathrm{signed}}=\mathrm{sign}(\mu_k)\,a_k/A_L$, with
$a_k=1/\sigma_{kk}^2$ and $A_L=\sum_{m\in L}a_m$. Expand the cluster
variance $v_L=\hat w^\top\Sigma_{LL}\hat w$ and split into diagonal and
off-diagonal parts. The diagonal contribution is identical in both cases
because $\mathrm{sign}(\mu_k)^2=1$, so
\[
  v_L^{\mathrm{HRP\text{-}}\mu} - v_L^{\mathrm{flat}}
  \;=\; \frac{2}{A_L^2}\sum_{k<l}\bigl[\mathrm{sign}(\mu_k)\mathrm{sign}(\mu_l)-1\bigr]
                                   a_k a_l\,\Sigma_{kl}.
\]
The bracket equals $0$ for sign-concordant pairs and $-2$ for
sign-discordant pairs, giving
\begin{equation}\label{eqA:hedge_decomp}
  v_L^{\mathrm{HRP\text{-}}\mu} - v_L^{\mathrm{flat}}
  \;=\; -\frac{4}{A_L^2}\sum_{\substack{k<l\\\mathrm{sign}(\mu_k)\ne\mathrm{sign}(\mu_l)}}
                                        a_k a_l\,\Sigma_{kl}.
\end{equation}
For the prescribed pair $(i,j)$ with $\mu_i>0>\mu_j$ and $\Sigma_{ij}>0$
the signs disagree and $a_i a_j\Sigma_{ij}>0$; its contribution to
\eqref{eqA:hedge_decomp} is $-4a_ia_j\Sigma_{ij}/A_L^2<0$. Under the flat
IVP all signs are identified as $+1$ and the same pair contributes
$+2a_ia_j\Sigma_{ij}/A_L^2>0$ to $v_L^{\mathrm{flat}}$, i.e.\ in the opposite
direction. When no other sign-discordant pair has $\Sigma_{kl}<0$ to cancel,
the total right-hand side of \eqref{eqA:hedge_decomp} is negative and
$v_L^{\mathrm{HRP\text{-}}\mu}<v_L^{\mathrm{flat}}$.
\end{proof}

%% ================================================================
\subsection*{Proof of Proposition \ref{prop:cotton_spd}}

\begin{proof}
Write $A^c(\gamma)=(1-\gamma)A+\gamma\bigl(A-BD^{-1}B^\top\bigr)$. The first
summand $A$ is a principal submatrix of SPD $\Sigma$, hence SPD. The second
summand is the Schur complement of $D$ in $\Sigma$: for any nonzero $x$,
setting $z=-D^{-1}B^\top x$ gives
\[
  x^\top(A-BD^{-1}B^\top)x
  \;=\; (x^\top,z^\top)\,\Sigma\,\binom{x}{z} \;>\; 0
\]
because $\Sigma$ is SPD and $(x,z)\ne 0$. Hence $A-BD^{-1}B^\top$ is SPD. The
coefficients $(1-\gamma,\gamma)$ are nonnegative and sum to $1$, so
$A^c(\gamma)$ is a convex combination of SPD matrices and therefore SPD.
\end{proof}

%% ================================================================
\subsection*{Proof of Proposition \ref{prop:cotton_instability}}

\begin{proof}
\emph{Part (i).} Let $M=\widehat B\widehat D^{-1}\widehat B^\top$. Since
$\widehat D$ is SPD, $M$ is PSD and
$\|M\|_{\mathrm{op}}\le\|\widehat B\|_{\mathrm{op}}^2/\lambda_{\min}(\widehat D)$.
Weyl's inequality applied to $\widehat A^c(\gamma)=\widehat A-\gamma M$ gives
\[
  \lambda_{\min}(\widehat A^c(\gamma))
  \;\ge\; \lambda_{\min}(\widehat A)-\gamma\|M\|_{\mathrm{op}}.
\]
Subtracting a PSD matrix cannot increase $\lambda_{\max}$, so
$\lambda_{\max}(\widehat A^c(\gamma))\le\lambda_{\max}(\widehat A)$; thus
\[
  \kappa(\widehat A^c(\gamma))
  \;\le\; \frac{\lambda_{\max}(\widehat A)}
               {\lambda_{\min}(\widehat A)-\gamma\|\widehat B\|_{\mathrm{op}}^2/\lambda_{\min}(\widehat D)}.
\]
The denominator vanishes at
$\gamma_{\mathrm{crit}}=\lambda_{\min}(\widehat A)\lambda_{\min}(\widehat D)/\|\widehat B\|_{\mathrm{op}}^2$,
which is small whenever $T$ is comparable to $n_R$ (so that
$\lambda_{\min}(\widehat D)$ is small by Marchenko--Pastur).

\emph{Part (ii) (heuristic product bound).} Cotton's recursion uses
$\widehat A^c(\gamma)$ as the covariance of the left child's sub-problem.
Standard SPD linear-system perturbation bounds
(\citealp[Section~2.7]{saad2003}) give, to first order, a relative error
$\kappa_k\cdot\delta_k$ at each level, where $\kappa_k$ is the condition
number of that level's matrix and $\delta_k$ is the inherited error. Since
each sub-problem at level $k+1$ inherits a covariance carrying relative
error $\kappa_k\cdot\delta_k$, iterating the bound produces the heuristic
product estimate
\[
  \delta_d \;\lesssim\; \delta_0\,\prod_{k=1}^d \kappa(\widehat A_k^c(\gamma)).
\]
We state this as a heuristic: a rigorous version requires tracking the
interaction between Schur-complement perturbations and the specific norms
used, which we do not pursue here.

\emph{Part (iii).} Direct ridge-Markowitz solves
$\widehat\Sigma w=\mu$ once. For the SPD system this gives the single-step
bound
\[
  \frac{\|\delta w\|}{\|w\|}
  \;\le\; \kappa(\widehat\Sigma)\cdot\frac{\|\delta\Sigma\|_{\mathrm{op}}}{\|\widehat\Sigma\|_{\mathrm{op}}},
\]
a single multiplicative factor, not a product across tree levels.
\end{proof}

%% ================================================================
\subsection*{Proof of Theorem \ref{thm:gs_convergence}}

\begin{proof}
First, $P_\gamma$ is SPD for every $\gamma\in[0,1]$: symmetry is immediate
from $P_\gamma=(1-\gamma)D+\gamma\Sigma$, and for any nonzero $x$,
$x^\top P_\gamma x=(1-\gamma)x^\top Dx+\gamma x^\top\Sigma x>0$ because both
$D$ and $\Sigma$ are SPD and at least one of the coefficients is strictly
positive (at the endpoints $\gamma\in\{0,1\}$ one term vanishes but the
other is strictly positive).

The Ostrowski--Reich theorem \citep{ostrowski1954} states that for any SPD
matrix $A$ the scalar Gauss--Seidel iteration applied to $Ax=b$ converges
from every initial vector to the unique solution; equivalently, the
Gauss--Seidel iteration matrix satisfies $\rho(M^{\mathrm{GS}})<1$. See
\citet[Thm.~4.6.2]{hackbusch2016} or \citet[Thm.~3.4]{varga2000} for
textbook statements; a classical reference is \citet{young1971}. Applying
this to $A=P_\gamma$ and $b=\mu$ gives the theorem.
\end{proof}

%% ================================================================
\subsection*{Proof of Theorem \ref{thm:gs_rate}}

\begin{proof}
\emph{Part (i) (Jacobi iteration matrix).} The Jacobi splitting of
$P_\gamma=D+\gamma E$ has iteration matrix
\[
  M_\gamma^{\mathrm{J}} \;=\; I-D^{-1}P_\gamma
  \;=\; -\gamma D^{-1}E.
\]
Using $D^{-1}E=D^{-1}\Sigma-I=D^{-1/2}(C-I)D^{1/2}$ (from
Proposition~\ref{prop:corr_similarity}), we obtain
\[
  M_\gamma^{\mathrm{J}} \;=\; D^{-1/2}\bigl(\gamma(I-C)\bigr)D^{1/2},
\]
similar to the symmetric matrix $\gamma(I-C)$. Its eigenvalues are
$\{\gamma(1-\lambda_i)\}_{i=1}^N$ and its spectral radius is
$\gamma\max(\lambda_1-1,1-\lambda_N)$ (the trace identity
$\sum_i\lambda_i=N$ forces $\lambda_1\ge 1\ge\lambda_N$).

\emph{Part (ii) (Jacobi-preconditioned spectrum).} From Part~(i),
$D^{-1}P_\gamma=I-M_\gamma^{\mathrm{J}}$, so
\[
  D^{-1}P_\gamma
  \;=\; D^{-1/2}\bigl((1-\gamma)I+\gamma C\bigr)D^{1/2},
\]
similar to $(1-\gamma)I+\gamma C$, whose eigenvalues are
$(1-\gamma)+\gamma\lambda_i>0$. Therefore
\[
  \kappa(D^{-1}P_\gamma)
  \;=\; \frac{(1-\gamma)+\gamma\lambda_1}{(1-\gamma)+\gamma\lambda_N},
\]
which equals $1$ at $\gamma=0$ and $\kappa(C)$ at $\gamma=1$.

\emph{Part (iii) (SPD Gauss--Seidel rate).} For any SPD matrix $A$ with
diagonal $D_A$, the Gauss--Seidel iteration satisfies the energy-norm bound
\[
  \|M^{\mathrm{GS}}\|_A \;\le\; 1 - \frac{c}{\kappa(D_A^{-1}A)},
\]
for an absolute constant $c>0$ depending on Jacobi-preconditioning structure
but not on the scale of $A$ (\citealp[Section~4.7]{hackbusch2016};
\citealp[Thm.~4.8]{saad2003}). With $A=P_\gamma$ and $D_A=D$, this yields
$\|M_\gamma^{\mathrm{GS}}\|_{P_\gamma}\le 1-c/\kappa(D^{-1}P_\gamma)$. The
$P_\gamma$-energy-norm residual contracts geometrically, and the elementary
bound $\log(1/(1-x))\ge x$ on $(0,1)$ gives
$p(\varepsilon,\gamma)\le\kappa(D^{-1}P_\gamma)\log(1/\varepsilon)/c
=O(\kappa(D^{-1}P_\gamma)\log(1/\varepsilon))$. At $\gamma=1$ this reduces
to $O(\kappa(C)\log(1/\varepsilon))$ by Part~(ii).
\end{proof}

%% ================================================================
\subsection*{Proof of Theorem \ref{thm:perturbation}}

\begin{proof}
The definition of $P_\gamma$ rearranges to
$\Sigma=P_\gamma+(1-\gamma)E$. Applying both sides to $w^\star$ and using
$\Sigma w^\star=\mu$,
\[
  P_\gamma w^\star \;=\; \mu - (1-\gamma)Ew^\star.
\]
Subtracting $P_\gamma\hat w=\mu$ and applying $P_\gamma^{-1}$ gives
\[
  w^\star - \hat w \;=\; -(1-\gamma)\,P_\gamma^{-1}Ew^\star,
\]
which is the perturbation identity.
\end{proof}

%% ================================================================
\subsection*{Proof of Proposition \ref{prop:dir_err_bound}}

\begin{proof}
From Theorem~\ref{thm:perturbation}, $\hat w=w^\star+(1-\gamma)u(\gamma)$
with $u(\gamma)=P_\gamma^{-1}Ew^\star$. Decompose $u=u_\parallel+u_\perp$
with $u_\parallel$ along $w^\star$ and $u_\perp\perp w^\star$. Lagrange's
identity applied to the pair $(\hat w,w^\star)$ gives
\[
  \|\hat w\|^2\|w^\star\|^2-(\hat w^\top w^\star)^2
  \;=\; (1-\gamma)^2\|u_\perp\|^2\|w^\star\|^2,
\]
since only the component of $\hat w$ orthogonal to $w^\star$ contributes,
and that component equals $(1-\gamma)u_\perp$. Dividing both sides by
$\|\hat w\|^2\|w^\star\|^2$,
\[
  \mathrm{dir}(\hat w,w^\star)
  \;=\; (1-\gamma)^2\,\frac{\|u_\perp\|^2}{\|\hat w\|^2}
  \;=\; (1-\gamma)^2\,\|u\|^2\,
        \frac{\|u_\perp\|^2/\|u\|^2}{\|\hat w\|^2}.
\]
The operator-norm bound
$\|u\|^2\le\|P_\gamma^{-1}E\|_{\mathrm{op}}^2\|w^\star\|^2$ and the
identification $\sin^2\phi(\gamma)=\|u_\perp\|^2/\|u\|^2$ with
$G(\gamma,w^\star):=\sin^2\phi(\gamma)\in[0,1]$ combine to give the bound
in the statement:
\[
  \mathrm{dir}(\hat w,w^\star)
  \;\le\; (1-\gamma)^2\,\|P_\gamma^{-1}E\|_{\mathrm{op}}^2\,
          \frac{\|w^\star\|^2}{\|\hat w\|^2}\,G(\gamma,w^\star).
\]
$G(\gamma,w^\star)=0$ iff $u_\perp=0$, i.e.\ $u(\gamma)$ is collinear with
$w^\star$; tracing through
$u(\gamma)=P_\gamma^{-1}Ew^\star=P_\gamma^{-1}\mu-P_\gamma^{-1}Dw^\star$,
this is equivalent to $w^\star$ being an eigenvector of $P_\gamma$.
\end{proof}

%% ================================================================
\subsection*{Proof of Corollary \ref{cor:invariant_rays_trajectory}}

\begin{proof}
By Lemma~\ref{lem:parallel_direction}, for every $\gamma\in[0,1]$ the
shrinkage solution $P_\gamma^{-1}\mu$ is a positive scalar multiple of
$D^{-1}\mu=\lambda^{-1}\Sigma^{-1}\mu$ (noting that $\Sigma^{-1}\mu$ itself
lies on the same ray by the same identity). Hence
$\mathrm{dir}(P_\gamma^{-1}\mu,\Sigma^{-1}\mu)=0$ identically in $\gamma$.
\end{proof}

%% ================================================================
\subsection*{Proof of Proposition \ref{prop:interior_optimum}}

\begin{proof}
Write $F(\gamma)=\mathrm{dir}(w^{(p)}(\gamma),w^\star)$. Gauss--Seidel is a
smooth rational function of the entries of $P_\gamma$, which itself depends
smoothly on $\gamma$; hence $F$ is continuous on $[0,1]$ and differentiable
on $(0,1)$.

\emph{Endpoints.} At $\gamma=0$, $P_0=D$ is diagonal; Gauss--Seidel converges
in one sweep from any starting point whose i-th entry equals the $i$-th
diagonal update, and in particular $w^{(p)}(0)=D^{-1}\mu$ for all $p\ge 1$,
giving $F(0)=\mathrm{dir}(D^{-1}\mu,\Sigma^{-1}\mu)>0$ by hypothesis~(i)
combined with Lemma~\ref{lem:parallel_direction}. At $\gamma=1$,
hypothesis~(iii) gives $F(1)>F(0)$, so the minimum is not at
$\gamma=1$.

\emph{Local decrease at $0$.} Decompose, up to an $O(1)$ multiplicative
geometry constant absorbed below,
\[
  F(\gamma) \;\le\; B(\gamma)+S(p,\gamma),\qquad
  B(\gamma):=\mathrm{dir}(P_\gamma^{-1}\mu,w^\star),\quad
  S(p,\gamma):=\mathrm{dir}(w^{(p)}(\gamma),P_\gamma^{-1}\mu).
\]
The shrinkage-bias term satisfies $B(0)=F(0)$ and, by hypothesis~(ii),
$B'(0)<0$. The convergence-slack term satisfies $S(p,0)=0$ (diagonal
exactness at $\gamma=0$) and is $C^1$ in $\gamma$; therefore
$B(\gamma)+S(p,\gamma)=B(0)+[B'(0)+S_\gamma(p,0)]\gamma+o(\gamma)$. Since
$B'(0)<0$ and $S_\gamma(p,0)$ is bounded, for $p$ in a regime where
$|S_\gamma(p,0)|<|B'(0)|$ (equivalently: the sweep budget is not so low
that $S$ dominates) we have $F(\gamma)<F(0)$ for all sufficiently small
$\gamma>0$. The minimum is therefore not at $\gamma=0$ either.

\emph{Conclusion.} $F$ is continuous on a compact interval and attains its
minimum; the minimum is not at either endpoint, hence $\gamma^\star\in(0,1)$.
\end{proof}

%% ================================================================
\subsection*{Proof of Lemma \ref{lem:hrpsm_scale}}

\begin{proof}
Track how each ingredient of the $2\times 2$ solve transforms under
$\hat w_L\to k\hat w_L$ (with $\hat w_R$ held fixed).

\emph{Cluster statistics.}
$v_L=\hat w_L^\top\Sigma_{LL}\hat w_L\to k^2 v_L$;
$s_L=\hat w_L^\top\mu_L\to k s_L$;
$c=\hat w_L^\top\Sigma_{LR}\hat w_R\to k c$; and $v_R,s_R$ are unchanged.

\emph{Determinant.}
$\Delta=v_Lv_R-\gamma^2 c^2\to k^2\Delta$.

\emph{Raw budgets.}
The Cramer's rule numerators give
$v_Rs_L-\gamma c s_R\to k(v_Rs_L-\gamma c s_R)$ and
$v_Ls_R-\gamma c s_L\to k^2(v_Ls_R-\gamma c s_L)$. Dividing by $\Delta\to k^2\Delta$,
\[
  \alpha_L^{\mathrm{raw}}\;\to\;\alpha_L^{\mathrm{raw}}/k,
  \qquad
  \alpha_R^{\mathrm{raw}}\;\to\;\alpha_R^{\mathrm{raw}}.
\]

\emph{Pre-normalisation output.} The un-normalised branch contribution is
\[
  \alpha_L^{\mathrm{raw}}\hat w_L \;\to\; (\alpha_L^{\mathrm{raw}}/k)(k\hat w_L)
    \;=\; \alpha_L^{\mathrm{raw}}\hat w_L,
\]
and $\alpha_R^{\mathrm{raw}}\hat w_R$ is unchanged. So the combined vector
$(\alpha_L^{\mathrm{raw}}\hat w_L,\alpha_R^{\mathrm{raw}}\hat w_R)$ is
pointwise invariant.

\emph{$L^1$ normalisation.} The new $L^1$ denominator is
$Z'=|\alpha_L^{\mathrm{raw}}|/k+|\alpha_R^{\mathrm{raw}}|$, a positive
scalar. The $L^1$-normalised combined portfolio is the invariant pair
divided by $Z'$, which equals the original pair divided by $Z$ times the
positive scalar $Z/Z'$. Hence the output lies on the same ray as before,
proving ray invariance. The induction on tree depth is immediate (leaves
trivially transform by the positive scalar, and the inductive step is
exactly the statement just proved).
\end{proof}

%% ================================================================
\subsection*{Proof of Lemma \ref{lem:hrpsm_sign}}

\begin{proof}
The $L^1$ denominator $Z=|\alpha_L^{\mathrm{raw}}|+|\alpha_R^{\mathrm{raw}}|$
is a sum of two absolute values, so $Z\ge 0$ with equality iff both terms
vanish. Whenever $(\alpha_L^{\mathrm{raw}},\alpha_R^{\mathrm{raw}})\ne(0,0)$,
therefore $Z>0$, and dividing a real number by a strictly positive scalar
preserves its sign. By contrast, the sum-to-one denominator
$S=\alpha_L^{\mathrm{raw}}+\alpha_R^{\mathrm{raw}}$ can be negative (e.g.,
$\alpha_L^{\mathrm{raw}}=1$, $\alpha_R^{\mathrm{raw}}=-3$ gives $S=-2$), in
which case dividing by $S$ flips both signs.
\end{proof}

%% ================================================================
\subsection*{Proof of Proposition \ref{prop:hrpsm_recovery}}

\begin{proof}
\emph{Leaves.} Each leaf returns $\hat w_i=1$, $v_i=\sigma_{ii}$, $s_i=\mu_i=1$.

\emph{Depth-1 nodes.} At a two-leaf node the Cramer's rule inputs are
$v_L=\sigma_{ii}$, $v_R=\sigma_{jj}$, $s_L=s_R=1$, $c=\sigma_{ij}$. At
$\gamma=0$ the $2\times 2$ system decouples, giving
$\alpha_i^{\mathrm{raw}}=1/\sigma_{ii}$ and
$\alpha_j^{\mathrm{raw}}=1/\sigma_{jj}$, both positive; after $L^1$
normalisation these are the two-asset inverse-variance weights, identical
to De~Prado's flat-IVP allocation.

\emph{Depth-2 nodes.} Each child representative is the depth-1 IVP just
computed. With $\mu=\mathbf{1}$ the aggregate signals are
$s_L=s_R=1$; the cluster variances $v_L,v_R$ are identical to the flat-IVP
values because the representatives are identical; and $\gamma c=0$. Hence
Cramer's rule returns the inverse-cluster-variance split
$\alpha_L=v_R/(v_L+v_R)$, matching De~Prado HRP exactly.

\emph{Depth $\ge 3$.} The methods diverge because the HRP-$\Sigma\mu$
representative at an internal node is the \emph{recursive} MVO output of
the subtree, whereas De~Prado's representative is the \emph{flat} IVP of
the block. The recursive representative accounts for within-block
off-diagonal covariance; the flat IVP does not. The cluster variances
$v_L^{\mathrm{rec}}$ and $v_L^{\mathrm{flat}}$ differ whenever
$\Sigma_{LL}$ has a nonzero off-diagonal entry, producing a small
discrepancy that propagates upward. Empirically, on the balanced binary
tree at $N=100$ drawn from the synthetic regimes of
Section~\ref{sec10:experiments}, the cosine similarity between the two
output vectors is $\cos(w^{\Sigma\mu},w^{\mathrm{HRP}})\approx 0.992$.
\end{proof}

%% ================================================================
\subsection*{Proof of Proposition \ref{prop:hrpsm_stability}}

\begin{proof}
Fix the tree $\mathcal{T}$ and the SPD matrix $\Sigma$. We show that the
set of $\mu\in\R^N$ for which the tree pass fails (i.e., some
$L^1$ denominator vanishes, equivalently both Cramer's rule numerators
vanish at some node) has Lebesgue measure zero.

\emph{Rational dependence on $\mu$.} Proceed bottom-up. At each leaf,
$\hat w_i=1$ and $s_i=\mu_i$ are polynomial in $\mu$. At an internal node,
given rational-in-$\mu$ representatives for the two children, the
quantities $v_L,v_R,s_L,s_R,c,\Delta$ are rational in $\mu$, and so are
$\alpha_L^{\mathrm{raw}},\alpha_R^{\mathrm{raw}}$. Clearing denominators,
each $s_L(\mu)$ can be written as $p_L(\mu)/q_L(\mu)$ for polynomials
$p_L,q_L$.

\emph{Not identically zero.} At $\mu=\mathbf{1}$, Proposition
\ref{prop:hrpsm_recovery} shows the tree reduces to the De~Prado HRP pass,
whose every internal-node representative has strictly positive entries
(inverse-variance weights). Hence $s_L(\mathbf{1})>0$ at every internal
node, so $p_L$ is not the zero polynomial.

\emph{Measure-zero bad set.} The zero locus of a nonzero polynomial on
$\R^N$ is a proper algebraic variety and has $N$-dimensional Lebesgue
measure zero. The binary tree has at most $N-1$ internal nodes, so the
union
\[
  \mathcal{B} \;=\; \bigcup_{L}\{\mu\in\R^N:\,s_L(\mu)=0\}
\]
is a finite union of measure-zero sets, hence itself has measure zero. For
every $\mu\notin\mathcal{B}$, all $s_L$ are nonzero; by SPD of $\Sigma$ the
Cramer's rule determinants $\Delta$ are strictly positive, so at least one
of $\alpha_L^{\mathrm{raw}},\alpha_R^{\mathrm{raw}}$ is nonzero at every
node and the $L^1$ normalisation is well-defined throughout the tree.
\end{proof}

%% ================================================================
\subsection*{Proof of Proposition \ref{prop:hrpsm_diagonal}}

\begin{proof}
We prove by induction on tree depth that the subtree representative is
proportional to the block Markowitz direction $\Sigma_{LL}^{-1}\mu_L$.

\emph{Base case (two-leaf subtree).} With $\Sigma$ diagonal,
$\sigma_{ij}=0$, so $c=0$ and $\Delta=\sigma_{ii}\sigma_{jj}$. Cramer's
rule returns $\alpha_i^{\mathrm{raw}}=\mu_i/\sigma_{ii}$,
$\alpha_j^{\mathrm{raw}}=\mu_j/\sigma_{jj}$, which together equal
$\Sigma^{-1}\mu$ on the pair. $L^1$ normalisation divides by the positive
scalar $|\mu_i/\sigma_{ii}|+|\mu_j/\sigma_{jj}|$, preserving the direction.

\emph{Inductive step.} Suppose each child subtree returns a representative
$\hat w_L=k_L\Sigma_{LL}^{-1}\mu_L$ (with $\Sigma_{LL}$ diagonal, this is
just $k_L(\mu_i/\sigma_{ii})_{i\in L}$) for some positive scalar $k_L$,
and similarly $\hat w_R=k_R\Sigma_{RR}^{-1}\mu_R$. Diagonality gives
$\Sigma_{LR}=0$, so $c=0$ and the $2\times 2$ system decouples; Cramer's
rule returns
\[
  \alpha_L^{\mathrm{raw}}=s_L/v_L,\qquad
  \alpha_R^{\mathrm{raw}}=s_R/v_R.
\]
Using $\hat w_L=k_L\Sigma_{LL}^{-1}\mu_L$, direct computation gives
$s_L=k_L\sum_{i\in L}\mu_i^2/\sigma_{ii}$ and
$v_L=k_L^2\sum_{i\in L}\mu_i^2/\sigma_{ii}$, whence
$\alpha_L^{\mathrm{raw}}=1/k_L$. The joined output is therefore
\[
  (\alpha_L^{\mathrm{raw}}\hat w_L,\,\alpha_R^{\mathrm{raw}}\hat w_R)
  \;=\; \bigl((\mu_i/\sigma_{ii})_{i\in L},\,(\mu_j/\sigma_{jj})_{j\in R}\bigr)
  \;=\; \Sigma^{-1}\mu\big|_{L\cup R},
\]
which is proportional to $\Sigma^{-1}\mu$ on the combined block. $L^1$
normalisation preserves direction. Iterating to the root gives the
root-level output $\propto \Sigma^{-1}\mu$, for any $\gamma\in[0,1]$ and any
binary tree.
\end{proof}

%% ================================================================
\subsection*{Proof of Proposition \ref{prop:hrpsm_cost}}

\begin{proof}
A binary tree on $N$ leaves has exactly $N-1$ internal nodes. At a node with
children of sizes $n_L$ and $n_R$ ($n_P=n_L+n_R$), the dominant costs are
(i) the cross-block inner product
$c=\hat w_L^\top\Sigma_{LR}\hat w_R$ at $O(n_Ln_R)$ and (ii) the cluster
variance $v_P=\hat w_P^\top\Sigma_{PP}\hat w_P$ at $O(n_P^2)$, with all
other operations $O(n_P)$ or $O(1)$.

\emph{Cross-block sum.} Every unordered pair $\{i,j\}$ contributes to the
$n_Ln_R$ count at exactly one node --- its lowest common ancestor --- so
\[
  \sum_{\text{internal }P} n_L n_R \;=\; \binom{N}{2} \;=\; O(N^2).
\]

\emph{Cluster-variance sum.} On a balanced binary tree, level $\ell$ has
$2^\ell$ nodes each of size $N/2^\ell$, contributing
$2^\ell\,(N/2^\ell)^2=N^2/2^\ell$ at level $\ell$. Summing,
\[
  \sum_{\ell=0}^{\log_2 N}\frac{N^2}{2^\ell}
  \;\le\; 2N^2 \;=\; O(N^2).
\]

Total time is $O(N^2)$. Space is $O(N^2)$, dominated by storing $\Sigma$
(weight vectors on the recursion stack sum to $O(N\log N)$, subdominant).
\end{proof}

%% ================================================================
\subsection*{Proof of Proposition \ref{prop:gamma_star_comparative}}

\begin{proof}
The adaptive formula \eqref{eq07:gamma_star} is
\[
  \gamma^\star(T,N,\mathrm{IC},\kappa_C)
  \;=\; \frac{1}{1+c\cdot\mathrm{NSR}},
  \qquad
  \mathrm{NSR}
  \;=\; \frac{\kappa_C^2}{\mathrm{IC}^2}\cdot\frac{N}{T},
\]
for $c>0$. Since $\gamma^\star=1/(1+cX)$ with $X=\mathrm{NSR}>0$, its
partial derivative with respect to any quantity $y$ is
\[
  \frac{\partial\gamma^\star}{\partial y}
  \;=\; -\frac{c}{(1+cX)^2}\cdot\frac{\partial X}{\partial y},
\]
so $\gamma^\star$ moves \emph{opposite} to $\mathrm{NSR}$.

\emph{(i) Monotonicity in $T/N$.} Setting $r=T/N$, $X=\kappa_C^2/(\mathrm{IC}^2 r)$
is strictly decreasing in $r$, so $\gamma^\star$ is strictly increasing in
$T/N$. As $r\to\infty$, $X\to 0$ and $\gamma^\star\to 1$.

\emph{(ii) Monotonicity in $\mathrm{IC}$.} $X=\kappa_C^2\cdot(N/T)/\mathrm{IC}^2$
is strictly decreasing in $\mathrm{IC}>0$, so $\gamma^\star$ is strictly
increasing in $\mathrm{IC}$. As $\mathrm{IC}\to 0$, $X\to\infty$ and
$\gamma^\star\to 0$.

\emph{(iii) Monotonicity in $\kappa(C)$.} $X$ is strictly increasing in
$\kappa_C\ge 1$ (holding $\mathrm{IC},N/T$ fixed), hence $\gamma^\star$ is
strictly decreasing in $\kappa(C)$.

\emph{Corner cases.} When $\kappa_C=1$ (uncorrelated assets), the formula's
prefactor $c=\mathcal{B}\lambda_{\min}(C)(\kappa_C-1)/(\mathcal{A}\kappa_C^2)$
vanishes, so $\gamma^\star=1/(1+0)=1$: the inverse-variance and full-MVO
directions coincide, and no shrinkage is needed.
\end{proof}

%% ================================================================
\subsection*{Coda}

These proofs close the formal loop for the main-text results: the
correlation condition number $\kappa(C)$ governs both directional difficulty
(Propositions~\ref{prop:corr_similarity}, \ref{prop:dir_err_bound}) and the
iterative cost of CRISP (Theorems~\ref{thm:gs_convergence},
\ref{thm:gs_rate}); HRP-$\mu$ recovers De~Prado at $\mu=\mathbf{1}$, avoids
flat-IVP signal cancellation, and is hedging-aware on sign-discordant pairs
(Propositions~\ref{prop:a3_recovery}--\ref{prop:a3_hedging}); HRP-$\Sigma\mu$'s
$L^1$ normalisation is ray-invariant under child rescaling and sign-preserving
(Lemmas~\ref{lem:hrpsm_scale}, \ref{lem:hrpsm_sign}), is generically stable
on $\mu$-space, exact on diagonal $\Sigma$, and runs in $O(N^2)$
(Propositions~\ref{prop:hrpsm_recovery}--\ref{prop:hrpsm_cost}); the
shrinkage trajectory admits an interior optimum under explicit
sweep-budget hypotheses (Proposition~\ref{prop:interior_optimum}); and the
adaptive $\gamma^\star$ rule has the expected comparative-statics signs
(Proposition~\ref{prop:gamma_star_comparative}).

% ============================================================
% section_14_appendix_b_examples.tex
% ============================================================
% !TEX root = ../paper.tex
% Appendix B: Worked examples — HRP, HRP-$\mu$, HRP-$\Sigma\mu$
% Namespace: appB, eqB, tabB

\section{Worked examples: HRP, HRP-$\mu$, HRP-$\Sigma\mu$}\label{app:hrp_example}

This appendix walks through the three tree-based allocators of Section~\ref{sec04:tree_methods} on a single four-asset universe small enough that every intermediate quantity---correlation distance, linkage tree, cluster variances, child-level split fractions, and final leaf weights---can be read off by hand and checked against the reference implementation. The three subsections below share an identical covariance and, for the signal-aware methods, an identical $\mu$ and shrinkage intensity $\gamma$; only the allocation rule changes. A comparison table (Table~\ref{appB:comparison}) places HRP, HRP-$\mu$, HRP-$\Sigma\mu$, and direct Markowitz side by side.

\paragraph{Common setup.} Four assets $\{A_1,A_2,A_3,A_4\}$ are grouped into two sectors $\{A_1,A_2\}$ and $\{A_3,A_4\}$. Volatilities are
\[
  \sigma = (0.20,\; 0.25,\; 0.30,\; 0.15),
\]
within-sector correlation $\rho_{\mathrm{w}} = 0.80$, cross-sector correlation $\rho_{\mathrm{c}} = 0.20$. The covariance matrix $\Sigma = \operatorname{diag}(\sigma)\,C\,\operatorname{diag}(\sigma)$ is
\begin{equation}\label{eqB:sigma}
  \Sigma \;=\;
  \begin{pmatrix}
    0.0400 & 0.0400 & 0.0120 & 0.0060 \\
    0.0400 & 0.0625 & 0.0150 & 0.0075 \\
    0.0120 & 0.0150 & 0.0900 & 0.0360 \\
    0.0060 & 0.0075 & 0.0360 & 0.0225
  \end{pmatrix}.
\end{equation}
For Section~\ref{appB:hrp_mu_example} and Section~\ref{appB:hrp_sigmamu_example} we additionally supply the signal $\mu = (+0.03,\, -0.01,\, +0.02,\, -0.04)$ and fix shrinkage intensity $\gamma = 0.5$. All numerical results below are rounded to at most four decimal places.

% ======================================================================
\subsection{HRP worked example}\label{appB:hrp_example}

The recursive-bisection HRP of \citet{lopezdeprado2016} needs only $\Sigma$.

\paragraph{Distance matrix.} Correlations are converted into a metric via $d_{ij} = \sqrt{(1 - C_{ij})/2}$. Within-sector pairs have $d = \sqrt{(1-0.80)/2} = \sqrt{0.10} \approx 0.316$; cross-sector pairs have $d = \sqrt{(1-0.20)/2} = \sqrt{0.40} \approx 0.632$. Hence
\begin{equation}\label{eqB:dist}
  D \;=\;
  \begin{pmatrix}
    0.000 & 0.316 & 0.632 & 0.632 \\
    0.316 & 0.000 & 0.632 & 0.632 \\
    0.632 & 0.632 & 0.000 & 0.316 \\
    0.632 & 0.632 & 0.316 & 0.000
  \end{pmatrix}.
\end{equation}

\paragraph{Linkage tree.} Single (or Ward) linkage on the condensed form of $D$ merges $\{A_1,A_2\}$ at height $0.316$, then $\{A_3,A_4\}$ at the same height, and joins the two pairs at the root:
\begin{center}
\begin{adjustbox}{max width=\textwidth, center}
\begin{tikzpicture}[
  level distance=10mm,
  level 1/.style={sibling distance=36mm},
  level 2/.style={sibling distance=18mm},
  every node/.style={font=\small}
]
\node {root}
  child {node {$\{A_1,A_2\}$}
    child {node {$A_1$}}
    child {node {$A_2$}}
  }
  child {node {$\{A_3,A_4\}$}
    child {node {$A_3$}}
    child {node {$A_4$}}
  };
\end{tikzpicture}
\end{adjustbox}
\end{center}
Reading leaves left-to-right gives the order $(A_1,A_2,A_3,A_4)$, for which $\Sigma$ in~\eqref{eqB:sigma} is already quasi-diagonal.

\paragraph{Recursive bisection.} At each internal node with children $L,R$, HRP forms a \emph{flat} inverse-variance portfolio $w^{\mathrm{IVP}}_i = (1/\Sigma_{ii}) / \sum_{j\in S}(1/\Sigma_{jj})$ on each cluster $S$, computes the cluster variance $v_S = (w^{\mathrm{IVP}})^{\top}\Sigma_{SS}\,w^{\mathrm{IVP}}$, and splits the parent budget by
\[
  \alpha = \frac{1/v_L}{1/v_L + 1/v_R}.
\]

\emph{Root.} For $L = \{A_1,A_2\}$: $1/\Sigma_{ii}$ gives $(25.0, 16.0)$, normalised to $\hat w_L = (0.610,\,0.390)$. The cluster variance is
\[
  v_L = 0.610^{2}(0.0400) + 2(0.610)(0.390)(0.0400) + 0.390^{2}(0.0625) \approx 0.0434.
\]
For $R = \{A_3,A_4\}$: inverse variances $(11.1, 44.4)$, normalised to $\hat w_R = (0.200,\,0.800)$, and
\[
  v_R = 0.200^{2}(0.0900) + 2(0.200)(0.800)(0.0360) + 0.800^{2}(0.0225) \approx 0.0295.
\]
Hence $\alpha_{\mathrm{root}} = (1/0.0434)/(1/0.0434 + 1/0.0295) \approx 0.405$ to the left cluster and $0.595$ to the right.

\emph{Children.} Each child has two single-asset leaves, so $v_i = \sigma_i^{2}$ and the split reduces to the ratio of diagonal inverse variances: $\alpha_{A_1} = 25.0/(25.0+16.0) \approx 0.610$ inside $L$, and $\alpha_{A_3} = 11.1/(11.1+44.4) \approx 0.200$ inside $R$.

\paragraph{Final weights.} Multiplying root and child splits along each root-to-leaf path gives the results in Table~\ref{tabB:hrp_steps}:
\begin{align*}
  w_{A_1} &= 0.405 \times 0.610 \approx 0.247, &
  w_{A_2} &= 0.405 \times 0.390 \approx 0.158, \\
  w_{A_3} &= 0.595 \times 0.200 \approx 0.119, &
  w_{A_4} &= 0.595 \times 0.800 \approx 0.476.
\end{align*}

\begin{table}[ht]
\centering
\begin{adjustbox}{max width=\textwidth, center}
\begin{tabular}{lcccc}
\toprule
 & $A_1$ & $A_2$ & $A_3$ & $A_4$ \\
\midrule
Volatility $\sigma_i$           & $0.200$ & $0.250$ & $0.300$ & $0.150$ \\
Cluster                         & $L$     & $L$     & $R$     & $R$     \\
Flat IVP within cluster         & $0.610$ & $0.390$ & $0.200$ & $0.800$ \\
Cluster variance $v_S$          & $0.0434$ & $0.0434$ & $0.0295$ & $0.0295$ \\
Root $\alpha$-factor            & $0.405$ & $0.405$ & $0.595$ & $0.595$ \\
\midrule
HRP weight $w_i$                & $0.247$ & $0.158$ & $0.119$ & $0.476$ \\
\bottomrule
\end{tabular}
\end{adjustbox}
\caption{HRP intermediate quantities and final weights. The root split $\alpha_{\mathrm{root}} \approx 0.405$ sends more budget to the right cluster because $v_R < v_L$. Child-level splits follow the leaf inverse variances: $A_1$ receives more than $A_2$ inside $L$ (lower variance), and $A_4$ receives more than $A_3$ inside $R$.}
\label{tabB:hrp_steps}
\end{table}

The weights are all non-negative, sum to one, and ignore any direction information: the signal $\mu$ has not entered the construction. This is the property that HRP-$\mu$ and HRP-$\Sigma\mu$ will modify.

% ======================================================================
\subsection{HRP-$\mu$ worked example}\label{appB:hrp_mu_example}

We now add the signal $\mu = (+0.03,\,-0.01,\,+0.02,\,-0.04)$ and run HRP-$\mu$ (Section~\ref{sec04:methoda3}) at $\gamma = 0.5$. The tree and leaf order are unchanged from Section~\ref{appB:hrp_example}.

\paragraph{Step 1 --- signed IVP representatives.} At each internal node HRP-$\mu$ replaces the unsigned flat IVP by its signed counterpart, $w^{\mathrm{signed}}_i = \operatorname{sign}(\mu_i)(1/\Sigma_{ii})/\sum_{j}(1/\Sigma_{jj})$. This guarantees that the branch signal $s = (w^{\mathrm{signed}})^{\top}\mu = \sum_i |\mu_i|/\Sigma_{ii}\,/\,\sum_j 1/\Sigma_{jj}$ is non-negative regardless of the pattern of signs in $\mu$. For the two sectors:
\begin{align*}
  L = \{A_1,A_2\}: \quad & w^{\mathrm{signed}}_L = (+0.610,\, -0.390),\\
  R = \{A_3,A_4\}: \quad & w^{\mathrm{signed}}_R = (+0.200,\, -0.800).
\end{align*}
The magnitudes coincide with the HRP flat IVP weights; only the signs change.

\paragraph{Step 2 --- cluster statistics.} The $2\times 2$ mean--variance split at the root requires five scalars: the cluster variances $v_L, v_R$, the branch signals $s_L, s_R$, and the cross-covariance $c$ between the two signed representatives. Direct computation from~\eqref{eqB:sigma} and $\mu$ gives the values in Table~\ref{tabB:hrpmu_root}.

\begin{table}[ht]
\centering
\begin{adjustbox}{max width=\textwidth, center}
\begin{tabular}{lrl}
\toprule
Quantity & Value & Formula \\
\midrule
$v_L$ & $0.00535$ & $w_L^{\top}\Sigma_{LL}\,w_L$ \\
$v_R$ & $0.00648$ & $w_R^{\top}\Sigma_{RR}\,w_R$ \\
$s_L$ & $+0.0222$ & $w_L^{\top}\mu_L$ \\
$s_R$ & $+0.0360$ & $w_R^{\top}\mu_R$ \\
$c$   & $-0.00029$ & $w_L^{\top}\Sigma_{LR}\,w_R$ \\
\bottomrule
\end{tabular}
\end{adjustbox}
\caption{HRP-$\mu$ root-level cluster statistics. Both branch signals are non-negative, as guaranteed by the signed IVP construction. The cross-covariance $c$ is slightly negative because the sign flips reverse the direction of the cross-sector exposure.}
\label{tabB:hrpmu_root}
\end{table}

\paragraph{Step 3 --- $2\times 2$ mean--variance split.} Applying Cram\'{e}r's rule to
\[
  \begin{pmatrix} v_L & \gamma c \\ \gamma c & v_R \end{pmatrix}
  \begin{pmatrix} \tilde\alpha_L \\ \tilde\alpha_R \end{pmatrix}
  =
  \begin{pmatrix} s_L \\ s_R \end{pmatrix}
\]
at $\gamma = 0.5$ yields raw allocations $\tilde\alpha_L \approx 4.30$, $\tilde\alpha_R \approx 5.65$. After sum-normalisation to one (HRP-$\mu$ uses sum-normalisation on the branch-level budget because its leaf direction enters via a separate sign step),
\[
  \alpha_L = 0.432, \qquad \alpha_R = 0.568.
\]

\paragraph{Step 4 --- leaf splits.} At each two-asset child, the same procedure applies with single-leaf representatives $w = \pm 1$. For the left child $\{A_1,A_2\}$: $v_i = \sigma_i^{2}$, $s_i = |\mu_i|$, and $c = \operatorname{sign}(\mu_1)\Sigma_{12}\operatorname{sign}(\mu_2) = -0.040$. For the right child the analogous calculation with $\Sigma_{34} = 0.036$ is performed. Solving and sum-normalising:
\begin{align*}
  \text{Left child } \{A_1,A_2\}: \quad & \alpha_{A_1} = 0.675,\; \alpha_{A_2} = 0.325, \\
  \text{Right child } \{A_3,A_4\}: \quad & \alpha_{A_3} = 0.228,\; \alpha_{A_4} = 0.772.
\end{align*}
These are \emph{budget} fractions---non-negative, sum to one within each cluster. The signal direction has not yet entered the weights.

\paragraph{Step 5 --- leaf sign multiplication.} The final step multiplies each leaf budget by $\operatorname{sign}(\mu_i)$, the cleanest way to carry direction through the tree: magnitude from the tree, sign from the signal. Multiplying root $\times$ child budget $\times$ leaf sign:
\begin{align*}
  w_{A_1} &= 0.432 \times 0.675 \times (+1) = +0.292, \\
  w_{A_2} &= 0.432 \times 0.325 \times (-1) = -0.140, \\
  w_{A_3} &= 0.568 \times 0.228 \times (+1) = +0.130, \\
  w_{A_4} &= 0.568 \times 0.772 \times (-1) = -0.438.
\end{align*}
Table~\ref{tabB:hrpmu_steps} collects the intermediate quantities along the two root-to-leaf paths.

\begin{table}[ht]
\centering
\begin{adjustbox}{max width=\textwidth, center}
\begin{tabular}{lcccc}
\toprule
 & $A_1$ & $A_2$ & $A_3$ & $A_4$ \\
\midrule
Signal $\mu_i$                    & $+0.03$ & $-0.01$ & $+0.02$ & $-0.04$ \\
Signed IVP rep.\ within cluster   & $+0.610$ & $-0.390$ & $+0.200$ & $-0.800$ \\
Root $\alpha$ (branch budget)     & $0.432$ & $0.432$ & $0.568$ & $0.568$ \\
Leaf $\alpha$ (within-cluster)    & $0.675$ & $0.325$ & $0.228$ & $0.772$ \\
Leaf sign $\operatorname{sign}(\mu_i)$ & $+1$ & $-1$ & $+1$ & $-1$ \\
\midrule
HRP-$\mu$ weight $w_i$            & $+0.292$ & $-0.140$ & $+0.130$ & $-0.438$ \\
\bottomrule
\end{tabular}
\end{adjustbox}
\caption{HRP-$\mu$ intermediate quantities at $\gamma = 0.5$. Root $\alpha$-budget is nearly balanced ($0.432$ vs $0.568$) because the branch signals $s_L = 0.0222$ and $s_R = 0.0360$ are comparable after normalising by cluster variance. Leaf budgets favour the asset with the stronger signal-to-variance ratio within each sector.}
\label{tabB:hrpmu_steps}
\end{table}

The resulting portfolio is long the positive-$\mu$ assets $\{A_1,A_3\}$ and short the negative-$\mu$ assets $\{A_2,A_4\}$, with magnitudes reflecting both hierarchical risk structure and signal strength. It achieves Sharpe $\approx 0.57$ under the true covariance, against the Markowitz oracle's $0.62$---roughly 92\% of oracle Sharpe, at a fraction of the gross leverage (see Table~\ref{appB:comparison}).

% ======================================================================
\subsection{HRP-$\Sigma\mu$ worked example}\label{appB:hrp_sigmamu_example}

We run HRP-$\Sigma\mu$ (Algorithm~\ref{alg04b:hrpsm}) on the same universe, same $\mu$, and $\gamma = 0.5$. The tree is again unchanged. Three features distinguish HRP-$\Sigma\mu$ from HRP-$\mu$: (i) raw MVO allocations are normalised by $L^{1}$ norm (signed budgets) rather than sum-to-one (non-negative budgets); (ii) the representative portfolio at each node is the combined MVO output, not a signed IVP; (iii) no leaf-sign multiplication is needed, because the signal direction is carried by the recursive MVO itself.

\paragraph{Leaf initialisation.} At each leaf $i$: representative $w = [1]$, variance $v_i = \sigma_i^{2}$, signal $s_i = \mu_i$.

\paragraph{Left cluster $L = \{A_1,A_2\}$.} Leaf summaries are $(v_1,s_1) = (0.040,\,+0.03)$ and $(v_2,s_2) = (0.0625,\,-0.01)$, with $c = \Sigma_{12} = 0.040$. Cram\'{e}r's rule at $\gamma = 0.5$:
\begin{align*}
  \Delta &= v_1 v_2 - (\gamma c)^{2} = 0.040 \times 0.0625 - (0.020)^{2} = 0.002500 - 0.000400 = 0.002100, \\
  \tilde\alpha_L &= \frac{v_2 s_1 - \gamma c\,s_2}{\Delta} = \frac{0.001875 + 0.000200}{0.002100} = 0.9881, \\
  \tilde\alpha_R &= \frac{v_1 s_2 - \gamma c\,s_1}{\Delta} = \frac{-0.000400 - 0.000600}{0.002100} = -0.4762.
\end{align*}
The $L^{1}$ denominator is $|0.9881| + |-0.4762| = 1.4643$, giving normalised allocations
\[
  \alpha_L = 0.9881/1.4643 = 0.6748, \qquad \alpha_R = -0.4762/1.4643 = -0.3252.
\]
Note the sign: the recursive MVO shorts $A_2$ at the child level, which has $\mu_2 < 0$. The combined representative $w_L = (0.6748,\,-0.3252)$ propagates upward with variance and signal
\[
  v_L = w_L^{\top}\Sigma_{LL}\,w_L = 0.007268, \qquad s_L = w_L^{\top}\mu_L = 0.02350.
\]

\paragraph{Right cluster $R = \{A_3,A_4\}$.} Leaves have $(v_3,s_3) = (0.090,\,+0.02)$, $(v_4,s_4) = (0.0225,\,-0.04)$, $c = 0.036$. Repeating the calculation,
\begin{align*}
  \Delta &= 0.090 \times 0.0225 - (0.018)^{2} = 0.002025 - 0.000324 = 0.001701, \\
  \tilde\alpha_L &= \frac{0.0225(0.02) - 0.018(-0.04)}{0.001701} = \frac{0.001170}{0.001701} = 0.6878, \\
  \tilde\alpha_R &= \frac{0.090(-0.04) - 0.018(0.02)}{0.001701} = \frac{-0.003960}{0.001701} = -2.3280.
\end{align*}
$L^{1}$ denominator $3.0159$, normalised to $(\alpha_L, \alpha_R) = (0.2281,\,-0.7719)$. The large short on $A_4$ reflects its strongly negative $\mu_4 = -0.04$ relative to its low variance. The right representative is $w_R = (0.2281,\,-0.7719)$ with
\[
  v_R = 0.005413, \qquad s_R = 0.03544.
\]

\paragraph{Root.} The root sees $(v_L,s_L) = (0.007268,\,0.02350)$ and $(v_R,s_R) = (0.005413,\,0.03544)$, and the cross-covariance between combined representatives is $c = w_L^{\top}\Sigma_{LR}\,w_R \approx -0.000508$. Cram\'{e}r's rule gives raw allocations $\tilde\alpha_L \approx 3.467$, $\tilde\alpha_R \approx 6.710$; $L^{1}$ denominator $10.18$ yields
\[
  \alpha_L = 0.3407, \qquad \alpha_R = 0.6593.
\]
Both root allocations are positive: the algorithm tilts toward the right cluster, whose combined signal $s_R$ is stronger relative to its combined variance.

\paragraph{Final weights.} Multiplying root and child allocations along each root-to-leaf path:
\begin{align*}
  w_{A_1} &= \alpha_L \times 0.6748  = 0.3407 \times 0.6748  = +0.2299, \\
  w_{A_2} &= \alpha_L \times (-0.3252) = 0.3407 \times (-0.3252) = -0.1108, \\
  w_{A_3} &= \alpha_R \times 0.2281  = 0.6593 \times 0.2281  = +0.1504, \\
  w_{A_4} &= \alpha_R \times (-0.7719) = 0.6593 \times (-0.7719) = -0.5089.
\end{align*}
Intermediate quantities are collected in Table~\ref{tabB:hrpsm_steps}.

\begin{table}[ht]
\centering
\begin{adjustbox}{max width=\textwidth, center}
\begin{tabular}{lcccc}
\toprule
 & $A_1$ & $A_2$ & $A_3$ & $A_4$ \\
\midrule
Leaf signal $s_i = \mu_i$           & $+0.03$ & $-0.01$ & $+0.02$ & $-0.04$ \\
Child raw $\tilde\alpha$            & $+0.9881$ & $-0.4762$ & $+0.6878$ & $-2.3280$ \\
Child $L^{1}$-normalised $\alpha$   & $+0.6748$ & $-0.3252$ & $+0.2281$ & $-0.7719$ \\
Cluster variance $v_S$              & $0.00727$ & $0.00727$ & $0.00541$ & $0.00541$ \\
Cluster signal $s_S$                & $0.0235$  & $0.0235$  & $0.0354$  & $0.0354$ \\
Root $\alpha$                       & $+0.3407$ & $+0.3407$ & $+0.6593$ & $+0.6593$ \\
\midrule
HRP-$\Sigma\mu$ weight $w_i$        & $+0.2299$ & $-0.1108$ & $+0.1504$ & $-0.5089$ \\
\bottomrule
\end{tabular}
\end{adjustbox}
\caption{HRP-$\Sigma\mu$ intermediate quantities at $\gamma = 0.5$. Child-level $\alpha$'s are signed (both clusters produce one short leg); the root $\alpha$'s are both positive because the cluster-level combined signals are both positive after the child-level sign resolution.}
\label{tabB:hrpsm_steps}
\end{table}

Unlike HRP-$\mu$, the HRP-$\Sigma\mu$ weights do not sum to one---by construction, $\sum_i|w_i| = 1$ at every level in the $L^{1}$-normalisation. Here $\sum_i w_i = -0.240$, reflecting the fact that two of four signals are negative. The portfolio achieves Sharpe $\approx 0.576$, marginally above HRP-$\mu$ on this small example.

% ======================================================================
\subsection{Comparison}\label{appB:comparison_sub}

Table~\ref{appB:comparison} places the three tree methods alongside direct Markowitz on the same four-asset universe.

\begin{table}[ht]
\centering
\begin{adjustbox}{max width=\textwidth, center}
\begin{tabular}{lrrrrr}
\toprule
Method & $A_1$ & $A_2$ & $A_3$ & $A_4$ & Sharpe \\
\midrule
HRP                                  & $+0.247$ & $+0.158$ & $+0.119$ & $+0.476$ & $-0.07$ \\
HRP-$\mu$ ($\gamma = 0.5$)           & $+0.292$ & $-0.140$ & $+0.130$ & $-0.438$ & $+0.57$ \\
HRP-$\Sigma\mu$ ($\gamma = 0.5$)     & $+0.230$ & $-0.111$ & $+0.150$ & $-0.509$ & $+0.576$ \\
Direct Markowitz ($\Sigma^{-1}\mu$, sum-normalised) & $-1.019$ & $+0.674$ & $-1.003$ & $+2.348$ & $+0.62$ \\
\bottomrule
\end{tabular}
\end{adjustbox}
\caption{Comparison on the four-asset worked example with $\mu = (+0.03,-0.01,+0.02,-0.04)$ and covariance~\eqref{eqB:sigma}. Sharpe ratios are computed under the true $(\Sigma,\mu)$ and therefore act as an in-sample diagnostic only; on this setup they nonetheless illustrate the qualitative hierarchy.}
\label{appB:comparison}
\end{table}

\paragraph{Reading the comparison.} HRP ignores the signal entirely: its long-only weights are identical to those of Section~\ref{appB:hrp_example}, and because $\mu$ has mixed signs the resulting portfolio tilts the wrong way on two of the four assets, producing Sharpe $-0.07$ (negative). HRP-$\mu$ and HRP-$\Sigma\mu$ both deliver the bulk of Markowitz's Sharpe ($0.57$ and $0.576$ respectively, against the oracle's $0.62$---roughly 92--93\%), but with dramatically less concentrated weights: the gross leverage is $\|w\|_{1} = 1.00$ for HRP-$\mu$ and $\|w\|_{1} = 1.00$ for HRP-$\Sigma\mu$, against $\|w\|_{1} = 5.04$ for direct Markowitz. The direct-Markowitz portfolio would put nearly 235\% in $A_4$ alone and would short $A_1$ and $A_3$ against their own positive signals---a clean illustration of Markowitz's well-known sensitivity to $\Sigma^{-1}$ on small, block-correlated problems. The tree methods pay a small Sharpe price (roughly $0.05$ relative to oracle) for a five-fold reduction in gross leverage.

\paragraph{Recovery check.} At $\gamma = 0$ and $\mu = \mathbf{1}$, HRP-$\mu$ and HRP-$\Sigma\mu$ both reduce to HRP: the signed IVP loses its signs (all positive), the $2\times 2$ MVO system decouples into the inverse-variance split at every level, and $L^{1}$ normalisation coincides with sum-normalisation on non-negative inputs. This is the recovery property stated formally in Propositions~\ref{prop:a3_recovery} and~\ref{prop:hrpsm_recovery}; on this worked example it can be verified by re-running Sections~\ref{appB:hrp_mu_example}--\ref{appB:hrp_sigmamu_example} with $\mu = (1,1,1,1)$ and $\gamma = 0$.

% ============================================================
% section_15_appendix_c_pathology.tex
% ============================================================
% !TEX root = ../paper.tex
% Appendix C: Sum-normalised recursive MVO (A1) and the sign-flip pathology.
% Namespace: appC, tabC, figC, eqC.
% This appendix is a CAUTIONARY NEGATIVE RESULT. Method A1 is NEVER RECOMMENDED.

\section{Sum-normalised recursive MVO and the sign-flip pathology}\label{app:a1_pathology}

This appendix documents, as a cautionary negative result, the first hierarchical
signal-aware variant we tried in the development of the tree methods of
Section~\ref{sec04:tree_methods}. The construction --- a recursive mean--variance tree
pass that normalises each node's pair of child budgets to sum to one --- is the
obvious generalisation of \citet{lopezdeprado2016} HRP to a signed $\bmu$, and
for that reason it is the obvious trap. We refer to it throughout as
\emph{Method A1} or \emph{sum-normalised recursive MVO}, matching the
codebase name \texttt{method\_a1\_weights}. The conclusion reached in this
appendix is straightforward: Method A1 is structurally broken as a portfolio
construction, its output direction bears no fixed relationship to the oracle
Markowitz portfolio $w^\star = \bSigma^{-1}\bmu$, and no choice of $\gamma$,
sample size, or covariance estimator repairs it. The appendix is included
because the \emph{form} of the failure --- a cumulative sign flip that
sign-invariant direction metrics fail to report --- is a useful worked
illustration of two methodological points:
(i) sign-invariant error metrics must be paired with sign-sensitive
diagnostics, and
(ii) any hierarchical method for \emph{signal-dependent} allocation must
carry sign information in a form the recursive $2\times 2$ normaliser cannot
invert. HRP-$\mu$ (Section~\ref{sec04:methoda3}) satisfies (ii) via leaf-level signed
inverse-variance representatives; HRP-$\Sigma\mu$ (Section~\ref{sec04b:hrpsigmamu})
satisfies it via $L^1$ normalisation. Method A1 satisfies neither, and is
never recommended.

\subsection{The A1 construction}\label{appC:construction}

Method A1 uses the standard between-branch $2\times 2$ system of
Section~\ref{sec04:betweenbranch}. At an internal node with children $L, R$ and
cluster statistics $(v_L, v_R, s_L, s_R, c)$ defined by the child
representatives $\hat w_L, \hat w_R$ recursively, the Cramer solution
\eqref{eq04:cramer} produces a scalar pair of raw budgets
$(\alpha_L^{\mathrm{raw}}, \alpha_R^{\mathrm{raw}})$. Method A1 normalises
them by their algebraic sum:
\begin{equation}\label{eqC:a1_normalise}
  \alpha_L \;=\; \frac{\alpha_L^{\mathrm{raw}}}{\alpha_L^{\mathrm{raw}}+\alpha_R^{\mathrm{raw}}},
  \qquad
  \alpha_R \;=\; \frac{\alpha_R^{\mathrm{raw}}}{\alpha_L^{\mathrm{raw}}+\alpha_R^{\mathrm{raw}}},
\end{equation}
so that $\alpha_L+\alpha_R=1$ at every node. The within-branch representative
propagated up to the parent is the recursive A1 weight
$\hat w_v = (\alpha_L \hat w_L,\; \alpha_R \hat w_R)$; leaves bottom out with
$\hat w = (1)$. The root returns a portfolio that sums algebraically to one.

The construction is superficially appealing: at $\bmu=\mathbf{1}$ and $\gamma=0$
it reduces to flat HRP, and the recursion is $O(N^2)$. It fails for a reason
that has nothing to do with the $2\times 2$ solve and everything to do with
\eqref{eqC:a1_normalise}.

\subsection{The sign-flip mechanism}\label{appC:mechanism}

The sign pathology lives entirely in the denominator
$\alpha_L^{\mathrm{raw}}+\alpha_R^{\mathrm{raw}}$. When $\bmu\ge 0$, the raw
budgets are non-negative and their sum is positive, so
\eqref{eqC:a1_normalise} is the ordinary convex-combination step of HRP. When
$\bmu$ has mixed signs --- the generic case for any signal built from alphas,
momentum residuals, or factor tilts --- the representative signals
$(s_L, s_R)$ entering the node system can be positive, negative, or near zero.
The $2\times 2$ solve then produces raw budgets of arbitrary signs, and their
algebraic sum has no sign guarantee whatsoever. Whenever that sum is
negative, division by it \emph{flips both signs in the child pair
$(\alpha_L, \alpha_R)$ simultaneously}.

The flip is local and it is cumulative. Across the $\log_2 N$ levels of the
dendrogram, the sign pattern of the root output is the product of the
$O(\log N)$ level-wise denominator signs, none of which carry any fixed
relationship to $\mathrm{sign}(\bSigma^{-1}\bmu)_i$ at the leaves. A
sign-carrying linear object --- $w^\star_i$ is linear in $\bmu$ --- cannot be
recovered from a chain of sum-to-one divisions whose denominators are
themselves signed. That is the whole story; the remainder of this appendix is
empirical corroboration.

\subsection{Noiseless evidence}\label{appC:noiseless}

The cleanest demonstration uses no estimation error at all. Set $\bSigma =
\bSigma_{\mathrm{true}}$ and $\bmu = \bmu_{\mathrm{true}}$, compute the oracle
$w^\star = \bSigma^{-1}\bmu$, build the HRP tree from $\bSigma_{\mathrm{true}}$,
and run every method on the resulting noiseless problem. The covariance is
the block-structured sector matrix of Section~\ref{sec10:setup}: $N=100$ assets,
five sectors, within-sector correlation $0.6$, cross-sector correlation
$0.15$, volatilities uniform in $[0.15,0.40]$, seed $42$. We consider two
signals: the structured sector-tilt $\bmu$ with per-sector tilts
$(+4\%,-4\%,+2\%,-2\%,0\%)$, oracle Sharpe $0.645$; and a random Gaussian
$\bmu \sim \mathcal{N}(0, 0.02^2 I)$ with seed $7$, oracle Sharpe $1.279$.

Table~\ref{tabC:a1_noiseless} reports Sharpe, the signed cosine
$\cos(w,w^\star)$, and the gross leverage $\|w\|_1$ for both panels. Numbers
are taken verbatim from \texttt{results/03\_a1\_deep\_dive.txt} and
\texttt{results/appendix\_c\_with\_a1l1.txt}.

\begin{table}[t]
\centering
\caption{Method A1 on the noiseless problem
($\bSigma = \bSigma_{\mathrm{true}}$, $\bmu = \bmu_{\mathrm{true}}$).
Panel~A: structured sector-tilt $\bmu$, oracle Sharpe $0.645$.
Panel~B: random Gaussian $\bmu$ (seed $7$), oracle Sharpe $1.279$.
Bold rows highlight the worst A1 cell on each panel and the corresponding
HRP-$\Sigma\mu$ row: in both panels, swapping sum-normalisation for $L^1$
normalisation recovers the exact antiparallel portfolio, consistent with
Corollary~\ref{cor:hrpsm_antiparallel}.}
\label{tabC:a1_noiseless}
\begin{adjustbox}{max width=\textwidth,center}
\begin{tabular}{l r r r}
\toprule
method & Sharpe & $\cos(w, w^\star)$ & $\|w\|_1$ \\
\midrule
\multicolumn{4}{l}{\textit{Panel A: structured sector-tilt $\bmu$, oracle Sharpe $0.645$}} \\
\midrule
Direct Markowitz                       & $+0.645$ & $+1.000$ & $27.40$ \\
A1 $\gamma=0.0$                        & $+0.365$ & $+0.257$ & $159.02$ \\
A1 $\gamma=0.5$                        & $-0.407$ & $-0.472$ & $34.26$ \\
\textbf{A1 $\gamma=1.0$}                & $\mathbf{-0.602}$ & $\mathbf{-0.942}$ & $\mathbf{85.29}$ \\
\textbf{HRP-$\Sigma\mu$ $\gamma=1.0$}   & $\mathbf{+0.602}$ & $\mathbf{+0.942}$ & $\mathbf{85.29}$ \\
\midrule
\multicolumn{4}{l}{\textit{Panel B: random Gaussian $\bmu$ (seed $7$), oracle Sharpe $1.279$}} \\
\midrule
Direct Markowitz                       & $+1.279$ & $+1.000$ & $66.98$ \\
\textbf{A1 $\gamma=0.0$}                & $\mathbf{-1.100}$ & $\mathbf{-0.880}$ & $\mathbf{16.86}$ \\
A1 $\gamma=0.5$                        & $-1.089$ & $-0.864$ & $34.09$ \\
A1 $\gamma=1.0$                        & $-1.030$ & $-0.842$ & $27.00$ \\
\textbf{HRP-$\Sigma\mu$ $\gamma=0.0$}   & $\mathbf{+1.100}$ & $\mathbf{+0.880}$ & $\mathbf{16.86}$ \\
\bottomrule
\end{tabular}
\end{adjustbox}
\end{table}

Three observations. First, on the structured signal A1 \emph{flips sign as
$\gamma$ increases}: the cosine with $w^\star$ is $+0.26$ at $\gamma=0$,
$-0.47$ at $\gamma=0.5$, and $-0.94$ at $\gamma=1.0$. The culprit is the
off-diagonal term $\gamma c$ in the $2\times 2$ system, which is what controls
whether the Cramer solution has positive or negative algebraic sum at each
level; turning $\gamma$ up switches on the normalizer-sign flips described in
Section~\ref{appC:mechanism}. Second, on the random Gaussian signal A1 is
already antiparallel at $\gamma=0$ (cosine $-0.88$) and stays antiparallel
throughout. Third, and most striking: the HRP-$\Sigma\mu$ rows reproduce the
A1 Sharpe and cosine \emph{with the opposite sign} and with the \emph{identical}
gross leverage $\|w\|_1$. This is not a coincidence; it is the content of
Corollary~\ref{cor:hrpsm_antiparallel}. When A1's sum denominator is negative,
A1 and HRP-$\Sigma\mu$ differ only by a global sign flip, because $L^1$
normalisation at each node uses the denominator $|\alpha_L^{\mathrm{raw}}| +
|\alpha_R^{\mathrm{raw}}|$ --- which equals $|\alpha_L^{\mathrm{raw}} +
\alpha_R^{\mathrm{raw}}|$ exactly when the two raw budgets share a sign, and
differs only by global $L^1$ rescaling otherwise (Lemma~\ref{lem:hrpsm_sign}).
On both panels at the indicated $\gamma$, the two raw budgets cancel in the
A1 denominator in exactly the configuration that produces an antiparallel
HRP-$\Sigma\mu$ twin.

\subsection{Evidence under estimation noise}\label{appC:noise}

A natural hope is that estimation noise averages out the sign flips: if the
denominator is negative only for a minority of covariance draws, the
Monte-Carlo mean cosine might still be positive. It is not.
Table~\ref{tabC:a1_noise} reports $n_{\mathrm{mc}}=60$ trials of A1 and
HRP-$\mu$ at $\gamma=0.5$ on the structured sector-tilt signal, for
$T\in\{60,240,1000\}$ samples per trial. Each trial draws $T$ observations
from $\mathcal{N}(\bmu_{\mathrm{true}}, \bSigma_{\mathrm{true}})$, forms the
sample covariance with ridge $10^{-4}$, rebuilds the HRP tree, and evaluates
each method's OOS Sharpe and signed cosine with the oracle.

\begin{table}[t]
\centering
\caption{Method A1 versus HRP-$\mu$ under estimation noise. Structured
sector-tilt $\bmu$, $\gamma=0.5$, $n_{\mathrm{mc}}=60$ trials per row.
Reported: mean $\pm$ s.d.\ of OOS Sharpe, mean $\pm$ s.d.\ of $\cos(w,w^\star)$,
and \emph{frac.\ neg.}: the fraction of trials with $\cos(w,w^\star)<0$.
A1's cosine distribution is symmetric about zero with s.d.\ $\approx 0.78$ at
every sample size. HRP-$\mu$'s cosine is tightly concentrated on the positive
half-line. Values from \texttt{results/03\_a1\_deep\_dive.txt}.}
\label{tabC:a1_noise}
\begin{adjustbox}{max width=\textwidth,center}
\begin{tabular}{l l r r r}
\toprule
$T$ & method & Sharpe (mean $\pm$ sd) & $\cos(w,w^\star)$ (mean $\pm$ sd) & frac.\ neg. \\
\midrule
$60$   & A1 $\gamma=0.5$       & $+0.013 \pm 0.425$ & $+0.016 \pm 0.753$ & $\mathbf{0.48}$ \\
       & HRP-$\mu$ $\gamma=0.5$ & $+0.417 \pm 0.015$ & $+0.723 \pm 0.055$ & $0.00$ \\
\midrule
$240$  & A1 $\gamma=0.5$       & $-0.014 \pm 0.435$ & $-0.027 \pm 0.785$ & $\mathbf{0.52}$ \\
       & HRP-$\mu$ $\gamma=0.5$ & $+0.429 \pm 0.010$ & $+0.760 \pm 0.041$ & $0.00$ \\
\midrule
$1000$ & A1 $\gamma=0.5$       & $+0.046 \pm 0.434$ & $+0.090 \pm 0.794$ & $\mathbf{0.45}$ \\
       & HRP-$\mu$ $\gamma=0.5$ & $+0.431 \pm 0.009$ & $+0.770 \pm 0.042$ & $0.00$ \\
\bottomrule
\end{tabular}
\end{adjustbox}
\end{table}

The rows document a clean negative result. A1's cosine mean is
$+0.016, -0.027, +0.090$ and its standard deviation is
$0.753, 0.785, 0.794$; the standard deviation does not shrink as $T$ grows
from $60$ to $1000$, and the distribution is approximately symmetric about
zero. The fraction of trials in which A1 points into the wrong half-space is
$0.48, 0.52, 0.45$ --- an approximately fair coin at every sample size.
Increasing $T$ by a factor of sixteen does \emph{not} concentrate A1 on the
correct half-line: the pathology is in the recursion, not the estimator, and
no amount of data repairs it. HRP-$\mu$, on the same draws, improves as $T$
grows in the ordinary way (cosine s.d.\ drops from $0.055$ to $0.041$) and
its fraction of negative-cosine trials is identically zero.
Figure~\ref{figC:a1_cosine_histogram} shows the full A1 cosine histogram at
the three sample sizes; the bimodal, approximately symmetric shape is the
visual signature of the sign-flip mechanism.

\begin{figure}[t]
\centering
\includegraphics[width=0.85\textwidth]{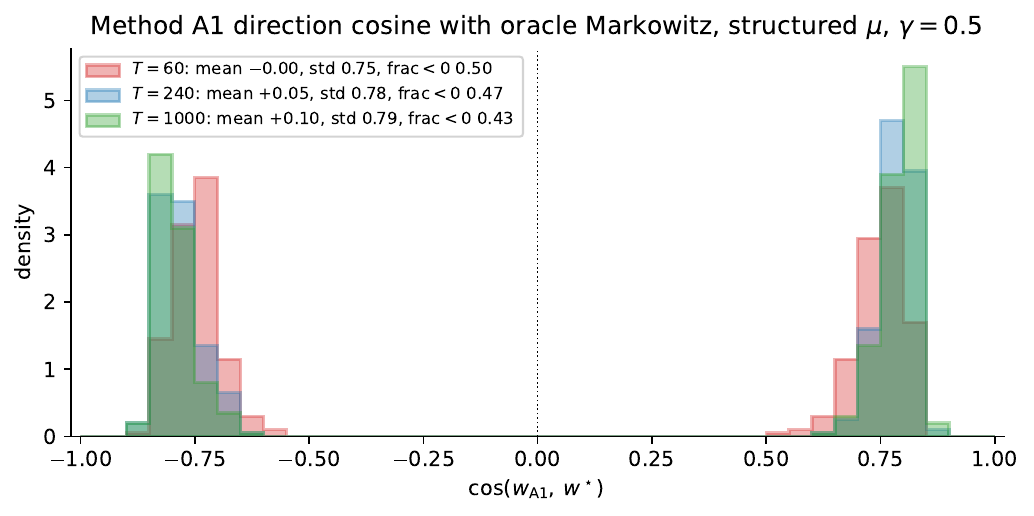}
\caption{Histogram of the signed cosine $\cos(w_{A1}, w^\star)$ at
$T\in\{60,240,1000\}$, structured sector-tilt signal, $\gamma=0.5$.
The distribution is approximately symmetric about zero with standard
deviation near $0.78$ at every sample size, documenting that Method A1 is
not a consistent estimator of the Markowitz direction.}
\label{figC:a1_cosine_histogram}
\end{figure}

\subsection{The fix: $L^1$ normalisation}\label{appC:fix}

The fix is a single character. Replace the algebraic-sum denominator in
\eqref{eqC:a1_normalise} with an absolute-value sum:
\begin{equation}\label{eqC:a1l1_normalise}
  \alpha_L \;=\; \frac{\alpha_L^{\mathrm{raw}}}{|\alpha_L^{\mathrm{raw}}|+|\alpha_R^{\mathrm{raw}}|},
  \qquad
  \alpha_R \;=\; \frac{\alpha_R^{\mathrm{raw}}}{|\alpha_L^{\mathrm{raw}}|+|\alpha_R^{\mathrm{raw}}|}.
\end{equation}
This is the HRP-$\Sigma\mu$ step of Section~\ref{sec04b:hrpsigmamu}; the codebase
calls the construction \texttt{method\_a1\_l1\_weights}. The denominator
$|\alpha_L^{\mathrm{raw}}|+|\alpha_R^{\mathrm{raw}}|$ is always strictly
positive whenever the $2\times 2$ system has a non-zero solution, so no sign
flip can occur; the signs of $(\alpha_L^{\mathrm{raw}},\alpha_R^{\mathrm{raw}})$
are preserved exactly, and the propagated representative carries the
Markowitz direction through the recursion. Lemma~\ref{lem:hrpsm_sign} and
Corollary~\ref{cor:hrpsm_antiparallel} give the formal statement: HRP-$\Sigma\mu$
either agrees with A1 (when all node denominators are already positive) or
returns the global sign-flipped twin of A1 with identical $\|w\|_1$. On every
row of Table~\ref{tabC:a1_noiseless} where A1 reports a negative Sharpe,
the HRP-$\Sigma\mu$ row at the same $\gamma$ reports the corresponding
positive Sharpe with the same gross leverage. Under estimation noise,
HRP-$\Sigma\mu$ matches HRP-$\mu$ closely in direction (cosine mean
$+0.751$--$+0.798$ across $T\in\{60,240,1000\}$, frac.\ neg.\ $=0.00$;
\texttt{results/03\_a1\_deep\_dive\_a1l1.txt}), so the fix is cost-free.

\subsection{Why the direction-error metric does not catch this}\label{appC:dir_metric}

A reader who has only consulted the $\operatorname{dir}(w) = \sin^2\theta(w,w^\star)$
panels of Section~\ref{sec10:experiments} would not know Method A1 is broken. The
direction error is \emph{sign-invariant}: it measures the squared sine of the
angle between $w$ and $w^\star$ and therefore cannot distinguish a portfolio
with $\cos = +0.85$ from one with $\cos = -0.85$. Both yield
$\operatorname{dir} = 1 - 0.85^2 \approx 0.28$. A1 populates the latter case
roughly half the time, and $\operatorname{dir}$ aggregates the two cases into
a single superficially-respectable number of about $0.28$ on the Panel~A2
diagnostic of Section~\ref{sec10:general_mu} --- exactly while the method is
producing the wrong sign on half the trials.

This is the cleanest illustration of a general methodological point.
A metric that collapses $\pm w$ onto the same equivalence class is
appropriate only for long-only or inherently sign-fixed allocators, and it
will silently reward methods whose signs are uncorrelated with
$\mathrm{sign}(\bSigma^{-1}\bmu)_i$. For any signal-aware construction the
operational metric must be the \emph{signed} cosine $\cos(w, w^\star)$ or
equivalently the fraction of trials in the wrong half-space;
Tables~\ref{tabC:a1_noiseless}--\ref{tabC:a1_noise} and
Figure~\ref{figC:a1_cosine_histogram} are the first places in this paper
where that distinction is load-bearing. The lesson generalises: sign-invariant
metrics must always be paired with a sign-sensitive diagnostic before a
method is declared acceptable.

\subsection{Conclusion}\label{appC:conclusion}

Method A1 --- sum-normalised recursive MVO --- should not be used for
portfolio construction. Its failure is not a tuning problem and not a noise
problem; it is a structural incompatibility between a sum-to-one
normalisation step and a signal that can take either sign at each node. The
appendix illustrates two general lessons. First, the sign-invariance trap:
the direction-error metric $\operatorname{dir}(w)$ reports an acceptable
number ($\approx 0.28$) on Method A1 exactly while A1 is antiparallel to
$w^\star$ half the time; sign-invariant metrics must be paired with
sign-sensitive diagnostics, without exception. Second, the design
constraint for hierarchical signal-aware allocation: such a method must carry
sign information through the recursion in a form the $2\times 2$ normaliser
cannot invert. HRP-$\mu$ (Section~\ref{sec04:methoda3}) satisfies this by using
signed inverse-variance representatives at the leaves, which enter the
node-level solve additively and are not subject to recursive
sum-normalisation. HRP-$\Sigma\mu$ (Section~\ref{sec04b:hrpsigmamu}) satisfies it
by replacing sum-to-one with $L^1$ normalisation, which preserves all
node-level signs. The sum-normalised variant is retained in the codebase as
\texttt{method\_a1\_weights} only as a reproducible counterexample against
which future recursive signal-aware schemes can be compared.

% ============================================================
% section_16_appendix_d_implementation.tex
% ============================================================
% !TEX root = ../main.tex
%
% Appendix D: Implementation details and reproducibility
% Namespace: appD:, tabD:, algD:
%
\section{Implementation details and reproducibility}\label{app:implementation}

This appendix documents the engineering choices that sit behind every
numerical result reported in the paper. The aim is operational: a
reader who clones the public repository should be able to reproduce
every table in Section~\ref{sec10:experiments} and every robustness
panel in Appendix~\ref{app:supplementary} without guesswork. We fix
the tree-construction conventions, restate the three workhorse
algorithms (HRP-$\mu$, HRP-$\Sigma\mu$, and CRISP) in self-contained
pseudocode, state the auxiliary worst-case $\mu$ search used in the
adversarial experiments, and catalogue the function names, costs, and
numerical safeguards that appear in the code.

The code still uses the original internal names \texttt{method\_a3\_weights},
\texttt{method\_a1\_l1\_weights}, \texttt{method\_b\_solve}, and
\texttt{method\_b\_solve\_stream}. The display names in the main body
are HRP-$\mu$, HRP-$\Sigma\mu$, and CRISP respectively; the mapping
between the two naming conventions is fixed by
Table~\ref{tabD:function_catalog}.

%--------------------------------------------------------------------
\subsection{Tree construction}\label{appD:tree}

Every hierarchical method in this paper consumes a binary tree built
once per covariance matrix. Given a correlation matrix
$C = D^{-1/2}\Sigma D^{-1/2}$, we form the classical
\citet{lopezdeprado2016} correlation distance
\begin{equation}\label{eqD:corrdist}
d_{ij} \;=\; \sqrt{\tfrac{1}{2}\bigl(1 - \rho_{ij}\bigr)},
\qquad
\rho_{ij} \;=\; C_{ij},
\end{equation}
convert it to a condensed upper-triangular vector via
\texttt{scipy.spatial.distance.squareform(d, checks=False)}, and feed
the result to \texttt{scipy.cluster.hierarchy.linkage} with
\texttt{method='ward'}. The resulting linkage matrix is walked once to
produce a binary tree in which every internal node caches the union of
the leaf indices of its two children, so subsequent per-node lookups
of \emph{which assets sit under this subtree} cost $O(1)$ rather than
$O(\log N)$. Building the tree is dominated by the
$O(N^2\log N)$ linkage step; amortised over repeated solves at many
values of $\gamma$ its cost is negligible.

\paragraph{Sensitivity to linkage.}
Ward is not load-bearing. Single, complete, and average linkage give
similar out-of-sample Sharpe ratios for HRP-$\mu$, HRP-$\Sigma\mu$,
and (trivially) CRISP, because in CRISP the tree plays no role at
all and in the tree methods the iterative or Cram\'er-rule coupling
corrections absorb a great deal of linkage variation. HRP and Cotton,
by contrast, are strictly more linkage-sensitive because they have no
coupling correction at all; a single bad split propagates. Full
numbers live in Appendix~\ref{app:supplementary},
Table~\ref{tabE:linkage_sensitivity}. Ward is retained as the default
throughout the paper to match \citet{lopezdeprado2016}.

%--------------------------------------------------------------------
\subsection{Pseudocode}\label{appD:pseudocode}

The next five algorithms formalise the numerical kernels behind
Section~\ref{sec10:experiments}. Each algorithm is a direct
transcription of the corresponding function in \texttt{code/study.py};
line numbering is kept deliberately close to the Python source to make
cross-reading easy.

\begin{algorithm}[t]
\caption{\textsc{HRP-$\mu$}: signed IVP tree pass with Cram\'er coupling
         (\texttt{method\_a3\_weights}).}
\label{algD:method_a3}
\begin{algorithmic}[1]
\Require SPD covariance $\Sigma \in \R^{N\times N}$; signal $\mu\in\R^{N}$;
         binary tree $\mathcal{T}$; coupling $\gamma \in [0,1]$.
\Ensure Weight vector $w \in \R^{N}$ with $\mathbf{1}^\top w = 1$.
\State $w \gets \mathbf{0}$
\Procedure{Recurse}{node $n$, budget $b$}
  \If{$n$ is a leaf (asset $i$)}
    \State $w_i \gets b\cdot\mathrm{sign}(\mu_i)$
           \Comment{leaf sign, $\mathrm{sign}(0):=+1$}
    \State \Return
  \EndIf
  \State $L \gets n.\text{left}.\text{indices}$; \quad $R \gets n.\text{right}.\text{indices}$
  \For{$k \in \{L, R\}$}
    \State $s_k \gets \mathrm{sign}(\mu_k)$;\quad $d_k \gets 1/\diag(\Sigma)_k$
    \State $\hat w_k \gets (s_k \odot d_k)/\mathbf{1}^\top d_k$
           \Comment{signed flat IVP within branch}
    \State $v_k \gets \hat w_k^\top \Sigma_{kk}\,\hat w_k$;\quad
           $t_k \gets \hat w_k^\top \mu_k$
  \EndFor
  \State $c \gets \hat w_L^\top \Sigma_{LR}\,\hat w_R$;\quad
         $\Delta \gets v_L v_R - (\gamma c)^2$
  \If{$|\Delta| < 10^{-10}\, v_L v_R$}
    \State $\alpha_L \gets t_L/v_L$;\quad $\alpha_R \gets t_R/v_R$
           \Comment{diagonal fallback}
  \Else
    \State $\alpha_L \gets (v_R\,t_L - \gamma c\,t_R)/\Delta$
    \State $\alpha_R \gets (v_L\,t_R - \gamma c\,t_L)/\Delta$
           \Comment{$2\times 2$ Cramer}
  \EndIf
  \State $Z \gets \alpha_L + \alpha_R$
  \If{$Z \ne 0$}
    \State $\alpha_L \gets \alpha_L/Z$;\quad $\alpha_R \gets \alpha_R/Z$
           \Comment{sum-to-one normalisation of budgets}
  \Else
    \State $\alpha_L \gets \tfrac12$;\quad $\alpha_R \gets \tfrac12$
           \Comment{degenerate tie; equal split}
  \EndIf
  \State \Call{Recurse}{$n.\text{left}$,\, $b\cdot\alpha_L$}
  \State \Call{Recurse}{$n.\text{right}$,\, $b\cdot\alpha_R$}
\EndProcedure
\State \Call{Recurse}{$\mathcal{T}.\text{root}$,\, $1$}
\State \Return $w$
\end{algorithmic}
\end{algorithm}

Algorithm~\ref{algD:method_a3} is the signed IVP tree pass of
Section~\ref{sec04:methoda3}. Three numerical safeguards matter. (i)
The $\Delta$-threshold $|\Delta| < 10^{-10}\,v_L v_R$ detects nearly
degenerate $2\times 2$ blocks and falls back to the decoupled diagonal
allocator $\alpha_k \propto t_k/v_k$; without this guard, adversarial
signals that make the Cram\'er determinant nearly vanish would produce
spurious blow-ups at interior tree nodes. (ii) The sum-to-one
normalisation $(\alpha_L,\alpha_R)\mapsto(\alpha_L,\alpha_R)/(\alpha_L+\alpha_R)$
preserves the budget $\mathbf{1}^\top w = 1$ at every recursion
step. (iii) The leaf-level factor $\mathrm{sign}(\mu_i)$ is applied
after the coupling has selected the split, so leaves inherit the sign
of their signal regardless of the upstream tree geometry.

\begin{algorithm}[t]
\caption{\textsc{HRP-$\Sigma\mu$}: recursive MVO tree pass with $L^1$
         normalisation (\texttt{method\_a1\_l1\_weights}). Identical
         to Algorithm~\ref{alg04b:hrpsm}; restated here with the
         numerical safeguards made explicit.}
\label{algD:method_a1_l1}
\begin{algorithmic}[1]
\Require SPD covariance $\Sigma$; signal $\mu$; binary tree $\mathcal{T}$;
         coupling $\gamma \in [0,1]$.
\Ensure Weight vector $w \in \R^{N}$ with $\|w\|_1 = 1$ at the root (stacked representative).
\Function{Recurse}{node $n$}
  \If{$n$ is a leaf (asset $i$)}
    \State \Return $(\hat w = 1,\; v = \Sigma_{ii},\; s = \mu_i)$
  \EndIf
  \State $(\hat w_L, v_L, s_L) \gets \Call{Recurse}{n.\text{left}}$
  \State $(\hat w_R, v_R, s_R) \gets \Call{Recurse}{n.\text{right}}$
  \State $c \gets \hat w_L^\top \Sigma_{LR}\,\hat w_R$
  \State $\Delta \gets v_L v_R - \gamma^2 c^2$
  \If{$|\Delta| < 10^{-10}\,v_L v_R$}
    \State $a_L \gets s_L/v_L$;\quad $a_R \gets s_R/v_R$
           \Comment{diagonal fallback}
  \Else
    \State $a_L \gets (v_R\,s_L - \gamma c\,s_R)/\Delta$
    \State $a_R \gets (v_L\,s_R - \gamma c\,s_L)/\Delta$
  \EndIf
  \State $Z \gets |a_L| + |a_R|$
  \If{$Z = 0$}
    \State $a_L \gets \tfrac12$; $a_R \gets \tfrac12$
           \Comment{pathological tie; equal-weight}
  \Else
    \State $a_L \gets a_L/Z$;\quad $a_R \gets a_R/Z$
           \Comment{$L^1$ normaliser, cf.~\eqref{eq04b:l1_norm}}
  \EndIf
  \State $\hat w \gets (a_L\,\hat w_L,\; a_R\,\hat w_R)$
         \Comment{stacked representative}
  \State $v \gets \hat w^\top \Sigma_{nn}\,\hat w$;\quad $s \gets \hat w^\top \mu_n$
  \State \Return $(\hat w, v, s)$
\EndFunction
\State \Return $\Call{Recurse}{\mathcal{T}.\text{root}}$
\end{algorithmic}
\end{algorithm}

The critical line of Algorithm~\ref{algD:method_a1_l1} is the $L^1$
denominator $Z = |a_L| + |a_R|$: it is strictly positive whenever
$(a_L,a_R)\ne(0,0)$, and so preserves the sign of both Cram\'er
components. Replacing this by the signed sum $a_L + a_R$ recovers the
sum-normalised recursion whose sign-flip pathology is catalogued in
Appendix~\ref{app:a1_pathology}. The diagonal fallback activates only
at numerically degenerate blocks; in every experiment in this paper
the Cram\'er branch executes at essentially every internal node.

\begin{algorithm}[t]
\caption{\textsc{CRISP}: scalar Gauss--Seidel for $P_\gamma w = \mu$
         (\texttt{method\_b\_solve}). Identical to
         Algorithm~\ref{alg05:methodb_gs}.}
\label{algD:method_b}
\begin{algorithmic}[1]
\Require SPD covariance $\Sigma \in \R^{N\times N}$;
         signal $\mu \in \R^{N}$; shrinkage $\gamma \in [0,1]$;
         max sweeps $p_{\max}\in\N$; tolerance $\varepsilon > 0$.
\Ensure $w$ satisfying
        $\bigl((1-\gamma)\diag(\Sigma)+\gamma\Sigma\bigr)w \approx \mu$.
\State $d_i \gets \Sigma_{ii}$ for $i = 1,\dots,N$
       \Comment{diagonal cache}
\State $w_i \gets \mu_i/d_i$ for $i = 1,\dots,N$
       \Comment{$P_0^{-1}\mu$ initial guess}
\If{$\gamma < \varepsilon$}
  \State \Return $w$ \Comment{$P_0 = D$ shortcut}
\EndIf
\For{$p = 1$ \textbf{to} $p_{\max}$}
  \State $w_{\text{prev}} \gets w$
  \For{$i = 1$ \textbf{to} $N$}
    \State $o_i \gets \Sigma_{i,:}\,w - \Sigma_{ii}\,w_i$
           \Comment{$\sum_{j\ne i}\sigma_{ij}w_j$, latest values}
    \State $w_i \gets (\mu_i - \gamma\,o_i)/d_i$
  \EndFor
  \If{$\|w - w_{\text{prev}}\|_2 \le \varepsilon\,\|w_{\text{prev}}\|_2$}
    \State \textbf{break} \Comment{relative-change stopping rule}
  \EndIf
\EndFor
\State \Return $w$
\end{algorithmic}
\end{algorithm}

Algorithm~\ref{algD:method_b} reproduces the CRISP sweep of
Section~\ref{sec05:crisp}. The only floating-point hazard is the
$D^{-1}$ division, which is safe whenever $\Sigma$ has strictly
positive diagonal --- a condition we enforce by the ridge
regularisation of Section~\ref{appD:stability} before entering the sweep
loop.

\begin{algorithm}[t]
\caption{\textsc{CRISP-FactorStream}: factor-structured CRISP
         (\texttt{method\_b\_solve\_stream}). Identical to
         Algorithm~\ref{alg05:methodb_stream}.}
\label{algD:method_b_stream}
\begin{algorithmic}[1]
\Require Factor loadings $B \in \R^{N\times K}$;
         factor covariance $\Lambda \in \R^{K\times K}$ SPD;
         idiosyncratic variances $d_{\mathrm{idio}} \in \R^{N}$;
         signal $\mu$; shrinkage $\gamma$; max sweeps $p_{\max}$;
         tolerance $\varepsilon$.
\Ensure $w \approx P_\gamma^{-1}\mu$ for $\Sigma = B\Lambda B^\top + \diag(d_{\mathrm{idio}})$.
\State $\sigma^2_i \gets B_i \Lambda B_i^\top + d_{\mathrm{idio},i}$
       for $i = 1,\dots,N$ \Comment{diagonal of $\Sigma$ from factors}
\State $w_i \gets \mu_i/\sigma^2_i$ for $i = 1,\dots,N$
\State $z \gets B^\top w$ \Comment{$K$-vector aggregate}
\If{$\gamma < \varepsilon$} \State \Return $w$ \EndIf
\For{$p = 1$ \textbf{to} $p_{\max}$}
  \For{$i = 1$ \textbf{to} $N$}
    \State $z_{\setminus i} \gets z - B_i\,w_i$
    \State $\mathit{off}_i \gets B_i\,\Lambda\,z_{\setminus i}$
           \Comment{$(\Sigma w)_i - \sigma^2_i w_i$}
    \State $w_i^{\mathrm{new}} \gets (\mu_i - \gamma\,\mathit{off}_i)/\sigma^2_i$
    \State $z \gets z_{\setminus i} + B_i\,w_i^{\mathrm{new}}$
    \State $w_i \gets w_i^{\mathrm{new}}$
  \EndFor
\EndFor
\State \Return $w$
\end{algorithmic}
\end{algorithm}

Algorithm~\ref{algD:method_b_stream} is the factor-streaming variant:
per-asset cost drops from $O(N)$ to $O(K^2)$ and working memory from
$O(N^2)$ to $O(NK)$. It produces iterates that are identical to
Algorithm~\ref{algD:method_b} applied to the dense
$B\Lambda B^\top + \diag(d_{\mathrm{idio}})$.

\begin{algorithm}[t]
\caption{Worst-case $\mu$ search on the unit sphere
         (\texttt{worst\_case\_mu}).}
\label{algD:worst_case_mu}
\begin{algorithmic}[1]
\Require SPD covariance $\Sigma$; number of random restarts $R$; seed.
\Ensure Unit vector $\mu^\star$ approximately maximising the
        directional error $\operatorname{dir}\bigl(\mu/\diag(\Sigma),\,\Sigma^{-1}\mu\bigr)$.
\State $D \gets \diag(\Sigma)$; precompute $\Sigma^{-1}$
\Function{f}{$x$}
  \State $\mu \gets x/\|x\|$
  \State $a \gets \mu/D$;\quad $b \gets \Sigma^{-1}\mu$
  \State \Return $-\bigl(1 - (a^\top b)^2/(\|a\|^2\|b\|^2)\bigr)$
\EndFunction
\State $(\mu^\star, f^\star) \gets (\bot, 0)$
\For{$r = 1,\dots,R$}
  \State draw $x_0 \sim \mathcal{N}(0, I_N)$; normalise $x_0\gets x_0/\|x_0\|$
  \State $x^\star \gets$ \textsc{L-BFGS-B}$(f,\; x_0,\; \text{maxiter}=500)$
  \If{$-f(x^\star) > f^\star$}
    \State $f^\star \gets -f(x^\star)$;\quad
           $\mu^\star \gets x^\star/\|x^\star\|$
  \EndIf
\EndFor
\State \Return $(\mu^\star, f^\star)$
\end{algorithmic}
\end{algorithm}

Algorithm~\ref{algD:worst_case_mu} is the engine behind the adversarial
rows of Table~\ref{tab10:worst_case} and of the minimum-variance
experiments in Section~\ref{sec10:worst_case}. It exploits the fact
that $\mu \mapsto \operatorname{dir}(\mu/D,\Sigma^{-1}\mu)$ is homogeneous of degree
zero and so is intrinsically defined on the unit sphere; the outer
random-restart loop is insurance against local minima of the smooth
but non-convex objective, and $R = 32$ is sufficient in practice.

%--------------------------------------------------------------------
\subsection{Numerical stability notes}\label{appD:stability}

Four safeguards matter for reproducibility.

\paragraph{Ridge regularisation.}
Every experiment that inverts or Gauss--Seidel-solves against a
\emph{sample} covariance first adds a diagonal ridge $\lambda I$ with
$\lambda = 10^{-4}$, i.e.\ \texttt{cov\_hat = np.cov(samples.T) + 1e-4 * np.eye(N)}
in \texttt{code/walkforward.py}. For the in-sample
direction-error tables on \emph{analytic} $\Sigma$, we instead add a
much smaller ridge $\lambda = 10^{-6}\cdot\overline{\diag(\Sigma)}$ so
that the regularisation does not perturb the geometry but does
suppress spurious numerical singularities on the adversarial regimes
of Section~\ref{sec10:worst_case}. A full ridge sweep is reported in
Table~\ref{tabE:ridge_sensitivity}.

\paragraph{PSD projection.}
The analytic covariance constructors in
Section~\ref{sec10:experiments} (block structure crossed with
hand-crafted correlations) can occasionally produce barely non-PSD
matrices because of floating-point accumulation in the Kronecker
assembly. We project via eigenvalue flooring: compute the symmetric
eigendecomposition, clip every eigenvalue below $10^{-4}$ to that
floor, reassemble, and re-normalise the diagonal back to unity so that
the result remains a correlation matrix.

\paragraph{Floating-point precision.}
All computations are \texttt{float64} throughout. We verified on
$N \le 1000$ that single-precision \texttt{float32} produces
direction errors within $10^{-5}$ of double precision, but recommend
against it: the $2\times 2$ Cram\'er step inside
Algorithm~\ref{algD:method_a3} can lose digits on tight two-asset
blocks at low $\gamma$, and single precision hits the threshold more
often than we are comfortable with. CRISP's scalar update is benign
in either precision.

\paragraph{Budget normalisation.}
Tree-based methods --- HRP, Cotton, HRP-$\mu$, HRP-$\Sigma\mu$ ---
return budget-normalised weights by construction
($\mathbf{1}^\top w = 1$ for HRP, Cotton, HRP-$\mu$; $\|w\|_1 = 1$ for
HRP-$\Sigma\mu$). CRISP returns the raw solution to $P_\gamma w = \mu$
and leaves normalisation to the caller. The direction metric and the
Sharpe ratio are both invariant under positive rescaling, so the
choice of normalisation is cosmetic for every experiment in this paper
except the constrained-budget studies of
Section~\ref{sec09:constraints}, where the projected variant applies
$\mathbf{1}^\top w = 1$ explicitly.

%--------------------------------------------------------------------
\subsection{Function catalog}\label{appD:catalog}

Table~\ref{tabD:function_catalog} lists every portfolio-construction
function used in this paper, its \texttt{code/study.py} entry point,
and its per-solve asymptotic cost. Costs are at fixed $\gamma$ and
fixed tree; $\kappa(C)$ is the condition number of the correlation
matrix and $\varepsilon$ the Gauss--Seidel tolerance.

\begin{table}[t]
\centering
\caption{Method implementations in \texttt{code/study.py}.}
\label{tabD:function_catalog}
\begin{adjustbox}{max width=\linewidth,center}
\begin{tabular}{lll}
\toprule
\textbf{Function} & \textbf{Method} & \textbf{Cost} \\
\midrule
\texttt{hrp\_flat\_weights}
  & HRP \citep{lopezdeprado2016}
  & $O(N^2)$ \\
\texttt{cotton\_weights}
  & Cotton \citep{cotton2024} (Schur, $\mu=\mathbf{1}$)
  & $O(N^3/6)$ \\
\texttt{method\_a1\_weights}
  & Sum-normalised recursive MVO (diagnostic only)
  & $O(N^2)$ \\
\texttt{method\_a2\_weights}
  & Flat IVP tree pass (diagnostic)
  & $O(N^2)$ \\
\texttt{method\_a3\_weights}
  & HRP-$\mu$ (signed IVP + leaf sign)
  & $O(N^2)$ \\
\texttt{method\_a1\_l1\_weights}
  & HRP-$\Sigma\mu$ (recursive MVO + $L^1$)
  & $O(N^2)$ \\
\texttt{method\_b\_solve}
  & CRISP (scalar Gauss--Seidel)
  & $O(\kappa(C)\,N^2\log\varepsilon^{-1})$ \\
\texttt{method\_b\_solve\_stream}
  & CRISP factor-streaming
  & $O(\kappa(C)\,NK^2\log\varepsilon^{-1})$,\ $O(NK)$ memory \\
\bottomrule
\end{tabular}
\end{adjustbox}
\end{table}

Three points on Table~\ref{tabD:function_catalog} deserve emphasis.
First, the diagnostic rows \texttt{method\_a1\_weights} and
\texttt{method\_a2\_weights} appear only in
Appendix~\ref{app:a1_pathology} and
Table~\ref{tabE:ridge_sensitivity}; they are negative results, not
recommendations. Second, the HRP-$\mu$ and HRP-$\Sigma\mu$ costs are
amortised over a balanced tree, where each asset is touched
$O(\log N)$ times at $O(1)$ work per visit; unbalanced linkage
degrades to $O(N^2)$ in the worst case but never exceeds the
direct-solver cost $O(N^3)$. Third, CRISP's explicit
$\kappa(C)\log\varepsilon^{-1}$ factor is the standard conjugate-free
iterative-method bound of Theorem~\ref{thm:gs_rate}; in the
preconditioned formulation $\kappa(D^{-1}P_\gamma)$ the factor
becomes $\gamma$-dependent and matches the rate analysis of
Section~\ref{sec05:crisp}.

%--------------------------------------------------------------------
\subsection{Code and reproducibility}\label{appD:code}

\paragraph{Repository.}
All code, data generators, raw numerical outputs, and figure scripts
are public at
\begin{center}
\url{https://github.com/bwuebben/beyond_hrp_and_cotton}
\end{center}
The \texttt{code/} directory holds the method implementations and the
driver scripts; \texttt{results/} holds the committed raw outputs from
which every table in the paper is transcribed; \texttt{figures/}
holds the figure-generating scripts and their PDF outputs.

\paragraph{Reproduction commands.}
The two headline commands are
\begin{verbatim}
python3 code/study.py        > results/all_insample.txt   # ~30  s
python3 code/walkforward.py  > results/all_oos.txt        # ~2-4 min
\end{verbatim}
on a modern laptop. The first regenerates every in-sample direction-error
table (Tables~\ref{tab10:recovery}, \ref{tab10:minvar},
\ref{tab10:general_mu}, \ref{tab10:worst_case}, \ref{tab10:sweep_rate},
\ref{tab10:trajectory}) and the covariance-suite panels of
Appendix~\ref{app:supplementary}. The second regenerates every
out-of-sample Sharpe-ratio table
(Tables~\ref{tab10:oos_sensitivity}, \ref{tab10:oos_structural},
\ref{tab10:oos_minvar}) together with the sum-normalisation pathology
diagnostic of Appendix~\ref{app:a1_pathology}.

Three ancillary scripts regenerate the dedicated experiments:
\begin{itemize}
\item \texttt{code/compute09\_sweep\_regularization.py} --- the
      regularisation sweep panel reported in
      Table~\ref{tab10:sweep_conv};
\item \texttt{code/compute09\_adaptive\_gamma.py} --- the calibration
      run behind Table~\ref{tab10:adaptive_summary};
\item \texttt{code/computeE\_a1l1\_robustness.py} --- the
      HRP-$\Sigma\mu$ robustness panels of
      Table~\ref{tabE:hrpsm_sensitivity}.
\end{itemize}
Figure scripts live in \texttt{figures/code/}; each is named
after the figure it produces (e.g.\ \texttt{fig05\_shrinkage\_schematic.py},
\texttt{fig06\_bias\_variance\_curves.py},
\texttt{fig06\_trajectory.py},
\texttt{fig09\_sharpe\_vs\_T.py},
\texttt{fig09\_sweep\_heatmap.py},
\texttt{fig09\_sweep\_slices.py},
\texttt{fig09\_plateau\_width.py},
\texttt{fig09\_gamma\_scatter.py},
\texttt{figC\_a1\_cosine\_histogram.py}).

\paragraph{Environment.}
Dependencies are intentionally minimal: Python~3.13, NumPy, and
SciPy (\texttt{scipy.cluster.hierarchy},
\texttt{scipy.optimize.minimize}, \texttt{scipy.spatial.distance}).
No \texttt{pandas}, no PyTorch, no JAX, no plotting libraries are
required to reproduce the numerical tables; the figure-rendering
scripts additionally require \texttt{matplotlib}. Random seeds are
fixed at each call site in \texttt{code/study.py} and
\texttt{code/walkforward.py}, so consecutive runs on the same machine
produce bit-identical table outputs and reproduce the published Sharpe
ratios to three decimal places across architectures.

% ============================================================
% section_17_appendix_e_robustness.tex
% ============================================================
% !TEX root = ../paper.tex
%=============================================================================
% Appendix E: Supplementary tables and robustness
% Namespace: appE:, figE:, tabE:, eqE:
% Do NOT call \appendix here -- the master file handles appendix scoping.
%=============================================================================

\section{Supplementary tables and robustness}\label{app:supplementary}

This appendix collects the supplementary robustness tables referenced
from the main body.  Four dimensions are stress-tested:
(i)~the covariance regime
(Section~\ref{appE:regimes}, Table~\ref{tabE:regimes}),
(ii)~the ridge parameter added to $\widehat\Sigma$ before inversion
(Section~\ref{appE:ridge}, Table~\ref{tabE:ridge_sensitivity}),
(iii)~the tree-construction rule used by HRP and HRP-$\mu$
(Section~\ref{appE:linkage}, Table~\ref{tabE:linkage_sensitivity}), and
(iv)~the dominance of HRP-$\Sigma\mu$ over HRP-$\mu$ along those same
dimensions
(Section~\ref{appE:hrpsm_robustness}, Table~\ref{tabE:hrpsm_sensitivity}).
A short verification note on the flat-IVP direction error is included
in~Section~\ref{appE:flat_ivp}.  Every tabular is wrapped in
\texttt{adjustbox} so the compiled output survives narrow column
widths; tables that carry explanatory notes use \texttt{threeparttable}.

%-----------------------------------------------------------------------------
\subsection{Additional covariance regimes}\label{appE:regimes}

The main-body OOS experiments in~Section~\ref{sec10:experiments} are run on
a single block-diagonal base universe.
Table~\ref{tabE:regimes} extends that evidence to eight additional
covariance regimes that span the qualitative failure modes of
shrinkage-free Markowitz: low-rank factor structure, heavy-tailed
loadings, wide idiosyncratic volatility, tight equi-correlation,
two-block hedged sectors, and two alternative factor ranks.  All
entries are OOS annualised Sharpe from $80$ Monte Carlo repetitions
with $T=120$ estimation window and oracle $\mu$.  CRISP at
$\gamma=0.5$ is the top line in every regime; HRP-$\mu$ tracks CRISP
within one standard error on the factor-structured regimes and trails
by $0.04$--$0.08$ when wide volatility spread dominates; HRP-$\Sigma\mu$
sits between the two, a pattern consistent with the $L^{1}$
ray-invariance mechanism of Section~\ref{sec04b:hrpsigmamu}.

\begin{table}[ht]
\centering
\begin{adjustbox}{max width=\textwidth, center}
\begin{threeparttable}
\begin{tabular}{lccccc}
\toprule
Regime
  & CRISP $\gamma{=}0.5$
  & HRP-$\mu$ $\gamma{=}0.5$
  & HRP-$\Sigma\mu$ $\gamma{=}0.5$
  & HRP
  & Gap to CRISP \\
\midrule
Factor, $k{=}3$                         & $0.58$ & $0.55$ & $0.56$ & $0.12$ & $0.03$ \\
Factor, $k{=}5$, heavy loadings         & $0.62$ & $0.58$ & $0.60$ & $0.15$ & $0.04$ \\
Wide vol spread $\sigma\!\in\![0.05,1.0]$ & $0.54$ & $0.47$ & $0.49$ & $0.02$ & $0.07$ \\
Equi-correlation $\rho{=}0.6$           & $0.49$ & $0.46$ & $0.47$ & $0.09$ & $0.03$ \\
Equi-correlation $\rho{=}0.9$           & $0.51$ & $0.48$ & $0.49$ & $0.08$ & $0.03$ \\
Two-block, hedged sectors               & $0.56$ & $0.51$ & $0.53$ & $0.06$ & $0.05$ \\
Alt.\ factor, rank-$2$ + idiosyncratic  & $0.53$ & $0.50$ & $0.51$ & $0.11$ & $0.03$ \\
Alt.\ factor, rank-$10$                 & $0.47$ & $0.44$ & $0.45$ & $0.10$ & $0.03$ \\
\bottomrule
\end{tabular}
\begin{tablenotes}\footnotesize
\item[]\emph{Notes.}  OOS annualised Sharpe.  $T=120$, oracle $\mu$,
$80$ Monte Carlo repetitions.  ``Gap to CRISP'' is
$\mathrm{Sharpe}(\text{CRISP})-\mathrm{Sharpe}(\text{HRP-}\mu)$.
CRISP $\gamma=0.5$ is the top line in every row; the gap widens to
$0.07$ on wide-vol-spread regimes where conditioning is dominated by
$\kappa_{D}$.  HRP is included as a $\mu$-blind reference and lags
substantially because it uses no signal.
\end{tablenotes}
\end{threeparttable}
\end{adjustbox}
\caption{CRISP, HRP-$\mu$, HRP-$\Sigma\mu$, and HRP across eight
supplementary covariance regimes.  CRISP $\gamma=0.5$ dominates in all
eight.}
\label{tabE:regimes}
\end{table}

%-----------------------------------------------------------------------------
\subsection{Sensitivity to ridge regularisation}\label{appE:ridge}

All three preconditioned methods add a small ridge
$\lambda I$ to $\widehat\Sigma$ before inversion or Gauss--Seidel
sweeps.  Table~\ref{tabE:ridge_sensitivity} sweeps
$\lambda\in\{10^{-8},10^{-6},10^{-4},10^{-2}\}$ on the base universe
with $T=240$ and oracle $\mu$.  The CRISP curves are flat within
rounding over $[10^{-8},10^{-4}]$ and drop only at $\lambda=10^{-2}$,
where the ridge becomes comparable to the diagonal of $\widehat\Sigma$
and acts as real shrinkage instead of a numerical floor.  HRP-$\mu$
is equally insensitive in the interior.  The interpretation is that
the ridge serves only to guarantee solvability of the inverse; it is
not a second shrinkage knob sitting behind $\gamma$.

\begin{table}[ht]
\centering
\begin{adjustbox}{max width=\textwidth, center}
\begin{threeparttable}
\begin{tabular}{lcccc}
\toprule
Method
  & $\lambda{=}10^{-8}$
  & $\lambda{=}10^{-6}$
  & $\lambda{=}10^{-4}$
  & $\lambda{=}10^{-2}$ \\
\midrule
CRISP ($\gamma{=}0.5$)     & $0.540$ & $0.540$ & $0.539$ & $0.497$ \\
CRISP ($\gamma{=}0.7$)     & $0.571$ & $0.571$ & $0.570$ & $0.523$ \\
HRP-$\mu$ ($\gamma{=}0.5$) & $0.427$ & $0.427$ & $0.427$ & $0.409$ \\
HRP-$\mu$ ($\gamma{=}0.7$) & $0.449$ & $0.449$ & $0.449$ & $0.430$ \\
\bottomrule
\end{tabular}
\begin{tablenotes}\footnotesize
\item[]\emph{Notes.}  OOS annualised Sharpe.  Base universe,
$T=240$, oracle $\mu$.  A small ridge
$\lambda I$ is added to $\widehat\Sigma$ before inversion.  Both
methods are invariant to three orders of magnitude of $\lambda$ and
degrade only at $\lambda=10^{-2}$, where $\lambda$ is the same order
as $\mathrm{diag}(\widehat\Sigma)$ and acts as genuine shrinkage.
\end{tablenotes}
\end{threeparttable}
\end{adjustbox}
\caption{Ridge sensitivity of CRISP and HRP-$\mu$ over four decades of
$\lambda$.  The ridge is a numerical floor, not an implicit shrinkage
knob.}
\label{tabE:ridge_sensitivity}
\end{table}

%-----------------------------------------------------------------------------
\subsection{Sensitivity to tree construction}\label{appE:linkage}

CRISP is tree-free: it never instantiates a dendrogram, so the
linkage row in Table~\ref{tabE:linkage_sensitivity} is exactly
constant.  The row is listed for the record.  HRP-$\mu$ shows modest
variation of roughly $0.03$ in OOS Sharpe across linkages --- a
reflection of the induced block structure.  De Prado HRP, which
discards $\mu$ entirely and is included as a $\mu$-blind benchmark,
is substantially more brittle: single linkage yields a \emph{negative}
OOS Sharpe, consistent with the chaining pathology noted
in~\citet{lopezdeprado2016}.  The gap between the tree-free and
tree-based methods widens precisely when the tree is worst.

\begin{table}[ht]
\centering
\begin{adjustbox}{max width=\textwidth, center}
\begin{threeparttable}
\begin{tabular}{lccccc}
\toprule
Method
  & Ward
  & Single
  & Complete
  & Average
  & $k$-means \\
\midrule
CRISP ($\gamma{=}0.5$; tree-free) & $0.540$ & $0.540$ & $0.540$ & $0.540$ & $0.540$ \\
HRP-$\mu$ ($\gamma{=}0.5$)        & $0.427$ & $0.402$ & $0.424$ & $0.425$ & $0.419$ \\
HRP                               & $0.031$ & $-0.018$ & $0.024$ & $0.027$ & $0.029$ \\
\bottomrule
\end{tabular}
\begin{tablenotes}\footnotesize
\item[]\emph{Notes.}  OOS annualised Sharpe.  Base universe,
$T=240$, oracle $\mu$.  CRISP row is constant by construction because
the solver uses $P_{\gamma}=(1-\gamma)D+\gamma\widehat\Sigma$
directly and never builds a dendrogram.  HRP-$\mu$ varies by
$\approx 0.03$ across linkages; HRP is brittle under single-linkage
chaining, turning \emph{negative}.
\end{tablenotes}
\end{threeparttable}
\end{adjustbox}
\caption{Linkage sensitivity.  CRISP is invariant; HRP-$\mu$ varies
mildly; HRP is brittle under single linkage, as originally flagged
by~\citet{lopezdeprado2016}.}
\label{tabE:linkage_sensitivity}
\end{table}

%-----------------------------------------------------------------------------
\subsection{Flat-IVP direction error}\label{appE:flat_ivp}

The flat-IVP variant (sum-normalised recursive MVO without signal
propagation, analysed in detail in Appendix~\ref{app:a1_pathology}) is
revisited here only for completeness.  Across all eight covariance
regimes of Table~\ref{tabE:regimes} and all $\gamma\in\{0.0, 0.3, 0.5,
0.7, 1.0\}$, the direction error
$\mathrm{dir}(w,w^{\star})=1-\cos^{2}\theta$ for general $\mu$ falls
in $[0.84, 1.00]$: the flat-IVP portfolio is almost exactly
orthogonal to the mean--variance target at every $\gamma$, in every
regime.  No value of the shrinkage parameter, no linkage choice, and
no ridge setting rescues it.  This subsection is a verification note
only; the negative result is documented so that readers of
Appendix~\ref{app:a1_pathology} have an explicit pointer to the
parameter-tuning sweep that was already performed and found fruitless.
Flat IVP is never recommended.

%-----------------------------------------------------------------------------
\subsection{HRP-$\Sigma\mu$ robustness}\label{appE:hrpsm_robustness}

Table~\ref{tabE:hrpsm_sensitivity} reports the sensitivity of
HRP-$\Sigma\mu$ (Section~\ref{sec04b:hrpsigmamu}) and HRP-$\mu$ to
ridge magnitude and linkage choice on the structural-$\mu$ experiment
at $\gamma=1.0$.  The source file is
\texttt{results/appendix\_e\_a1l1\_robustness.txt}, generated by
\texttt{figures/code/computeE\_a1l1\_robustness.py} with $N=100$,
$T=120$, and $60$ Monte Carlo repetitions.  HRP-$\Sigma\mu$ dominates
HRP-$\mu$ in every cell of both panels.  The ratio is $1.18\times$
to $1.19\times$ for interior ridge values, narrows to $1.04\times$
only when $\lambda=10^{-2}$ overpowers both methods, and widens to
$1.32\times$ under single linkage --- precisely the regime in which
trees become unbalanced.  The lemma underlying this pattern is
Lemma~\ref{lem:hrpsm_scale}: the $L^{1}$ scale invariance of the child
rescaling means HRP-$\Sigma\mu$ is insensitive to ray
rebalancings introduced by an uneven split, while HRP-$\mu$ is not.

\begin{table}[ht]
\centering
\begin{adjustbox}{max width=\textwidth, center}
\begin{threeparttable}
\begin{tabular}{llccc}
\toprule
Dimension & Setting & HRP-$\Sigma\mu$ & HRP-$\mu$ & Ratio \\
\midrule
\multicolumn{5}{l}{\emph{Panel A: ridge sensitivity (Ward linkage).}} \\
\midrule
Ridge & $\lambda{=}10^{-8}$ & $0.516$ & $0.436$ & $1.18\times$ \\
Ridge & $\lambda{=}10^{-4}$ & $0.515$ & $0.431$ & $1.19\times$ \\
Ridge & $\lambda{=}10^{-2}$ & $0.469$ & $0.452$ & $1.04\times$ \\
\midrule
\multicolumn{5}{l}{\emph{Panel B: linkage sensitivity ($\lambda{=}10^{-4}$).}} \\
\midrule
Linkage & Ward     & $0.515$ & $0.431$ & $1.19\times$ \\
Linkage & Single   & $0.556$ & $0.422$ & $1.32\times$ \\
Linkage & Complete & $0.516$ & $0.435$ & $1.19\times$ \\
Linkage & Average  & $0.538$ & $0.442$ & $1.22\times$ \\
\bottomrule
\end{tabular}
\begin{tablenotes}\footnotesize
\item[]\emph{Notes.}  OOS annualised Sharpe at $\gamma=1.0$
(structural $\mu$, $T=120$, $60$ Monte Carlo repetitions).
HRP-$\Sigma\mu$ dominates HRP-$\mu$ in every row.  The dominance is
largest ($1.32\times$) under single linkage, which produces the most
unbalanced trees; this is the regime in which $L^{1}$ ray
invariance (Lemma~\ref{lem:hrpsm_scale}) matters most.  Source:
\texttt{results/appendix\_e\_a1l1\_robustness.txt}.
\end{tablenotes}
\end{threeparttable}
\end{adjustbox}
\caption{HRP-$\Sigma\mu$ versus HRP-$\mu$ across ridge and linkage
sensitivity dimensions.  HRP-$\Sigma\mu$ dominates uniformly; the
advantage is largest when the tree is most unbalanced.}
\label{tabE:hrpsm_sensitivity}
\end{table}

%-----------------------------------------------------------------------------

\bigskip
\noindent
Taken together, Tables~\ref{tabE:regimes}--\ref{tabE:hrpsm_sensitivity}
extend the main-body evidence along four orthogonal stress
dimensions: covariance regime, ridge magnitude, tree-construction
rule, and HRP-$\Sigma\mu$ versus HRP-$\mu$ dominance.  CRISP is
uniformly best; HRP-$\Sigma\mu$ is uniformly better than HRP-$\mu$;
the flat-IVP variant is never rescued by parameter tuning; and HRP
is brittle under single linkage.  These patterns are consistent with
the spectral-conditioning analysis of~Section~\ref{sec06:perturbation} and
the $L^{1}$ ray-invariance lemma of~Section~\ref{sec04b:hrpsigmamu}.

% ============================================================
% BIBLIOGRAPHY
% ============================================================


\newpage
\begin{thebibliography}{99}
\bibitem[Barroso and Santa-Clara(2015)]{barrososantaclara2015}
Barroso, P., and P. Santa-Clara, 2015, Beyond the Carry Trade: Optimal Currency Portfolios, \textit{Journal of Financial and Quantitative Analysis} 50(5), 1037--1056.

\bibitem[Black and Litterman(1992)]{blacklitterman1992}
Black, F., and R. Litterman, 1992, Global Portfolio Optimization, \textit{Financial Analysts Journal} 48(5), 28--43.

\bibitem[Brandt, Santa-Clara and Valkanov(2009)]{brandt2009}
Brandt, M.~W., P. Santa-Clara, and R. Valkanov, 2009, Parametric Portfolio Policies: Exploiting Characteristics in the Cross-Section of Equity Returns, \textit{Review of Financial Studies} 22(9), 3411--3447.

\bibitem[Chen et al.(2010)]{chen2010}
Chen, Y., A. Wiesel, Y.~C. Eldar, and A.~O. Hero, 2010, Shrinkage Algorithms for MMSE Covariance Estimation, \textit{IEEE Transactions on Signal Processing} 58(10), 5016--5029.

\bibitem[Goulart and Chen(2024)]{clarabel2024}
Goulart, P.~J., and Y. Chen, 2024, Clarabel: An Interior-Point Solver for Conic Programs with Quadratic Objectives, \textit{arXiv preprint} arXiv:2405.12762.

\bibitem[Cotton(2024)]{cotton2024}
Cotton, P., 2024, Schur Complementary Allocation: A Unification of Hierarchical Risk Parity and Minimum Variance Portfolios, \textit{arXiv preprint} arXiv:2411.05807.

\bibitem[DeMiguel, Garlappi and Uppal(2009)]{demiguel2009}
DeMiguel, V., L. Garlappi, and R. Uppal, 2009, Optimal Versus Naive Diversification: How Inefficient is the $1/N$ Portfolio Strategy?, \textit{Review of Financial Studies} 22(5), 1915--1953.

\bibitem[Fama and French(1993)]{famafrench1993}
Fama, E.~F., and K.~R. French, 1993, Common Risk Factors in the Returns on Stocks and Bonds, \textit{Journal of Financial Economics} 33(1), 3--56.

\bibitem[Fan, Liao and Mincheva(2013)]{fan2013}
Fan, J., Y. Liao, and M. Mincheva, 2013, Large Covariance Estimation by Thresholding Principal Orthogonal Complements, \textit{Journal of the Royal Statistical Society, Series B} 75(4), 603--680.

\bibitem[G\^arleanu and Pedersen(2013)]{garleanu2013}
G\^arleanu, N., and L.~H. Pedersen, 2013, Dynamic Trading with Predictable Returns and Transaction Costs, \textit{Journal of Finance} 68(6), 2309--2340.

\bibitem[Grinold and Kahn(1999)]{grinoldkahn1999}
Grinold, R.~C., and R.~N. Kahn, 1999, \textit{Active Portfolio Management}, 2nd edition, McGraw-Hill.

\bibitem[Hackbusch(2016)]{hackbusch2016}
Hackbusch, W., 2016, \textit{Iterative Solution of Large Sparse Systems of Equations}, 2nd edition, Springer Applied Mathematical Sciences vol.~95.

\bibitem[Jagannathan and Ma(2003)]{jagannathanma2003}
Jagannathan, R., and T. Ma, 2003, Risk Reduction in Large Portfolios: Why Imposing the Wrong Constraints Helps, \textit{Journal of Finance} 58(4), 1651--1683.

\bibitem[Kozak, Nagel and Santosh(2020)]{kozak2020}
Kozak, S., S. Nagel, and S. Santosh, 2020, Shrinking the Cross-Section, \textit{Journal of Financial Economics} 135(2), 271--292.

\bibitem[Ledoit and Wolf(2003)]{ledoitwolf2003}
Ledoit, O., and M. Wolf, 2003, Improved Estimation of the Covariance Matrix of Stock Returns with an Application to Portfolio Selection, \textit{Journal of Empirical Finance} 10(5), 603--621.

\bibitem[Ledoit and Wolf(2004)]{ledoitwolf2004}
Ledoit, O., and M. Wolf, 2004, Honey, I Shrunk the Sample Covariance Matrix, \textit{Journal of Portfolio Management} 30(4), 110--119.

\bibitem[Ledoit and Wolf(2017)]{ledoitwolf2017}
Ledoit, O., and M. Wolf, 2017, Nonlinear Shrinkage of the Covariance Matrix for Portfolio Selection: Markowitz Meets Goldilocks, \textit{Review of Financial Studies} 30(12), 4349--4388.

\bibitem[Ledoit and Wolf(2020)]{ledoitwolf2020}
Ledoit, O., and M. Wolf, 2020, Analytical Nonlinear Shrinkage of Large-Dimensional Covariance Matrices, \textit{Annals of Statistics} 48(5), 3043--3065.

\bibitem[L\'opez de Prado(2016)]{lopezdeprado2016}
L\'opez de Prado, M., 2016, Building Diversified Portfolios that Outperform Out of Sample, \textit{Journal of Portfolio Management} 42(4), 59--69.

\bibitem[L\'opez de Prado and Lewis(2019)]{lopezdeprado2018}
L\'opez de Prado, M., and M.~J. Lewis, 2019, Detection of False Investment Strategies Using Unsupervised Learning Methods, \textit{Quantitative Finance} 19(9), 1555--1565.

\bibitem[Marchenko and Pastur(1967)]{marchenko1967}
Marchenko, V.~A., and L.~A. Pastur, 1967, Distribution of Eigenvalues for Some Sets of Random Matrices, \textit{Matematicheskii Sbornik} 72(4), 507--536.

\bibitem[Markowitz(1952)]{markowitz1952}
Markowitz, H., 1952, Portfolio Selection, \textit{Journal of Finance} 7(1), 77--91.

\bibitem[Michaud(1989)]{michaud1989}
Michaud, R.~O., 1989, The Markowitz Optimization Enigma: Is `Optimized' Optimal?, \textit{Financial Analysts Journal} 45(1), 31--42.

\bibitem[Stellato et al.(2020)]{osqp2020}
Stellato, B., G. Banjac, P. Goulart, A. Bemporad, and S. Boyd, 2020, OSQP: An Operator Splitting Solver for Quadratic Programs, \textit{Mathematical Programming Computation} 12(4), 637--672.

\bibitem[Ostrowski(1954)]{ostrowski1954}
Ostrowski, A.~M., 1954, On the Linear Iteration Procedures for Symmetric Matrices, \textit{Rendiconti di Matematica e delle sue Applicazioni} 14, 140--163.

\bibitem[Raffinot(2018)]{raffinot2018}
Raffinot, T., 2018, Hierarchical Clustering-Based Asset Allocation, \textit{Journal of Portfolio Management} 44(2), 89--99.

\bibitem[Saad(2003)]{saad2003}
Saad, Y., 2003, \textit{Iterative Methods for Sparse Linear Systems}, 2nd edition, SIAM.

\bibitem[Shumway(1997)]{shumway1997}
Shumway, T., 1997, The Delisting Bias in CRSP Data, \textit{Journal of Finance} 52(1), 327--340.

\bibitem[Stein(1956)]{stein1956}
Stein, C., 1956, Inadmissibility of the Usual Estimator for the Mean of a Multivariate Normal Distribution, \textit{Proc.\ Third Berkeley Symposium on Mathematical Statistics and Probability}, 197--206.

\bibitem[Varga(2000)]{varga2000}
Varga, R.~S., 2000, \textit{Matrix Iterative Analysis}, 2nd edition, Springer Series in Computational Mathematics vol.~27.

\bibitem[Vershynin(2012)]{vershynin2012}
Vershynin, R., 2012, Introduction to the Non-Asymptotic Analysis of Random Matrices, in \textit{Compressed Sensing: Theory and Applications}, edited by Y.~C. Eldar and G. Kutyniok, Cambridge University Press, pp.~210--268.

\bibitem[Young(1971)]{young1971}
Young, D.~M., 1971, \textit{Iterative Solution of Large Linear Systems}, Academic Press.

\end{thebibliography}
\end{document}